\documentclass[12pt]{amsart}
\usepackage{amssymb,amsmath}
\usepackage{amsfonts}
\usepackage{latexsym}
\usepackage[hyperindex,breaklinks]{hyperref}

\usepackage{eepic,pict2e}

\newtheorem{theorem}{Theorem}[section]
\newtheorem{lemma}[theorem]{Lemma}
\newtheorem{exercise}[theorem]{Exercise}
\newtheorem{proposition}[theorem]{Proposition}
\newtheorem{corollary}[theorem]{Corollary} 
\theoremstyle{definition}  
\newtheorem{definition}[theorem]{Definition}

\newtheorem{example}[theorem]{Example}
  
\newtheorem{remark}[theorem]{Remark}

\newcommand{\Tr}{\text{Tr}}%{\text{Tr}\,} 

\newcommand{\G}{\mathcal{G}}

\newcommand{\g}{\mathfrak{g}}

\newcommand{\la}{\langle\,} 
\newcommand{\ra}{\,\rangle}

\newcommand{\ben}{\begin{enumerate}}
\newcommand{\een}{\end{enumerate}}

\renewcommand{\Bbb}{\mathbb}
\renewcommand{\frak}{\mathfrak}
\renewcommand{\bold}{\mathbf}

\numberwithin{equation}{section}
\makeindex
\begin{document}

\title{Mathematical ideas and notions of quantum field theory}

\author{Pavel Etingof} 

\address{Department of Mathematics, MIT, Cambridge, MA 02139, USA}

\maketitle

\centerline{\bf To Pierre Deligne on his 80th birthday with admiration} 

\tableofcontents

%%%%%%%%%%%%%%%%%%%%%%%%%%%%%%%%%%%%%%%%%%%%%%%%%%%%%%%%%%%%%%%%%%%%%%%%%%%%%%

\section*{Introduction} 

Physics has always been a major source of both motivation and applications for several central fields of mathematics, such as analysis, differential equations and probability. However, the development of quantum field theory and string theory in the
last four decades has taken interactions between these two disciplines to an unprecedented level, incorporating into physics such traditionally ``pure'' areas of mathematics as algebraic topology, category theory, differential and algebraic geometry, representation theory, combinatorics, 
and even number theory. This interaction has been 
highly fruitful in both directions, and led to a necessity for 
physicists to know the basics of modern mathematics and for mathematicians to
know the basics of modern physics. Physicists have been quick to learn, and 
nowadays good physicists often understand relevant areas of mathematics 
as deeply as professional mathematicians. On the other hand, 
many mathematicians have been dragging their feet, deterred by lack of
rigor in physical texts, and, more importantly, by a 
different manner of presentation. In particular, even the basic
setting of quantum field theory, necessary for
understanding its more advanced (and more mathematically exciting)
parts, is already largely unfamiliar to mathematicians. 
Nevertheless, many of the basic ideas of quantum field theory 
can in fact be presented in a rigorous and mathematically natural way. 
Doing so is the main goal of this text.  

Namely, these are slightly expanded lecture notes for a graduate course on basic mathematical structures of quantum field theory that I gave at the MIT Mathematics Department in 2002 and then again in 2023. The reader should not hope to learn quantum field theory from this text - this is impossible, for instance, because I know less about this subject than a beginning physics graduate student. Rather, as mentioned above, its aim is to present the basic setup of quantum field theory in a mathematically motivated manner, highlighting its connections with various fields of mathematics. As such, it could serve to prepare the reader for more advanced texts in this genre, such as \cite{QFS}, or for reading a regular QFT textbook or lecture notes, such as \cite{Co,W,IZ,PS} from a mathematician's viewpoint. Note that a lot of important material is contained in exercises, which I strongly recommend the reader to solve while reading the text. 

We begin with a general discussion of classical and quantum mechanics and field theory (Chapter 1). Then we proceed to prove the steepest descent and stationary phase formulas in classical asymptotic analysis, which serve as a finite dimensional model for perturbative computations with path integrals (Chapter 2). Then, in Chapter 3, we develop Feynman calculus, the main combinatorial tool in perturbative quantum field theory. To illustrate Feynman calculus, we give a number of ins applications to enumerative combinatorics (the matrix-tree theorem and its specializations). 

In Chapter 4, we extend Feynman calculus to matrix integrals, and show that for such integrals Feynman graphs are replaced by fat graphs (surfaces), so that 
the coefficients of the asymptotic expansion in $1/N$ (where $N$ is the matrix size) 
are sums over fat graphs of a given genus. This allows us, in Chapter 5, to prove 
Harer-Zagier's theorem on the Euler characteristic of the moduli space of curves, 
and in Chapter 6 to obtain non-trivial counts of planar graphs.  

All this material is, however, about quantum field theory in $0$ spacetime dimensions, or, as one may jokingly say, in $-1$-dimensional space. To connect to real physics, 
we must go up at least one dimension, i.e., consider quantum field theory 
in $0+1$ spacetime (or $0$ space) dimensions, which is quantum mechanics (Chapter 7). We begin with a review of Lagrangian formalism of classical mechanics (Lagrangians, least action principle), and then proceed to quantize this formalism, developing the path integral approach to quantum mechanics. Namely, we describe perturbative expansion of quantum-mechanical path integrals using Feynman diagrams and give several examples. We also explain that quantum mechanical path integrals are related to (rigorously defined) Wiener integrals in the theory of stochastic processes by the Wick rotation of the time, $t\mapsto it$. 

In Chapter 8, after reviewing Hamiltonian formalism in classical mechanics, 
we describe its quantization, which gives a rigorous basis for non-perturbative quantum mechanics. We also prove the Feynman-Kac formula which relates 
the correlation functions obtained in the Lagrangian and Hamiltonian approaches. 

In Chapters 9 and 10, we discuss the super-generalization of the material of the previous chapters, i.e., describe classical and quantum mechanics for fermions. We begin with a review of supergeometry and Berezin's integration theory on supermanifolds and proceed to extend Feynman calculus to the super-case. Then we discuss classical and quantum mechanics for fermions. 

In Chapter 11 we finally get to the actual quantum field theory, in $d+1$ spacetime dimension with $d\ge 1$. We start with reviewing Lagrangian classical field theory and then pass to its quantization, in particular describing the theories of free bosons and fermions. We review the classical theory of spinors (in particular, real Bott periodicity for them, modulo 8) and use it to describe 
the possible kinetic terms and mass terms in fermionic lagrangians.  Then we turn to hamiltonian formalism in both classical and quantum field theory and discuss Wightman axioms. We conclude with describing the quantum theory of a free scalar boson from this point of view. 

In Chapter 12, we describe the basics of the perturbative renomalization theory. 
In particular, we discuss ultraviolet divergences of Feynman amplitudes and regularization of such divergences by introducing counterterms in the Lagrangian depending on the cutoff $\Lambda$ in the momentum space. We define 
super-renormalizable, renormalizable and non-renormalizable theories, critical dimensions for various theories and terms in the Lagrangian, and discuss the key examples. 

Finally, in Chapter 13 we give a brief introduction to 2-dimensional conformal field theory. After a review of classical field theory of a massless scalar in 1+1 dimensions, we quantize it and construct its Hilbert space from the Fock representation of the infinite dimensional Heisenberg Lie algebra. We show that the partition function of this theory (normalized using the zeta function regularization) is modular invariant, reflecting its conformal symmetry. Then we show that the Hilbert space of the theory carries two commuting projective actions of the Lie algebra $W$ of polynomial vector fields on $\Bbb C^\times$, which expresses the infinitesimal conformal symmetry. We explain that this action is truly projective, i.e., both copies of $W$ are replaced by its non-trivial central extension - the Virasoro algebra (conformal anomaly). Then we discuss a circle-valued version of this theory, vertex operators and $T$-duality. We also briefly discuss the quantum theory of a free fermion in 1+1 dimensions and the Wess-Zumino-Witten model. 

{\bf Acknowledgements.} I am very grateful to my coauthors of \cite{QFS}; this text would definitely not have appeared had we not collaborated on this project almost 30 years ago. In particular, I'd like to thank David Kazhdan, who prompted me to study the basics of quantum field theory, Edward Witten, from whom I learned almost everything I know about it, Dan Freed, without whose careful notes and explanations this would have been impossible, and Pierre Deligne, who infused and greatly facilitated our learning with deep mathematical insights, clarity and elegance. I am also indebted to the participants of the MIT courses in 2002 and 2023 which gave rise to this text. This work was partially supported by the NSF grant DMS-2001318.

\section{Generalities on quantum field theory}

\subsection{Classical mechanics} In classical mechanics, we study the motion 
of a particle (say, of mass $1$) in a Euclidean space $V$ . This motion is described by a function of one variable, $q=q(t)\in V$, representing the position 
of the particle at a time $t$. This function must satisfy the {\it Newton equation 
of motion}, \index{Newton equation 
of motion}
$$
\ddot{q}=-U'(q),
$$ 
where $U$ is the potential energy.  

Another way 
to express this law of motion is to say that $q(t)$ must be 
a solution of a certain variational problem. Namely, one
introduces the {\it Lagrangian} \index{Lagrangian}
$$
{\mathcal L}(q):=\frac{{\dot{q}}^2}{2}-U(q)
$$
(the difference of kinetic and potential energy), 
and the {\it action} functional \index{action}
$$
S(q):=\int_{a}^{b}{\mathcal L}(q)dt
$$
(for some fixed $a<b$). 
Then the law of motion can be expressed as the {\it least action
  principle}\index{least action principle}: $q(t)$ must be a critical point of $S$ 
on the space of all functions with given
$q(a)$ and $q(b)$, i.e., the Newton equation is the Euler-Lagrange
equation for a solution of the variational problem defined by
$S$. 
Indeed, using integration by parts, for $\varepsilon \in C^1[a,b]$ with 
$\varepsilon(a)=\varepsilon(b)=0$ we have 
$$
\tfrac{d}{ds}|_{s=0}\int_a^b{\mathcal L}(q+s\varepsilon)dt=
\int_a^b(\tfrac{\partial \mathcal L}{\partial q}\varepsilon+
\tfrac{\partial \mathcal L}{\partial \dot{q}}\dot{\varepsilon})dt=
$$
$$
\int_a^b(-U'(q)\varepsilon+\dot{q}\dot{\varepsilon})dt=
-\int_a^b(U'(q)+\ddot{q})\varepsilon dt,
$$
and this vanishes for all $\varepsilon$ iff 
$q$ satisfies the Newton equation $\ddot{q}=-U'(q)$. 

\begin{remark} The name ``least action principle'' comes from the fact that 
in some cases (for example when $U''\le 0$) 
the action is not only extremized 
but also minimized at the solution $q(t)$.
In general, however, this is not the case, and 
the trajectory of the particle may be not a (local) minimum, 
but only a critical point of the action. 
Therefore, the law of motion is better formulated 
as the ``extremal (or stationary) action principle''; this is the way we will think 
of it in the future.
\end{remark}

\begin{exercise}\label{unii} (i) Consider the motion of a particle in a Euclidean space $V$.  
Show that if the potential is concave ($U''(q)\le 0$)
then for any $\bold a,\bold b\in V$ and $a<b\in\Bbb R$ there exists at most one solution 
of the Newton equation with $q(a)=\bold a$ and $q(b)=\bold b$, and it is the strict global minimum for the action with these boundary conditions (if exists). 

(ii) Show that the conclusion of (i) holds if $U''(q)<\frac{\pi^2}{(b-a)^2}$
(prove and use Wirtinger's inequality: if $\varepsilon\in C^1[a,b]$ and $\varepsilon(a)=\varepsilon(b)=0$ then $\int_a^b \varepsilon'(t)^2dt\ge \frac{\pi^2}{(b-a)^2}\int_a^b \varepsilon(t)^2dt$).  

(iii) Compute the unique solution in (i) if $U(q)=-\frac{1}{2}B(q,q)$, where 
$B$ is a nonnegative definite symmetric bilinear form on $V$. 

(iv) Show that the statements of (i) fail for $\dim V=1$, $U(q)=\frac{1}{2}q^2$ and $b-a\ge \pi$.

(v) Let $\dim V=1$ and $U$ be a smooth potential on $\Bbb R$. 
Suppose that $\limsup_{|x|\to \infty}\frac{U(x)}{x^2}\le 0$. 
Show that a solution in (i) (possibly non-unique) exists for any $a,b,\bold a,\bold b$.  
Give an example of a smooth potential $U$ for which a solution in (i) 
does not always exist.\footnote{One can show 
using calculus of variations that for any $\dim V$, if $U(q)\le 0$ for all $q$ then the solution always exists.}
\end{exercise} 

\begin{remark} 
Physicists often consider solutions of Newton's equation on the whole line
rather than on a fixed interval $[a,b]$. In this case, 
the naive definition of an extremal 
does not make sense,
since the action integral 
$S(q)=\int_{\Bbb R}{\mathcal L}(q)dt$ is improper and in general 
diverges. Instead, one makes the following ``corrected'' definition: 
a function $q(t)$ 
on $\Bbb R$ is an extremal of $S$ if the expression 
$$
\tfrac{d}{ds}|_{s=0}\int_{\Bbb R}{\mathcal L}(q+s\varepsilon)dt:=
\int_{\Bbb R}(\tfrac{\partial \mathcal L}{\partial q}\varepsilon+
\tfrac{\partial \mathcal L}{\partial \dot{q}}\dot{\varepsilon})dt,
$$
where $\varepsilon(t)$ is any compactly supported perturbation,
is identically zero. 
With this definition, the extremals are exactly the solutions of
Newton's equation (which, as before, is easily seen by integration by parts).  
\end{remark}

\begin{remark} Note that this formalism also describes the motion of a system of $n$ particles, if we combine the vectors representing their positions in a Euclidean space $V$ into a single vector in $V^{n}$.  More generally, we may consider a particle moving on a Riemannian manifold $M$. In this case $q(t)$ is a path on $M$, and the motion is described by the same equation, where $\ddot{q}$ means the covariant derivative $\nabla_{\dot q}\dot q$ of $\dot{q}$ with respect to the Levi-Civita connection. For example, if $U=0$, this is the {\it geodesic flow}\index{geodesic flow}, whose trajectories are the geodesics on $M$. The same applies to a system of $n$ particles on $M$, in which case $q(t)$ is a path on the configuration space $M^n$.
Finally, a similar analysis applies to more general Lagrangians, which are 
arbitrary smooth functions of (finitely many) derivatives of $q$. 
\end{remark} 

\subsection{Classical field theory} In classical field theory, the situation is
similar, but with infinitely many particles. Namely, in this case we should think not of a single particle or a finite system of particles, but
rather of a ``continuum of particles'' (e.g. a string, a membrane, a jet
of fluid); so in a $d+1$-dimensional classical field theory 
the motion is described by a {\it classical field}\index{classical field} -- a 
(vector-valued) function 
$\phi(x,t)$ depending on both space and time coordinates
($x\in \Bbb R^d$, $t\in \Bbb R$). 
Consequently, the equation of motion is a partial differential
equation. For example, for a string or a membrane  
the equation of motion is the {\it wave equation}\index{wave equation}
$\square \phi=0$, where $\square$ is the {\it D'Alembertian}\index{D'Alembertian} 
$\partial_t^2-v^{2}\Delta$ (here $\Delta$ is the Laplacian with
respect to the space coordinates, and $v$ the velocity of wave propagation, e.g. for the string $v^2$ is proportional to the string tension). 

As in classical mechanics, in classical field theory 
there is a Lagrangian ${\mathcal L}(\phi)$ (a smooth function of finitely many partial derivatives of $\phi$),
whose integral 
$$
S(\phi)=\int_{D}{\mathcal L}(\phi)dxdt
$$ 
over a compact region $D$ in the {\it spacetime}\index{spacetime} $\Bbb R^{d+1}$ is called the {\it action}.\index{action}
The law of motion can be expressed as the condition that 
the action must be extremized over any such region $D$ with fixed boundary conditions; 
so the equations of motion (also called 
the {\it field equations}\index{field equations}) are the Euler-Lagrange equations
for this variational problem. 
For example, in the case of string or membrane, 
the Lagrangian is 
$$
{\mathcal L}(\phi)=\tfrac{1}{2}(\phi_t^2-v^{2}(\nabla \phi)^2). 
$$

\begin{remark} Like in mechanics, in field theory solutions of 
the equations of motion on the whole space-time (rather than a compact
region $D$) are extremals of the action in the sense that 
$$
\tfrac{d}{ds}|_{s=0}\int_{\Bbb R^{d+1}}{\mathcal L}(u+s\varepsilon)dxdt=0,
$$
where $\varepsilon$ is a compactly supported perturbation. 
\end{remark} 

\subsection{Brownian motion} One of the main differences between classical and quantum
mechanics is, roughly speaking, that quantum particles do not have to 
obey the classical equations of motion, but can randomly deviate 
from their classical trajectories. 
Therefore, given the position and velocity of the
particle at a given time, we cannot determine its position at a
later time, but can only determine the density of probability 
that at this later time the particle will be found at a given
point. In this sense quantum particles are similar
to random (Brownian) particles. Brownian particles are a bit 
easier to understand conceptually, so let us begin with them.

The motion of a Brownian particle in $\Bbb R^k$ 
in a potential field $$U:\Bbb R^k\to \Bbb R$$ 
is described by a stochastic process 
$q=q(t)$, $q=(q_1, \ldots ,q_k)\in \Bbb R^k$. That is, 
for each real $t$ we have 
a random variable $q(t)\in \Bbb R^k$ (the position of the particle at a time $t$), 
such that the dependence of $t$ is 
regular in some sense. Namely, for $\bold a,\bold b\in \Bbb R^k$ the random dynamics of the particle conditioned to have $q(a)=\bold a$, $q(b)=\bold b$ is ``defined'' as follows:\footnote{We put the word ``defined'' in quotation marks 
because this definition is obviously heuristic and not rigorous;
see below for more explanations.} 
if $y:[a,b]\to \Bbb R^k$ 
is a continuously differentiable function with $y(a)=\bold a,y(b)=\bold b$, then the density of 
probability that $q(t)=y(t)$ for $t\in [a,b]$ is 
proportional to $e^{-S(y)/\kappa}$, where 
$$
S(y):=\int_a^b
(\tfrac{1}{2}{y'}^2+U(y))dt
$$ 
is the action 
and $\kappa>0$ is the diffusion coefficient. 
Thus, the likeliest 
$q(t)$ is the one that 
minimizes $S$ (in particular, solves the classical 
equations of motion $\ddot{q}=U'(q)$), while the likelihood 
of the other paths decays exponentially with the deviation 
of the action of these paths from the minimal possible.  

\begin{remark} 1. This discussion thus assumes that the extremum of $S$ at $q$ is 
actually a minimum, which we know is not always the case, but is so when $U$ is convex, i.e., $U''(q)\ge 0$ for all $q$ (see Exercise \ref{unii}). 

2. The reader must have noticed that compared to the discussion of classical mechanics, the sign in front of the potential $U$ has been changed to the opposite one. This is not a misprint! It has to do with the fundamental fact discussed below that statistical mechanics is related to usual (quantum) mechanics by the {\it Wick rotation} $t\mapsto it$, where $i=\sqrt{-1}$. 
In particular, this means that Brownian motion is well defined 
in the physically important case of convex potential, such as the 
multidimensional {\it harmonic oscillator}\index{harmonic oscillator} potential $\frac{1}{2}B(q,q)$ where $B$ is a positive definite bilinear form. 
\end{remark} 

All the information we can hope to get 
about the stochastic process $q(t)$ is contained in the {\it correlation functions}\index{correlation functions}
$$
\langle q_{j_1}(t_1) \ldots q_{j_n}(t_n)\rangle,
$$ 
which by definition are the expectation values 
of the products of random variables 
$q_{j_1}(t_1), \ldots, q_{j_n}(t_n)$, 
(more specifically, by Kolmogorov's theorem
the stochastic process $q(t)$ is completely determined by 
these functions). So such functions should be regarded as 
the output, or answer, of the theory of the Brownian particle. 

Thus the main question is how to compute the correlation functions. 
Physicists write down the following ``answer" motivated by the above definition: 
given points $t_1, \ldots ,t_n\in [a,b]$,
\begin{equation}\label{brmo}
\langle q_{j_1}(t_1) \ldots q_{j_n}(t_n)\rangle=
\int_{P_{\bold a,\bold b}} q_{j_1}(t_1) \ldots q_{j_n}(t_n)e^{-\frac{S(q)}{\kappa}}Dq, 
\end{equation}
where integration is carried out
over the space $P_{\bold a,\bold b}$ of paths 
$$
q: [a,b]\to \Bbb R^n,\ q(a)=\bold a,\ q(b)=\bold b,
$$ 
and 
$Dq$ is a Lebesgue measure on this space 
such that 
$$
\int_{P_{\bold a,\bold b}} e^{-\frac{S(q)}{\kappa}}Dq=1.
$$ 
Alternatively, when they do not want to normalize the Lebesgue measure, they write 
\begin{equation}\label{brmo1}
\langle q_{j_1}(t_1) \ldots q_{j_n}(t_n)\rangle=
\frac{1}{Z}\int_{P_{\bold a,\bold b}} q_{j_1}(t_1) \ldots q_{j_n}(t_n)e^{-\frac{S(q)}{\kappa}}Dq, 
\end{equation}
where 
$$
Z:=\int_{P_{\bold a,\bold b}} e^{-\frac{S(q)}{\kappa}}Dq
$$
is the {\it partition function}.\index{partition function}  Such an integral is called 
a {\it path integral},\index{path integral} since it is an integral over the space of
paths. 

It is clear, however, that such definition and answer are a priori 
not satisfactory from the mathematical viewpoint, since the infinite 
dimensional integration requires justification.  
In the case of Brownian motion, such a justification is actually possible within the framework 
of the Lebesgue measure theory, and the corresponding integration 
theory is called the theory of {\it Wiener integral}.\index{Wiener integral}
(To be more precise, one cannot define 
the measure $Dq$, but one can define the measure 
$e^{-\frac{S(q)}{\kappa}}Dq$ for sufficiently nice potentials $U(q)$). 

\begin{remark}\label{classlim} As $\kappa\to 0$, the non-optimal trajectories become increasingly less likely 
relatively to the optimal one, so in the limit we recover the deterministic system: 
$$
\langle q_{j_1}(t_1) \ldots q_{j_n}(t_n)\rangle\to \bold q_{j_1}(t_1) \ldots \bold q_{j_n}(t_n), 
$$
where $\bold q(t)$ is the classical trajectory with $\bold q(a)=\bold a,\bold q(b)=\bold b$ (note that if $U\ge 0$ then this trajectory is unique by Exercise \ref{unii}). 
\end{remark} 

\subsection{Quantum mechanics} Now let us turn to a quantum particle. 
Quantum mechanics is notoriously difficult to visualize,
and the randomness of the behavior of a quantum particle
is less intuitive and more subtle than that of a Brownian particle;
nevertheless, it was pointed out by Feynman that the behavior of a quantum 
particle in a potential field $U(q)$ is correctly described by the same model,
with the real positive parameter $\kappa$ replaced by 
the imaginary number $-i\hbar$ where $\hbar>0$ is the {\it Planck constant},\index{Planck constant} and the time $t$ is replaced by $it$. In other words, the dynamics
of a quantum particle can be expressed (we will discuss later how) via 
the {\it correlation functions}\index{correlation functions}
\begin{equation}\label{quantmech}
\langle q_{j_1}(t_1) \dots q_{j_n}(t_n)\rangle=
\int_{P_{\bold a,\bold b}} q_{j_1}(t_1) \ldots q_{j_n}(t_n)e^{\frac{iS(q)}{\hbar}}Dq, 
\end{equation}
where 
$Dq$ is normalized 
so that 
\begin{equation}\label{nor}
\int_{P_{\bold a,\bold b}} e^{\frac{iS(q)}{\hbar}}Dq=1,
\end{equation} 
and 
$S(q)$ is now given by the same formula 
as in classical mechanics (and differing by sign 
from Brownian motion):
$$
S(q)=\int_a^b (\tfrac{{\dot{q}}^2}{2}-U(q))dt.
$$

As before, we have to make sense of this path integral, 
and now the theory of Wiener integrals unfortunately does not work any more: for instance, the absolute value of the integrand in \eqref{nor} does not decay as the path $q(t)$ deviates from the classical trajectory (in fact, it identically equals to $1$!). 
So we will be able to make sense of \eqref{quantmech} only partially, and an effective mathematically rigorous approach to quantum mechanics is, in fact, based on different techniques (Hamiltonian formalism); this is discussed in more detail below. Still, formula \eqref{quantmech} is extremely helpful for motivational purposes and with appropriate care can be used for computation.

\begin{remark} Similarly to Brownian motion (cf. Remark \ref{classlim}), in the limit $\hbar\to 0$ we are supposed to recover the classical system: 
$$
\langle q_{j_1}(t_1) \ldots q_{j_n}(t_n)\rangle\to \bold q_{j_1}(t_1) \ldots \bold q_{j_n}(t_n). 
$$
However, now this is achieved not because individual non-optimal trajectories become less likely, but rather due to cancellation in the oscillatory integral \eqref{quantmech} which corresponds to the physical phenomenon of {\it quantum interference}.\index{quantum interference} We will observe how this cancellation 
occurs in finite-dimensional oscillatory integrals when we discuss the stationary phase formula below.  
\end{remark} 

\subsection{Quantum field theory}
The situation is the same in field theory, but with infinitely many particles. Namely, a useful 
theory of quantum fields 
(used in the study of interactions of elementary particles) 
is obtained when one considers
correlation functions 
\begin{equation}\label{quanft}
\langle\phi_{j_1}(x_1,t_1) \ldots \phi_{j_n}(x_n,t_n)\rangle= 
\int \phi_{j_1}(x_1,t_1) \ldots \phi_{j_n}(x_n,t_n)e^{\frac{iS(\phi)}{\hbar}}D\phi, 
\end{equation}
where 
$D\phi$ is normalized 
so that $\int e^{\frac{iS(\phi)}{\hbar}}D\phi=1$.

Of course, from the mathematical point of view, this setting is 
a priori even less satisfactory than the one for a quantum particle, since it involves integration with respect to the complex-valued measure $e^{\frac{iS(q)}{\hbar}}Dq$ on functions of $\ge 2$ variables which nobody knows how to define in general (even after the Wick rotation). Nevertheless, physicists imagine that certain integrals of this type exist and come to correct and interesting conclusions (both physical and mathematical). Therefore, making sense of such integrals 
is an interesting problem for mathematicians, and will be one of our main 
occupations.\footnote{To be more precise, we
will make sense of path integrals as power series in 
$\hbar$.} 

\section{The steepest descent and stationary phase formulas}

Now, let us forget for a moment that the integrals 
(\ref{brmo},\ref{quantmech},\ref{quanft}) are infinite dimensional and hence 
problematic to define, and ask ourselves the following question:
why should we expect to recover the usual classical mechanics or field theory when the parameter $\kappa$ or $\hbar$ goes to zero? 
The answer is that this expectation is based on the 
{\it steepest descent} (respectively, {\it stationary phase})\index{steepest descent principle}\index{stationary phase principle} principle
from classical analysis: if $f(x)$ is a function on $\Bbb R^d$ then 
the integrals $\int g(x)e^{-\frac{f(x)}{\kappa}}dx$, 
$\int g(x)e^{\frac{if(x)}{\hbar}}dx$ ``localize'' to minima, respectively 
critical points, of the function $f$. As this classical fact is of central 
importance to the whole subject, let us now discuss it in some detail. 

\subsection{Gaussian integrals} 

We start with auxiliary facts from linear algebra and analysis. Let $V$ be a real vector space of dimension $d$. Let $\bold M(V)$ be the set of non-degenerate complex-valued symmetric bilinear forms on $V$ with non-negative definite real part. We have an open dense subset $\bold M^\circ(V)\subset \bold M(V)$ of forms with positive definite real part. If $B=P+iQ\in \bold M^\circ(V)$ where $P,Q$ are the real and imaginary parts of $B$, then $P^{-1}Q: V\to V$ is a self-adjoint operator with respect to $P$, which therefore has real eigenvalues and diagonalizes in an orthonormal basis. In this basis $B(x,y)=\sum_{j=1}^d a_jx_jy_j$ where ${\rm Re}(a_j)=1$. 
Thus $B^{-1}\in \bold M^\circ(V^*)$. It follows that the map $B\mapsto B^{-1}$ is a homeomorphism $\bold M(V)\cong \bold M(V^*)$ which restricts to a homeomorphism $\bold M^\circ (V)\cong \bold M^\circ(V^*)$. 
 
Now fix a translation-invariant volume form $dx$ on $V$. Then for every complex-valued symmetric bilinear form $B$ on $V$ we can define its determinant $\det B$. Thus we can define a continuous function $(\det B)^{-\frac{1}{2}}$ on $\bold M(V)$ using the branch of the square root which is positive on positive definite forms (it exists and is unique because $\bold M(V)$ is star-like with respect to any point of $\bold M^\circ(V)$, hence simply connected). Note that 
if $B=iQ$ where $Q$ is a real non-degenerate form 
then $(\det B)^{-\frac{1}{2}}=e^{\frac{\pi i\sigma(Q)}{4}}|\det Q|^{-\frac{1}{2}}$, where $\sigma$ is the signature of $Q$. Indeed, it suffices to check the statement for diagonal forms, hence for $d=1$, in which case it is straighforward. 

Let $\mathcal S(V)$ be the {\it Schwartz space}\index{Schwartz space} of $V$, i,.e., the space of smooth functions on $V$ whose all derivatives are rapidly decaying at $\infty$ (faster than any power of $|x|$). In other words, $\mathcal S(V)$ is the space of smooth functions $f$ on $V$ such that 
$D(V)f\subset L^2(V)$, where $D(V)$ is the algebra of differential operators on $V$ with polynomial coefficients. 
The Schwartz space has a natural Fr\'echet topology defined by the seminorms 
$||Df||_{L^2}$, $D\in D(V)$. The topological dual space $\mathcal S'(V)$ is the space of {\it tempered distributions}\index{tempered distribution} on $V$. 
Note that we have natural inclusions $\mathcal S(V)\subset L^2(V)\subset \mathcal S'(V)$. Recall that the Fourier transform 
is the operator 
$$
\mathcal F: \mathcal S(V)\to \mathcal S(V^*)
$$
given by 
$$
\mathcal F(g)(p):=(2\pi)^{-\frac{d}{2}}\int_V g(x)e^{-i(p,x)}dx,
$$
which defines an isometry $L^2(V)\to L^2(V^*)$ such that 
$(\mathcal F^2g)(x)=g(-x)$. By duality, it defines an operator 
$$
\mathcal F: \mathcal S'(V)\to \mathcal S'(V^*) 
$$
which extends $\mathcal F$. 
For any complex symmetric bilinear form 
$B$ with ${\rm Re}B\ge 0$ the function $e^{-\frac{1}{2}B(x,x)}$ belongs 
to $\mathcal S'(V)$, and moreover to $\mathcal S(V)$ iff $B\in \bold M^\circ(V)$.  
Furthermore, it depends continuously on $B$ as an element of these spaces. 
We will call it the {\it complex Gaussian distribution}.\index{complex Gaussian distribution} 

\begin{lemma}\label{gi} (Gaussian integral) For any $B\in \bold M(V)$ we have 
$$
\mathcal F(e^{-\frac{1}{2}B(x,x)})=(\det B)^{-\frac{1}{2}} e^{-\frac{1}{2}B^{-1}(p,p)}.
$$
\end{lemma} 

\begin{proof} By continuity, it suffices to prove this when ${\rm Re}B>0$.
In this case $B$ is diagonalizable, so the statement reduces to the case $d=1$. 
In this case we have to show that for every $a\in \Bbb C$ with ${\rm Re}a>0$,  
$$
\frac{1}{\sqrt{2\pi}}\int_{-\infty}^\infty e^{-ipx-\frac{1}{2}ax^2}dx=
\frac{1}{\sqrt{a}}e^{-\frac{1}{2a}p^2}.
$$
Since both sides are holomorphic in $a$, it is enough to check the statement when $a$ is real. The integral in question can be written as 
$$
\frac{e^{-\frac{1}{2}a^{-1}p^2}}{\sqrt{2\pi}}\int_{-\infty}^\infty e^{-\frac{1}{2}a(x+ia^{-1}p)^2}dx.
$$
But using Cauchy's theorem, 
$$
\frac{1}{\sqrt{2\pi}}\int_{-\infty}^\infty e^{-\frac{1}{2}a(x+ia^{-1}p)^2}dx=
\frac{1}{\sqrt{2\pi}}\int_{\Bbb R+ia^{-1}p} e^{-\frac{1}{2}ax^2}dx=
\frac{1}{\sqrt{2\pi}}\int_{\Bbb R} e^{-\frac{1}{2}ax^2}dx.
$$
Thus the result follows from the Poisson integral
$$
\int_{-\infty}^\infty e^{-x^2}dx=\sqrt{\pi}.
$$
by rescaling $x$. 
\end{proof} 

In the sense of Lemma \ref{gi} we can say, setting $p=0$, that 
\begin{equation}\label{gaui}
(2\pi)^{-\frac{d}{2}}\int_V e^{-\frac{1}{2}B(x,x)}dx=(\det B)^{-\frac{1}{2}}.
\end{equation} 
Note that this equality is also true in the sense of absolute convergence 
if $B\in \bold M^\circ(V)$ and conditional convergence otherwise 
(check it!). 

\subsection{Gaussian integrals with insertions} \label{inse}
Now let $g\in \mathcal S(V)$. Consider the integral 
$$
I_g(\hbar):=\int_{V}g(\hbar^{\frac{1}{2}}x)e^{-\frac{1}{2}B(x,x)}dx,\ \hbar\ge 0,
$$
where for $\hbar=0$ we use \eqref{gaui}, so 
\begin{equation}\label{valzer}
I_g(0)=(2\pi)^{\frac{d}{2}}(\det B)^{-\frac{1}{2}}g(0). 
\end{equation}

Let $\Delta_B: \mathcal S(V)\to \mathcal S(V)$ 
be the Laplace operator corresponding to $B$: 
$\Delta_B=\sum_{j=1}^d \partial_{B^{-1}e_j^*}\partial_{e_j}$ 
for a basis $\lbrace e_i\rbrace$ of $V$. 

\begin{theorem} \label{contt2}
We have 
$$
I_g'(\hbar)=I_{\tfrac{1}{2}\Delta_B g}(\hbar),\ \hbar\ge 0.
$$
Thus $I_g\in C^\infty[0,\infty)$. In particular, if 
$g$ vanishes at the origin to order $2n+1$ then $I_g(0)=...=I_g^{(n)}(0)=0$.
\end{theorem}

The rest of the subsection is occupied by the proof of Theorem \ref{contt2}.

\begin{lemma}\label{contt} $I_g$ is a continuous function. 
\end{lemma} 

\begin{proof} Only continuity at $\hbar=0$ requires proof. By Plancherel's theorem and Lemma \ref{gi}, 
$$
I_g(\hbar)=(g(\hbar^{\frac{1}{2}} x),e^{-\frac{1}{2}B(x,x)})=
$$
$$
\hbar^{-\frac{d}{2}}(\det B)^{-\frac{1}{2}}(\widehat g(\hbar^{-\frac{1}{2}}p), e^{-\frac{1}{2}B^{-1}(p,p)})=(\det B)^{-\frac{1}{2}}(\widehat g(p), e^{-\frac{\hbar}{2}B^{-1}(p,p)}),
$$
where $\widehat g$ is the Fourier transform of $g$. But $e^{-\frac{\hbar}{2}B^{-1}(p,p)}\to 1$ in $\mathcal S'(V^*)$ as $\hbar\to 0$ (as the complex Gaussian distribution depends continuously of the bilinear form). Thus 
$$
\lim_{\hbar\to 0}I_g(\hbar)=(\det B)^{-\frac{1}{2}}(\widehat g(p),1)=
(2\pi)^{\frac{d}{2}}(\det B)^{-\frac{1}{2}}g(0)=I_g(0),
$$
as desired. 
\end{proof} 

\begin{lemma}\label{contt1} If $\ell\in V^*$ and $f\in \mathcal S(V)$ then 
$$
I_{\ell f}(\hbar)=\hbar I_{\partial_{B^{-1}\ell}f}(\hbar).
$$
\end{lemma}

\begin{proof}  
We have 
$$
I_{\ell f}(\hbar)=
\hbar^{\frac{1}{2}}(\ell(x)f(\hbar^{\frac{1}{2}}x),e^{-\frac{1}{2}B(x,x)})=
\hbar^{\frac{1}{2}}(f(\hbar^{\frac{1}{2}}x),\ell(x)e^{-\frac{1}{2}B(x,x)})=
$$
$$
-\hbar^{\frac{1}{2}}(f(\hbar^{\frac{1}{2}}x),\partial_{B^{-1}\ell}e^{-\frac{1}{2}B(x,x)})=
\hbar^{\frac{1}{2}}(\partial_{B^{-1}\ell}f(\hbar^{\frac{1}{2}}x),e^{-\frac{1}{2}B(x,x)})=
$$
$$
\hbar ((\partial_{B^{-1}\ell}f)(\hbar^{\frac{1}{2}}x),e^{-\frac{1}{2}B(x,x)})=\hbar I_{\partial_{B^{-1}\ell}f}(\hbar).
$$
This proves the lemma. 
\end{proof} 

Now we prove Theorem \ref{contt2}.
 If $\hbar>0$ then by direct differentiation we get  
$$
I_g'(\hbar)=\tfrac{1}{2}\hbar^{-1}I_{Eg}(\hbar),
$$
where $E:=\sum_{j=1}^d e_j^*\partial_{e_j}$ is the Euler vector field on $V$. 
Thus by Lemma \ref{contt1} we have 
\begin{equation}\label{diffo}
I_g'(\hbar)=I_{\tfrac{1}{2}\Delta_B g}(\hbar),\ \hbar>0.
\end{equation} 
So, using Lemma \ref{contt}, 
it suffices to show that $I_g\in C^1[0,\infty)$ (then smoothness will follow 
by repeated application of \eqref{diffo}).  
To this end, note that if $C$ is a positive definite form on $V$ then 
$$
I_{e^{-\frac{1}{2}C(x,x)}}(\hbar)=\int_{V}e^{-\frac{1}{2}(B+\hbar C)(x,x)}dx=
(2\pi)^{\frac{d}{2}}\det (B+\hbar C)^{-\frac{1}{2}},
$$
which is analytic, hence continuously differentiable on $[0,\infty)$. 
So subtracting from $g$ a multiple 
of such function, it suffices to prove that $I_g\in C^1[0,\infty)$ when $g(0)=0$. 
In this case $g$ is well known to be a linear combination of 
functions of the form $\ell f$ where $f\in \mathcal S(V)$ 
and $\ell\in V^*$. So it suffices to check that $I_g\in C^1[0,\infty)$ for $g=\ell f$. 
But then by Lemma \ref{contt1} 
$I_g'(0)=I_{\partial_{B^{-1}\ell}f}(0)=I_{\frac{1}{2}\Delta_Bg}(0)$,
as 
$$
\tfrac{1}{2}\Delta_Bg(0)=\tfrac{1}{2}\Delta_B(\ell f)(0)=\sum_j \ell(e_j)
\partial_{B^{-1}e_j^*}f(0)=\partial_{B^{-1}\ell}f(0).
$$
This completes the proof. 

\begin{exercise}\label{gencontt2} Let $\mathcal S_m(V)\subset C^m(V)$ be the subspace of functions whose derivatives of order $\le m$ are rapidly decaying. Prove that the differentiation formula of Theorem \ref{contt2} holds for $g\in \mathcal S_2(V)$. Deduce that if $g\in \mathcal S_{2n}(V)$ then $I\in C^n[0,\infty)$, and that if moreover $g$ vanishes at $0$ to order $2n+1$ then $I_g(0)=...=I_g^{(n)}(0)=0$. 
\end{exercise} 

\subsection{The steepest descent formula} \index{steepest descent formula}
Let $a<b$ be real numbers and $f,g: [a,b]\to \Bbb R$ be
continuous functions which are smooth on $(a,b)$.  

\begin{theorem}\label{statphase1} (Steepest descent formula)
Assume that 
$f$ attains a global minimum 
at a unique point $c\in [a,b]$, such that $a<c<b$ and $f''(c)>0$. 
Then one has 
\begin{equation}\label{inte}
\int_a^bg(x)e^{-\frac{f(x)}{\hbar}}dx=\hbar^{\frac{1}{2}}e^{-\frac{f(c)}{\hbar}}I(\hbar),
\end{equation}
where $I(\hbar)$ extends to a smooth function on $[0,\infty)$
such that $$I(0)=\sqrt{2\pi}\frac{g(c)}{\sqrt{f''(c)}}.$$ 
\end{theorem} 

\begin{proof} Without loss of generality we may put $c=0,f(c)=0$. 
Let $f''(c)=M$. Making a change of variable, we may reduce to a situation where 
$f(x)=\frac{M}{2}x^2$ when $x$ is in some neighborhood $U$ of $0$. Let $h$ be a ``bump" function - a smooth function supported in $U$ which equals $1$ in a smaller neighborhood $0\in U'\subset U$. Write 
$g=g_1+g_2$, where $g_1=hg$ and $g_2=(1-h)g$. 
Let $I$ be defined by equation (\ref{inte}), and $I_1,I_2$ be defined by the same equation for $g$ replaced by $g_1,g_2$, so $I=I_1+I_2$. 
Since $f$ has a unique global minimum, we see by direct differentiation that for all $n$, $I_2^{(n)}(\hbar)$ is rapidly decaying as $\hbar\to 0$. 
Thus for $g=g_2$ the result is obvious, and our job is 
to prove it for $g=g_1$. In other words, we may assume without loss of 
generality that $g=g_1$ and $g_2=0$. We extend $g$ by zero to the whole real line. 

Let us make a change of variables 
$y:=\hbar^{-\frac{1}{2}}x$. Then we get
\begin{equation}\label{inte1}
I(\hbar)=\int_{-\infty}^{\infty}g(\hbar^{\frac{1}{2}} y)e^{-\frac{M}{2}y^2}dy.
\end{equation} 
Thus the result follows from \eqref{valzer} and Theorem \ref{contt2}. 
\end{proof} 

\begin{remark}\label{exform} Theorem \ref{statphase1}, in fact, provides an explicit formula 
for the Taylor coefficients of $I(\hbar)$. Namely, as in the proof of Theorem \ref{statphase1}, assume that $c=0$ and $f(x)=\frac{1}{2}p(x)^2$ near $0$, where $$
p'(0)=\sqrt{f''(0)}>0.
$$ 
Ignoring limits of integration (which, as we have seen, are irrelevant for the asymptotic expansion of $I(\hbar)$), we have\footnote{Recall 
that for $I\in C^\infty[0,\varepsilon)$ we write $I(\hbar)\sim \sum_{n=0}^\infty a_n\hbar^n$ 
if for every $N\ge 0$ we have $I(\hbar)=\sum_{n=0}^{N-1} a_n\hbar^n+O(\hbar^N)$ as $\hbar\to 0$.} 
$$
I(\hbar)=\hbar^{-\frac{1}{2}}\int g(x)e^{-\frac{p(x)^2}{2\hbar}}dx\sim \int_{-\infty}^\infty \widetilde g(\hbar^{\frac{1}{2}}y)e^{-\frac{y^2}{2}}dy
$$
where 
$$
\widetilde g(z):=g(p^{-1}(z))(p^{-1})'(z)=\frac{g(p^{-1}(z))}{p'(p^{-1}(z))}.
$$ 
By Theorem \ref{contt2}, the first $n+1$ terms of the Taylor expansion 
of this integral are given by the integral 
$$
I_N(\hbar):=\int_{-\infty}^\infty \widetilde g_N(\hbar^{\frac{1}{2}}y)e^{-\frac{y^2}{2}}dy
$$
where $\widetilde g_N$ is the $2N$-th Taylor polynomial of $\widetilde g$ at $0$. 
Thus if 
$\widetilde g(z)\sim \sum_{n=0}^\infty b_nz^n$ then 
$$
I(\hbar)\sim \sum_{n=0}^\infty b_{2n}\hbar^n\int_{-\infty}^\infty y^{2n}e^{-\frac{y^2}{2}}dy.
$$
But, setting $u=\frac{y^2}{2}$, we have  
\begin{equation}\label{gammaint}
\int_{-\infty}^\infty y^{2n}e^{-\frac{y^2}{2}}dy=
2^{n+\frac{1}{2}}\int_0^\infty u^{n-\frac{1}{2}} e^{-u}du=
2^{n+\frac{1}{2}}\Gamma(n+\tfrac{1}{2})=(2\pi)^{\frac{1}{2}}(2n-1)!!,
\end{equation}
where $(2n-1)!!:=\prod_{1\le j\le n}(2j-1)$.
Hence 
$$
I(\hbar)\sim \sum_{n=0}^\infty b_{2n}
2^{n+\frac{1}{2}}\Gamma(n+\tfrac{1}{2})\hbar^n. 
$$
\end{remark} 

\subsection{Stationary phase formula}
Theorem \ref{statphase1} has the following imaginary analog,
called the {\it stationary phase formula}.\index{stationary phase formula} 
\begin{theorem}\label{statphase2} (Stationary phase formula)
Let $f,g: [a,b]\to \Bbb R$ be smooth functions. 
Assume that $f$ has a unique critical point $c\in [a,b]$, such that $a<c<b$ and $f''(c)\ne
0$, and $g$ has vanishing derivatives of all orders at $a$ and $b$. Then 
$$
\int_a^bg(x)e^{\frac{if(x)}{\hbar}}dx=\hbar^{\frac{1}{2}}e^{\frac{if(c)}{\hbar}}I(\hbar),
$$
where $I(\hbar)$ extends to a smooth function on $[0,\infty)$
such that 
$$
I(0)=\sqrt{2\pi}e^{\pm\frac{\pi i}{4}}\frac{g(c)}{\sqrt{|f''(c)|}},
$$
where $\pm$ is the sign of $f''(c)$.\footnote{This is called the stationary phase formula because the main contribution comes from the point where the phase 
$\frac{f(x)}{\hbar}$ is stationary.}     
\end{theorem}

\begin{remark} It is important to assume that $g$ has vanishing derivatives of all orders at $a$ and $b$. Otherwise we will get 
additional boundary contributions. 
\end{remark} 

\begin{proof} The proof is analogous to the proof of the steepest descent formula, but slightly more subtle, as we have to keep track of cancellations. 
First we need the following very simple but important lemma which allows us 
to do so. 

\begin{lemma}\label{riem} (Riemann lemma) (i) Let $f: [a,b]\to \Bbb R$ be a smooth function
such that $f'(x)>0$ for all $x\in [a,b]$ and $g: [a,b]\to \Bbb R$ a 
$C^n$-function such that 
$$
g(a)=...=g^{(n-1)}(a)=g(b)=...=g^{(n-1)}(b)=0.
$$
Let 
$$
I(\hbar):=\int_a^b g(x)e^{\frac{if(x)}{\hbar}}dx.
$$
Then $I(\hbar)=O(\hbar^{n}),\ \hbar\to 0$.

(ii) Suppose $g$ is smooth on $[a,b]$ and all derivatives of $g$ at $a$ and $b$ are zero. Then $I$ extends (by setting $I(0):=0$) to a smooth function on $[0,\infty)$  whose all derivatives are rapidly decaying as $\hbar\to 0$. 
\end{lemma} 

\begin{proof} (i) By making a change of variables we may assume without loss of generality that $f(x)=x$. Then the proof is by induction in $n$. The base case $n=0$ is obvious. 
For $n>0$ note that 
 $$
\int_a^b g(x)e^{\frac{ix}{\hbar}}dx=i\hbar \int_a^b g'(x)e^{\frac{ix}{\hbar}}dx
$$
(integration by parts), which justifies the induction step. 

(ii) follows from (i) by repeated differentiation. 
\end{proof} 

Now we proceed to prove the theorem. As in the proof of the steepest 
descent formula, we may assume that $c=0$ and $f=\frac{M}{2}x^2$ near $0$ for some $M\ne 0$, and write $I$ as the sum $I_1+I_2$. 
Moreover, by Lemma \ref{riem}(ii)
$$
I_2(\hbar)=\int_a^bg_2(x)e^{\frac{if(x)}{\hbar}}dx
$$
is rapidly decaying with all derivatives, so it suffices to prove the theorem for $g=g_1$.  

Again following the proof of the steepest descent formula, we have 
\begin{equation}\label{inte12}
I(\hbar)=\int_{-\infty}^\infty
g(\hbar^{\frac{1}{2}} y)e^{\frac{iM}{2}y^2}dy,
\end{equation}
so as before the result follows from \eqref{valzer} and Theorem \ref{contt2}.
\end{proof} 

\begin{remark}\label{exform1} Since computation of the asymptotic expansion 
of $I(\hbar)$ is a purely algebraic procedure, the explicit formula for this expansion   
 in the imaginary case is the same as in the real case (Remark \ref{exform}) but with $\hbar$ replaced by $i\hbar$: 
$$
I(\hbar)\sim \sum_{n=0}^\infty b_{2n}2^{n+\frac{1}{2}}\Gamma(n+\tfrac{1}{2})(i\hbar)^n. 
$$
\end{remark}

\subsection{Non-analyticity of $I(\hbar)$ and Borel summation} Even though $I(\hbar)$ is smooth at $\hbar=0$, its Taylor series is usually only an asymptotic
expansion which diverges for any $\hbar\ne 0$, so that this function is not analytic at $0$. To illustrate this, 
consider the integral 
$$
\int_{-\infty}^\infty e^{-\frac{x^2+x^4}{2\hbar}}dx=
\hbar^{\frac{1}{2}}I(\hbar),
$$
where 
\begin{equation}\label{quartic} 
I(\hbar)=
\int_{-\infty}^\infty e^{-\frac{y^2+\hbar y^4}{2}}dy.
\end{equation} 
Since this integral is divergent for any $\hbar<0$, we cannot conclude its analyticity at $\hbar=0$, and it indeed fails to be so. Namely, as in Remark \ref{exform}, the asymptotic expansion of integral \eqref{quartic} is obtained by expanding the exponential $e^{-\frac{1}{2}\hbar y^4}$ into a Taylor series and integrating termwise using \eqref{gammaint}:
$$
I(\hbar)\sim \sum_{n=0}^\infty a_n\hbar^n,
$$
where 
$$
a_n=(-1)^n\int_{-\infty}^\infty 
e^{-\frac{y^2}{2}}\frac{y^{4n}}{2^n n!}dy=
$$
$$
(-1)^n\frac{2^{n+\frac{1}{2}}\Gamma(2n+\frac{1}{2})}{n!}=(-1)^n\sqrt{2\pi}\frac{(4n-1)!!}{2^nn!}.
$$
It is clear that this sequence has super-exponential growth, 
so the radius of convergence of the series is zero. 

Let us now discuss the question: to what extent
does the asymptotic expansion of the function $I(\hbar)$ 
(which we can find using Feynman diagrams as explained below)
actually determine this function? 

Suppose that 
$$
\widetilde I(\hbar)=\sum_{n\ge 0}a_n\hbar^n
$$ 
is 
a series with zero radius of convergence. In general, we cannot 
uniquely determine a function $I$ on $[0,\varepsilon)$ whose expansion 
is given by such a series: it always exists (check it!) but in general there is no canonical choice. However, assume that the exponential generating function of $a_n$  
$$
g(\hbar)=\sum_{n\ge 0}a_n\frac{\hbar^n}{n!}
$$ 
is convergent in some neighborhood of $0$, analytically continues 
to $[0,\infty)$, and has 
at most exponential growth as $\hbar\to \infty$. 
In this case there is a ``canonical''
way to construct a smooth function $I$ on $[0,\varepsilon)$ 
with (asymptotic) Taylor expansion $\widetilde I$, called the {\it Borel
summation}\index{Borel summation} of $\widetilde I$. Namely, the function 
$I$ is defined by the formula 
$$
I(\hbar)=\int_0^\infty g(\hbar u)e^{-u}du=\hbar^{-1}\int_0^\infty g(u)e^{-\frac{u}{\hbar}}du,
$$
i.e., $I(\hbar)=\hbar^{-1}(\mathcal Lg)(\hbar^{-1})$, where $\mathcal L$ is the Laplace transform (note that since $g$ grows at most exponentially at infinity, this is well defined for small enough $\hbar>0$). Note that 
$$
I(\hbar)=\int_{-\infty}^\infty |v|g(\hbar v^2)e^{-v^2}dv=\hbar^{-\frac{1}{2}}\int_{-\infty}^\infty g_*(\hbar^{\frac{1}{2}}v)e^{-v^2}dv,
$$
where $g_*(v)=|v|g(v^2)$. Thus Exercise \ref{gencontt2} 
implies that to compute the asymptotic expansion of $I$, 
we may replace $g$ by its Taylor polynomials at $0$. Hence 
the identity $\int_0^\infty x^ne^{-x}dx=n!$ implies that 
$I$ has the Taylor expansion $\widetilde I$. 

For example, consider the divergent series 
$$
\widetilde I:=\sum_{n\ge  0}(-1)^nn!\hbar^n.
$$ 
Then 
$$
g(\hbar)=
\sum_{n\ge
  0}(-1)^n\hbar^n=\frac{1}{1+\hbar}.
$$ 
Hence, the Borel summation yields 
$$
I(\hbar)=\int_0^\infty\frac{e^{-u}}{1+\hbar u}du=\hbar^{-1} e^{\hbar^{-1}}E_1(\hbar^{-1})
$$
where $E_1(x):=\int_{x}^\infty \frac{e^{-u}}{u}du$ is the integral exponential. 

Physicists expect that in physically interesting situations
perturbation expansions in quantum field theory
are Borel summable, and the actual answers are 
obtained from these expansions by Borel summation. 
The Borel summability of perturbation series 
has actually been established in a few 
nontrivial examples of QFT. 

\begin{exercise} Show that the function given by \eqref{quartic} equals the Borel sum of its asymptotic expansion. 

{\bf Hint.} The function $g(z)$ in this example is a special case of the hypergeometric function ${}_2F_{1}$ which does not express in elementary functions. But it satisfies a hypergeometric differential equation. 
Write down this equation and show that the Laplace transform turns it into another second order linear differential equation, and that the function $I(\hbar)$ given by \eqref{quartic} satisfies this equation. 
\end{exercise} 

\subsection{Application of steepest descent} 
Let us give an application of 
Theorem \ref{statphase1}.
Consider the integral 
$$
\Gamma(s+1)=\int_0^\infty t^se^{-t}dt,\ s>0.
$$
By doing a change of variable $t=sx$, we get 
$$
\frac{\Gamma(s+1)}{s^{s+1}}=\int_0^\infty x^se^{-sx}dx=
\int_0^\infty e^{-s(x-\log x)}dx.
$$
Thus, we can apply Theorem \ref{statphase1} 
for $\hbar=\frac{1}{s}$, $f(x)=x-\log x$, $g(x)=1$
(of course, the interval $[a,b]$ is now infinite, 
and the function $f$ blows up on the boundary, 
but one can easily see that the theorem is still applicable, with the same proof). 
The function $f(x)=x-\log x$ has a unique critical point 
on $[0,\infty)$, which is $c=1$, and we have $f''(c)=1$.  
Then we get
\begin{equation}\label{stir}
\Gamma(s+1)\sim s^se^{-s}\sqrt{2\pi s}(1+\tfrac{a_1}{s}+\tfrac{a_2}{s^2}+ \cdots).
\end{equation}
This is the celebrated {\it Stirling formula}.\index{Stirling formula} 

Moreover, we can compute the coefficients $a_1,a_2,...$ using Remark \ref{exform}. Namely, 
$$
p(x)=\sqrt{2(x-\log(1+x))}=x\sqrt{1-\tfrac{2x}{3}+\tfrac{x^2}{2}-\dots}=x-\tfrac{x^2}{3}+\tfrac{7x^3}{36}+...
$$ 
Thus 
$$
p^{-1}(z)=z+\tfrac{z^2}{3}+\tfrac{z^3}{36}+...,
$$
hence 
$$
(p^{-1})'(z)=1+\tfrac{2z}{3}+\tfrac{z^2}{12}+...,
$$
So for instance by Remark \ref{exform} $a_1=b_2=\frac{1}{12}$.

\begin{remark} Another way to compute this asymptotic expansion is to use the Euler product formula 
for the Gamma function. Differentiating the logarithm of this formula twice, we obtain (for $z>0$):
$$
(\log \Gamma)''(z)=\sum_{n=0}^\infty\frac{1}{(z+n)^2}=\sum_{n=0}^\infty \int_0^\infty te^{-(z+n)t}dt=
\int_0^\infty \frac{te^{-zt}}{1-e^{-t}}dt. 
$$
Recall that the {\it Bernoulli numbers}\index{Bernoulli numbers} are defined by the 
generating function 
$$
\sum_{n\ge 0}\frac{B_nt^n}{n!}=\frac{t}{1-e^{-t}},
$$
e.g. $B_0=1,B_1=\frac{1}{2}$, $B_{2n+1}=0$ for $n\ge 1$. 
Thus we get for $z\to \infty$
$$
(\log \Gamma)''(z)\sim \sum_{n\ge 0}B_nz^{-n-1}.
$$
Integrating, we get 
$$
(\log \Gamma)'(z)\sim \log z+C_1 -\sum_{n\ge 1}\frac{B_n}{n}z^{-n},
$$
so integrating again and adding $\log z$, we get 
$$
\log \Gamma(z+1)\sim z\log z-z +C_1z+\frac{1}{2}\log z+C_2+\sum_{n\ge 2}\frac{B_n}{n(n-1)}z^{-n+1}.
$$
From Stirling's formula we have $C_1=0,C_2=\frac{1}{2}\log(2\pi)$, so in the end we get 
\begin{equation}\label{loggapri}
(\log \Gamma)'(z)\sim \log z -\sum_{n\ge 1}\frac{B_n}{n}z^{-n},
\end{equation} 
\begin{equation}\label{loggapri1}
\log \Gamma(z+1)\sim z\log z +\frac{1}{2}\log z +\frac{1}{2}\log(2\pi)+\sum_{n\ge 2}\frac{B_n}{n(n-1)}z^{-n+1}.
\end{equation} 
So
$$
1+\tfrac{a_1}{s}+\tfrac{a_2}{s^2}+ \cdots=\exp(\sum_{n\ge 2}\tfrac{B_n}{n(n-1)}s^{-n+1}). 
$$
In particular, since $B_2=\frac{1}{6}$, we get $a_1=\frac{1}{12}$. 
\end{remark} 

\begin{exercise} Calculate $\int_0^\pi \sin^nxdx$ for nonnegative integers $n$  
using integration by parts. Then apply steepest descent to this integral 
and discover a formula for $\pi$ (the so called {\it Wallis formula}\index{Wallis formula}). 
\end{exercise} 

\begin{exercise} The Bessel function $I_0(a)$ is defined by the formula 
$$
I_0(a)=\frac{1}{2\pi}\int_0^{2\pi} e^{a\cos \theta} d\theta.
$$
It is an even entire function with Taylor expansion 
$$
I_0(a)=\sum_{n=0}^\infty \frac{a^{2n}}{2^{2n}n!^2}.
$$
Use the steepest descent/stationary phase formulas to 
find the asymptotic expansion of $I_0(a)$ as $a\to +\infty$ and $a\to i\infty$. Compute the first two terms of the expansion
(cf. Remark \ref{severa}). 
\end{exercise}

\subsection{Multidimensional versions of steepest descent and
  stationary phase} 
Theorems \ref{statphase1},\ref{statphase2} 
have multidimensional analogs. To formulate them, let $V$ 
be a real vector space of dimension $d$
with a fixed volume element $dx$ and $D\subset V$ be a compact region with smooth boundary.\footnote{The condition of smooth boundary 
is introduced for simplicity of exposition only and is not essential. The same results and proofs apply with trivial modifications to more general regions, e.g. those whose boundary is only piecewise smooth in an appropriate sense.}

\begin{theorem}\label{statphase3} (Multidimensional steepest descent formula) \linebreak
Let $f,g: D\to \Bbb R$ be continuous functions which are smooth in the interior 
of $D$. Assume that $f$ achieves global minimum on $D$ at a unique point $c$, such that $c$ is an interior point and $f''(c)>0$. Then 
\begin{equation}\label{intemult}
\int_{D} g(x)e^{-\frac{f(x)}{\hbar}}dx=\hbar^{\frac{d}{2}}e^{-\frac{f(c)}{\hbar}}I(\hbar),
\end{equation}
where $I(\hbar)$ extends to a smooth function on $[0,\infty)$
such that 
$$
I(0)=
(2\pi)^{\frac{d}{2}}\frac{g(c)}{\sqrt{\det f''(c)}}.
$$ 
\end{theorem} 

\begin{theorem}\label{statphase4} (Multidimensional stationary phase formula)
Let \linebreak $f,g: D\to \Bbb R$ be smooth functions.  
Assume that $f$ has a unique critical point $c$ in $D$,
such that $c$ is an interior point and $\det f''(c)
\ne 0$, and $g$ 
has vanishing derivatives of all orders on $\partial D$. Then 
\begin{equation}\label{intemult1}
\int_{D} g(x)e^{\frac{if(x)}{\hbar}}dx=\hbar^{\frac{d}{2}}e^{\frac{if(c)}{\hbar}}I(\hbar),
\end{equation}
where $I(\hbar)$ extends to a smooth function on $[0,\infty)$
such that 
$$
I(0)=
(2\pi)^{\frac{d}{2}}e^{\frac{\pi i\sigma}{4}}\frac{g(c)}{\sqrt{|\det f''(c)|}},
$$ 
where $\sigma$ is the signature of the symmetric bilinear form $f''(c)$.
\end{theorem} 

\subsection{Morse lemma} For the proof of these theorems it is convenient to use a fundamental result in multivariable calculus called {\it the Morse lemma}\index{Morse lemma}. This lemma easily follows by induction in dimension from the following theorem. 

\begin{theorem}\label{sepvar} (Separation of variables) 
Let $f$ be a smooth function on an open ball $0\in B\subset \Bbb R^d$ which has a non-degenerate critical point at $0$, and suppose $f(0)=0$. 
Then there is a local coordinate system near $0$ (possibly defined in a smaller ball) 
in which 
$$
f(x_1,...,x_n)=f(x_1,...,x_{d-1})\pm x_d^2. 
$$
\end{theorem} 

\begin{proof} By making a linear change of variables, we can 
assume that the quadratic part of $f$ has the form $Q(y)\pm u^2$, where 
$y:=(x_1,...,x_{d-1})$, $u:=x_d$. Consider the hypersurface $S$
defined by the equation 
$$
\partial_uf(y,u)=0.
$$ 
The linear part 
of $\partial_uf(y,u)$ is $\pm 2u$, so by the implicit function theorem there is a change of coordinates $F$ near $0$ (with $dF(0)=1$) in which $u$ is replaced by $v:=\pm\frac{1}{2}\partial_uf(y,u)$ and $y$ is kept unchanged; so
$u=g(y,v)$ for some function $g$ with $(\partial_vg)(0,0)\ne 0$. Let 
$$
f_*(y,v):=f(y,u)=f(y,g(y,u)).
$$ 
Then by the chain rule
$$
\partial_v f_*(y,v)=\partial_u f_*(y,v)\tfrac{\partial u}{\partial v}=\partial_uf(y,u)\tfrac{\partial u}{\partial v}=\pm 2v\partial_{v}g(y,v).
$$
Thus the hypersurface $S$ in the new coordinates is defined by the equation 
$v=0$. So we may assume without loss of generality that $S$ is given by the equation $u=0$ to start with. Then $(\partial_uf)(y,0)=0$, so 
$$
f(y,u)-f(y,0)=h(y,u)u^2,
$$
where $h$ is a smooth function in $B$ with $h(0,0)=\pm 1$.
By replacing $u$ with $\widetilde u:=\sqrt{|h(y,u)|}u$ and keeping $y$ unchanged, we may assume that $h=\pm 1$. Then 
$$
f(u,y)=f(0,y)\pm u^2,
$$
as claimed. 
\end{proof} 

\begin{corollary}\label{morsele} (Morse lemma)
Let $f$ be a smooth function on an open ball $0\in B\subset \Bbb R^d$ which has a non-degenerate critical point at $0$, and suppose $f(0)=0$. Then 
there is a local coordinate system $(x_1,...,x_d)$ near $0$ (possibly defined in a smaller ball) in which $$f=x_1^2+...+x_m^2-x_{m+1}^2-...-x_d^2.$$ In other words, near a non-degenerate critical point a smooth function is equivalent 
by a change of coordinates to its quadratic part. 
\end{corollary} 

\begin{proof} As mentioned above, this follows easily from Theorem \ref{sepvar} by induction in dimension. 
\end{proof}

\begin{exercise} Let $f$ be a smooth function on $\Bbb R^2$ 
which is a cubic polynomial in $x$: 
$$
f(x,y)=a(y)+b(y)x+c(y)x^2+d(y)x^3.
$$
Assume that $a(0)=a'(0)=0$, $b(0)=b'(0)=0$, $a''(0)=c(0)=2$. 
Find explicitly local coordinates $u=u(x,y),v=v(x,y)$ near $0$ in which $f(x,y)=u^2+v^2$. 
\end{exercise} 

\subsection{Proof of the multidimensional steepest descent and stationary phase formulas} 
The proofs of the multidimensional steepest descent and stationary phase formulas are parallel to the proofs of their 
one-dimensional versions, using the Morse lemma. 
Namely, the Morse lemma allows us to assume without loss of 
generality that $f$ is quadratic near the critical point. After this, 
the proof of the steepest descent formula is identical to the 1-variable case. The same applies to the stationary phase formula, using 
the following multivariable analog of the Riemann lemma.

\begin{lemma}\label{mrl} Let $f,g: D\to \Bbb R$ be smooth functions such that 
all derivatives of $g$ vanish on $\partial D$ and $df$ does not vanish anywhere on the support of $g$. Then the function 
$$
I(\hbar):=\int_{D} g(x)e^{\frac{if(x)}{\hbar}}dx
$$
extends to a smooth function on $[0,\infty)$ and 
has rapidly decaying derivatives of all orders as $\hbar\to 0$. 
\end{lemma} 
 
\begin{proof} Since $df$ does not vanish on ${\rm supp}g$, we can cover ${\rm supp}g$ by local charts $U_i$ in which $f(x)$ is the last coordinate $x_d$. By compactness this cover can be chosen finite. By using a partition of unity $\lbrace h_i\rbrace$ on ${\rm supp}g$ subordinate to this cover and replacing $g$ with $h_ig$, 
we may assume without loss of generality that $g$ is supported on a single chart. 
Then changing variables, we may also assume that $f(x)=x_d$. 
Then integrating out the variables $x_1,...,x_{d-1}$, we 
reduce to the 1-dimensional case covered by Lemma \ref{riem}. 
\end{proof} 

\begin{remark}\label{severa} It is clear from the proof of the stationary phase formula that it extends 
to the case when $f$ may have several critical points but all of them are interior and non-degenerate. In this case the asymptotic expansions coming from different critical points are simply added together. The same applies to the steepest descent formula if 
the global minimum is attained at several points all of which are interior and non-degenerate. 
\end{remark} 

\section{Feynman calculus}

\subsection{Wick's theorem}
Let $V$ be a real vector space of dimension $d$ with volume element $dx$. 
Let $S(x)$ be a smooth function on a compact region $D\subset V$ with smooth boundary which attains its minimum at a unique point 
$c\in D$ in the interior of $D$, and let $g$ be any smooth function on $D$. 
In the previous section we proved the steepest descent formula which implies that the function 
$$
I(\hbar)=\hbar^{-\frac{d}{2}}e^{\frac{S(c)}{\hbar}}\int_{D}g(x)e^{-\frac{S(x)}{\hbar}}dx
$$
admits an asymptotic power series expansion in $\hbar$:
\begin{equation}\label{A012}
I(\hbar)=a_0+a_1\hbar+ \cdots +a_m\hbar^m+ \cdots
\end{equation}
Our main question now will be: how to compute the coefficients $a_i$?

Our proof of the steepest descent formula shows that although 
the problem of computing $I(\hbar)$ is
transcendental, the problem of computing the coefficients 
$a_i$ is, in fact, purely algebraic, and involves only differentiation 
of the functions $S$ and $g$ at the point $c$. Indeed, recalling the proof 
of equation \eqref{A012},  
we see that the calculation of $a_i$ reduces to calculation 
of integrals of the form 
$$
\int_{V}P(x)e^{-\frac{B(x,x)}{2}}dx,
$$
where $P$ is a polynomial and $B$ is a positive definite bilinear form
(in fact, $B(v,u)=(\partial_v\partial_uS)(c)$). 
But such integrals can be exactly evaluated. Namely, it is sufficient 
to consider the case
when $P$ is a product of linear functions, in which case the answer 
is given by the following elementary formula, known to physicists as 
{\it Wick's theorem.}\index{Wick's theorem}

For a positive integer $k$, consider the set  
$\lbrace{1, \ldots ,2k\rbrace}$. By a {\it matching}\index{matching}
$\sigma$ on this set we will mean its partition  
into $k$  disjoint two-element subsets (pairs). 
A matching can be visualized by drawing $2k$ points and 
connecting two points with an edge if they 
belong to the same pair (see Fig.~\ref{fig:1}). This will give $k$ edges
which are not connected to each other. 

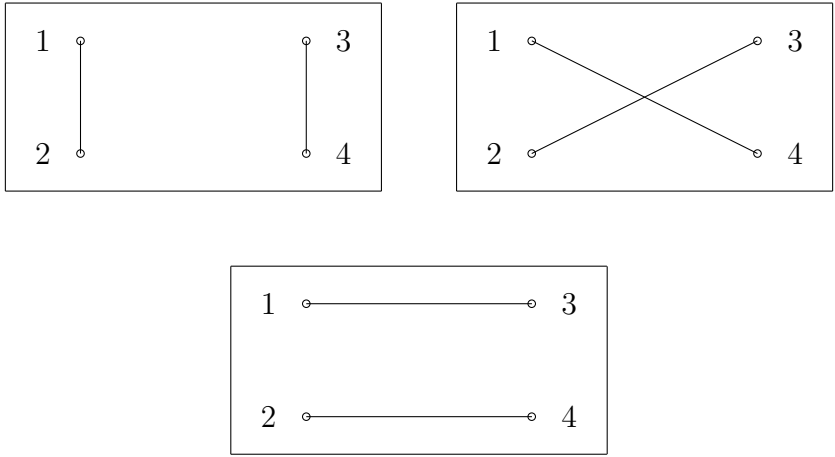
\begin{figure}[htbp]
  \begin{center}

    \setlength{\unitlength}{0.5cm}

    \begin{picture}(22,12)(0,4)
      %%% make first figure
      %%% make box
      \put(0,11){\line(1,0){10}}
      \put(0,16){\line(1,0){10}}
      \put(0,11){\line(0,1){5}}
      \put(10,11){\line(0,1){5}}
      %%% make graph
      \put(2,12){\line(0,1){3}}
      \put(8,12){\line(0,1){3}}
      \put(2,12){\circle*{0.2}}
      \put(2,15){\circle*{0.2}}
      \put(8,12){\circle*{0.2}}
      \put(8,15){\circle*{0.2}}
      \put(1,12){\makebox(0,0)[c]{$2$}}
      \put(1,15){\makebox(0,0)[c]{$1$}}
      \put(9,12){\makebox(0,0)[c]{$4$}}
      \put(9,15){\makebox(0,0)[c]{$3$}}

      %%% make second figure
      %%% make box
      \put(12,11){\line(1,0){10}}
      \put(12,16){\line(1,0){10}}
      \put(12,11){\line(0,1){5}}
      \put(22,11){\line(0,1){5}}
      %%% make graph
      \put(14,12){\line(2,1){6}}
      \put(20,12){\line(-2,1){6}}
      \put(14,12){\circle*{0.2}}
      \put(14,15){\circle*{0.2}}
      \put(20,12){\circle*{0.2}}
      \put(20,15){\circle*{0.2}}
      \put(13,12){\makebox(0,0)[c]{$2$}}
      \put(13,15){\makebox(0,0)[c]{$1$}}
      \put(21,12){\makebox(0,0)[c]{$4$}}
      \put(21,15){\makebox(0,0)[c]{$3$}}

      %%% make third figure
      %%% make box
      \put(6,4){\line(1,0){10}}
      \put(6,9){\line(1,0){10}}
      \put(6,4){\line(0,1){5}}
      \put(16,4){\line(0,1){5}}
      %%% make graph
      \put(8,5){\line(1,0){6}}
      \put(8,8){\line(1,0){6}}
      \put(8,5){\circle*{0.2}}
      \put(8,8){\circle*{0.2}}
      \put(14,5){\circle*{0.2}}
      \put(14,8){\circle*{0.2}}
      \put(7,5){\makebox(0,0)[c]{$2$}}
      \put(7,8){\makebox(0,0)[c]{$1$}}
      \put(15,5){\makebox(0,0)[c]{$4$}}
      \put(15,8){\makebox(0,0)[c]{$3$}}

    \end{picture}

    \caption{Matchings of the set $\{1, 2, 3, 4 \}$}
    \label{fig:1}
  \end{center}
\end{figure}

Let us denote the set of matchings on a set $T$ by $\Pi(T)$ and 
the set $\Pi(\lbrace{1, \ldots ,2k\rbrace})$ by 
$\Pi_k$. It is clear that $|\Pi_k|=\frac{(2k)!}{2^k\cdot k!}=(2k-1)!!$. 
For any $\sigma\in \Pi_k$,  
we can think of $\sigma$ as a permutation of $\lbrace{1, \ldots ,2k\rbrace}$, 
such that $\sigma^2=1$ and $\sigma$ has no fixed points. 
Namely, $\sigma$ maps any element $i$ to the second element 
$\sigma(i)$ of the pair containing $i$. 

\begin{theorem}\label{wick} (Wick's theorem) Let $B^{-1}$ denote the inverse form to $B$ on $V^*$, 
and $\ell_1, \ldots ,\ell_N\in V^*$. 
Then, if $N$ is even, we have 
$$
\int_{V}\ell_1(x) \ldots \ell_N(x)e^{-\frac{B(x,x)}{2}}dx=
\frac{(2\pi)^{\frac{d}{2}}}{\sqrt{\det B}}
\sum_{\sigma\in \Pi_{N/2}}\prod_{i\in \lbrace{1, \ldots ,N\rbrace}/\sigma}
B^{-1}(\ell_i,\ell_{\sigma(i)})
$$
If $N$ is odd, the integral is zero. 
\end{theorem}

\begin{proof} If $N$ is odd, the statement is obvious, because the integrand is an odd
function. So consider the even case $N=2k$. 
Since both sides of the equation are symmetric 
polylinear forms in $\ell_1, \ldots ,\ell_N$, 
it suffices to prove the result when $\ell_1= \cdots =\ell_N=\ell$. 
Further, it is clear that the formula in question is 
stable under linear changes of variable, 
so we can choose a coordinate system  
in such a way that $B(x,x)=x_1^2+ \cdots +x_d^2$, and $\ell(x)=x_1$. 
Therefore, it is sufficient to assume that $d=1$ and $\ell(x)=x$. 
In this case, the theorem says that 
$$
\int_{-\infty}^{\infty}x^{2k}e^{-\frac{x^2}{2}}dx=(2\pi)^{\frac{1}{2}}(2k-1)!!,
$$
which is formula \eqref{gammaint}. 
\end{proof}

\begin{example} We have
$$
\int_{V}\ell_1(x)\ell_2(x)e^{-\frac{B(x,x)}{2}}dx=
\frac{(2\pi)^{\frac{d}{2}}}{\sqrt{\det B}}
B^{-1}(\ell_1,\ell_2),
$$
\scriptsize
$$
\int_{V}\ell_1(x)\ell_2(x)\ell_3(x)\ell_4(x)e^{-\frac{B(x,x)}{2}}dx=
$$
$$
\frac{(2\pi)^{\frac{d}{2}}}{\sqrt{\det B}}
(B^{-1}(\ell_1,\ell_2)B^{-1}(\ell_3,\ell_4)+
B^{-1}(\ell_1,\ell_3)B^{-1}(\ell_2,\ell_4)+
B^{-1}(\ell_1,\ell_4)B^{-1}(\ell_2,\ell_3)).
$$
\normalsize
\end{example} 

Wick's theorem shows that the problem of computing $a_i$ is 
of combinatorial nature. In fact, the central role in this computation 
is played by certain finite graphs, which are called  
{\it Feynman diagrams}.\index{Feynman diagram} They are the main subject of the remainder of this section. 

\subsection{Feynman diagrams and Feynman's theorem} \label{fdft}
We come back to the problem of computing the coefficients $a_i$.
Since each particular $a_i$ depends only on a finite number of
derivatives of $g$ at $c$, it suffices to assume that $g$ is a
polynomial, or, more specifically, a product of linear functions:
$g=\ell_1 \ldots \ell_N$, $\ell_i\in V^*$.  Thus, it suffices to be
able to compute the series expansion of the integral
\begin{equation}\label{expaa}
\langle\ell_1 \ldots \ell_N\rangle:= \hbar^{-\frac{d}{2}}e^{\frac{S(c)}{\hbar}}\int_{D}
\ell_1(x) \ldots\ell_N(x)e^{-\frac{S(x)}{\hbar}}dx.
\end{equation} 

Without loss of generality we may assume that 
$c=0$ and $S(c)=0$. 
Then the (asymptotic) Taylor expansion of $S$ 
at $c$ is 
$$
S(x)=\frac{B(x,x)}{2}-\sum_{i\ge 3}\frac{B_i(x,\ldots,x)}{i!},
$$ 
where $B_i:=d^if(0)$. Therefore, regarding 
the left hand side of \eqref{expaa} as a power series in $\hbar$
and making a change of variable $x\mapsto \hbar^{\frac{1}{2}}x$
(like in the last section), we get 
$$
\langle\ell_1 \ldots \ell_N\rangle=
\hbar^{\frac{N}{2}}\int_V \ell_1(x) \ldots \ell_N(x)e^{-\frac{B(x,x)}{2}+
\sum_{i\ge 3}\hbar^{\frac{i}{2}-1}\frac{B_i(x, \ldots ,x)}{i!}}dx.
$$
Note that this is only an identity of asymptotic expansions in $\hbar$, as we ignored the rapidly decaying error which comes from replacing the region $D$ by the
whole space. But it implies in particular that $\langle\ell_1 \ldots \ell_N\rangle=O(\hbar^{\lceil\frac{N}{2}\rceil})$ as $\hbar\to 0$ (as the expansion contains only integer powers of $\hbar$). 

The theorem below, due to Feynman, 
gives the value of this integral in terms of Feynman diagrams.  
This theorem is easy to prove but is central in quantum field theory,
and will be one of our main theorems. Before formulating 
Feynman's theorem, let us introduce some notation. 

Let $G_{\ge 3}(N)$ be the set of isomorphism classes of graphs
with $N$ \linebreak 1-valent ``external'' vertices, labeled by $1, \ldots ,N$, 
and a finite number of unlabeled ``internal'' vertices, of any 
valency $\ge 3$. Note that here and below graphs are allowed
to have multiple edges between two vertices and loops from a
vertex to itself
(see Fig.~\ref{fig:2}).

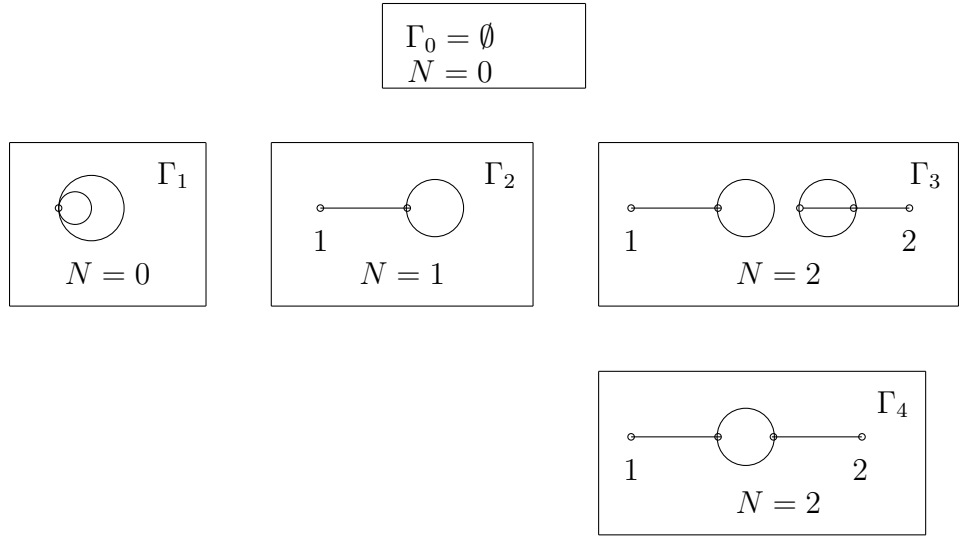
\begin{figure}[b]
  \begin{center}

      \setlength{\unitlength}{0.45cm}
      \begin{picture}(6,2.5)(0,0)

        \put(0,0){\line(1,0){6}}
        \put(0,2.5){\line(1,0){6}}
        \put(0,0){\line(0,1){2.5}}
        \put(6,0){\line(0,1){2.5}}
        \put(2,.5){\makebox(0,0)[c]{$N = 0$}}
        \put(2,1.5){\makebox(0,0)[c]{$\Gamma_0=\emptyset$}}
      \end{picture}
      
    \bigskip
    \bigskip

    \setlength{\unitlength}{0.435cm}
    \begin{picture}(29,12)(0,4)

      %%% make first figure
      %%% make box
      \put(0,11){\line(1,0){6}}
      \put(0,16){\line(1,0){6}}
      \put(0,11){\line(0,1){5}}
      \put(6,11){\line(0,1){5}}
      %%% make graph
      \put(1.5,14){\circle*{0.2}}
      \put(2,14){\circle{1}}
      \put(2.5,14){\circle{2}}
      \put(3,12){\makebox(0,0)[c]{$N = 0$}}
      \put(5,15){\makebox(0,0)[c]{$\Gamma_1$}}

      %%% make second figure
      %%% make box
      \put(8,11){\line(1,0){8}}
      \put(8,16){\line(1,0){8}}
      \put(8,11){\line(0,1){5}}
      \put(16,11){\line(0,1){5}}
      %%% make graph
      \put(9.5,14){\circle*{0.2}}
      \put(9.5,14){\line(1,0){2.75}}
      \put(12.15,14){\circle*{0.2}}
      \put(13,14){\circle{1.75}}
      \put(12,12){\makebox(0,0)[c]{$N = 1$}}
      \put(15,15){\makebox(0,0)[c]{$\Gamma_2$}}
      \put(9.5,13){\makebox(0,0)[c]{$1$}}

      %%% make third figure
      %%% make box
      \put(18,11){\line(1,0){11}}
      \put(18,16){\line(1,0){11}}
      \put(18,11){\line(0,1){5}}
      \put(29,11){\line(0,1){5}}
      %%% make graph
      \put(19,14){\circle*{0.2}}
      \put(19,14){\line(1,0){2.75}}
      \put(21.65,14){\circle*{0.2}}
      \put(22.5,14){\circle{1.75}}
      \put(23.5,12){\makebox(0,0)[c]{$N = 2$}}
      \put(28,15){\makebox(0,0)[c]{$\Gamma_3$}}
      \put(19,13){\makebox(0,0)[c]{$1$}}
      \put(25,14){\circle{1.75}}
      \put(24.25,14){\line(1,0){3.25}}
      \put(24.15,14){\circle*{0.2}}
      \put(25.8,14){\circle*{0.2}}
      \put(27.5,14){\circle*{0.2}}
      \put(27.5,13){\makebox(0,0)[c]{$2$}}

      %%% make fourth figure
      %%% make box
      \put(18,4){\line(1,0){10}}
      \put(18,9){\line(1,0){10}}
      \put(18,4){\line(0,1){5}}
      \put(28,4){\line(0,1){5}}
      %%% make graph
      \put(19,7){\circle*{0.2}}
      \put(19,7){\line(1,0){2.75}}
      \put(21.65,7){\circle*{0.2}}
      \put(22.5,7){\circle{1.75}}
      \put(19,6){\makebox(0,0)[c]{$1$}}
      \put(23.3,7){\line(1,0){2.75}}
      \put(26,6){\makebox(0,0)[c]{$2$}}
      \put(23.35,7){\circle*{0.2}}
      \put(26.05,7){\circle*{0.2}}
      \put(23.5,5){\makebox(0,0)[c]{$N = 2$}}
      \put(27,8){\makebox(0,0)[c]{$\Gamma_4$}}

    \end{picture}
    \caption{Examples of elements of $G_{\ge 3} (N)$.}
    \label{fig:2}
  \end{center}
\end{figure}

For each graph $\Gamma\in G_{\ge 3}(N)$, 
we define the {\it Feynman amplitude}\index{Feynman amplitude} of $\Gamma$ as follows. 

1. Put 
the covector $\ell_j$ at the $j$-th external vertex. 

2. Put the tensor $B_i$ at each $i$-valent internal vertex. 

3. Take the contraction of the tensors along edges 
of $\Gamma$, using the bilinear form $B^{-1}$. 
This will produce a number, called the {\it (Feynman) amplitude}\index{Feynman amplitude} of $\Gamma$ and 
denoted $F_\Gamma(\ell_1, \ldots ,\ell_N)$. 

\begin{remark} If $\Gamma$ is not connected, then 
$F_\Gamma$ is defined to be the product of numbers obtained
from the connected components. Also, the amplitude 
of the empty diagram is defined to be $1$. 
\end{remark} 

\begin{example} Let 
$$
B_3:=\sum_i b_i^{13}\otimes b_i^{23}\otimes b_i^{33},\ 
B_4:=\sum_j b_j^{14}\otimes b_j^{24}\otimes b_j^{34}\otimes b_j^{44},
$$ 
where $b_i^{jk}\in V^*$. Then for the graph $\Gamma_3$ in Fig.~\ref{fig:2}
the amplitude equals 
\scriptsize
$$
F_{\Gamma_3}(\ell_1,\ell_2)=
$$
$$
\sum_i B^{-1}(\ell_1,b_i^{13})B^{-1}(b_i^{23},b_i^{33})\cdot \sum_{i,j} B^{-1}(b_i^{13},b_j^{14})B^{-1}(b_i^{23},b_j^{24})B^{-1}(b_i^{33},b_j^{34})B^{-1}(b_j^{44},\ell_2).
$$
\normalsize
\end{example} 

\begin{theorem}\label{Feyn} (Feynman)
One has
\begin{equation}\label{Feynsum}
\langle\ell_1 \ldots \ell_N\rangle=
\frac{(2\pi)^{\frac{d}{2}}}{\sqrt{\det B}}
\sum_{\Gamma\in G_{\ge 3}(N)}\frac{\hbar^{b(\Gamma)}}{|{\rm Aut}(\Gamma)|}
F_\Gamma(\ell_1, \ldots ,\ell_N),
\end{equation} 
where $b(\Gamma)$
is the number of edges minus the number of internal vertices of $\Gamma$.
\end{theorem}
Here ${\rm Aut}(\Gamma)$ denotes the group of automorphisms of $\Gamma$, and by an automorphism of $\Gamma$ we mean a permutation of
vertices {\bf and} edges (possibly flipping the self-loops) which fixes each external vertex and preserves the graph structure, see Fig.~\ref{fig:3}. Thus there can exist nontrivial automorphisms 
which act trivially on the set of vertices and even ones also acting trivially on the set of edges. For example, there is an automorphism 
of $\Gamma_4$ that flips the upper and lower arc, and 
an automorphism of $\Gamma_2$ that flips the self-loop.  

\begin{figure}[htbp]
  \begin{center}
    \setlength{\unitlength}{0.5cm}

    \begin{picture}(14,5)(1,0)
      %%% make first figure
      %%% make box
      \put(1,0){\line(1,0){14}}
      \put(1,5){\line(1,0){14}}
      \put(1,0){\line(0,1){5}}
      \put(15,0){\line(0,1){5}}
      %%% make graph
      \put(2,2.5){\circle*{0.2}}
      \put(2,2.5){\line(1,0){3}}
      \put(1.75,1.5){$1$}
      %%% first oval
      \put(5,2.5){\circle*{0.2}}
      \put(8,2.5){\circle*{0.2}}
      \qbezier(5,2.5)(6.5,4)(8,2.5)
      \qbezier(5,2.5)(6.5,1)(8,2.5)
      \put(6.5,3.15){\vector(0,-1){1.2}}
      %%% second oval
      \put(8,2.5){\line(1,0){2.5}}
      \put(10,2.5){\line(1,0){2.5}}
      \put(11.25,1){\line(0,1){3}}
      \put(10,2.5){\circle*{0.2}}
      \put(11.25,2.5){\circle*{0.2}}
      \put(12.5,2.5){\circle*{0.2}}
      \put(11.25,4){\circle*{0.2}}
      \put(11.25,1){\circle*{0.2}}
      \qbezier(11.25,1)(8.73,2.5)(11.25,4)
      \qbezier(11.25,1)(13.77,2.5)(11.25,4)
      %%% last arrow
      \qbezier(12.75,1.25)(15.25,2.5)(12.75,3.75)
      \put(12.75,1.25){\vector(-4,-3){.2}}
    \end{picture}
    \caption{An automorphism of a graph}
    \label{fig:3}
  \end{center}
\end{figure}
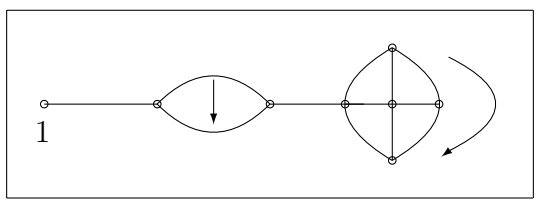

\begin{remark} 1. Note that this sum is infinite, but 
$\hbar$-adically convergent. 

2. Theorem \ref{Feyn} is a generalization of Wick's theorem: 
the latter is obtained if $S(x)=\frac{B(x,x)}{2}$. Indeed, in this case graphs
which give nonzero amplitudes do not have internal vertices, and thus reduce to graphs
corresponding to matchings $\sigma$. 
\end{remark}

Let us now make some comments about the terminology. In quantum 
field theory, the function $\langle\ell_1 \ldots \ell_N\rangle$ is called the 
{\it N-point correlation function},\index{N-point correlation function} and graphs
$\Gamma$ are called {\it Feynman diagrams}.\index{Feynman diagram}  The form $B^{-1}$
which is put on the edges is called the {\it propagator}.\index{propagator}.  The
cubic and higher terms $\frac{B_i}{i!}$ in the expansion of the function
$S$ are called {\it interaction terms},\index{interaction term} since such terms (in the
action functional) describe interaction between particles. The
situation in which $S$ is quadratic (i.e., there is no
interaction) is called a {\it free theory}\index{free theory}; i.e. for the free
theory the correlation functions are determined by Wick's
formula.

\begin{remark}\label{G3*} Sometimes it is convenient to consider 
normalized correlation functions 
$$
\langle\ell_1\ldots\ell_N\rangle_{\rm norm}:=\frac{\langle\ell_1\ldots\ell_N\rangle}{\langle\emptyset\rangle}
$$ 
where $\langle\emptyset\rangle$ denotes the integral without insertions. 
Feynman's theorem implies that they are given by the formula 
$$
\langle\ell_1 \ldots \ell_N\rangle_{\rm norm}=
\sum_{\Gamma\in G_{\ge 3}^*(N)}\frac{\hbar^{b(\Gamma)}}{|{\rm Aut}(\Gamma)|}
F_\Gamma(\ell_1, \ldots ,\ell_N),
$$
where $G^*_{\ge 3}(N)$ is the subset of all graphs 
in $G_{\ge 3}(N)$ which have no components without external 
vertices. 
\end{remark}  

\subsection{A weighted version of Feynman's theorem}
Before proving Theorem \ref{Feyn}, we would
like to slightly modify and generalize it. Namely, in quantum field theory 
it is often useful to consider an interacting theory as a 
deformation of a free theory. 
This means that $S(x)=\frac{B(x,x)}{2}+\widetilde S(x)$, 
where $\widetilde S(x)$ is a perturbation 
$$
\widetilde S(x):=-\sum_{i\ge 0}g_i\frac{B_i(x, \ldots ,x)}{i!}
$$ 
in which $g_r, r\ge 0$ are (formal) parameters. One benefit of these parameters is that they will 
allow us to group the amplitudes of Feynman diagrams in the sum \eqref{Feynsum} 
by the numbers of vertices of each valency. 
Namely, consider the {\it partition function}\index{partition function}
$$
Z=\hbar^{-\frac{d}{2}}\int_V e^{-\frac{S(x)}{\hbar}}dx
$$
as a series in $g_i$.
Let $\bold n=(n_0,n_1,n_2, \ldots )$ be a sequence 
of nonnegative integers, almost all zero. 
Let $G(\bold n)$ denote the set of isomorphism classes of 
graphs with $n_0$ 0-valent vertices, $n_1$ 1-valent vertices, $n_2$ 2-valent vertices, etc. 
(thus, now we are considering graphs without external vertices). For $\Gamma\in G(\bold n)$, 
let $F_\Gamma$ is the amplitude of $\Gamma$ defined as before. Thus 
$$
F_\Gamma=\prod_i{g_i^{n_i}}\cdot \Bbb F_\Gamma, 
$$
where $\Bbb F_\Gamma$ is the Feynman amplitude computed without the factors $g_j$. 

\begin{theorem}\label{Feyn1}
One has
$$
Z=
\frac{(2\pi)^{\frac{d}{2}}}{\sqrt{\det B}}
\sum_{\bold n}
\sum_{\Gamma\in G(\bold n)}\frac{\hbar^{b(\Gamma)}}{|{\rm Aut}(\Gamma)|}
F_\Gamma=
$$
$$
\frac{(2\pi)^{\frac{d}{2}}}{\sqrt{\det B}}
\sum_{\bold n}\prod_i({g_i\hbar^{\frac{i}{2}-1})^{n_i}}
\sum_{\Gamma\in G(\bold n)}\frac{\Bbb F_\Gamma}{|{\rm Aut}(\Gamma)|},
$$
where $b(\Gamma)=\sum_i n_i(\frac{i}{2}-1)$ 
is the number of edges minus the number of vertices of $\Gamma$. 
\end{theorem}

Note that we may view $Z$ as an element of the algebra
$$
\Bbb C[g_0\hbar^{-\frac{3}{2}},g_1\hbar^{-1},g_2\hbar^{-\frac{1}{2}}; g_j,j\ge 3][[\hbar^{\frac{1}{2}}]],
$$ 
i.e., it can be specialized to numerical values of 
$$
g_0\hbar^{-\frac{3}{2}},g_1\hbar^{-1},g_2\hbar^{-\frac{1}{2}},g_3,g_4,...,
$$ 
giving an element of $\Bbb C[[\hbar^{\frac{1}{2}}]]$. 
Also $Z$ can be specialized to $\hbar=1$, giving an element 
of $\Bbb C[[g_j, j\ge 0]]$, and the theorem is, in fact, equivalent to this specialization. 
Still we choose to keep $\hbar$ to be able to take the classical limit $\hbar\to 0$. 

We will prove Theorem \ref{Feyn1} in the next subsection. 
Meanwhile, let us show that Theorem \ref{Feyn} is in fact a special case 
of Theorem \ref{Feyn1}.
Indeed, because of symmetry of the correlation functions 
with respect to $\ell_1, \ldots ,\ell_N$, it is sufficient to consider 
the case $\ell_1= \cdots =\ell_N=\ell$. 
In this case, denote the correlation 
function $\langle\ell^N\rangle$ (expectation value of $\ell^N$). Clearly, to compute 
$\langle\ell^N\rangle$ for all $N$, it is sufficient to compute the generating function 
$$
\langle e^{\ell}\rangle=\hbar^{-\frac{d}{2}}\int_Ve^{\ell(x)-\frac{S(x)}{\hbar}}dx:=\sum_{N=0}^\infty \frac{\langle\ell^N\rangle}{N!},
$$
which up to scaling and multiplication of $\ell$ by $i$ is the Fourier transform of the Feynman density $e^{-\frac{S(x)}{\hbar}}dx$. But this 
expectation value is exactly 
the one given by Theorem \ref{Feyn1} for $g_i=1$, $i\ge 3$, 
$g_0=g_2=0$, $g_1=\hbar $, $B_1=\ell$, $B_0=0$, $B_2=0$. Thus, Theorem \ref{Feyn1} 
implies Theorem \ref{Feyn} (the factor $N!$ in the denominator 
is accounted for by the fact that in Theorem \ref{Feyn1} 
we consider unlabeled, rather than labeled, 
1-valent vertices). 

\subsection{Proof of Feynman's theorem} 
Now we will prove Theorem \ref{Feyn1}.
Let us make a change of variable $y=\hbar^{-\frac{1}{2}}x$. 
Expanding the exponential in a Taylor series, we obtain 
$$
Z=\sum_{\bold n}Z_\bold n, 
$$
where
$$
Z_\bold n=\int_V e^{-\frac{B(y,y)}{2}}
\prod_i \frac{g_i^{n_i}}{i!^{n_i}n_i!}
(\hbar^{\frac{i}{2}-1}B_i(y, \ldots ,y))^{n_i}dy.
$$
Writing $B_i$ as a sum of products of linear functions, and 
using Wick's theorem, we find that the value of the integral for each $\bold n$
can be expressed combinatorially as follows. 

1. Attach 
to each factor $B_i$ a ``flower'' --- a vertex with $i$ outgoing
edges (see Fig.~\ref{fig:4}).

\begin{figure}[htbp]
  \begin{center}
    \setlength{\unitlength}{0.5cm}
    \begin{picture}(12,13)

      %%% make first figure
      %%% make box
      \put(0,10){\line(1,0){12}}
      \put(0,13){\line(1,0){12}}
      \put(0,10){\line(0,1){3}}
      \put(12,10){\line(0,1){3}}
      %%% make graph
      \put(1,11.5){\circle*{0.2}}
      \put(3,11.5){\makebox(0,0)[l]{$0$-valent flower}}

      %%% make second figure
      %%% make box
      \put(0,5){\line(1,0){12}}
      \put(0,8){\line(1,0){12}}
      \put(0,5){\line(0,1){3}}
      \put(12,5){\line(0,1){3}}
      %%% make graph
      \put(1,6.5){\circle*{0.2}}
      \put(1,6.5){\line(1,0){3}}
      \put(5.5,6.5){\makebox(0,0)[l]{$1$-valent flower}}

      %%% make third figure
      %%% make box
      \put(0,0){\line(1,0){12}}
      \put(0,3){\line(1,0){12}}
      \put(0,0){\line(0,1){3}}
      \put(12,0){\line(0,1){3}}
      %%% make graph
      \put(1,1.5){\line(1,0){3}}
      \put(4,1.5){\circle*{0.2}}
      \put(4,1.5){\line(1,1){1.2}}
      \put(4,1.5){\line(1,-1){1.2}}
      \put(6,1.5){\makebox(0,0)[l]{$3$-valent flower}}

    \end{picture}
    \caption{ }
    \label{fig:4}
  \end{center}
\end{figure}

2. Consider the set $T_{\bold n}$ of ends of these outgoing edges (see
Fig.~\ref{fig:5}), and for any matching $\sigma$ of this set,
consider the corresponding contraction of the tensors $B_i$ using
the form $B^{-1}$.  This will produce a scalar $\Bbb F(\sigma)$.

\begin{figure}[btp]
  \begin{center}

    \setlength{\unitlength}{0.5cm}
    \begin{picture}(9,13.2)(0,.4)

      %%% make first figure
      \put(0,11){\line(1,0){2.5}}
      \put(0,11){\line(1,1){2.5}}
      \put(0,11){\line(1,-1){2.5}}
      \put(2.7,8.4){\circle{0.5}}
      \put(2.7,11){\circle{0.5}}
      \put(2.7,13.6){\circle{0.5}}      

      %%% make second figure
      \put(0,3){\line(1,0){2.5}}
      \put(0,3){\line(1,1){2.5}}
      \put(0,3){\line(1,-1){2.5}}
      \put(2.7,0.4){\circle{0.5}}
      \put(2.7,3){\circle{0.5}}
      \put(2.7,5.6){\circle{0.5}}      

      %%% make third figure
%      \put(9,11){\line(-1,0){2.5}}
%      \put(9,11){\line(-1,1){2.5}}
%      \put(9,11){\line(-1,-1){2.5}}
%      \put(6.3,8.4){\circle{0.5}}
%      \put(6.3,11){\circle{0.5}}
%      \put(6.3,13.6){\circle{0.5}}      
      \put(9,11){\line(-1,1){2.5}}
      \put(9,11){\line(-3,1){2.5}}
      \put(9,11){\line(-3,-1){2.5}}
      \put(9,11){\line(-1,-1){2.5}}
      \put(6.3,8.4){\circle{0.5}}
      \put(6.3,10.13){\circle{0.5}}
      \put(6.3,11.86){\circle{0.5}}
      \put(6.3,13.6){\circle{0.5}}      

    \end{picture}
    
    \caption{The set $T_\bold n$ for $\bold n = (0, 0, 0, 2, 1, 0, 0,
      \ldots)$ (the set of white circles)} 
    \label{fig:5}
    
  \end{center}
\end{figure}
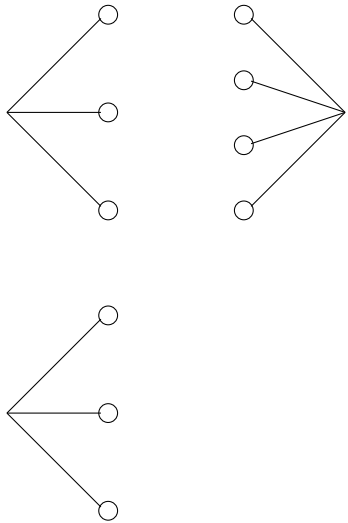

3. The integral $Z_{\bold n}$ is given by 
\begin{equation}\label{sumpair}
Z_{\bold n}=\frac{(2\pi)^{\frac{d}{2}}}{\sqrt{\det B}}\prod_i
\frac{g_i^{n_i}}{i!^{n_i}n_i!}\hbar^{n_i(\frac{i}{2}-1)} \sum_{\sigma\in \Pi(T_{\bold n}) }\Bbb F(\sigma).
\end{equation}

Now, recall that matchings on a set can be visualized by drawing
its elements as points and connecting them with edges. If we do
this with the set $T_{\bold n}$, all ends of outgoing edges will become
connected with each other in some way, i.e. we will obtain a
certain (unoriented) graph $\Gamma=\Gamma_\sigma$ (see Fig.~\ref{fig:6}).
Moreover, it is easy to see that the scalar $\Bbb  F(\sigma)$ is
nothing but the amplitude $\Bbb F_\Gamma$.

\begin{figure}[htbp]
  \begin{center}

    \setlength{\unitlength}{0.5cm}
    \begin{picture}(21,15)(0,.4)

      \put(4.5,15){\makebox(0,0)[c]{$\sigma$:}}
      %*\put(17.5,15){\makebox(0,0)[c]{$r$:}}

      %%% make first figure
      \put(0,11){\line(1,0){2.5}}
      \put(0,11){\line(1,1){2.5}}
      \put(0,11){\line(1,-1){2.5}}
      \put(2.7,8.4){\circle{0.5}}
      \put(2.7,11){\circle{0.5}}
      \put(2.7,13.6){\circle{0.5}}      

      %%% make second figure
      \put(0,3){\line(1,0){2.5}}
      \put(0,3){\line(1,1){2.5}}
      \put(0,3){\line(1,-1){2.5}}
      \put(2.7,0.4){\circle{0.5}}
      \put(2.7,3){\circle{0.5}}
      \put(2.7,5.6){\circle{0.5}}      

      %%% make third figure
%      \put(9,11){\line(-1,0){2.5}}
%      \put(9,11){\line(-1,1){2.5}}
%      \put(9,11){\line(-1,-1){2.5}}
%      \put(6.3,8.4){\circle{0.5}}
%      \put(6.3,11){\circle{0.5}}
%      \put(6.3,13.6){\circle{0.5}}      
      \put(9,11){\line(-1,1){2.5}}
      \put(9,11){\line(-3,1){2.5}}
      \put(9,11){\line(-3,-1){2.5}}
      \put(9,11){\line(-1,-1){2.5}}
      \put(6.3,8.4){\circle{0.5}}
      \put(6.3,10.13){\circle{0.5}}
      \put(6.3,11.86){\circle{0.5}}
      \put(6.3,13.6){\circle{0.5}}      

      %%% make lines and curves
%      \put(2.7,13.6){\circle{0.5}}      
%      \put(6.3,11.86){\circle{0.5}}
      \put(2.92,13.6){\line(2,-1){3.23}}
      \put(2.92,11){\line(4,3){3.23}}
%      \put(2.7,8.4){\circle{0.5}}
      \qbezier(2.92,8.35)(9,4)(2.92,0.4)
%      \put(2.7,3){\circle{0.5}}
%      \put(2.7,5.6){\circle{0.5}}      
      \qbezier(2.92,5.6)(5,4.3)(2.92,3)
      \qbezier(6.06,8.4)(4,9.35)(6.06,10.18)

      %%% make fourth figure
      %*\put(14,11){\line(1,0){7}}
      %*\qbezier(14,11)(17.5,14)(21,11)
%      \put(0,11){\line(1,1){2.5}}
%      \put(0,11){\line(1,-1){2.5}}
      %*\put(14,11){\circle*{0.2}}
      %*\put(21,11){\circle*{0.2}}
%      \put(2.7,13.6){\circle{0.5}}      
      %*\spline(14,11)(15.15,10.5)(16.2,10)(17,9)(16.5,8)(15.5,8.1)(15,8.7)(15,9)
      %*\put(15,9){\circle*{0.2}}
      %*\spline(15,9)(15.5,9.7)(16.2,9.75)(15.7,9.1)(15,9)
%%%      \put(15.8,8.7){\ellipse{1}{.5}}
      %*\spline(21,11)(19.25,10.7)(18.5,9.25)(20.1,9.4)(21,11)

    \end{picture}
    
    \caption{A matching $\sigma$ of $T_{\bold n}$ and the corresponding
      graph $\Gamma$.} 
    \label{fig:6}
    
  \end{center}
\end{figure}
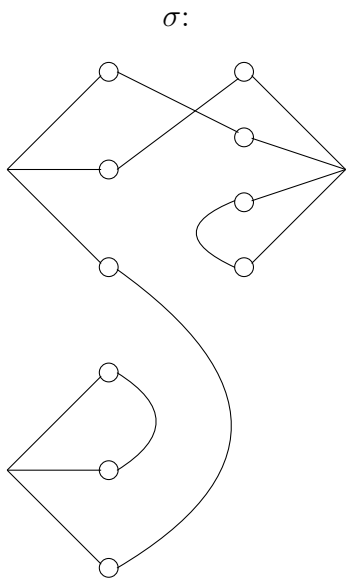

It is clear that any graph $\Gamma$ with $n_i$ $i$-valent vertices
for each $i$ can be obtained in this way. However, the same graph
can be obtained in many different ways, so if we want to collect identical terms in
the sum over $\sigma$, and turn it into a sum over $\Gamma$, we
must find the number of $\sigma$ which yield a given $\Gamma$.

For this purpose, we will consider the group $\Bbb G_{\bold n}$ of permutations of $T_{\bold n}$, 
which preserves ``flowers'' (i.e. endpoints of any two edges outgoing
from the same flower end up again in the same flower). 
This group involves 

1) permutations of ``flowers'' with a given valency;

2) permutation of the $i$ edges inside each $i$-valent ``flower''. 

More precisely, the group $\Bbb G_{\bold n}$ is the semidirect product of symmetric groups
$$
\Bbb G_{\bold n}=\prod_i (S_{n_i}\ltimes S_i^{n_i}).
$$ 
Note that $|\Bbb G_{\bold n}|=\prod_i i!^{n_i}n_i!$, which is the product of the numbers 
in the denominator of formula (\ref{sumpair}). 

The group $\Bbb G_{\bold n}$ acts   
on the set $\Pi(T_{\bold n})$ of all matchings $\sigma$ of $T_{\bold n}$. Moreover, it acts 
transitively on the set 
$\Pi_\Gamma(T_{\bold n})$ of matchings of $T_{\bold n}$ which yield a given graph $\Gamma$. 
Furthermore, it is easy to see that the stabilizer of a given matching is 
${\rm Aut}(\Gamma)$. Thus, the number of matchings giving $\Gamma$ is 
$$
N_\Gamma=\frac{\prod_i i!^{n_i}n_i!}{|{\rm Aut}(\Gamma)|}.
$$
Hence, 
$$
\sum_{\sigma\in \Pi(T_{\bold n})} \Bbb F(\sigma)=\sum_\Gamma \frac{\prod_i i!^{n_i}n_i!}{|{\rm Aut}(\Gamma)|}
\Bbb F_\Gamma.
$$
Finally, note that the exponent of $\hbar$ in equation (\ref{sumpair})
is $\sum_in_i(\frac{i}{2}-1)$, which is the number of edges of $\Gamma$ minus the number 
of vertices, i.e. $b(\Gamma)$. 
Substituting this into (\ref{sumpair}), we get the result.

\begin{example}\label{treeexa} Let $d=1$, $V=\Bbb R$, $g_i=g$, $B_i=z^i$ for all $i\ge 0$ (where $z$ is a formal variable), $\hbar=1$. 
Then we find the asymptotic expansion
$$
\frac{1}{\sqrt{2\pi}}\int_{-\infty}^\infty e^{-\frac{x^2}{2}+ge^{zx}}=
\sum_{n\ge 0}g^n\sum_{\Gamma\in G(n,k)}\frac{z^{2k}}{|{\rm Aut}(\Gamma)|},
$$
where $G(n,k)$ is the set of isomorphism classes of graphs with $n$ vertices
and $k$ edges.\footnote{This integral converges for $g<0$, $z\in \Bbb R$, but this is not important for us here, since we consider the integral formally.}
Expanding the left hand side, we get 
$$
\sum_k\sum_{\Gamma\in G(n,k)}\frac{z^{2k}}{|{\rm Aut}(\Gamma)|}=\frac{e^{\frac{z^2n^2}{2}}}{n!},
$$
and hence
$$
\sum_{\Gamma\in G(n,k)}\frac{1}{|{\rm Aut}(\Gamma)|}=\frac{n^{2k}}{2^kk!n!}.
$$
\end{example} 

\begin{exercise} Check this by direct combinatorics.
\end{exercise} 

\subsection{Sum over connected diagrams} 
Now we will show that the logarithm of the partition function $Z$ 
is also given by summation over diagrams, but with only connected diagrams 
taken into account. This significantly simplifies the analysis of $Z$ 
in the first few orders of perturbation theory, since the number of connected 
diagrams with a given number of vertices and edges 
is significantly smaller than the number of all diagrams. 
\begin{theorem}\label{Feyn2} Let $Z_0=\frac{(2\pi)^{\frac{d}{2}}}{\sqrt{\det B}}$. Then one has 
$$
\log\frac{Z}{Z_0}=\sum_{\bold n}
\prod_i(g_i\hbar^{\frac{i}{2}-1})^{n_i}
\sum_{\Gamma\in G_c(\bold n)}\frac{\Bbb F_\Gamma}{|{\rm Aut}(\Gamma)|}
$$
where $G_c(\bold n)$ is the set of connected graphs in $G(\bold n)$.\footnote{We define a connected graph as a graph with exactly one connected component. So the empty graph, which has zero connected components, is not considered connected.} 
\end{theorem}

\begin{proof}
  For any graphs $\Gamma_1$, $\Gamma_2$, let $\Gamma_1\Gamma_2$
  stand for the disjoint union of $\Gamma_1$ and $\Gamma_2$, and
  for any graph $\Gamma$ let $\Gamma^n$ denote the disjoint union
  of $n$ copies of $\Gamma$.  Then every graph can be uniquely
%%%  written as $\Gamma_1^{k_1}...\Gamma_l^{k_l}$, where $\Gamma_j$
  written as $\Gamma_1^{k_1} \ldots \Gamma_l^{k_l}$, where $\Gamma_j$
  are connected non-isomorphic graphs.  Moreover, it is clear
  that $\Bbb F_{\Gamma_1\Gamma_2}=\Bbb F_{\Gamma_1}\Bbb F_{\Gamma_2}$,
  $b(\Gamma_1\Gamma_2)=b(\Gamma_1)+b(\Gamma_2)$, and 
  $$
  |{\rm
    Aut}(\Gamma_1^{k_1} \ldots \Gamma_l^{k_l})|= \prod_j|{\rm
    Aut}(\Gamma_j)|^{k_j}k_j!.
    $$ 
    Thus, exponentiating the
  equation of Theorem \ref{Feyn2}, and using the above facts
  together with the Taylor series for the function $e^x$, we
  arrive at Theorem \ref{Feyn1}.  As Theorem \ref{Feyn1} has
  been proved, so is Theorem \ref{Feyn2}
\end{proof}

\subsection{The loop expansion} 
Note that since summation in Theorem
\ref{Feyn2} is over connected Feynman diagrams, the number
$b(\Gamma)$ is the number of loops in $\Gamma$ minus 1.  In
particular, the lowest coefficient in $\hbar$ is that of
$\hbar^{-1}$, and it is the sum over all trees; the next
coefficient is to $\hbar^0$, and it is the sum over all diagrams
with one loop (cycle); the next coefficient to $\hbar$ is the sum
over two-loop diagrams, and so on. Therefore, physicists refer to
the expansion of Theorem \ref{Feyn2} as the {\it loop expansion}.\index{loop expansion}

Let us study the two most singular terms in this expansion 
(with respect to $\hbar$), i.e. the terms given by the sum over trees
and 1-loop graphs. 

Let $x_0$ be the critical point of the function $S$. 
It exists and is unique, since $g_i$ are assumed to be formal parameters. 
Let $G^{(j)}(\bold n)$ denote the set of classes of graphs in $G_c(\bold n)$ 
with $j$ loops. 
Let 
$$
\left(\log\frac{Z}{Z_0}\right)_j:=
\sum_{\bold n}\prod_ig_i^{n_i}
\sum_{\Gamma\in G^{(j)}(\bold n)}\frac{\Bbb F_\Gamma}{|{\rm Aut}(\Gamma)|},
$$
so that 
$$
\log\frac{Z}{Z_0}=\sum_{j=0}^\infty \left(\log\frac{Z}{Z_0}\right)_j \hbar^{j-1}.
$$

\begin{theorem}\label{trees1l}
\begin{equation}\label{treeeq}
\left(\log\frac{Z}{Z_0}\right)_0=-S(x_0),
\end{equation}
and 
\begin{equation}\label{1loopeq}
\left(\log\frac{Z}{Z_0}\right)_1=\frac{1}{2}\log \frac{\det B}{\det S''(x_0)}.
\end{equation}
\end{theorem}

\begin{proof}
First note that the statement is purely combinatorial. 
This means, in particular, that it is sufficient to check 
that the statement yields the correct asymptotic expansion 
of the right hand sides of equations (\ref{treeeq}),(\ref{1loopeq}) 
in the case when $S$ is a polynomial with real coefficients
of the form $\frac{B(x,x)}{2}-\sum_{i=0}^Ng_i\frac{B_i(x, \ldots,x)}{i!}$ and $\hbar>0$. To do so, let $Z:=\hbar^{-\frac{d}{2}}\int_{\bold B}e^{-\frac{S(x)}{\hbar}}dx$, 
where $\bold B$ is a ball centered at $0$.
For sufficiently small $g_i$, the function $S$ has a unique 
global minimum point $x_0$ in $\bold B$, which is non-degenerate. 
Thus, 
by the steepest descent formula, we have 
$$
\frac{Z}{Z_0}=e^{-\frac{S(x_0)}{\hbar}}I(\hbar), 
$$
where $I(\hbar)\sim\sqrt{\frac{\det B}{\det
    S''(x_0)}}(1+a_1\hbar+a_2\hbar^2+ \cdots )$ (asymptotically).
Thus,
$$
\log\frac{Z}{Z_0}=-S(x_0)\hbar^{-1}+\frac{1}{2}\log
\frac{\det B}{\det S''(x_0)}+O(\hbar).
$$
This implies the result.
\end{proof}
 
Physicists call the expression 
$(\log\frac{Z}{Z_0})_0$ the {\it classical (or tree) approximation}\index{classical approximation}\index{tree approximation}
to the quantum mechanical quantity $\hbar\log\frac{Z}{Z_0}$, and the sum 
$(\log\frac{Z}{Z_0})_0+\hbar(\log\frac{Z}{Z_0})_1$ the {\it one-loop approximation}.\index{one-loop approximation}
Similarly one defines higher loop approximations. Note that the
classical approximation 
is obtained by finding the critical point and value of the classical action $S(x)$, 
which in the classical mechanics and field theory situation corresponds to solving the classical 
equations of motion.

\subsection{Nonlinear equations and trees}\label{treepar}
As we have noted, Theorem \ref{trees1l} 
does not involve integrals and is purely combinatorial. 
Therefore, there should exist 
a purely combinatorial proof 
of this theorem. Such a proof indeed exists. Here we will give
a combinatorial proof of the first statement of the Theorem  
(formula (\ref{treeeq})).

Consider the equation $S'(x)=0$, defining the critical point 
$x_0$. This equation can be written as $x=\beta(x)$, where
$$
\beta(x):=\sum_{i\ge 1}g_i\frac{B^{-1}B_i(x, \ldots ,x,-)}{(i-1)!},
$$
where $B^{-1}: V^*\to V$ is the operator corresponding 
to the form $B^{-1}$. 

In the sense of power series norm, $\beta$ is a contracting mapping. 
Thus, $x_0=\lim_{N\to \infty}\beta^N(x)$, for any initial vector, for example $0\in V$. 
In other words, we will obtain $x_0$ if we keep substituting 
the series $\beta(x)$ into itself. This leads to summation over trees
(explain why!).
More precisely, we get the following expression
for $x_0$:
$$
x_0=\sum_{\bold n}\prod_i g_i^{n_i}\sum_{\Gamma\in G^{(0)}(\bold n,1)}
\frac{\Bbb F_\Gamma}{|{\rm Aut}(\Gamma)|},
$$
where $G^{(0)}(\bold n,1)$ is the set of trees with one external vertex and 
$n_i$ internal vertices of degree $i$. 
Now, since $S(x)=\frac{B(x,x)}{2}-\sum_i g_i\frac{B_i(x, \ldots ,x)}{i!}$, 
the expression $-S(x_0)$ equals the sum of expressions 
$\prod_i g_i^{n_i}\frac{\Bbb F_\Gamma}{|{\rm Aut}(\Gamma)|}$
over all trees (without external vertices).
Indeed, the term $\frac{B(x_0,x_0)}{2}$ corresponds to gluing 
two trees with external vertices (identifying 
the two external vertices, so that 
they disappear); so it corresponds to 
summing over trees with a marked edge, i.e. counting each tree
as many times as it has edges. On the other hand, 
the term $g_i\frac{B_i(x_0, \ldots ,x_0)}{i!}$ corresponds 
to gluing $i$ trees with external vertices 
together at these vertices (making a tree with 
a marked vertex). So $\sum_i g_i\frac{B_i(x_0, \ldots ,x_0)}{i!}$
corresponds to summing over trees with a marked vertex, i.e.
counting each tree as many times as it has vertices. 
But the number of vertices of a tree exceeds the number of edges by $1$. 
Thus, the difference $-S(x_0)$ of the above two 
contributions corresponds to summing over trees, counting each exactly once.
This implies formula (\ref{treeeq}).

\subsection{The case $d=1$}

In the case $d=1$ we can compute the tree sum $-S(x_0)$ 
even more explicitly. Namely, let 
$$
S(x):=\frac{x^2}{2}-gh(x)
$$
where $h(x)=\sum_{n\ge 0}c_nx^n$ with $c_1\ne 0$.   
Then $x_0$ is the solution of the equation 
$x=gh'(x)$, i.e., $x_0=f(g)$ where $x=f(y)$
is the inverse function to $y=\frac{x}{h'(x)}$. 
So the tree approximation takes the form 
$-S(x_0)=F(g)$ where 
$$
F(g)=-\frac{f(g)^2}{2}+gh(f(g)). 
$$
Thus 
$$
F'(g)=-f(g)f'(g)+h(f(g))+gh'(f(g))f'(g).
$$
But $h'(f(g))=\frac{f(g)}{g}$, so the first and third summands cancel and we get 
$$
F'(g)=h(f(g)), 
$$
hence 
\begin{equation}\label{deq1}
-S(x_0)=\int_0^g h(f(a))da.
\end{equation} 

\subsection{Counting trees and Cayley's theorem}
In this section we will apply Theorem \ref{trees1l} 
to tree counting problems, in particular will prove a classical
theorem due to Cayley that the number of labeled trees with $n$
vertices is $n^{n-2}$. 

We consider essentially the same situation as we considered above in Example \ref{treeexa}: 
$d=1$, $B_i=1$, $g_i=g$. Thus, we have 
$S(x)=\frac{x^2}{2}-ge^x$. 
By Theorem \ref{trees1l}, we have 
$$
\sum_{n\ge 0}g^n\sum_{\Gamma \in T(n)}\frac{1}{|{\rm Aut}(\Gamma)|}=
-S(x_0),
$$
where $T(n)$ is the set of 
isomorphism classes of trees with $n$ vertices, 
and $x_0$ is the root of the equation $S'(x)=0$, i.e. 
$x=ge^x$. 

In other words, let $x=f(y)$ be the function inverse to the function $y=xe^{-x}$ near $x=0$, then $x_0=f(g)$. The function $f(y)$ is related to (the principal branch of) the {\it Lambert function}\index{Lambert function} $W(y)$ by the formula $f(y)=-W(-y)$. By \eqref{deq1} 
$$
-S(x_0)=\int_0^g e^{f(a)}da=\int_0^g\frac{f(a)}{a}da.
$$

Thus it remains to find the Taylor expansion of $f$. This expansion is given by the following classical result. 

\begin{proposition}
One has
$$
f(g)=\sum_{n\ge 1}\frac{n^{n-2}}{(n-1)!}g^n.
$$
\end{proposition}

\begin{proof}
Let $f(g)=\sum_{n\ge 1}a_ng^n$. 
Then 
$$
a_n=\frac{1}{2\pi i}\oint\frac{f(g)}{g^{n+1}}dg=
\frac{1}{2\pi i}\oint\frac{x}{(xe^{-x})^{n+1}}d(xe^{-x})=
$$
$$
\frac{1}{2\pi i}\oint e^{nx}\frac{1-x}{x^n}dx=
\frac{n^{n-1}}{(n-1)!}-\frac{n^{n-2}}{(n-2)!}=
\frac{n^{n-2}}{(n-1)!}.
$$
\end{proof}

So we get  
$$
-S(x_0)=\int_0^g\frac{f(a)}{a}da=\sum_{n\ge 1}\frac{n^{n-2}}{n!}g^n.
$$
This shows that 
$$
\sum_{\Gamma\in T(n)}\frac{1}{|{\rm Aut}(\Gamma)|}=\frac{n^{n-2}}{n!}.
$$
But each isomorphism class of unlabeled trees 
with $n$ vertices has $\frac{n!}{|{\rm
    Aut}(\Gamma)|}$ nonisomorphic labelings. Thus we obtain
    
\begin{corollary} (A. Cayley) The number of labeled trees with
$n$ vertices is $n^{n-2}$.
\end{corollary}

\subsection{Counting trees with conditions}

In a similar way we can count labeled trees with conditions on
vertices. For example, let us compute the number of labeled
trivalent trees with $m$ vertices (i.e. trees 
that have only 1-valent and 3-valent vertices). 
Clearly, $m=2k$, otherwise there is no such trees.
The relevant action functional is 
$$
S(x)=\tfrac{x^2}{2}-g(x+\tfrac{x^3}{6}).
$$ 
Then the critical point $x_0$ is obtained from 
the equation 
$$
x=g(1+\tfrac{x^2}{2}),
$$ which yields
$$
x_0=\frac{1-\sqrt{1-2g^2}}{g}.
$$ 
Thus, by \eqref{deq1} the tree sum equals
$$
-S(x_0)=\int_0^g\left(\tfrac{1-\sqrt{1-2a^2}}{a}+\tfrac{(1-\sqrt{1-2a^2})^3}{6a^3}\right)da=
$$
$$
\frac{2}{3}\int_0^g\tfrac{1-(1+a^2)\sqrt{1-2a^2}}{a^3}da=
\frac{(1-2g^2)^{\frac{3}{2}}-(1-3g^2)}{3g^2}.
$$
Expanding this in a Taylor series, we find 
$$
-S(x_0)=\sum_{n=1}^\infty \frac{1\cdot 3\cdot \, \cdots \,\cdot
  (2n-3)}{(n+1)!}g^{2n}.
$$
Hence, we get 
\begin{corollary}
The number $N_k$ of trivalent labeled trees with $2n$ vertices is
$(2k-3)!!\frac{(2k)!}{(k+1)!}$.
\end{corollary}

For example, $N_1=1$ (a single edge), $N_2=4$ (a single tree with $4!$ labelings modulo a group of order $6$), $N_3=90$ (a single tree with $6!$ labelings modulo a group of order $8$), etc.

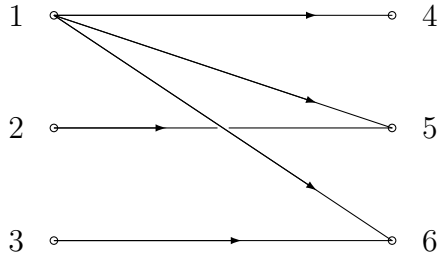
\begin{figure}[bp]
  \begin{center}
    \setlength{\unitlength}{0.5cm}
    \begin{picture}(11,8)(0,2)

      %%% straight lines and end dots
      \put(1,9){\line(1,0){9}}
%%%      \put(1,6){\line(1,0){9}}
      \put(1,6){\line(1,0){4.35}}
      \put(5.65,6){\line(1,0){4.35}}
      \put(1,3){\line(1,0){9}}
      \put(1,9){\circle*{0.2}}
      \put(1,6){\circle*{0.2}}
      \put(1,3){\circle*{0.2}}
      \put(10,9){\circle*{0.2}}
      \put(10,6){\circle*{0.2}}
      \put(10,3){\circle*{0.2}}

      %%% numbers
      \put(0,9){\makebox(0,0)[c]{$1$}}
      \put(0,6){\makebox(0,0)[c]{$2$}}
      \put(0,3){\makebox(0,0)[c]{$3$}}
      \put(11,9){\makebox(0,0)[c]{$4$}}
      \put(11,6){\makebox(0,0)[c]{$5$}}
      \put(11,3){\makebox(0,0)[c]{$6$}}

      %%% vectors
      \put(1,9){\vector(1,0){7}}      
      \put(1,6){\vector(1,0){3}}
      \put(1,3){\vector(1,0){5}}

      %%% angled lines & vectors
      \put(1,9){\line(3,-2){9}}
      \put(1,9){\line(3,-1){9}}
      \put(1,9){\vector(3,-2){7}}
      \put(1,9){\vector(3,-1){7}}

    \end{picture}
    \caption{A labeled oriented tree with 3 sources and 3 sinks.}
    \label{fig:7}
  \end{center}
\end{figure}

\subsection{Counting oriented trees}

Feynman calculus can be used to count not only non-oriented, but
also oriented graphs. For example, suppose we want to count
labeled oriented trees, whose vertices are either sources
or sinks (see Fig.~\ref{fig:7}). In this case, it is easy to see (check it!)
that the relevant integration problem is in two dimensions,
with the action $S=xy-be^x-ae^y$ (the form $xy$ is not positive definite, but this is immaterial since our computations are purely formal). 
So the critical point is found from the equations 
$$
xe^{-y}=a,\ ye^{-x}=b.
$$
Like before, look for a solution $(x,y)=(x_0,y_0)$ in the form 
$$
x=a+\sum_{p\ge 1,q\ge 1} c_{pq}a^pb^q,\ y=b+\sum_{p\ge 1,q\ge 1} d_{pq}a^pb^q.
$$
A calculation with residues similar to the one we did for 
unoriented trees yields
$$
c_{pq}=\frac{1}{(2\pi i)^2}\oint\oint\frac{x}{a^{p+1}b^{q+1}}da\wedge db=
$$
$$
\frac{1}{(2\pi i)^2}\oint\oint \frac{e^{qx+py}}{x^py^{q+1}}(1-xy)dx\wedge dy=
\frac{q^{p-1}p^{q-1}}{(p-1)!q!}.
$$
Similarly, $d_{pq}=\frac{q^{p-1}p^{q-1}}{p!(q-1)!}$. 
Now, similarly to the unoriented case, we find that
$-a\partial_aS(x,y)=x$, $-b\partial_bS(x,y)=y$, so
$$
-S(x,y)=b+\int_0^a\frac{x}{u}du=
a+b+\sum_{p,q\ge 1}\frac{p^{q-1}q^{p-1}}{p!q!}a^pb^q
$$
This implies that the number of labeled trees
with $p$ sources and $q$ sinks is
$p^{q-1}q^{p-1}\frac{(p+q)!}{p!q!}$.
In particular, if we specify which vertices are sources
and which are sinks, the number of labeled trees is $p^{q-1}q^{p-1}$.

\begin{exercise} Do this calculation in detail.
\end{exercise} 

\subsection{The matrix-tree theorem}

These calculations can be generalized to compute the number of {\it colored} labeled trees. 
For this we first need to define the {\it Kirchhoff polynomial}\index{Kirchhoff polynomial} $K_m(\bold u)$. 
Namely, for a collection 
of variables $\bold u:=(u_{ik})$, $1\le i\ne k\le m$, $u_{ik}=u_{ki}$
consider the quadratic form 
$$
U(\bold y):=\sum_{1\le i<k\le m}u_{ik}(y_i-y_k)^2.
$$ 
Generically it has a 1-dimensional kernel spanned by $\bold 1=(1,...,1)$, so it is nondegenerate 
on the subspace defined by the equation $\sum_i y_i=0$. This subspace 
carries a volume form $\omega_0(v_1,...,v_{m-1}):=\omega(v_1,...,v_{m-1},\bold 1)$, 
where $\omega$ is the standard volume form  on $\Bbb R^m$, and with respect to this 
form we have 
$$
K_m(\bold u):=\det U=\det(\delta_{i\ell}\sum_{k\ne \ell} u_{k\ell }-u_{i\ell})_{(j)}
$$
for any $1\le j\le m$, where the subscript $(j)$ means that the $j$-th row and column are removed. 
The polynomial $K_m$ is called the {\it Kirchhoff polynomial}.
For instance, $K_2=u_{12}$, $K_3=u_{12}u_{13}+u_{13}u_{23}+u_{12}u_{23}$, etc. 

Now let $\bold p=(p_1,...,p_m)$ be a $m$-tuple of positive integers 
and $\bold r=(r_{ij},1\le i\le j\le m)$ be a collection of nonnegative integers with 
$|\bold r|=|\bold p|-1$, where $|\bold r|:=\sum_{i\le j}r_{ij}$, $|\bold p|:=\sum_k p_k$.
Suppose vertices of the tree are given colors $1,...,m$, and we want to compute the number $N(\bold p,\bold r)$ of labeled trees with the first $p_1$ vertices colored with $1$, the next $p_2$ with $2$,..., the last $p_m$ with $m$, and 
$r_{ij}$ edges going between vertices 
of color $i$ and vertices of color $j$. 

It suffices to compute the polynomial 
$$
Q_{\bold p}(\bold z):=\sum_{\bold r: |\bold r|=|\bold p|-1} N(\bold p,\bold r)\prod_{i\le j}z_{ij}^{r_{ij}}. 
$$

\begin{theorem}\label{mt} We have 
$$
Q_\bold p(\bold z)=
(p_1...p_m)^{-1}K(p_kz_{k\ell}p_\ell, k\ne \ell)\prod_\ell(\sum_k p_kz_{k\ell})^{p_\ell-1}.
$$
\end{theorem} 

Note that for $m=1$ and $\bold z=1$ this recovers Cayley's theorem, while for $m=2$ 
and $\bold z=\begin{pmatrix} 0& 1\\ 1& 0\end{pmatrix}$ it recovers our count of 
oriented trees. 

\begin{proof} 
We attach to each color $j$ a real variable $x_j$. Then the corresponding action is 
$$
S(x,y)=\frac{1}{2}x^TB x-\sum_{j=1}^m a_je^{x_j},
$$
where $B=(b_{ij})$ is inverse to the matrix $\bold z:=(z_{ij})$ with $z_{ij}=z_{ji}$. 
Then by Theorem \ref{trees1l}, $Q_{\bold p}(\bold z)$ is the coefficient to $\prod_k a_k^{p_k}$ in $-S(x)$, where $x$ is the critical point of $S$. 

The equation for the critical point of $S$ is  
$$
\sum_i x_ib_{ij}e^{-x_j}=a_j.
$$
Let $X_j:=\sum_i x_ib_{ij}$, then 
$x_i=\sum_jz_{ij}X_j$, 
$a_i=X_ie^{-x_i}$, and 
$$
-S(x)=\int X_j\frac{da_j}{a_j}
$$ 
for all $j$. In other words, the coefficient to $\prod_k a_k^{p_k}$ in $-S(x)$ equals
the coefficient to the same monomial in $X_j(\bold z,\bold a)$ divided by $p_j$. 
Thus, denoting by $D_T(\bold z)$ the principal minor of $\bold z$ 
corresponding to a subset $T\subset \lbrace 1,...,m\rbrace$, we get 
$$
Q_{\bold p}(\bold z)=\tfrac{p_j^{-1}}{(2\pi i)^m}\oint X_j(\prod_k a_k^{-p_k-1})d\bold a=
$$
$$
\tfrac{p_j^{-1}}{(2\pi i)^m}\oint X_j(\prod_k (X_ke^{-x_k})^{-p_k-1})d(X_1e^{-x_1})\wedge...\wedge d(X_me^{-x_m})=
$$
$$
\tfrac{p_j^{-1}}{(2\pi i)^m}\oint \sum_{T\subset \lbrace 1,...,m\rbrace}(-1)^{|T|}D_T(\bold z)X_j(\prod_{\ell\notin T}X_\ell^{-1})(\prod_\ell X_\ell^{-p_\ell})e^{\sum_{k,\ell} p_kz_{k\ell}X_\ell}dX_1\wedge...\wedge dX_m
$$
$$
=p_j^{-1}\sum_{T\subset \lbrace 1,...,m\rbrace}(-1)^{|T|}\frac{p_j-1+\delta_{jT^c}}{\sum_k p_kz_{kj}}D_T(\bold z)\prod_\ell\frac{(\sum_k p_kz_{k\ell})^{p_\ell-\delta_{\ell T}}}{(p_\ell-\delta_{\ell T})!}=
$$
\scriptsize
$$
p_j^{-1}
\left(\frac{p_j-1}{\sum_k p_kz_{kj}}\det(\delta_{i\ell}\sum_k p_kz_{k\ell}-z_{i\ell}p_\ell)+
\det(\delta_{i\ell}\sum_k p_kz_{k\ell}-z_{i\ell}p_\ell)_{(j)}\right)\prod_\ell\frac{(\sum_k p_kz_{k\ell})^{p_\ell-1}}{p_\ell!},
$$
\normalsize 
where $\delta_{\ell T}=1$ if $\ell\in T$ and $0$ otherwise. 
The first determinant is zero, so we get 
$$
Q_{\bold p}(\bold z)=(p_1...p_m)^{-1}\det(\delta_{i\ell}\sum_k p_kz_{k\ell}p_\ell-p_iz_{i\ell}p_\ell)_{(j)}\prod_\ell\frac{(\sum_k p_kz_{k\ell})^{p_\ell-1}}{p_\ell!}.
$$
This implies the theorem. 
\end{proof} 

Theorem \ref{mt} is a weighted version of {\it Kirchhoff's matrix-tree theorem},\index{Kirchhoff's matrix-tree theorem}
which is a generalization of Cayley's theorem. More precisely, take 
$\bold z=A_\Gamma$ to be the adjacency matrix of a graph $\Gamma$ (without self-loops), $m$ the number of vertices of $\Gamma$, and $p_i=1$ for all $i$. Then 
$Q_{\bold p}(\bold z)=N_\Gamma$ is the number of spanning trees of 
$\Gamma$, and Theorem \ref{mt} says that 
$$
N_\Gamma=\det U,
$$
where $U$ is the quadratic form 
$$
U(\bold y)=\sum_{i<j}(A_\Gamma)_{ij}(y_i-y_j)^2=(\Delta_\Gamma \bold y,\bold y),
$$ 
where $\Delta_\Gamma=D_\Gamma-A_\Gamma$ is the Laplace operator of $\Gamma$ 
($D_\Gamma$ being the diagonal matrix of vertex degrees). Thus we get 

\begin{corollary} (The matrix-tree theorem)
$$
N_\Gamma=\frac{1}{m}\lambda_1...\lambda_{m-1},
$$
where $\lambda_i$ are the non-zero eigenvalues of $\Delta_\Gamma$. 
\end{corollary} 

Cayley's theorem is obtained from this result when $\Gamma$ is a complete graph, in which case $\lambda_i=m$ for all $i$, so we get $N_\Gamma=m^{m-2}$.

\subsection{1-particle irreducible diagrams and the effective
  action}

Let $Z=Z_S$ be the partition function corresponding to the action
$S$. In the previous subsections we have seen that the ``classical''
(or ``tree'') part 
$(\log\frac{Z_S}{Z_0})_0$ 
of the quantity $\hbar\log\frac{Z_S}{Z_0}$ is quite elementary to
compute -- it is just minus the critical value of the action
$S(x)$. Thus, if we could find a new ``effective'' action 
$S_{\rm eff}$ (a ``deformation'' of $S$) such that  
$$
\hbar^{-1}(\log\tfrac{Z_{\rm S_{\rm eff}}}{Z_0})_0=\log\tfrac{Z_S}{Z_0}
$$
(i.e. the classical answer for the effective action is the quantum
answer for the original one), then we can consider the quantum theory 
for the action $S$ solved. In other words, 
the problem of solving the quantum theory attached to $S$
(i.e. finding the corresponding integrals) essentially reduces 
to the problem of computing the effective action $S_{\rm eff}$. 

We will now give a recipe of computing the effective
action in terms of amplitudes of Feynman diagrams, 
and see that it is computationally easier than computing the sum over connected diagrams. 

\begin{definition}
An edge $e$ of a connected graph $\Gamma$ is said to be a {\it bridge}\index{bridge}
if the graph $\Gamma\setminus\ e$ is disconnected. 
A connected graph without bridges is called {\it 1-particle irreducible}\index{1-particle irreducible diagram} (1PI).\footnote{This is the physical terminology. The mathematical 
term is ``2-connected''.}  
\end{definition}

To compute the effective action, 
we will need to consider graphs with external edges
(but having at least one internal vertex). 
Such a graph $\Gamma$ (with $N$ external edges) 
will be called 1-particle irreducible if so 
is the corresponding ``amputated'' graph (i.e. the graph
obtained from $\Gamma$ by removal of the external edges).
In particular, a graph with one internal vertex is always 1-particle
irreducible, while a single edge graph
without internal vertices is defined {\it not} to be 1-particle irreducible. 
The notions of a bridge and a 1-particle irreducible graph 
are illustrated by Fig.~\ref{fig:8}.

\begin{figure}[htp]
  \begin{center}
    \setlength{\unitlength}{0.5cm}
    \begin{picture}(17,19)(3,1)

      %%% make first figure
      %%% make box
      \put(5,12){\line(1,0){12}}
      \put(5,20){\line(1,0){12}}
      \put(5,12){\line(0,1){8}}
      \put(17,12){\line(0,1){8}}

      %%% make graph
      \put(5.5,17){\line(1,0){5}}
      \put(6.25,17){\circle{1.5}}
      \put(9.75,17){\circle{1.5}}

      %%% make arrows and descriptions
      \put(8,14){\vector(0,1){2.8}}
      \put(8,13){\makebox(0,0)[c]{a bridge}}
      \put(11.5,19){\vector(-2,-3){1}}
      \put(11.75,19){\makebox(0,0)[l]{not a bridge}}

      %%% make secondfigure
      %%% make box
      \put(3,1){\line(1,0){16}}
      \put(3,10){\line(1,0){16}}
      \put(3,1){\line(0,1){9}}
      \put(19,1){\line(0,1){9}}

      %%% make upper graph
      \put(6.75,7.5){\circle{1.5}}
      \put(6.75,6.75){\line(0,1){1.5}}
      \put(6,7.5){\line(-1,0){1}}
      \put(7.5,7.5){\line(1,0){1}}

      %%% make lower graph
      \put(5.5,3.5){\circle{1.5}}
      \put(8,3.5){\circle{1.5}}
      \put(6.25,3.5){\line(1,0){1}}
      \put(4.95,4){\line(-1,1){.75}}
      \put(4.95,3){\line(-1,-1){.75}}
      \put(8.55,4){\line(1,1){.75}}
      \put(8.55,3){\line(1,-1){.75}}

      %%% make descriptions
      \put(10.5,7.5){\makebox(0,0)[l]{\shortstack[l]{1PI graph with \\
                                    two external edges}}}
      \put(10.5,3.5){\makebox(0,0)[l]{\shortstack[l]{non-1PI graph with \\
                                     four external edges}}}

    \end{picture}
    \caption{ }
    \label{fig:8}
  \end{center}
\end{figure}

Denote by $G_{\rm 1PI}(\bold n,N)$ the set of
isomorphism classes of 1-particle irreducible graphs with $N$
external edges and $n_i$ $i$-valent internal vertices for each $i$
(where isomorphisms are not allowed to move external edges).

\begin{theorem} \label{effac}
The effective action $S_{\rm eff}$ is given by the formula 
$$
S_{\rm eff}(x)=
 \frac{B(x,x)}{2}-
\sum_{i\ge 0}\frac{{\mathcal B}_i(x,...,x)}{i!},
$$
where 
$$
{\mathcal B}_N(x, \ldots ,x)=\hbar\sum_{\bold n}\prod_i (g_i\hbar^{\frac{i}{2}-1})^{n_i}
\sum_{\Gamma\in G_{\rm 1PI}(\bold
  n,N)}\frac{\Bbb F_\Gamma(Bx, \ldots ,Bx)}{|{\rm Aut}(\Gamma)|}.
  $$
 \end{theorem}

Thus, $S_{\rm eff}=S+\hbar S_1+\hbar^2 S_2+..$ 
The expressions $\hbar^j S_j$ are called 
the {\it $j$-loop corrections}\index{$j$-loop correction} to the effective action.

This theorem allows physicists to worry only about 1-particle 
irreducible diagrams, and is the reason why you will rarely see 
other diagrams in a QFT textbook. 
As before, it is very useful in doing low order computations, since 
the number of 1-particle irreducible diagrams with a given number of loops 
is much smaller than the number of connected diagrams with the same 
number of loops. 

\begin{proof} The proof is based on the following lemma from graph theory. 

\begin{lemma}\label{obv} Any connected graph $\Gamma$ can be uniquely represented as a tree whose vertices are 1-particle irreducible subgraphs (with external edges), 
and edges are the bridges of $\Gamma$.
\end{lemma}

The lemma is obvious. Namely, let us remove all bridges from $\Gamma$.
Then $\Gamma$ will turn into a disjoint union of 1-particle irreducible graphs 
which should be taken to be the vertices of the said tree. 

The tree corresponding to the graph $\Gamma$ is called the
{\it skeleton}\index{skeleton of a graph} of $\Gamma$
(see Fig.~\ref{fig:9}). 

\begin{figure}[htbp]
  \begin{center}
    \setlength{\unitlength}{0.5cm}
    \begin{picture}(16,8)(2,0)

      %%% make descriptions
      \put(5,7){\makebox(0,0)[c]{Graph:}}
      \put(15,7){\makebox(0,0)[c]{Skeleton:}}

      %%% make first graph
      \put(5,2.5){\circle{1.5}}
      \put(2,4){\circle{1.5}}
      \put(8,1){\circle{1.5}}
      \put(7,5.25){\circle{1.5}}
      \put(4.3,2.85){\line(-2,1){1.65}}
      \put(2,3.25){\line(0,1){1.5}}
%%%      \put(5.7,2.15){\line(2,-1){1.65}}
      \put(5.7,2.15){\line(2,-1){3}}
      \put(5.5,3){\line(2,3){1.1}}

      %%% make second graph
      \put(15,2.5){\circle*{.5}}
      \put(12,4){\circle*{.5}}
      \put(18,1){\circle*{.5}}
      \put(17,5.25){\circle*{.5}}
      \put(15,2.5){\line(-2,1){3}}
%%%      \put(5.7,2.15){\line(2,-1){1.65}}
      \put(15,2.5){\line(2,-1){3}}
      \put(15,2.5){\line(2,3){2}}

    \end{picture}
    \caption{The skeleton of a graph.}
    \label{fig:9}
  \end{center}
\end{figure}
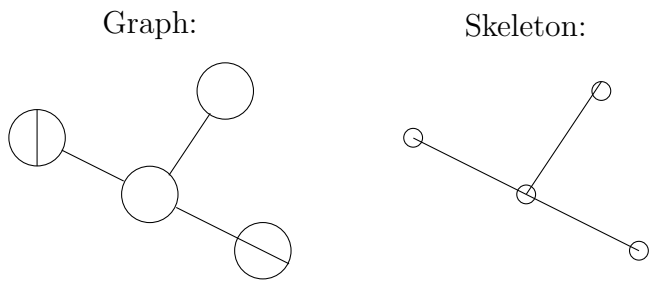

It is easy to see that Lemma \ref{obv} implies Theorem \ref{effac}. Indeed,
it implies that the sum over 
all connected graphs occuring in the expression of
$\log\frac{Z_S}{Z_0}$ can be written as a sum over skeleton trees, 
so that the contribution from each tree is 
(proportional to) the contraction of tensors ${\mathcal B}_i$ put in 
its vertices, and ${\mathcal B}_i$ is the (weighted) sum of amplitudes of 
all 1-particle irreducible graphs with $i$ external edges. 
\end{proof}

\subsection{1-particle irreducible diagrams and the Legendre transform}

Recall the notion of {\it Legendre transform}.\index{Legendre transform}
Let $f$ be a smooth function on a vector space $Y$, such that 
the map $Y\to Y^*$ given by $x\to df(x)$ is a diffeomorphism. 
Then one can define the Legendre transform of $f$ as follows. 
For $p\in Y^*$, let $x_0=x_0(p)$ be the critical point of the function 
$(p,x)-f(x)$ (i.e. the unique solution of the equation $df(x)=p$).
Then the Legendre transform of $f$ is the function on $Y^*$ 
defined by 
$$
L(f)(p)=(p,x_0)-f(x_0).
$$
It is easy to see that the differential of $L(f)$ is also a diffeomorphism
$Y^*\to Y$ (in fact, inverse to $df(x)$), and that $L^2(f)=f$. 

\begin{example} Let $f(x)=\frac{ax^2}{2}$, $a\ne 0$. 
Then $px-f=px-\frac{x^2}{2}$ has a critical point at 
$p=\frac{x}{a}$, and the critical value is $\frac{p^2}{2a}$. Thus $L(\frac{ax^2}{2})=
\frac{p^2}{2a}$. More generally, if $f(x)=\frac{B(x,x)}{2}$ where $B$ is a non-degenerate symmetric form on $Y$ then $L(f)(p)=\frac{B^{-1}(p,p)}{2}$. E.g.,
the Legendre transform of a Lagrangian $\frac{mv^2}{2}-U(x)$ 
of a particle of mass $m$ with respect to velocity $v=\dot{x}$ 
is its Hamiltonian (energy) $\frac{p^2}{2m}+U(x)$, and vice versa. 
This is, in fact, so in complete generality, which is why
Legendre transform plays an important role in classical mechanics and field theory.   
\end{example} 

Note that the stationary phase formula implies that the Legendre transform is the classical analog of the Fourier transform. Indeed, 
the leading term of the asymptotics as $\hbar\to 0$ of the logarithm of the (suitably normalized) Fourier transform 
$\hbar^{-\frac{d}{2}}\int_V 
e^{\frac
{i(-(p,x)+S(x))}
{\hbar}}dx$ of the Feynman density $e^{\frac
{iS(x)}
{\hbar}}dx$
(where the integral is understood in the sense of distributions) is $-\frac{iL(S)(p)}{\hbar}$.

Now let us consider Theorem \ref{effac} in the situation of Theorem \ref{Feyn}.
Thus, $S(x)=\frac{B(x,x)}{2}+O(x^3)$, and we look at 
$$
Z(p)=\hbar^{-\frac{d}{2}}\int_Ve^{\frac{(p,x)-S(x)}{\hbar}}dx.
$$
By Theorem \ref{effac}, 
one has 
$$
\log\frac{Z(p)}{Z_0}=-\hbar^{-1}S_{\rm eff}(x_0,p),
$$
where the effective action $S_{\rm eff}(x,p)$ is the sum over 
$1$-particle irreducible graphs and $x_0=x_0(p)$ is its critical point. 

Now, we must have $S_{\rm eff}(x,p)=-p\cdot x+S_{\rm eff}(x)$, 
since the only 1PI graph which contains 1-valent internal vertices (corresponding
to $p$) is the graph with one edge, connecting an internal vertex with an external one
(so it yields the term $-p\cdot x$, and other graphs contain no $p$-vertices).
This shows that $\hbar\log\frac{Z(p)}{Z_0}$ is the critical value of $p\cdot x-S_{\rm eff}(x)$.
Thus we have proved the following.

\begin{proposition} We have 
$$
S_{\rm eff}(x)=L(\hbar\log\tfrac{Z(p)}{Z_0}),\ \hbar\log\tfrac{Z(p)}{Z_0}=L(S_{\rm eff}(x)).
$$
\end{proposition}

Physicists formulate this result as follows: 
the effective action is the Legendre transform 
of $\hbar$ times the logarithm of the generating function for 
quantum correlators
(and vice versa).

\begin{exercise} Compute the 1-loop contribution to 
$\log\frac{Z}{Z_0}$ 
for 
$$
S(x)=\tfrac{x^2}{2}-g(x+\tfrac{x^3}{6}).
$$ 
Using this, compute the number of labeled n-vertex 1-loop 
graphs with 1-valent and 3-valent vertices only (be careful with double edges and self-loops!). Check your answer by directly enumerating such graphs with small number of vertices.  
\end{exercise}

\begin{exercise} 
Find the exponential generating function 
$\sum_n a_n\frac{z^n}{n!}$ for the numbers $a_n$ of labeled 
n-vertex trees with 1-valent and 4-valent vertices. 
You may express the answer via inverse functions to polynomials.  
\end{exercise} 

\begin{exercise} Find the one-loop contribution to the 
effective action for $S(x)=\frac{x^2}{2}-\frac{gx^3}{6}$. 
That is, one has $S_{\rm eff}=S+\hbar S_1+O(\hbar^2)$, and 
you need to find $S_1$. Which Feynman diagrams need to be considered? 
\end{exercise} 

\begin{exercise} Consider the heat equation $u_t=\frac{1}{2}\Delta_B u$, where 
$\Delta_B$ is the Laplace operator attached to $B$ defined in Subsection \ref{inse}. 
It is solved by the heat flow $u(x,t)=e^{\frac{t\Delta_B}{2}}u(x,0)$. 
Show that the effective action $S_{\rm eff}$ for the action $S(x)=\frac{B(x,x)}{2}-\widetilde S(x)$ can be computed as the sum of contributions of 1PI Feynman diagrams without self-loops 
for the action $S^\circ(x):=\frac{B(x,x)}{2}-\widetilde S^\circ(x)$ where $\widetilde S^\circ(x):=
e^{\frac{\hbar\Delta_B}{2}}\widetilde S(x)$ obtained by transforming $\widetilde S$ by the heat flow for time $\hbar$. 
\end{exercise}

\section{Matrix integrals}

Let ${\frak h}_N$ be the space of Hermitian matrices of size $N$. 
The inner product on ${\frak h}_N$ is given by $B(A_1,A_2)=\Tr(A_1A_2)$.
In this section we will consider integrals of the form
$$
Z_N:=\hbar^{-\frac{N^2}{2}}\int_{{\frak h}_N}e^{-\frac{S(A)}{\hbar}}dA,
$$
where the Lebesgue measure $dA$ is normalized by the condition
$$
\int_{{\frak h}_N} e^{-\frac{{\rm Tr}(A^2)}{2}}dA=1
$$ 
(so we don't have to drag around the $\sqrt{2\pi}$-factors), and 
$$
S(A):=\frac{\Tr(A^2)}{2}-\sum_{m\ge 1} g_m\frac{\Tr(A^m)}{m}
$$
is the action functional.\footnote{Note that 
we divide by $m$ and not by $m!$. We will see below why 
such normalization will be more convenient.}
We will be interested in the behavior of the coefficients of the expansion 
of $Z_N$ in $g_i$ for large $N$. The study of this behavior 
will lead us to considering not simply Feynman graphs, but actually fat (or ribbon)
graphs, which are in fact 2-dimensional surfaces. 
Thus, before we proceed further, we need to do some 
2-dimensional combinatorial topology. 

\subsection{Fat graphs}
Recall from the proof of Feynman's theorem 
that given a finite collection of flowers 
and a matching $\sigma$ on the set $T$ of endpoints of their edges, 
we can obtain a graph $\Gamma_\sigma$ by connecting (or gluing)
the points which fall into the same pair. 

Now, given an $i$-flower, let us inscribe 
it in a closed disk $D$ (so that the ends of the edges 
are on the boundary). 
Then take its small tubular neighborhood in $D$. 
This produces a region with piecewise smooth boundary. We will 
equip this region and its boundary with the standard orientation, and call it 
a {\it fat i-valent flower}.\index{fat flower} The boundary of a fat $i$-valent flower  
has the form $P_1Q_1P_2Q_2 \ldots P_iQ_iP_1$, where 
$P_i,Q_i$ are the angle points, the intervals $P_jQ_j$ are 
arcs on $\partial D$, and $Q_jP_{j+1}$ are (smooth) arcs lying
inside $D$
(see Fig.~\ref{fig:10}).

\begin{figure}[htbp]
  \begin{center}
%%% figure 10
    \setlength{\unitlength}{2ex}
    \begin{picture}(22,12)(2,0)
      %%% make first figure
      \put(4,11.5){\makebox(0,0)[c]{3-valent flower}}
      \put(4,5){\circle*{.5}}
      \put(4,5){\line(1,0){3.6055}}
      \put(4,5){\line(-2,3){2}}
      \put(4,5){\line(-2,-3){2}}
 
      \put(18,11.5){\makebox(0,0)[c]{fat 3-valent flower}}
      \put(18,5){\circle{7.2111}}
      \put(18,5){\circle*{.4}}
      \put(18,5){\line(1,0){3.6055}}
      \put(18,5){\line(-2,3){2}}
      \put(18,5){\line(-2,-3){2}}

      \put(15.1,7.2){\circle*{.3}}
%%%      \put(13.5,7.5){$P_1$}
      \put(14.25,8){\makebox(0,0)[c]{$Q_1$}}

      \put(15.1,2.8){\circle*{.3}}
%%%      \put(13.5,1.8){$Q_3$}
      \put(14.25,2){\makebox(0,0)[c]{$P_2$}}

      \put(17,8.5){\circle*{.3}}
%%%      \put(16.5,9.3){$Q_1$}
      \put(16.75,9.5){\makebox(0,0)[c]{$P_1$}}

      \put(17,1.5){\circle*{.3}}
%%%      \put(16.5,0){$P_3$}
      \put(16.75,0.5){\makebox(0,0)[c]{$Q_2$}}

      \put(21.4,6){\circle*{.3}}
%%%      \put(22,6){$P_2$}
      \put(22.6,6){\makebox(0,0)[c]{$Q_3$}}

      \put(21.4,4){\circle*{.3}}
%%%      \put(22,3.5){$Q_2$}
      \put(22.6,4){\makebox(0,0)[c]{$P_3$}}

      \thicklines
      \qbezier(17,8.5)(18.5,5.5)(21.4,6)
      \qbezier(21.4,4)(18.5,4.5)(17,1.5)
      \qbezier(15.1,2.8)(17.5,5)(15.1,7.2)

%      \qbezier(15.1,2.8)(16.05,2.15)(17,1.5)
      \qbezier(15.1,2.8)(15.85,1.85)(17,1.5)

%      \qbezier(15.1,7.2)(16.05,7.85)(17,8.5)
      \qbezier(15.1,7.2)(15.85,8.15)(17,8.5)

%      \qbezier(21.4,6)(21.4,5)(21.4,4)
      \qbezier(21.4,6)(21.7,5)(21.4,4)

    \end{picture}

    \caption{ }
    \label{fig:10}
  \end{center}
\end{figure}

Now, given a collection of usual flowers and a matching $\sigma$ as above,
we can consider the corresponding fat flowers, and glue them, respecting 
the orientation, along intervals $P_jQ_j$ according to $\sigma$. This will produce a compact
oriented surface with boundary (the boundary is glued from intervals $Q_{j}P_{j+1}$). 
We will denote this surface by $\widetilde \Gamma_{\sigma}$, and call it the {\it 
fattening}\index{fattening of a graph} of $\Gamma$ with respect to $\sigma$. 
A fattening of a graph will be called a {\it fat (or ribbon) graph}.\index{fat graph}\index{ribbon graph}

Thus, a fat graph is not just an oriented
 surface with boundary, but 
such a surface together with a partition of this surface into fat flowers. 

Note that the same graph $\Gamma$ can have many 
fattenings which are non-homeomorphic 
(albeit homotopy equivalent) surfaces, and in particular the genus ${\rm g}$ of the
fattening is {\it not} determined by $\Gamma$ (see Fig.~\ref{fig:11}). 

\begin{figure}[htbp]
  \begin{center}
    
    %%% figure 11
    \setlength{\unitlength}{1ex}
    \begin{picture}(60,40)

      %%% TOP LEFT
      \qbezier(7,28)(.5,28)(.5,34)
      \qbezier(7,40)(.5,40)(.5,34)
      \qbezier(7,38)(3,38)(3,35)
      \qbezier(7,30)(3,30)(3,33)
      \put(3,35){\line(1,0){4}}
      \put(3,33){\line(1,0){4}}
      \put(7,38){\line(0,1){2}}
      \put(7,33){\line(0,1){2}}
      \put(7,28){\line(0,1){2}}

      \put(7,40){\line(1,0){12}}
      \put(7,38){\line(1,0){12}}
      \put(7,35){\line(1,0){12}}
      \put(7,33){\line(1,0){12}}
      \put(7,30){\line(1,0){12}}
      \put(7,28){\line(1,0){12}}
      
      \qbezier(19,28)(25.5,28)(25.5,34)
      \qbezier(19,40)(25.5,40)(25.5,34)
      \qbezier(19,38)(23,38)(23,35)
      \qbezier(19,30)(23,30)(23,33)
      \put(19,35){\line(1,0){4}}
      \put(19,33){\line(1,0){4}}
      \put(19,38){\line(0,1){2}}
      \put(19,33){\line(0,1){2}}
      \put(19,28){\line(0,1){2}}
      
%%%      \put(8,25){$\Gamma_1$}
      \put(8,25){\makebox(0,0)[c]{$\Gamma_1$}}
%%%      \put(20,25){$g=0$}
      \put(20,25){\makebox(0,0)[c]{${\rm g}=0$}}

      %%% TOP RIGHT
      \qbezier(41,28)(34.5,28)(34.5,34)
      \qbezier(41,40)(34.5,40)(34.5,34)
      \qbezier(41,38)(37,38)(37,35)
      \qbezier(41,30)(37,30)(37,33)
      \put(37,35){\line(1,0){4}}
      \put(37,33){\line(1,0){4}}
      \put(41,38){\line(0,1){2}}
      \put(41,33){\line(0,1){2}}
      \put(41,28){\line(0,1){2}}
      
      \cbezier(41,40)(45.5,40)(49.5,35)(53,35)
      \cbezier(41,38)(44.5,38)(48.5,33)(53,33)
      \cbezier(41,35)(45.5,35)(49.5,30)(53,30)
      \cbezier(41,33)(44.5,33)(48.5,28)(53,28)
      \qbezier(53,40)(50,40)(48,37)
      \qbezier(53,38)(51,38)(49.6,35.7)
      \put(45.8,34.1){\line(2,3){.75}}
      \put(47.3,32.7){\line(2,3){.75}}
      \qbezier(41,28)(44,28)(46,31)
      \qbezier(41,30)(43,30)(44.4,32.3)
      
      \qbezier(53,28)(59.5,28)(59.5,34)
      \qbezier(53,40)(59.5,40)(59.5,34)
      \qbezier(53,38)(57,38)(57,35)
      \qbezier(53,30)(57,30)(57,33)
      \put(53,35){\line(1,0){4}}
      \put(53,33){\line(1,0){4}}
      \put(53,38){\line(0,1){2}}
      \put(53,33){\line(0,1){2}}
      \put(53,28){\line(0,1){2}}
      
%%%      \put(42,25){$\Gamma_2$}
      \put(42,25){\makebox(0,0)[c]{$\Gamma_2$}}
%%%      \put(54,25){$g=0$}
      \put(54,25){\makebox(0,0)[c]{${\rm g}=0$}}

      %%% BOTTOM LEFT
      \qbezier(7,4)(.5,4)(.5,10)
      \qbezier(7,16)(.5,16)(.5,10)
      \qbezier(7,14)(3,14)(3,11)
      \qbezier(7,6)(3,6)(3,9)
      \put(3,11){\line(1,0){4}}
      \put(3,9){\line(1,0){4}}
      \put(7,14){\line(0,1){2}}
      \put(7,9){\line(0,1){2}}
      \put(7,4){\line(0,1){2}}
      
      \put(7,4){\line(1,0){12}}
      \put(7,6){\line(1,0){12}}
      \cbezier(7,16)(11.5,16)(15.5,11)(19,11)
      \cbezier(7,14)(10.5,14)(14.5,9)(19,9)
      \qbezier(7,11)(10,11)(11.5,12.5)
      \qbezier(7,9)(11,9)(13,11)
      \qbezier(19,16)(15,16)(13,14)
      \qbezier(19,14)(16,14)(14.5,12.5)
      
      \qbezier(19,4)(25.5,4)(25.5,10)
      \qbezier(19,16)(25.5,16)(25.5,10)
      \qbezier(19,14)(23,14)(23,11)
      \qbezier(19,6)(23,6)(23,9)
      \put(19,11){\line(1,0){4}}
      \put(19,9){\line(1,0){4}}
      \put(19,14){\line(0,1){2}}
      \put(19,9){\line(0,1){2}}
      \put(19,4){\line(0,1){2}}

%%%      \put(8,1){$\Gamma_3$}
      \put(8,1){\makebox(0,0)[c]{$\Gamma_3$}}
%%%      \put(20,1){$g=1$}
      \put(20,1){\makebox(0,0)[c]{${\rm g}=1$}}

    \end{picture}
    
    \caption{Gluing a fat graph from fat flowers}
    
    \label{fig:11}
  \end{center}
\end{figure}

\subsection{Matrix integrals in large $N$ limit, planar graphs, and the genus expansion}

Let us now return to the study of the integral $Z_N$. 
We have 
$$
B_m(A,...,A)=(m-1)!\Tr(A^m).
$$ 
Thus 
by Feynman's theorem, 
$$
\log Z_N=\sum_{\bold n}\prod_i \frac{(g_i\hbar^{\frac{i}{2}-1})^{n_i}}{i!^{n_i}n_i!}
\sum_{\sigma\in \Pi_c(T_{\bold n})}\Bbb F(\sigma),
$$
where the summation is taken over the set $\Pi_c(T_{\bold n})$ of all matchings 
of $T=T_{\bold n}$ that produce a connected graph $\Gamma_\sigma$,
and $\Bbb F(\sigma)$ denotes the contraction of 
the tensors $(m-1)!\Tr(A^m)$ using $\sigma$. 
So let us compute $\Bbb F(\sigma)$. 

Let $\lbrace e_i\rbrace$ be the standard basis of $\Bbb C^N$, and 
$\lbrace e_i^*\rbrace$ the dual basis of the dual space.
Then the tensor $\Tr(A^m)$ can be written as 
$$
\Tr(A^m)=\sum_{i_1, \ldots ,i_m=1}^N
(e_{i_1}\otimes e_{i_2}^*\otimes e_{i_2}\otimes e_{i_3}^*\otimes \cdots
\otimes e_{i_m}\otimes e_{i_1}^*,A^{\otimes m}).
$$
Thus 
$$
B_m=\sum_{s\in S_{m-1}}\sum_{i_1, \ldots ,i_m=1}^N
s(e_{i_1}\otimes e_{i_2}^*\otimes e_{i_2}\otimes e_{i_3}^*\otimes \cdots
\otimes e_{i_m}\otimes e_{i_1}^*)
$$
(sum over all possible cyclic orderings of edges of an $m$-valent flower). 
Hence
$$
\Bbb F(\sigma)=\sum_{s\in \prod_i S_{i-1}^{n_i}} \widetilde {\Bbb F}(s\sigma),
$$
where $\widetilde {\Bbb F}(\sigma)$ is obtained by contracting the 
tensors 
\begin{equation}\label{tenso}
\sum_{i_1,...,i_m=1}^Ne_{i_1}\otimes e_{i_2}^*\otimes e_{i_2}\otimes e_{i_3}^*\otimes \cdots
\otimes e_{i_m}\otimes e_{i_1}^*
\end{equation} 
according to the fat graph $\widetilde\Gamma_{\sigma}$. 
It follows that 
$$
\log Z_N=
\sum_{\bold n}\prod_i \frac{g_i^{n_i}\hbar^{n_i(\frac{i}{2}-1)}}{i!^{n_i}n_i!}
\sum_{\sigma\in \Pi(T_{\bold n})}\sum_{s\in \prod_i S_{i-1}^{n_i}}\widetilde {\Bbb F}(s\sigma)=
$$
$$
\sum_{\bold n}\prod_i \frac{g_i^{n_i}\hbar^{n_i(\frac{i}{2}-1)}}{i^{n_i}n_i!}
\sum_{\sigma}\widetilde {\Bbb F}(\sigma)
$$
(the product $\prod_i i!^{n_i}$ in the denominator got replaced by $\prod_i i^{n_i}$ since in the sum $\sum_{s,\sigma}\widetilde {\Bbb F}(s\sigma)$ every term 
$\widetilde {\Bbb F}(\sigma)$ occurs $|\prod_i S_{i-1}^{n_i}|=\prod_i (i-1)!^{n_i}$ times). 

For a surface $\Sigma$ with boundary, let $\nu(\Sigma)$ denote the number 
of connected components of the boundary. 

\begin{proposition}\label{fatt}
$\widetilde {\Bbb F}(\sigma)=N^{\nu(\widetilde\Gamma_{\sigma})}$.
\end{proposition}

\begin{proof}
One can visualize each summand in the sum \eqref{tenso} as a labeling of 
the angle points $P_1,Q_1, \ldots ,P_m,Q_m$ on the boundary of a fat 
$m$-valent flower by $i_1,i_2,i_2,i_3, \ldots ,i_m,i_1$.  
Now, the contraction using $\sigma$ of some set of such monomials 
is nonzero iff the subscript is constant along each boundary component
of $\widetilde\Gamma_{\sigma}$ (see Fig. 12). This implies the result. 
\end{proof}

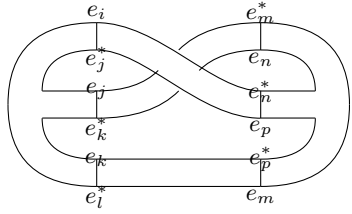
\begin{figure}[htbp]
  \begin{center}
    
    %%% figure 12
    \setlength{\unitlength}{1ex}
    \begin{picture}(26,16)(5,-1)

      \qbezier(7,0)(.5,0)(.5,6)
      \qbezier(7,12)(.5,12)(.5,6)
      \qbezier(7,10)(3,10)(3,7)
      \qbezier(7,2)(3,2)(3,5)
      \put(3,7){\line(1,0){4}}
      \put(3,5){\line(1,0){4}}
      \put(7,10){\line(0,1){2}}
      \put(7,5){\line(0,1){2}}
      \put(7,0){\line(0,1){2}}
      
      \put(7,0){\line(1,0){12}}
      \put(7,2){\line(1,0){12}}
      \cbezier(7,12)(11.5,12)(15.5,7)(19,7)
      \cbezier(7,10)(10.5,10)(14.5,5)(19,5)
      \qbezier(7,7)(10,7)(11.5,8.5)
      \qbezier(7,5)(11,5)(13,7)
      \qbezier(19,12)(15,12)(13,10)
      \qbezier(19,10)(16,10)(14.5,8.5)

      \qbezier(19,0)(25.5,0)(25.5,6)
      \qbezier(19,12)(25.5,12)(25.5,6)
      \qbezier(19,10)(23,10)(23,7)
      \qbezier(19,2)(23,2)(23,5)
      \put(19,7){\line(1,0){4}}
      \put(19,5){\line(1,0){4}}
      \put(19,10){\line(0,1){2}}
      \put(19,5){\line(0,1){2}}
      \put(19,0){\line(0,1){2}}

%%%      \put(7,13){$e_i$}
      \put(7,12.75){\makebox(0,0)[c]{\scriptsize $e_i$}}
%%%      \put(7,9){$e_j^*$}
      \put(7,9.3){\makebox(0,0)[c]{\scriptsize $e_j^*$}}
%%%      \put(7,7){$e_j$}
      \put(7,7){\makebox(0,0)[c]{\scriptsize $e_j$}}
%%%      \put(7,4){$e_k^*$}
      \put(7,4.2){\makebox(0,0)[c]{\scriptsize $e_k^*$}}
%%%      \put(7,2){$e_k$}
      \put(7,2){\makebox(0,0)[c]{\scriptsize $e_k$}}
%%%      \put(7,0){$e_l^*$}
      \put(7,-0.75){\makebox(0,0)[c]{\scriptsize $e_l^*$}}

%%%      \put(19,13){\hbox to 0pt{\hss{$e_m^*$}}}
      \put(19,12.75){\makebox(0,0)[c]{\scriptsize $e_m^*$}}
%%%      \put(19,9){\hbox to 0pt{\hss{$e_n$}}}
      \put(19,9.3){\makebox(0,0)[c]{\scriptsize $e_n$}}
%%%      \put(19,7){\hbox to 0pt{\hss{$e_n^*$}}}
      \put(19,7){\makebox(0,0)[c]{\scriptsize $e_n^*$}}
%%%      \put(19,4){\hbox to 0pt{\hss{$e_n^*$}}}
      \put(19,4.2){\makebox(0,0)[c]{\scriptsize $e_p$}}
%%%      \put(19,2){\hbox to 0pt{\hss{$e_n^*$}}}
      \put(19,2){\makebox(0,0)[c]{\scriptsize $e_p^*$}}
%%%      \put(19,0){\hbox to 0pt{\hss{$e_n^*$}}}
      \put(19,-0.75){\makebox(0,0)[c]{\scriptsize $e_m$}}

      \put(35,6){\makebox(0,0)[l]{\shortstack[l]{Contraction 
            nonzero iff \\ 
            $i = r$, $j = p$, $j = m$, $k = r$, \\
            $k = p$, $i = m$, \\
            that is \\
            $i = r = k = p = j = m$.}}}
    \end{picture}

    \caption{Contraction defined by a fat graph.}
    
    \label{fig:12}
  \end{center}
\end{figure}

Let $\widetilde G_c(\bold n)$ be the set of isomorphism classes of connected  
fat graphs with $n_i$ $i$-valent vertices for $i\ge 1$. For $\widetilde\Gamma\in \widetilde G_c(\bold n)$,
let $b(\widetilde\Gamma)$ be the number of edges minus the number of 
vertices of the underlying usual graph $\Gamma$. 

\begin{corollary}\label{fatt1}
$$
\log Z_N=\sum_{\bold n}\prod_i (g_i\hbar^{\frac{i}{2}-1})^{n_i}
\sum_{\widetilde\Gamma\in \widetilde G_c(\bold n)}\frac{N^{\nu(\widetilde\Gamma)}}
{|{\rm Aut}(\widetilde\Gamma)|}=
$$
$$
\sum_{\bold n}\prod_i g_i^{n_i}
\sum_{\widetilde\Gamma\in \widetilde G_c(\bold n)}\frac{N^{\nu(\widetilde\Gamma)}\hbar^{b(\widetilde\Gamma)}}
{|{\rm Aut}(\widetilde\Gamma)|}.
$$
\end{corollary}

\begin{proof} 
Let $\Bbb G_{\bold n}^{\rm cyc}:
=\prod_i (S_{n_i}\ltimes (\Bbb Z/i\Bbb Z)^{n_i})$.
This group acts on $T_{\bold n}$, so that 
$\widetilde\Gamma_{\sigma}=\widetilde\Gamma_{g\sigma}$, 
for any $g\in \Bbb G_{\bold n}^{\rm cyc}$. 
Moreover, the group acts transitively on the set of $\sigma$ giving a fixed fat graph
$\widetilde\Gamma_{\sigma}$, and 
the stabilizer of any $\sigma$ is 
${\rm Aut}(\widetilde\Gamma_{\sigma})$. This implies the result, as $|\Bbb G_{\bold n}^{\rm cyc}|=\prod_i i^{n_i}n_i!$ 
which cancels the denominators.
\end{proof}

Now for any compact connected surface $\Sigma$ with boundary, let ${\rm g}(\Sigma)$ be the 
genus of $\Sigma$. Then for a connected fat graph 
$\widetilde\Gamma$, 
$$
b(\widetilde\Gamma)=
2{\rm g}(\widetilde\Gamma)-2+\nu(\widetilde\Gamma)
$$ 
(minus the Euler characteristic).
Thus, defining 
$$
\widehat Z_N(\hbar):=Z_N(\tfrac{\hbar}{N}),
$$
we obtain

\begin{theorem}\label{lnZ}
$$
\log\widehat Z_N=
\sum_{\bold n}\prod_i (g_i\hbar^{\frac{i}{2}-1})^{n_i}
\sum_{\widetilde\Gamma\in \widetilde G_c(\bold n)}
\frac{N^{2-2{\rm g}(\widetilde\Gamma)}
}{|{\rm Aut}(\widetilde\Gamma)|}.
$$
\end{theorem}

This implies the following important result, due to t'Hooft. 

\begin{theorem} \label{tHooft} (1) There exists a limit 
$W_\infty:=\lim_{N\to \infty}\frac{\log\widehat Z_N}{N^2}$. This 
limit is given by the formula 
$$
W_\infty=
\sum_{\bold n}\prod_i (g_i\hbar^{\frac{i}{2}-1})^{n_i}
\sum_{\widetilde\Gamma\in \widetilde G_c(\bold n)[0]}
\frac{1}{|{\rm Aut}(\widetilde\Gamma)|},
$$
where $\widetilde G_c(\bold n)[0]$ denotes the set of {\bf planar}
connected fat graphs, i.e. those which have genus zero. 

(2) Moreover, there exists an expansion
$$
\frac{\log\widehat Z_N}{N^2}=\sum_{{\rm g}\in \Bbb Z_{\ge  0}}a_{\rm g}N^{-2\rm g},
$$
where 
$$
a_{\rm g}=\sum_{\bold n}\prod_i (g_i\hbar^{\frac{i}{2}-1})^{n_i}
\sum_{\widetilde\Gamma\in \widetilde G_c(\bold n)[{\rm g}]}
\frac{1}{|{\rm Aut}(\widetilde\Gamma)|},
$$
and 
$\widetilde G_c(\bold n)[{\rm g}]$ denotes the set of 
connected fat graphs of genus ${\rm g}$.
\end{theorem}

\begin{remark} Genus zero fat graphs are said to be planar 
because the underlying usual graphs can be put on the 2-sphere (and hence on the plane)
without self-intersections. 
\end{remark}

\begin{remark} t'Hooft's theorem may be interpreted
in terms of the usual Feynman diagram expansion. 
Namely, it implies that for large $N$, 
the leading contribution to $\log Z_N(\frac{\hbar}{N})$ 
comes from the terms in the Feynman diagram expansion 
corresponding to planar graphs (i.e. those that admit 
an embedding into the 2-sphere).
\end{remark} 

\subsection{Integration over real symmetric matrices}

One may also consider the matrix integral over the space 
${\frak s}_N$ of real symmetric matrices of size $N$. 
Namely, one puts
$$
Z_N=\hbar^{-\frac{N(N+1)}{4}}\int_{{\frak s}_N} e^{-\frac{S(A)}{\hbar}}dA,
$$ 
where $S$ and $dA$ are as above. Let us generalize 
Theorem \ref{tHooft} to this case. 

As before, consideration of the large $N$ limit leads 
to consideration of fat flowers and gluing of them. 
However, the exact nature of gluing 
is now somewhat different. 
Namely, in the Hermitian case we had 
$(e_i\otimes e_j^*,e_k\otimes e_l^*)=
\delta_{il}\delta_{jk}$, which forced us to glue 
fat flowers preserving orientation. 
On the other hand, 
in the real symmetric case $e_i^*=e_i$, and 
the inner product of the functionals 
$e_i\otimes e_j$ on the space of symmetric matrices
is given by 
$(e_i\otimes e_j,e_k\otimes e_l)=
\delta_{ik}\delta_{jl}+\delta_{il}\delta_{jk}$.
This means that besides the usual (orientation preserving) gluing of fat flowers,
we now must allow gluing with a twist of the ribbon by $180^\circ$. 
Fat graphs thus obtained will be called {\it twisted fat graphs}.\index{twisted fat graph}
That means, a twisted fat graph is a surface with boundary (possibly 
not orientable), together with a partition into fat flowers, and 
orientations on each of them (which may or may not match 
at the cuts, see Fig.13). 

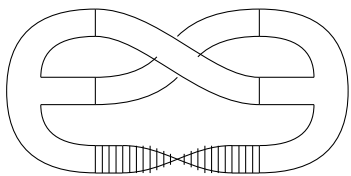
\begin{figure}[htbp]
  \begin{center}
    
    \setlength{\unitlength}{1ex}
    \begin{picture}(25.5,16)(0,3)

      \qbezier(7,4)(.5,4)(.5,10)
      \qbezier(7,16)(.5,16)(.5,10)
      \qbezier(7,14)(3,14)(3,11)
      \qbezier(7,6)(3,6)(3,9)
      \put(3,11){\line(1,0){4}}
      \put(3,9){\line(1,0){4}}
      \put(7,14){\line(0,1){2}}
      \put(7,9){\line(0,1){2}}
      \put(7,4){\line(0,1){2}}
      
%      \put(7,4){\line(1,0){12}}
%      \put(7,6){\line(1,0){12}}
      \put(7,4){\line(1,0){2}}
      \put(7,6){\line(1,0){2}}
      \put(17,4){\line(1,0){2}}
      \put(17,6){\line(1,0){2}}
      \cbezier(9,4)(12,4)(14,6)(17,6)
      \cbezier(9,6)(12,6)(14,4)(17,4)

%%% shading
%%% left side
      \put(7.5,4){\line(0,1){2}}
      \put(8,4){\line(0,1){2}}
      \put(8.5,4){\line(0,1){2}}
      \put(9,4){\line(0,1){2}}
      \put(9.5,4){\line(0,1){2}}
      \put(10,4){\line(0,1){2}}
      \put(10.5,4){\line(0,1){2}}
      \put(11,4){\line(0,1){2}}
      \put(11.5,4.2){\line(0,1){1.6}}
      \put(12,4.4){\line(0,1){1.2}}
      \put(12.5,4.6){\line(0,1){.8}}
%%% right side
      \put(18.5,4){\line(0,1){2}}
      \put(18,4){\line(0,1){2}}
      \put(17.5,4){\line(0,1){2}}
      \put(17,4){\line(0,1){2}}
      \put(16.5,4){\line(0,1){2}}
      \put(16,4){\line(0,1){2}}
      \put(15.5,4){\line(0,1){2}}
      \put(15,4){\line(0,1){2}}
      \put(14.5,4.2){\line(0,1){1.6}}
      \put(14,4.4){\line(0,1){1.2}}
      \put(13.5,4.6){\line(0,1){.8}}

      \cbezier(7,16)(11.5,16)(15.5,11)(19,11)
      \cbezier(7,14)(10.5,14)(14.5,9)(19,9)
      \qbezier(7,11)(10,11)(11.5,12.5)
      \qbezier(7,9)(11,9)(13,11)
      \qbezier(19,16)(15,16)(13,14)
      \qbezier(19,14)(16,14)(14.5,12.5)
      
      \qbezier(19,4)(25.5,4)(25.5,10)
      \qbezier(19,16)(25.5,16)(25.5,10)
      \qbezier(19,14)(23,14)(23,11)
      \qbezier(19,6)(23,6)(23,9)
      \put(19,11){\line(1,0){4}}
      \put(19,9){\line(1,0){4}}
      \put(19,14){\line(0,1){2}}
      \put(19,9){\line(0,1){2}}
      \put(19,4){\line(0,1){2}}

%%%      \put(8,1){\makebox(0,0)[c]{$\Gamma_3$}}
%%%      \put(20,1){\makebox(0,0)[c]{$g=1$}}

    \end{picture}
    
    \caption{Twisted fat graph}
    
    \label{fig:13}
  \end{center}
\end{figure}

Now one can show analogously to the 
Hermitian case that the $\frac{1}{N}$ expansion of $\log \widehat Z_N$ 
(where $\widehat Z_N:=Z_N(\frac{2\hbar}{N}))$ is given by the same formula 
as before, but with summation over 
the set $\widetilde G_c^{\rm tw}(\bold n)$ of twisted 
fat graphs:

\begin{theorem}\label{lnZ1}
$$
\log\widehat Z_N=
\sum_{\bold n}\prod_i (g_i\hbar^{\frac{i}{2}-1})^{n_i}
\sum_{\widetilde\Gamma\in \widetilde G_c^{\rm tw}(\bold n)}
\frac{N^{2-2{\rm g}(\widetilde\Gamma)}
}{|{\rm Aut}(\widetilde\Gamma)|}.
$$
\end{theorem}

Here the genus ${\rm g}$ of a (possibly non-orientable) surface is defined for closed surfaces by 
${\rm g}:=1-\frac{\chi}{2}$, where $\chi$ is the Euler characteristic. 
Thus the genus of $\Bbb R\Bbb P^2$ is $\frac{1}{2}$, the genus of the Klein bottle is $1$, and 
so on. 

In particular, we have the following analog of t'Hooft's theorem.

\begin{theorem} \label{tHooft1} (1) There exists a limit 
$W_\infty:=\lim_{N\to \infty}\frac{\log\widehat Z_N}{N^2}$. This 
limit is given by the formula 
$$
W_\infty=
\sum_{\bold n}\prod_i (g_i\hbar^{\frac{i}{2}-1})^{n_i}
\sum_{\widetilde\Gamma\in \widetilde G_c^{\rm tw}(\bold n)[0]}
\frac{1}{|{\rm Aut}(\widetilde\Gamma)|},
$$
where $\widetilde G_c^{\rm tw}(\bold n)[0]$ denotes the set of {\bf planar}
connected twisted fat graphs, i.e. those which have genus zero. 

(2) Moreover, there exists an expansion
$$
\frac{\log\widehat Z_N}{N^2}=\sum_{{\rm g}\in \frac{1}{2}\Bbb Z_{\ge 0}}a_{\rm g}N^{-2\rm g},
$$
where 
$$
a_{\rm g}=\sum_{\bold n}\prod_i (g_i\hbar^{\frac{i}{2}-1})^{n_i}
\sum_{\widetilde\Gamma\in \widetilde G_c^{\rm tw}(\bold n)[{\rm g}]}
\frac{1}{|{\rm Aut}(\widetilde\Gamma)|},
$$
and 
$\widetilde G_c^{\rm tw}(\bold n)[{\rm g}]$ denotes the set of 
connected twisted fat graphs which have genus ${\rm g}$.
\end{theorem}

\begin{exercise} Consider the matrix integral over the  
space ${\frak q}_N$  of quaternionic Hermitian matrices of size $N$. 
Show that in this case the results are the same as in the real case, except
that each twisted fat graph counts with a sign equal 
to $(-1)^\nu$, where $\nu$ is the number of boundary components. In other words, 
$\log Z_N^{\rm quat}(\hbar)$ equals
$\log Z_{2N}^{\rm real}(\hbar)$ with $N$ replaced by $-N$. 

{\bf Hint:} Use that the quaternionic unitary group $U(N,\Bbb H)$ is a real form of $Sp(2N)$, 
and ${\frak q}_N$ is a real form of the representation 
of $\Lambda^2V$, where $V$ is the standard (vector) representation
of $Sp(2N)$. Compare to the case of real symmetric matrices,
where the relevant representation is $S^2V$ for $O(N)$, and the case of complex 
Hermitian matrices, where it is $V\otimes V^*$ for $GL(N)$. 
\end{exercise} 

\subsection{The number of ways to glue a surface from a polygon and the Wigner semicircle law}

Matrix integrals are so rich that even the simplest 
possible example reduces to a nontrivial counting problem. 
Namely, consider the matrix integral $Z_N$ over 
complex Hermitian matrices with $\hbar=1$ in the case 
$S(A)=\frac{\Tr(A^2)}{2}-s\frac{\Tr(A^{2m})}{2m}$, where $s^2=0$ 
(i.e. we work over the ring $\Bbb C[s]/(s^2)$). 
Then from Theorem \ref{tHooft} we get 
$$
\int_{{\frak h}_N}\Tr(A^{2m})e^{-\frac{\Tr(A^2)}{2}}dA
=P_m(N),
$$
where $P_m(N)$ is a polynomial given by the formula 
$$
P_m(N)=\sum_{{\rm g}\ge 0}\varepsilon_{\rm g}(m)N^{m+1-2\rm g},
$$ 
and $\varepsilon_{\rm g}(m)$ is the number of ways to glue 
a surface of genus ${\rm g}$ from a $2m$-gon with labeled sides, i.e., to match 
the sides and then glue the matching ones to each other in an orientation-preserving manner.
Indeed, in this case we have only one fat flower of valency $2m$, 
which has to be glued with itself; so a direct application of our Feynman rules
leads to counting ways to glue a surface of a given genus from a polygon.

The value of this integral is given by the following 
non-trivial theorem. 

\begin{theorem} (Harer-Zagier, \cite{HZ} 1986) \label{howmany}
$$
P_m(x)=\frac{(2m)!}{2^mm!}\sum_{p=0}^m\begin{pmatrix}m\\ p\end{pmatrix}
%%%2^p\frac{x(x-1)...(x-p)}{(p+1)!}.
2^p\frac{x(x-1) \ldots (x-p)}{(p+1)!}.
$$
\end{theorem}

The theorem is proved in the next subsections. 

Looking at the leading coefficient of $P_m$, we get

\begin{corollary}\label{Catalan}
The number of ways to glue a sphere from a $2m$-gon is 
the Catalan number $C_m=\frac{(2m)!}{m!(m+1)!}=\frac{1}{m+1}\binom{2m}{m}$.
\end{corollary}

Corollary \ref{Catalan} actually has another (elementary
combinatorial) proof, which is as follows. For each matching
$\sigma$ on the set of sides of the $2m$-gon, let us connect the
midpoints of the matched sides by straight lines (Fig.14). It is
geometrically evident that if these lines don't intersect then
the gluing will give a sphere. We claim that the converse is true
as well. Indeed, assume the contrary, i.e. that for cyclically ordered edges $a,b,c,d$, the edge
$a$ connects to $c$ and $b$ to $d$. Then it is easy to see that gluing these two pairs of edges 
gives a torus with a hole (or without if $m=2$). But an (open) torus with a hole can't be embedded into a sphere (e.g. it contains a copy of $K_5$), contradiction. 

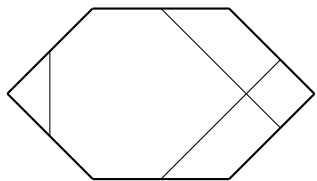
\begin{figure}[htbp]
  \begin{center}
    
    \setlength{\unitlength}{2.5ex}
    \begin{picture}(10,5)

      \thicklines
      \put(3,0){\line(1,0){4}}
      \put(3,0){\line(-1,1){2.5}}
      \put(0.5,2.5){\line(1,1){2.5}}
      \put(3,5){\line(1,0){4}}
      \put(7,0){\line(1,1){2.5}}
      \put(7,5){\line(1,-1){2.5}}

      \thinlines
      \put(1.75,1.25){\line(0,1){2.5}}
      \put(5,0){\line(1,1){3.5}}
      \put(5,5){\line(1,-1){3.5}}

    \end{picture}
    \caption{Matching of sides of a $6$-gon.}
    \label{fig:14}
  \end{center}
\end{figure}

Now it remains to count the number of ways to connect midpoints
of sides with lines without intersections. Suppose we draw one
such line, such that the number of sides on the left of it is
$2k$ and on the right is $2l$ (so that $k+l=m-1$). Then we face
the problem of connecting the two sets of $2k$ and $2l$ sides
without intersections. This shows that the number of gluings
$D_m$ satisfies the recursion
$$
D_m=\sum_{k+l=m-1}D_kD_l,\ D_0=1.
$$
%%%In other words, the generating function $\sum D_mx^m=1+x+...$ satisfies 
In other words, the generating function 
$$
h(x):=\sum_m D_mx^m=1+x+ \cdots
$$ 
satisfies the equation $h(x)-1=xh(x)^2$. This implies that 
$$
h(x)=\frac{1-\sqrt{1-4x}}{2x},
$$ 
which yields that $D_m=C_m$.  
We are done. 

Corollary \ref{Catalan} can be used to derive 
the following 
fundamental result from the theory 
of random matrices, discovered by Wigner in 1955. 

\begin{theorem} (Wigner's semicircle law)
Let $f$ be a continuous function on $\Bbb R$ 
of at most polynomial growth at infinity. 
Then 
$$
\lim_{N\to \infty} \frac{1}{N}\int_{{\frak h}_N}{\rm }{\rm Tr} f(\tfrac{A}{\sqrt{N}})e^{-\frac{{\rm Tr}(A^2)}{2}}=
\frac{1}{2\pi}\int_{-2}^2f(x)\sqrt{4-x^2}dx.
$$
\end{theorem}

This theorem is called the semicircle law because 
it says that the graph of the density of eigenvalues 
of a large random Hermitian matrix 
distributed 
according to the ``Gaussian unitary ensemble'' 
(i.e. with density $e^{-\frac{\Tr(A^2)}{2}}dA$)
is a semicircle. In particular, we see that for large $N$ 
almost all eigenvalues of $A$ 
belong to the interval $[-2\sqrt{N},2\sqrt{N}]$, so  
the limit does not depend on the values of $f$ 
outside $[-2,2]$. 

\begin{proof} 
By Weierstrass' theorem on uniform approximation of a continuous function on an interval
 by polynomials, we may assume 
that $f$ is a polynomial. (Exercise: Justify this step). 
Thus, it suffices to check the result if 
$f(x)=x^{2m}$. In this case, 
by Corollary \ref{Catalan}, the left hand side is $C_m$. 
On the other hand, an elementary computation yields
$$
\frac{1}{2\pi}\int_{-2}^2x^{2m}\sqrt{4-x^2}dx=C_m,
$$ 
which implies 
the theorem. 
\end{proof}

\subsection{Hermite polynomials}

The proof\footnote{I adopted this proof from 
D.Jackson's notes.} of Theorem \ref{howmany} given below
uses Hermite polynomials. So let us recall their properties. 

Hermite polynomials are defined by the formula 
$$
H_n(x)=(-1)^ne^{x^2}\frac{d^n}{dx^n}e^{-x^2}.
$$
So the leading term of $H_n(x)$ is $(2x)^n$. 

We collect the standard properties of $H_n(x)$ in the following theorem. 

\begin{theorem} \label{Hermite}
(i) The exponential generating function of $H_n(x)$ is 
$$
f(x,t)=\sum_{n\ge 0}H_n(x)\frac{t^n}{n!}
=e^{2xt-t^2}.
$$ 

(ii) $H_n(x)$ satisfy the differential equation
$f''-2xf'+2nf=0$. In other words,
$H_n(x)e^{-x^2/2}$ are eigenfunctions 
of the operator $L=-\frac{1}{2}\partial^2+\frac{1}{2}x^2$ 
(Hamiltonian of the quantum harmonic oscillator) with eigenvalues 
$n+\frac{1}{2}$. 

(iii) $H_n(x)$ are orthogonal:
$$
\frac{1}{\sqrt{\pi}}\int_{-\infty}^{\infty}e^{-x^2}H_m(x)H_n(x)dx=
2^nn!\delta_{mn}.
$$
Moreover, the functions $H_n(x)e^{-\frac{x^2}{2}}$ form an orthogonal basis of $L^2(\Bbb R)$.  

(iv) One has 
$$
\frac{1}{\sqrt{\pi}}\int_{-\infty}^{\infty}e^{-x^2}x^{2m}H_{2k}(x)dx=
\frac{(2m)!}{(m-k)!}2^{2(k-m)}
$$
(if $k>m$, the answer is zero).

(v) One has
$$
\frac{H_r^2(x)}{2^rr!}=\sum_{k=0}^r\frac{r!}{2^kk!^2(r-k)!}H_{2k}(x).
$$
\end{theorem}
 
\begin{proof} (sketch)
(i) Follows immediately 
from the fact that the operator $\sum_{n\ge 0} (-1)^n\frac{t^n}{n!}\frac{d^n}{dx^n}$
maps a function $g(x)$ to $g(x-t)$. 

(ii) Follows from (i) and the fact that the function $f(x,t)$ 
satisfies the PDE 
$$
f_{xx}-2xf_x+2tf_t=0.
$$ 

(iii) The orthogonality follows from (i) by direct integration:
$$
\frac{1}{\sqrt{\pi}}\int_{\Bbb R}f(x,t)f(x,u)e^{-x^2}dx=\frac{1}{\sqrt{\pi}}\int_{\Bbb R}e^{2ut-(x-u-t)^2}dx=e^{2ut}.
$$ 
Thus the functions $H_n(x)e^{-\frac{x^2}{2}}$ form an orthogonal system in $L^2(\Bbb R)$. 

To show that these functions are complete, denote by $E\subset L^2(\Bbb R)$ the closure of their span $\Bbb C[x]e^{-\frac{x^2}{2}}$. By approximating the function $e^{ipx}$ by its Taylor polynomials, it is easy to see that $e^{ipx-\frac{x^2}{2}}\in E$ for any $p\in \Bbb R$. 
Thus for any compactly supported smooth $\phi\in C_0^\infty(\Bbb R)$ we have 
$$
\phi(x)e^{-\frac{x^2}{2}}=\int_{\Bbb R} \widehat \phi(p)e^{ipx-\frac{x^2}{2}}dp\in E.
$$
where $\widehat \phi$ is the (suitably normalized) Fourier transform of $\phi$. In other words, $C_0^\infty(\Bbb R)$ is dense in $E$. 
But $C_0^\infty(\Bbb R)$ is clearly dense in $L^2(\Bbb R)$, so $E=L^2(\Bbb R)$, as claimed. 

(iv) By (i), one should calculate 
$\int_{\Bbb R}x^{2m}e^{2xt-t^2}e^{-x^2}dx$. 
This integral equals 
$$
\int_{\Bbb R}x^{2m}e^{-(x-t)^2}dx=\int_{\Bbb R}(y+t)^{2m}e^{-y^2}dy=
=\sqrt{\pi}
\sum_p \begin{pmatrix} 2m\\ 2p\end{pmatrix} \frac{(2m-2p)!}{2^{m-p}(m-p)!}t^{2p}.
$$
The result is now obtained by extracting individual coefficients. 

(v) By (iii), it suffices to show that 
$$
\frac{1}{\sqrt{\pi}}\int_{\Bbb R}H_r^2(x)H_{2k}(x)
e^{-x^2}dx=\frac{2^{r+k}r!^2(2k)!}{k!^2(r-k)!}
$$
To prove this identity, let us integrate the product of three 
generating functions. By (i), we have 
$$
\frac{1}{\sqrt{\pi}}\int_{\Bbb R}f(x,t)f(x,u)f(x,v)e^{-x^2}dx=
$$
$$
\frac{1}{\sqrt{\pi}}\int_{\Bbb R}e^{2(ut+uv+tv)-(x-u-t-v)^2}dx
=e^{2(ut+tv+uv)}.
$$
Extracting the coefficient of $t^ru^rv^{2k}$, we get the result.
\end{proof}

\subsection{Proof of Theorem \ref{howmany}}
We need to compute the integral
$$
\int_{{\frak h}_N}\Tr(A^{2m})e^{-\frac{\Tr(A^2)}{2}}dA.
$$
To do this, we note that the integrand 
is invariant with respect to conjugation
by unitary matrices. Therefore, 
the integral can be reduced to an 
%%%integral over the eigenvalues $\lambda_1,...,\lambda_N$ of $A$. 
integral over the eigenvalues $\lambda_1, \ldots ,\lambda_N$ of $A$. 

More precisely, 
consider the spectrum map $\sigma: {\frak h}_N\to \Bbb R^N/S_N$. 
It is well known (due to H.Weyl) that the direct image 
$\sigma_*dA$ is given by the formula 
$\sigma_*dA=Ce^{-\sum_i\frac{\lambda_i^2}{2}}\prod_{i<j} (\lambda_i-\lambda_j)^2d\lambda$, 
where $C>0$ is a normalization constant that will not be relevant to us. 
Thus, we have
$$
P_m(N)=\frac{NJ_m}{J_0},\ J_m:=\int_{\Bbb R^N}(\frac{1}{N}\sum_i\lambda_i^{2m})
e^{-\sum_i \frac{\lambda_i^2}{2}}\prod_{i<j} (\lambda_i-\lambda_j)^2d\lambda.
$$
To calculate $J_m$, we will use
Hermite polynomials.  Observe that since $H_n(x)$ are
polynomials of degree $n$ with highest coefficient $2^n$, we have
$$
\prod_{i<j}(\lambda_i-\lambda_j)= 2^{-\frac{N(N-1)}{2}}
\det (H_k(\lambda_\ell)),
$$ 
where $k$ runs through the set
$0,1, \ldots ,N-1$ and $\ell$ through $1,...,N$. Thus, we find
\begin{equation}
\begin{gathered}
J_m=2^{m+{N^2\over 2}}\int_{\Bbb R^N}\lambda_1^{2m}
e^{-\sum_i \lambda_i^2}\prod_{i<j} (\lambda_i-\lambda_j)^2d\lambda=\\ 
2^{m-{N(N-2)\over 2}}\int_{\Bbb R^N}\lambda_1^{2m}
e^{-\sum_i \lambda_i^2}\det(H_k(\lambda_j))^2d\lambda=\\ 
2^{m-{N(N-2)\over 2}}\sum_{\sigma,\tau\in S_N}(-1)^{\sigma}(-1)^{\tau}\int_{\Bbb R^N}\lambda_1^{2m}
e^{-\sum_i \lambda_i^2}
\prod_i H_{\sigma(i)}(\lambda_i)H_{\tau(i)}(\lambda_i)d\lambda.
\end{gathered}
\end{equation}
(Here $(-1)^\sigma$ denotes the sign of $\sigma$).

Since Hermite polynomials are orthogonal, 
the only terms in this sum which are nonzero are 
the terms with $\sigma(i)=\tau(i)$ for $i=2, \ldots ,N$. 
That is, the nonzero terms have $\sigma=\tau$.
Thus, we have
\begin{equation}
\begin{aligned}
J_m=2^{m-{N(N-2)\over 2}}\sum_{\sigma\in S_N}
\int_{\Bbb R^N}\lambda_1^{2m}
e^{-\sum_i \lambda_i^2}
\prod_i H_{\sigma i}(\lambda_i)^2d\lambda=\\
2^{m-{N(N-2)\over 2}}(N-1)!\gamma_0 \ldots \gamma_{N-1}\sum_{j=0}^{N-1}
\frac{1}{\gamma_j}\int_{-\infty}^\infty x^{2m}H_j(x)^2e^{-x^2}dx,
\end{aligned}
\end{equation}
where $\gamma_i:=\int_{-\infty}^\infty H_i(x)^2e^{-x^2}dx$
are the squared norms of the Hermite polynomials. 
Applying this for $m=0$ and dividing $NJ_m$ by $J_0$, we find 
$$
P_m(N)=2^m\sum_{j=0}^{N-1}
\frac{1}{\gamma_j}\int_{-\infty}^\infty x^{2m}H_j(x)^2e^{-x^2}dx.
$$
Using Theorem \ref{Hermite} (iii) and (v), we find that
$\gamma_i=2^ii!\sqrt{\pi}$, and hence
$$
P_m(N)=\frac{1}{\sqrt{\pi}}\int_{\Bbb R}\sum_{j=0}^{N-1}
\sum_{k=0}^j\frac{2^mx^{2m}H_{2k}(x)}{2^kk!^2(j-k)!}e^{-x^2}dx.
$$
Now, using Theorem \ref{Hermite} (iv), we get 
$$
P_m(N)=\frac{(2m)!}{2^m}\sum_{j=0}^{N-1}
\sum_{k=0}^j\frac{2^kj!}{(m-k)!k!^2(j-k)!}=
$$
$$
\frac{(2m)!}{2^mm!}\sum_{j=0}^{N-1}
\sum_{k=0}^j2^k\begin{pmatrix}m\\ k\end{pmatrix}
\begin{pmatrix}j\\ k\end{pmatrix}.
$$
The sum over $k$ can be represented as the constant term of a polynomial: 
$$
\sum_{k=0}^j2^k\begin{pmatrix}m\\ k\end{pmatrix}
\begin{pmatrix}j\\ k\end{pmatrix}=
C.T.((1+z)^m(1+2z^{-1})^j). 
$$
Therefore, summation over $j$ (using the formula for 
the sum of the geometric progression) yields
$$
P_m(N)=\frac{(2m)!}{2^mm!}
C.T.\left((1+z)^m\frac{(1+2z^{-1})^N-1}{2z^{-1}}\right)=
$$
$$
\frac{(2m)!}{2^mm!}\sum_{p=0}^m2^p
\begin{pmatrix}m\\ p\end{pmatrix}
\begin{pmatrix}N\\ p+1\end{pmatrix}.
$$
We are done. 

\begin{exercise}  Find the number of ways to glue an orientable
 surface of genus ${\rm g}\ge 1$ 
from a $4{\rm g}$-gon (the gluing must preserve orientation),
and prove your answer. 

{\bf Answer:} $\frac{(4{\rm g}-1)!!}{2{\rm g}+1}$.
\end{exercise} 

\begin{exercise} Consider a random Hermitian matrix $A\in
{\frak h}_N$, distributed with Gaussian density
$e^{-{\rm Tr}(A^2)}dA$. Show that the most likely
eigenvalues of $A$ are the roots of the $N$-th Hermite 
polynomial $H_N$. 

{\bf Hint.} 1) Write down the system of algebraic equations 
for the maximum of the density on eigenvalues. 

2) Introduce the polynomial $P(z)=\prod_i(z-\lambda_i)$,
where $\lambda_i$ are the most likely eigenvalues. 
Let $f=P'/P$. Compute $f'+f^2$ (look at the poles). 

3) Reduce the obtained Riccati equation for $f$ 
to a second order linear differential equation for $P$. 
Show that this equation is the Hermite's equation, and deduce
that $P=\frac{H_N}{2^N}$. 
\end{exercise} 

\section{The Euler characteristic of the moduli space of curves}

Matrix integrals (in particular, the computation of the polynomial
$P_m(x)$)
can be used to calculate the orbifold Euler characteristic 
of the moduli space of curves. This was done by Harer and Zagier
in 1986. Here we will give a review of this result (with some
omissions).

\subsection{Euler characteristics of groups}

We start with recalling some basic notions from algebraic
topology. 

Let $\Gamma$ be a discrete group, and $Y$ be a contractible 
finite dimensional CW complex, on
which $\Gamma$ acts cellularly. This means that $\Gamma$ acts by
homeomorphisms of $Y$ that map each cell homeomorphically to
another cell. We will assume that the stabilizer of each cell 
is a finite group (i.e. $Y$ is a proper
$\Gamma$-complex). 

Suppose first that the action of $\Gamma$ is free (i.e. the stabilizers of cells
are trivial). This is equivalent to saying that $\Gamma$ is
torsion free (i.e. has no nontrivial finite subgroups), since a finite
group cannot act without fixed points on a contractible finite
dimensional cell complex (as it has infinite cohomological dimension). 

In this case we can define a cell complex $Y/\Gamma$ 
(a classifying space for $\Gamma$), 
and we have $H^i(Y/\Gamma,A)=H^i(\Gamma,A)$ for any coefficient
group $A$. In particular, if $Y/\Gamma$ is finite then 
$\Gamma$ has finite cohomological dimension, and 
the Euler characteristic $\chi(\Gamma):=\sum_i (-1)^i\dim
H^i(\Gamma,\Bbb Q)$ is equal to $\sum_i (-1)^in_i(Y/\Gamma)$, where
$n_i(Y/\Gamma)$ denotes the number of cells in $Y/\Gamma$ of
dimension $i$. 

This setting, however, is very restrictive,
since it allows only groups of finite cohomological dimension, 
and in particular excludes all non-trivial finite groups.
So let us consider a more general setting: assume that 
some finite index subgroup $\Gamma'\subset \Gamma$,
rather than $\Gamma$ itself, satisfies 
the above conditions. In this case, on may define
the Euler characteristic of $\Gamma$ in the sense of Wall, 
which is the rational number
$[\Gamma:\Gamma']^{-1}\chi(\Gamma')$. 

It is easy to check that 
the Euler characteristic in the sense of Wall can be computed
using the following {\it Quillen's formula}\index{Quillen's formula}
$$
\chi(\Gamma)=\sum_{\sigma\in {\rm cells(Y)}/\Gamma}\frac{(-1)^{\dim
    \sigma}}{|{\rm Stab}\sigma|}.
$$
In particular, this number is independent of $\Gamma'$ (which is
also easy to check directly). 

\begin{example} If $G$ is a finite group then $\chi(G)=|G|^{-1}$ 
(one takes the trivial group as the subgroup of finite index).
\end{example}

\begin{example} $G=SL_2(\Bbb Z)$. This group contains 
a subgroup $F$ of index 12, which is free in two generators (check
it!). The group $F$ has Euler characteristic $-1$, since its
classifying space $Y/F$ is figure ``eight'' (i.e., $Y$ is the universal
cover of figure ``eight''). Thus, the Euler characteristic 
of $SL_2(\Bbb Z)$ is $-\frac{1}{12}$. 
\end{example} 

The Euler characteristic in the sense of Wall 
has a geometric interpretation in terms of orbifolds. 
Namely, suppose that $\Gamma$ is as above (i.e. 
$\chi(\Gamma)$ is a well defined rational number), 
and $M$ is a contractible manifold, on which 
$\Gamma$ acts freely and properly discontinuously. 
In this case, stabilizers of points are finite, and 
thus $M/\Gamma$ is an orbifold. This means, 
in particular, that to every point $x\in M/\Gamma$ is attached a
finite group ${\rm Aut}(x)$, of size $\le [\Gamma:\Gamma']$.
Let $X_m$ be the subset of $M/\Gamma$, consisting of points $x$
such that ${\rm Aut}(x)$ has order $m$. It often happens that 
$X_m$ has the homotopy type of a finite cell complex. In this
case, the {\it orbifold Euler characteristic}\index{orbifold Euler characteristic} of $M/\Gamma$ is defined
to be 
$$
\chi_{\rm orb}(M/\Gamma)=\sum_m\frac{\chi(X_m)}{m}.
$$ 

Now, we claim that $\chi_{\rm orb}(M/\Gamma)=\chi(\Gamma)$. 
Indeed, looking at the projection $M/\Gamma'\to M/\Gamma$, it is easy to see that 
$\chi_{\rm  orb}(M/\Gamma)=\frac{1}{[\Gamma:\Gamma']}\chi(M/\Gamma')$. 
But $M/\Gamma'$ is a classifying space for $\Gamma'$, so 
$\chi(M/\Gamma')=\chi(\Gamma')$, which implies the claim. 

\begin{example} Consider the group $\Gamma=SL_2(\Bbb Z)$ acting
on the upper half plane $H$. Then $H/\Gamma$ is the moduli space 
of elliptic curves. So as a topological space it is 
$\Bbb C$, where all points have automorphism group $\Bbb Z/2$,
except the point $i$ having automorphism group $\Bbb Z/4$, and 
$\rho=\frac{-1+i\sqrt{3}}{2}$ which has automorphism group $\Bbb Z/6$. 
Thus, the orbifold Euler characteristic of $H/\Gamma$ 
is $(-1)\frac{1}{2}+\frac{1}{4}+\frac{1}{6}=-\frac{1}{12}$. 
This is not surprising since we proved that $\chi_{\rm
orb}(H/\Gamma)=\chi(\Gamma)$, which was computed to be $-\frac{1}{12}$.
\end{example}

\subsection{The mapping class group}

Now let ${\rm g}\ge 1$ be an integer, and $\Sigma$ be a closed
oriented surface of genus ${\rm g}$. Let $p\in \Sigma$, and let
$\Gamma_{\rm g}^1$ be the group of isotopy classes of 
diffeomorphisms of $\Sigma$ which preserve $p$.
We will recall without proof some standard facts 
about this group, following the paper of Harer and Zagier, \cite{HZ}. 

The group $\Gamma_{\rm g}^1$ is not torsion free, 
but it has a torsion free subgroup of finite index. 
Namely, consider the homomorphism 
$\Gamma_{\rm g}^1\to {\rm Sp}(2{\rm g},\Bbb Z/n\Bbb Z)$ given by the action 
of $\Gamma_{\rm g}^1$ on $H_1(\Sigma,\Bbb Z/n\Bbb Z)$. 
Then for large enough $n$ (in fact, $n\ge 3$), 
the kernel $K_n$ of this map is torsion free.

It turns out that there exists a contractible finite dimensional
cell complex $Y_{\rm g}$, to be constructed below, 
on which $\Gamma_{\rm g}^1$ acts cellularly with finitely many cell
orbits. Thus, the Euler characteristic 
of $\Gamma_{\rm g}^1$ in the sense of Wall is well defined. 

\subsection{The Harer-Zagier theorem}  The Euler characteristic of $\Gamma_{\rm g}^1$ is given by the following theorem. 

\begin{theorem} \label{hz} (Harer-Zagier) One has 
$$
\chi(\Gamma_{\rm g}^1)=-\frac{B_{2{\rm g}}}{2{\rm g}},
$$ 
where $B_n$ are the Bernoulli numbers. 
\end{theorem}

\begin{remark} The group $\Gamma_{\rm g}^1$ acts on the Teichm\"uller space 
${\mathcal T}_{\rm g}^1$, which is, by definition, the space 
of pairs $((R,z),f)$, where $(R,z)$ is a complex Riemann surface
with a marked point $z$, and $f$  is an isotopy class
of diffeomorphisms $R\to \Sigma$ that map $z$ to $p$. 
One may show that ${\mathcal T}_{\rm g}^1$ is a contractible manifold 
of dimension $6{\rm g}-4$, and that the action of $\Gamma_{\rm g}^1$ 
on ${\mathcal T}_{\rm g}^1$ is properly discontinuous. 
In particular, we may define an orbifold
$M_{\rm g}^1={\mathcal T}_{\rm g}^1/\Gamma_{\rm g}^1$. This orbifold 
parametrizes pairs $(R,z)$ as above; therefore, it is called the moduli
space of Riemann surfaces (=smooth complex projective 
algebraic curves) of genus ${\rm g}$ with one marked point. Thus, Theorem \ref{hz} 
gives the orbifold Euler characteristic of the moduli space of
curves of genus ${\rm g}$ with one marked point. 
\end{remark}

\begin{remark} 
If ${\rm g}>1$, one may define the analogs of the above objects without
marked points, namely the mapping class group $\Gamma_{\rm g}$, 
the Teichm\"uller space ${\mathcal T}_{\rm g}$, and the moduli space of curves
$M_{\rm g}={\mathcal T}_g/\Gamma_g$ (one can do it for ${\rm g}=1$ as well, 
but in this case there is no difference with the case of one
marked point, since the translation group allows one to identify any
two points on $\Sigma$). It is easy to see that for ${\rm g}>1$ we have an exact sequence
$1\to \pi_1(\Sigma)\to \Gamma_{\rm g}^1\to \Gamma_{\rm g}\to 1$, 
which implies that $\chi(\Gamma_{\rm g})=\chi(\Gamma_{\rm g}^1)/\chi(\Sigma)$.
Thus, the Harer-Zagier theorem implies that $\chi(\Gamma_{\rm g})=\chi_{\rm orb}(M_{\rm g})=\frac{B_{2{\rm g}}}{4{\rm g}({\rm g}-1)}$.
\end{remark} 

\subsection{Construction of the complex $Y_{\rm g}$}

We begin the proof of Theorem \ref{hz} with the construction of
the complex $Y_{\rm g}$, following \cite{HZ}. 
We will first construct a simplicial complex
with a $\Gamma_{\rm g}^1$ action, and then use it to construct $Y_{\rm g}$. 

Let $(\alpha_1,...,\alpha_n)$ be a collection of closed simple 
unoriented curves on $\Sigma$, which begin and end at $p$, and do not
intersect other than at $p$. Such a collection is called an {\it arc
system}\index{arc system} if two conditions are satisfied:

(A) none of the curves is contractible to a point;

(B) none of the curves is contractible to another. 

Define a simplicial complex $A$, whose $n-1$-simplices 
are isotopy classes of arc systems consisting of $n\ge 1$ arcs, and 
the boundary of a simplex corresponding to
$(\alpha_1,...\alpha_n)$ is the union of simplices
corresponding to the arc system
$(\alpha_1,...,\widehat\alpha_i,...,\alpha_n)$ 
($\alpha_i$ is omitted). 

It is clear that the group $\Gamma_{\rm g}^1$ 
acts simplicially on $A$. 

\begin{example} Let ${\rm g}=1$, i.e. $\Sigma=S^1\times S^1$. Then 
$\Gamma_{\rm g}^1=SL_2(\Bbb Z)$. Up to its action, there 
are only three arc systems (Fig. 15). Namely, viewing $S^1$ as the unit
circle in the complex plane, and representing arcs
parametrically, we may write these three systems as follows:
$$
B_0=\lbrace{(e^{i\theta},1)\rbrace}; B_1=
\lbrace{(e^{i\theta},1),(1,e^{i\theta})\rbrace};
B_2=\lbrace{(e^{i\theta},1),(1,e^{i\theta}),(e^{i\theta},e^{i\theta})\rbrace}
$$
From this it is easy to find the simplicial complex $A$. 
Namely, let $T$ be the tree with root $t_0$ connected to three
vertices $t_1,t_2,t_3$, with each $t_i$ 
connected to two vertices $t_{i1},t_{i2}$, each $t_{ij}$ connected to
$t_{ij1},t_{ij2}$, etc. (Fig.16). Put at every vertex of $T$ a triangle,
with sides transversal to the three edges going out of this
vertex, and glue the triangles along the sides. 
This yields the complex $A$, Fig.17 (check it!). 
The action of $SL_2(\Bbb Z)$ (or rather $PSL_2(\Bbb Z)$) on this 
complex is easy to describe. 
Namely, recall that $PSL_2(\Bbb Z)$ is generated by $S,U$
with defining relations $S^2=U^3=1$. The action of $S,U$ on $T$ is defined as
follows: $S$ is the reflection with flip with respect to a side 
of the triangle $\Delta_0$ centered at $t_0$ (Fig.18), and $U$ is the rotation
by $2\pi/3$ around $t_0$. 
\end{example}

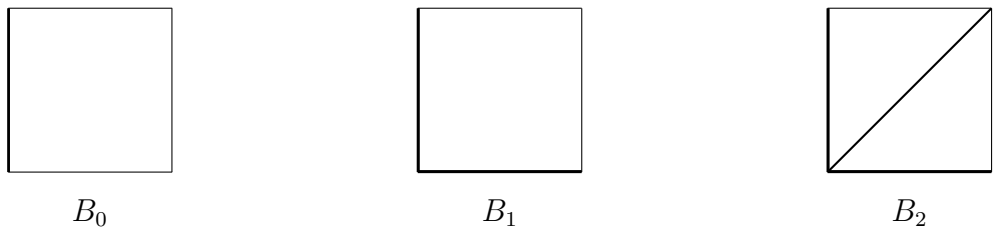
\begin{figure}[htbp]
  \begin{center}
    
    \setlength{\unitlength}{3ex}
    \begin{picture}(24,6)(0,-1)

      %%% left figure
      \thicklines
      \put(0,0){\line(0,1){4}}
      \put(0.02,0){\line(0,1){4}}

      \thinlines
      \put(0,0){\line(1,0){4}}
      \put(4,0){\line(0,1){4}}
      \put(0,4){\line(1,0){4}}

      %%% middle figure
      \thicklines
      \put(10,0){\line(0,1){4}}
      \put(10.02,0){\line(0,1){4}}
      \put(10,0){\line(1,0){4}}
      \put(10,0.02){\line(1,0){4}}

      \thinlines
      \put(14,0){\line(0,1){4}}
      \put(10,4){\line(1,0){4}}

      %%% right figure
      \thicklines
      \put(20,0){\line(0,1){4}}
      \put(20.02,0){\line(0,1){4}}
      \put(20,0){\line(1,0){4}}
      \put(20,0.02){\line(1,0){4}}
      \put(20,0){\line(1,1){4}}

      \thinlines
      \put(24,0){\line(0,1){4}}
      \put(20,4){\line(1,0){4}}

      \put(2,-1){\makebox(0,0)[c]{$B_0$}}
      \put(12,-1){\makebox(0,0)[c]{$B_1$}}
      \put(22,-1){\makebox(0,0)[c]{$B_2$}}

    \end{picture}
    \caption{Three arc systems.}
    \label{fig:15}
  \end{center}
\end{figure}

\begin{figure}
\begin{center}
\setlength{\unitlength}{2ex}
\begin{picture}(18,15)
%%%basic tree
\put(9,8){\line(0,-1){4.6}}
\put(9,8){\line(5,3){4}}
\put(9,8){\line(-5,3){4}}
\put(9,3.4){\line(2,-1){2}}
\put(9,3.4){\line(-2,-1){2}}
\put(7,2.4){\line(0,-1){2}}
\put(7,2.4){\line(-2,1){1.789}}
\put(11,2.4){\line(0,-1){2}}
\put(11,2.4){\line(2,1){1.789}}
\put(13,10.4){\line(0,1){2}}
\put(13,10.4){\line(2,-1){2}}
\put(13,12.4){\line(1,1){1.414}}
\put(13,12.4){\line(-1,1){1.414}}
\put(15,9.4){\line(2,1){1.789}}
\put(15,9.4){\line(0,-1){2}}
\put(5,10.4){\line(0,1){2}}
\put(5,10.4){\line(-2,-1){2}}
\put(5,12.4){\line(1,1){1.414}}
\put(5,12.4){\line(-1,1){1.414}}
\put(3,9.4){\line(-2,1){1.789}}
\put(3,9.4){\line(0,-1){2}}

%%%dots
\multiput(3,14.4)(2,0){3}{\circle*{.2}}
\multiput(11,14.4)(2,0){3}{\circle*{.2}}
\multiput(.2,9.8)(1,-1.5){3}{\circle*{.2}}
\multiput(17.8,9.8)(-1,-1.5){3}{\circle*{.2}}

\multiput(6.2,-.2)(-1,1.5){3}{\circle*{.2}}
\multiput(11.8,-.2)(1,1.5){3}{\circle*{.2}}

\end{picture}
\end{center}
\caption{The tree $T$}
    \label{fig:16}
\end{figure}
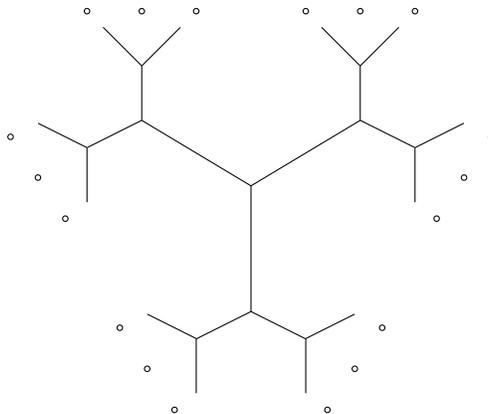

\begin{figure}
\begin{center}
\setlength{\unitlength}{2ex}
\begin{picture}(18,15)
%%%basic tree
\put(9,8){\line(0,-1){4.6}}
\put(9,8){\line(5,3){4}}
\put(9,8){\line(-5,3){4}}
\put(9,3.4){\line(2,-1){2}}
\put(9,3.4){\line(-2,-1){2}}
\put(7,2.4){\line(0,-1){2}}
\put(7,2.4){\line(-2,1){1.789}}
\put(11,2.4){\line(0,-1){2}}
\put(11,2.4){\line(2,1){1.789}}
\put(13,10.4){\line(0,1){2}}
\put(13,10.4){\line(2,-1){2}}
\put(13,12.4){\line(1,1){1.414}}
\put(13,12.4){\line(-1,1){1.414}}
\put(15,9.4){\line(2,1){1.789}}
\put(15,9.4){\line(0,-1){2}}
\put(5,10.4){\line(0,1){2}}
\put(5,10.4){\line(-2,-1){2}}
\put(5,12.4){\line(1,1){1.414}}
\put(5,12.4){\line(-1,1){1.414}}
\put(3,9.4){\line(-2,1){1.789}}
\put(3,9.4){\line(0,-1){2}}

%%%triangulation
\thicklines
\put(3,11.4){\circle*{.6}}
\put(6,6.4){\circle*{.6}}
\put(9,1.4){\circle*{.6}}
\put(9,11.4){\circle*{.6}}
\put(12,6.4){\circle*{.6}}
\put(15,11.4){\circle*{.6}}

\put(9,1.4){\line(-3,5){6}}
\put(9,1.4){\line(3,5){6}}
\put(3,11.4){\line(1,0){12}}
\put(6,6.4){\line(1,0){6}}
\put(9,11.4){\line(-3,-5){3}}
\put(9,11.4){\line(3,-5){3}}

\end{picture}
\end{center}
\caption{The complex A}
    \label{fig:17}
\end{figure}
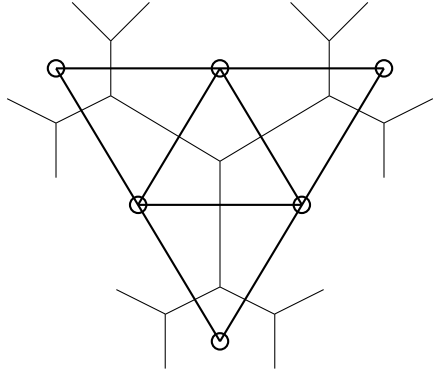

\begin{figure}[htbp]
  \begin{center}
    
    \setlength{\unitlength}{3.5ex}
    \begin{picture}(9,4.6)(0,0)

      \put(5,0){\line(3,2){3.5}}
      \put(5,0){\line(-3,2){3.5}}
      \put(5,4.6){\line(-3,-2){3.5}}
      \put(5,4.6){\line(3,-2){3.5}}

      \put(5,0){\line(0,1){4.6}}

%      %%% make cat
%      \put(3.8,3){\oval(1.5,3.5)[b]}
%      \put(3.05,3){\line(2,-3){.43}}
%      \put(4.55,3){\line(-2,-3){.43}}
%      \put(3.8,2.35){\oval(.65,.34)[t]}

%      %%% eyes
%      \put(3.425,2.15){\circle*{.1}}
%      \put(4.175,2.15){\circle*{.1}}

%      %%% mouth
%      \put(3.8,1.7){\oval(.32,.22)[b]}

      %%% make cat

      \put(3.8,1.6){\oval(1.5,3.5)[t]}
      \put(3.05,1.6){\line(2,3){.43}}
      \put(4.55,1.6){\line(-2,3){.43}}
      \put(3.8,2.25){\oval(.65,.34)[b]}

      %%% eyes
      \put(3.425,2.45){\circle*{.1}}
      \put(4.175,2.45){\circle*{.1}}

      %%% nose
      \put(3.8,2.7){\circle*{.05}}

      %%% mouth
      \put(3.8,2.9){\oval(.32,.22)[t]}

      %%% make upside-down cat

      \put(6.2,3){\oval(1.5,3.5)[b]}
      \put(6.95,3){\line(-2,-3){.43}}
      \put(5.45,3){\line(2,-3){.43}}
      \put(6.2,2.35){\oval(.65,.34)[t]}

      %%% eyes
      \put(6.575,2.15){\circle*{.1}}
      \put(5.825,2.15){\circle*{.1}}

      %%% nose
      \put(6.2,1.9){\circle*{.05}}

      %%% mouth
      \put(6.2,1.7){\oval(.32,.22)[b]}

    \end{picture}
    \caption{Reflection with a flip.}
    \label{fig:18}
  \end{center}
\end{figure}
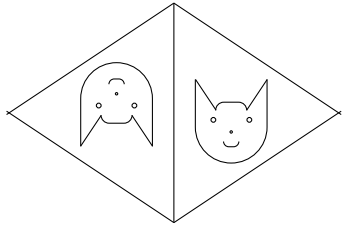

This example shows that the action 
of $\Gamma_{\rm g}^1$ on $A$ is not properly discontinuous, as 
some simplices have infinite stabilizers (in the example, it is
the 0-dimensional simplices). Thus, we would like to
throw away the ``bad'' simplices. To do so, let us say 
that an arc system $(\alpha_1,..,\alpha_n)$ {\it fills up} 
$\Sigma$ if it cuts $\Sigma$ into a union of 
regions diffeomorphic to the open disk. Let $A_\infty$ be the
union of the simplices in $A$ corresponding to arc systems that
do not fill up $\Sigma$. This is a closed subset, since the
property of not filling up $\Sigma$ is 
obviously stable under taking an arc
subsystem. Thus, $A\setminus A_\infty$ is an open subset 
of $A$. In the example above, it is the complex $A$ with
$0$-dimensional simplices removed. 

The following theorem shows that $A\setminus A_\infty$ is in fact
a combinatorial model for the Teichm\"uller space ${\mathcal T}_{\rm g}^1$,
with the action of $\Gamma_{\rm g}^1$.
 
\begin{theorem} \label{mum} (Mumford)
(a) The action of $\Gamma_{\rm g}^1$ on $A\setminus A_\infty$ is properly
discontinuous. 

(b) $A\setminus A_\infty$ is topologically a manifold, which 
is $\Gamma_{\rm g}^1$-equivariantly homeomorphic to the Teichm\"uller
space ${\mathcal T}_{\rm g}^1$; in particular, it is contractible. 
\end{theorem}

\begin{remark} Theorem \ref{mum} exhibits 
the significance of conditions (A) and (B).
Indeed, if either of these conditions were dropped, 
then one could consider arc systems 
$(\alpha_1,...,\alpha_n)$ with arbitrarily large $n$, while 
with conditions (A),(B), as seen from Theorem \ref{mum}, 
the largest value of $n$ is $6{\rm g}-3$.
\end{remark}

\begin{remark} If ${\rm g}=1$, Theorem \ref{mum} is clear from the
explicit description of $A$ (convince yourself of this!).
\end{remark}

Theorem \ref{mum} is rather deep, and we will not give 
its proof, which is beyond the scope
of this text. Rather, we will use it to define the ``Poincar\'e
dual'' CW complex $Y_{\rm g}$ of $A\setminus A_\infty$. Namely, to each
filling arc system $(\alpha_1,...,\alpha_n)$ we will assign 
a $6{\rm g}-3-n$-dimensional cell, and the boundary relation 
is opposite to the one before. The existence of this CW complex
follows from the fact that $A\setminus A_\infty$ is a manifold. 
For instance, in the case ${\rm g}=1$ the complex $Y_{\rm g}$ is the tree $T$. 

Now, the complex $Y_{\rm g}$ is contractible (since so is $A\setminus
A_\infty$), and admits a cellular action of $\Gamma_{\rm g}^1$ with
finitely many cell orbits and finite stabilizers. 
This means that the Euler characteristic of $\Gamma_{\rm g}^1$ is given
by Quillen's formula. 
$$
\chi(\Gamma_{\rm g}^1)=\sum_{\sigma\in {\rm
    cells}(Y_{\rm g})/\Gamma_{\rm g}^1}
(-1)^{\dim\sigma}\frac{1}{|{\rm Stab\sigma}|}.
$$

\begin{example} In the ${\rm g}=1$ case, $T$ has one orbit of $0$-cells
and one orbit of $1$-cells. The stabilizer of a 0-cell in
$SL_2(\Bbb Z)$ is $\Bbb Z/6$, and of a 1-cell is $\Bbb Z/4$. 
Hence, $\chi(SL_2(\Bbb
Z))=\frac{1}{6}-\frac{1}{4}=-\frac{1}{12}$, 
which was already computed before by other methods. 
\end{example}

\subsection{Enumeration of cells in $Y_{\rm g}/\Gamma_{\rm g}^1$}

Now it remains to count the cells in $Y_{\rm g}/\Gamma_{\rm g}^1$, i.e. to enumerate
arc systems which fill $\Sigma$ (taking into account signs and stabilizers) 
To do this, we note that by definition of ``filling'', any
filling arc system $S$ defines a cellular decomposition of $\Sigma$. 
Thus, let $S^*$ be the Poincare dual of this cellular
decomposition. Since $S$ has a unique zero cell, $S^*$ has a
unique 2-cell. Let $n$ be the number of 1-cells in $S$ (or
$S^*$). Then $(\Sigma, S^*)$ is obtained by gluing a 
$2n$-gon (=the unique 2-cell) according to a matching of its sides
preserving orientation. (Note that $S$ can be reconstructed as
$(S^*)^*$). 

This allows us to link the problem of enumerating
filling arc systems with the problem of counting such
gluings, which was solved using matrix integrals. 
Namely, the problem of enumerating filling arc systems 
is essentially solved 
modulo one complication: because of conditions (A)
and (B) on an arc system, the gluings we will get will be not
arbitrary gluings, but gluings which also must satisfy some
conditions. Namely, we have

\begin{lemma}
Let $(\alpha_1,...,\alpha_n)$ be a system of curves, satisfying
the axioms of a filling arc system, except maybe conditions (A)
and (B). Then

(i) $(\alpha_1,...,\alpha_n)$ satisfies condition (A) 
iff no edge in the corresponding gluing is glued to a neighboring
edge. 

(ii) $(\alpha_1,...,\alpha_n)$ satisfies condition (B)
iff no two consequtive edges are glued to another pair 
of consequtive edges in the opposite order.
\end{lemma}

\begin{figure}[htbp]
  \begin{center}
    
    \setlength{\unitlength}{2.5ex}
    \begin{picture}(31,6)(1,0)

      %%% left figure

      \thicklines
      \put(3,0){\line(1,0){4}}
      \put(3,0){\line(-1,1){2.5}}
      \put(0.5,2.5){\line(1,1){2.5}}
      \put(3,5){\line(1,0){4}}
      \put(7,0){\line(1,1){2.5}}
      \put(7,5){\line(1,-1){2.5}}

      \thinlines
      \put(1.8,1.25){\line(0,1){2.5}}

      %%% right figure

      \thicklines
      \put(27,0.2){\line(3,2){3.5}}
      \put(27,0.2){\line(-3,2){3.5}}
      \put(27,4.8){\line(-3,-2){3.5}}
      \put(27,4.8){\line(3,-2){3.5}}

      \put(23.5,2.5){\line(1,0){2.5}}      

      \thinlines
      \put(26,2.5){\circle{2}}      

      \put(26.95,2.5){\circle*{.2}}      

      \put(29.1,0.25){\vector(-1,1){2}}
      \put(28,-0.75){\makebox(0,0)[l]
          {\shortstack{loop homo-\\topic to $0$.}}}

      %%% squiggle in the middle
      \put(13,1.5){\line(0,1){2}}      
      \cbezier(13,2.5)(15,3)(17,2)(20,2.5)
      \put(19.75,2.5){\vector(1,0){0.25}}            

    \end{picture}
    \caption{ }
    \label{fig:19}
  \end{center}
\end{figure}

\begin{figure}[htbp]
  \begin{center}
    
    \setlength{\unitlength}{2.5ex}
    \begin{picture}(30,8)(2,-3.5)

      %%% left figure

      \thicklines
      \put(3,0){\line(1,0){4}}
      \put(3,0){\line(-1,1){2.5}}
      \put(0.5,2.5){\line(1,1){2.5}}
      \put(3,5){\line(1,0){4}}
      \put(7,0){\line(1,1){2.5}}
      \put(7,5){\line(1,-1){2.5}}

      \thinlines
      \put(1.7,1.3){\line(1,1){3.7}}
      \put(4.8,0){\line(1,1){3.6}}

      %%% right figure

      \thinlines
      \put(26,2.5){\line(1,0){3.6}}
      \put(25.1,2.5){\circle{1.8}}
      \put(28,2.5){\circle*{.2}}      
      \put(30.5,2.5){\circle{1.8}}      
      \put(30.5,2.5){\circle{3.6}}      
      \put(32.3,2.5){\circle*{.2}}      
      \put(29.8,2.5){\circle{5}}      
%%%      \put(28,2.5){\oval(9,6)}      
      \put(28.2,2.5){\circle{9}}      

      \cbezier(28.45,0.4)(29.25,-2.5)(30.5,-2.5)(30.3,0.58)
      \put(28.64,-0.25){\vector(-1,3){.2}}
      \put(30.3,0.5){\vector(0,1){.2}}
      \put(27,-3.5){\makebox(0,0)[l]
          {\shortstack{loops homotopic\\ to each other.}}}

%      \put(29.1,0.25){\vector(-1,1){2}}
%      \put(29.5,0){\makebox(0,0)[l]{loop homotopic to $0$.}}

      %%% squiggle in the middle
      \put(13,1.5){\line(0,1){2}}      
      \cbezier(13,2.5)(15,3)(17,2)(20,2.5)
      \put(19.75,2.5){\vector(1,0){0.25}}            

    \end{picture}
    \caption{ }
    \label{fig:20}
  \end{center}
\end{figure}

The lemma is geometrically evident, and its proof is obtained by 
drawing a picture (Fig.19 for (i), Fig.20 for (ii)). 
Motivated by the lemma, we will refer to the
conditions on a gluing in (i) and (ii) also as conditions (A) and
(B). 

Denote by $\varepsilon_{\rm g}(n),\mu_{\rm g}(n),\lambda_{\rm g}(n)$ the 
numbers of gluings of a (labeled) $2n$-gon into a surface of
genus ${\rm g}$, with no conditions, condition (A), and conditions
(A),(B), respectively (so $\varepsilon_{\rm g}(n)$ is the quantity we
already studied). 

\begin{proposition} One has\footnote{Note that $\lambda_{\rm g}(n)=0$ for almost all $n$, so this sum is finite.} 
$$
\chi(\Gamma_{\rm g}^1)=\sum_n (-1)^{n-1}\frac{\lambda_{\rm g}(n)}{2n}.
$$
\end{proposition}

\begin{proof}
Each filling arc system $\sigma$ arises from 
$2n/|{\rm Stab(\sigma)}|$ gluings (since the labeling 
of the polygon does not matter for the resulting surface
with an arc system). Thus, the result follows 
from Quillen's formula.
\end{proof}

\subsection{Computation of $\sum_n (-1)^{n-1}\frac{\lambda_{\rm g}(n)}{2n}$}
Now it remains to compute the sum on the right hand side. 
To do this, we will need to link $\lambda_{\rm g}(n)$ with
$\varepsilon_{\rm g}(n)$, which has already been computed. 
This is accomplished by the following lemma. 

\begin{lemma}\label{rel}
(i) One has
$$
\varepsilon_{\rm g}(n)=\sum_i \begin{pmatrix}2n\\ i\end{pmatrix}
\mu_{\rm g}(n-i).
$$
(ii) One has
$$
\mu_{\rm g}(n)=\sum_i \begin{pmatrix}n\\ i\end{pmatrix}
\lambda_{\rm g}(n-i).
$$
\end{lemma}

\begin{proof}(i)  Let $\sigma$ be a matching of 
the sides of a $2n$-gon $\Delta$ with labeled vertices. If there is a pair of consecutive
edges that are matched, we can glue them to each other to obtain a $2n-2$-gon. 
Proceeding like this as long as we can, we will arrive at a $2n-2i$-gon
$\Delta_\sigma$, with a matching $\sigma'$ of its sides which
satisfies condition (A). Note that $\Delta_\sigma$ and $\sigma'$
do  not depend on the order in which neighboring edges were glued
to each other, and $\Delta_\sigma$ has a canonical labeling 
by $1,...,2n-2i$, in the increasing order of the ``old'' labels. 
Now, we claim that each $(\Delta_\sigma,\sigma')$ 
is obtained in exactly $\begin{pmatrix}2n\\ i\end{pmatrix}$ ways;
this implies the required statement.

Indeed, let us consider the vertices of 
$\Delta$ that ended up in the interior of $\Delta_\sigma$. 
They have mapped to $i$ points in the interior (each gluing of a
pair of edges produces a new point). Let us call these points
$w_1,...,w_i$, and let $\nu_j$ be the smallest label of a vertex
of $\Delta$ that goes to $w_j$ (where we label the vertices so that the $k$-th 
edge connects vertex $k$ with vertex $k+1$). 
Then $\nu_1,...,\nu_i$ is a
subset of $\lbrace{1,...,2n\rbrace}$. 
This subset completely determines 
the matching $\sigma$ if $(\Delta_\sigma,\sigma')$ are given: 
namely, we should choose a $\nu_j$ such that $\nu_{j}+1\ne \nu_k$
for any $k$, and glue the two edges adjacent to $\nu_j$; then
relabel by $1,...,2n-2$ the remaining vertices
(in increasing order of ``old'' labels), and continue the
step again, and so on. From this it is also seen that any set of
$\nu_j$ may arise. This proves (i). 

(ii) Let $\sigma$ be a matching of $\Delta$ (with labeled edges) which 
satisfies condition (A) but not necessarily (B). 
If $a_1,a_2$ are consecutive edges that are glued to 
consecutive edges $b_2,b_1$ in the opposite order, then we may 
unite $a_1,a_2$ into a single edge $a$, and $b_2,b_1$ into $b$, 
and obtain a $2n-2$-gon with a matching. Continuing so as long as we
can, we will arrive at a $2n-2i$-gon $\Delta_\sigma$ with a new
matching $\sigma'$, which satisfies conditions (A) and (B). 
In $\Delta_\sigma$, each ($j$-th) pair of edges 
is obtained for $m_j+1$ pairs of edges in $\Delta$. 
Thus, $\sum_{j=1}^{n-i}m_j=i$. Furthermore, 
for any $(\Delta_\sigma,\sigma')$ the collection of numbers 
$m_1,...,m_{n-i}$ defines $(\Delta,\sigma)$ uniquely, 
up to deciding which of the $m_1+1$ edges constituting
the first edge of $\Delta_\sigma$ should be labeled by 1. 
Thus, each $(\Delta_\sigma,\sigma')$ 
arises in the number of ways given by the formula 
$$
\sum_{m_1,...,m_{n-i}:\sum_{j=1}^{n-i} m_j=i}(m_1+1).
$$
It is easy to show (check!) that this number is equal to 
$\begin{pmatrix}n\\ i\end{pmatrix}$.
This proves (ii). 
\end{proof}

The completion of the proof of Theorem \ref{hz} 
depends now on the following computational lemma. 

\begin{lemma}\label{tech}
Let $\varepsilon(n),\mu(n),\lambda(n)$, $n\ge 0$, be
 sequences satisfying the equations
$$
\varepsilon(n)=\sum_i \begin{pmatrix}2n\\ i\end{pmatrix}
\mu(n-i);
$$
$$
\mu(n)=\sum_i \begin{pmatrix}n\\ i\end{pmatrix}
\lambda(n-i).
$$
Assume also that $\varepsilon(n)=\binom{2n}{
  n}f(n)$, where $f$ is a polynomial such that $f(0)=0$. 
Then $\lambda(0)=0$, $\lambda(n)$ has finitely many nonzero
values, and 
$$
\sum_{n\ge 1} (-1)^{n-1}\frac{\lambda(n)}{2n}=f'(0).
$$ 
\end{lemma}

\begin{proof}
Let us first consider any sequences $\varepsilon(n)$,
$\mu(n)$, and $\lambda(n)$ linked by the equations of the lemma. 
Let $E(z)$, $M(z)$, and $L(z)$ be their generating functions
(i.e. $E(z)=\sum_{n\ge 0}\varepsilon(n)z^n$ etc.). 
We claim that 
$$
E(z)=\frac{1+\sqrt{1-4z}}{2(1-4z)}L\left(\frac{1-\sqrt{1-4z}}{2\sqrt{1-4z}}\right).
$$
To see this, it suffices to consider the case 
$\lambda_i=\delta_{ki}$ for some $k$. 
In this case, 
$$
E(z)=\sum_{i,n}\begin{pmatrix} 2n\\ i\end{pmatrix}
\begin{pmatrix} n-i\\ k\end{pmatrix}z^n=
\sum_{p,q\ge 0}\begin{pmatrix} 2p+2q\\ p\end{pmatrix}
\begin{pmatrix} q\\ k\end{pmatrix}z^{p+q}.
$$
But the function 
$$ 
F_r(z):=\sum_{p\ge 0}\begin{pmatrix} 2p+r\\ p\end{pmatrix}z^p
$$
equals
$$
F_r(z)=
\frac{1}{\sqrt{1-4z}}\biggl(\frac{1-\sqrt{1-4z}}{2z}\biggr)^r,
$$
as may be easily seen by induction from the recursion 
$$
F_r=z^{-1}(F_{r-1}-F_{r-2}),
$$ 
$r\ge 2$.  
Substituting this in the formula for $E(z)$, one gets 
(after trivial simplifications)
$$
E(z)=\frac{1+\sqrt{1-4z}}{2(1-4z)}\left(\frac{1-\sqrt{1-4z}}{2\sqrt{1-4z}}\right)^k,
$$
as desired. 

Now assume that $\varepsilon(n)$ satisfies the polynomiality condition. 
This means that $E(z)=P(z\partial)|_{z=0}\frac{1}{\sqrt{1-4z}}$, 
where $P$ is a polynomial with vanishing constant term. 
To prove our claim, it suffices to consider the case 
$P(z)=(1+a)^z-1$, where $a$ is a formal parameter (so $P'(0)=\log(1+a)$); indeed, the coefficients of this formal series 
are $\binom{z}{j}$, $j\ge 1$, which form a basis in the space of polynomials of $z$ with vanishing constant term.  
In this case we get 
$$
E(z)=\frac{1}{\sqrt{1-4(1+a)z}}-\frac{1}{\sqrt{1-4z}}.
$$
Hence, 
$$
L(u)=\frac{1}{1+u}\left(\frac{1}{\sqrt{1-4au(1+u)}}-1\right).
$$
Therefore, 
$$
\sum_k (-1)^{k-1}\frac{\lambda_k}{2k}=\frac{1}{2}\int_{-1}^0 L(u)\frac{du}{u}=
\frac{1}{2}\sum_{p\ge 1} \begin{pmatrix} 2p\\ p\end{pmatrix}
(-1)^{p-1}a^p\int_0^1x^{p-1}(1-x)^{p-1}dx.
$$
But $\int_0^1x^{p-1}(1-x)^{p-1}dx$ is an Euler Beta integral, and it equals 
$\frac{(p-1)!^2}{(2p-1)!}$. Thus, 
$$
\sum_k (-1)^{k-1}\frac{\lambda_k}{2k}=\sum_{p\ge 1}(-1)^{p-1}\frac{a^p}{p}=\log(1+a),
$$ as desired. 
\end{proof}

\subsection{End of proof of Theorem \ref{hz}}
Now we finish the proof of the Harer-Zagier theorem. 
Recall that using matrix integrals we have proved the formula 
\begin{equation}\label{Pn}
P_n(x):=\sum_{\rm g} \varepsilon_{\rm g}(n)x^{n+1-2{\rm g}}
=\frac{(2n)!}{2^nn!}\sum_{p\ge 0}\begin{pmatrix} n\\ p\end{pmatrix}
2^p\begin{pmatrix} x\\ p+1\end{pmatrix}.
\end{equation}

Let us set $q:=n-p$. Then expression (\ref{Pn}) takes the form 
\begin{equation}\label{Pn1}
P_n(x)=\begin{pmatrix} 2n\\ n\end{pmatrix}
\sum_{q\ge 0}2^{-q}\begin{pmatrix} n\\ q\end{pmatrix}\frac{n!}{(n-q+1)!}
x(x-1)...(x-n+q).
\end{equation}

We claim now that the coefficient of $x^{-2{\rm g}}$ (${\rm g}\ge 1$) in the polynomial 
$\frac{P_n(x)}{x^{n+1}}$ is of the form 
$\binom{2n}{n}f_{\rm g}(n)$, where $f_{\rm g}$ is a polynomial. 
Indeed, contributions to the coefficient of $x^{-2{\rm g}}$ come from terms with 
$q\le 2{\rm g}$ only, so it suffices to check that each of these contributions 
is as stated. This reduces to checking that  
the coefficients of the Laurent polynomial $Q(x,n)=(1-\frac{1}{x})...(1-\frac{n}{x})$ are polynomials in $n$,
which vanish at $-1$ (except, of course, the leading coefficient).
To see this, let $Q(x,a)=\frac{\Gamma(x)}{\Gamma(x-a)x^a}$ (this equals to 
$Q(x,n)$ if $a=n$). This function has an asymptotic Taylor expansion 
in $\frac{1}{x}$ as $x\to +\infty$ which is obtained from the Stirling asymptotic expansion of $\Gamma(x)$ given by \eqref{stir}, and it is easy to show that the coefficients are polynomials in 
$a$. Moreover, $Q(x,-1)=1$, which implies the required statement. 

Furthermore, we claim that $f_{\rm g}(0)=0$: again, this follows from the fact that 
the non-leading coefficients of the expansion of $Q(x,a)$ vanish at $a=0$. 
But this is clear, since $Q(x,0)=1$. 

Thus, we are in a situation where 
Lemma \ref{tech} can be applied. 
So it remains to compute $\sum_{{\rm g}\ge 1}f_{\rm g}'(0)x^{-2{\rm g}}$. 
To do this, observe that the terms with $q>1$ do not contribute to 
$f_{\rm g}'(0)$, as they are given by polynomials of $n$ that are divisible by $n^2$.  
So we only need to consider $q=0$ and $q=1$. 
For $q=1$, the contribution is the value of 
$$
\tfrac{1}{2x}(1-\tfrac{1}{x})...(1-\tfrac{n}{x})
$$ 
at $n=0$, i.e. it is $\frac{1}{2x}$. 
For $q=0$, the contribution is the derivative at $0$ with respect to $n$ of 
$\frac{1}{n+1}(1-\frac{1}{x})...(1-\frac{n}{x})$, i.e. it is 
$$
\tfrac{d}{da}|_{a=0}\tfrac{Q(x,a)}{a+1}=-1+\tfrac{d}{da}|_{a=0}Q(x,a).
$$ 
Thus, we have (asymptotically)
$$
\sum_{{\rm g}\ge 1}f_{\rm g}'(0)x^{-2{\rm g}}=\tfrac{1}{2x}+\tfrac{d}{da}|_{a=0}Q(x,a)=
\tfrac{1}{2x}+\tfrac{\Gamma'(x)}{\Gamma(x)}-\log x
$$
Now, the asymptotic expansion for $\Gamma'/\Gamma$ given by \eqref{loggapri}
implies that $f_{\rm g}'(0)=-\frac{B_{2{\rm g}}}{2{\rm g}}$. This completes the proof. 

\begin{exercise} Prove Theorem \ref{mum} for ${\rm g}=1$. 
\end{exercise} 
\begin{exercise} 
Let $\Gamma(N)$ be the congruence subgroup of $SL_2(\Bbb Z)$ 
which consists of matrices equal to $1$ modulo $N$. 

(a) Show that $\Gamma(N)$ is free for $N\ge 3$. 
(Hint: consider the action of $\Gamma(N)$ on the upper half-plane).
Show that $\Gamma(2)$ is the direct product of a free group $\Gamma_+(2)$ 
on two generators with $\Bbb Z/2\Bbb Z$. 

(b) Find the number of generators of $\Gamma(N)$, $N\ge 3$ which generate it without 
relations. (Hint: compute $\chi(\Gamma(N))$).
\end{exercise} 

\begin{exercise} Let $\Gamma$ be the group defined by the generators $a,b,c$ with defining relation
$ab=ba$. Find the Euler characteristic of $\Gamma$. 
\end{exercise} 

\begin{exercise} Consider a triangle $\Delta$ in the hyperbolic plane $H=\Bbb C_+$ with angles $\alpha=\frac{\pi}{2}$, $\beta=\frac{\pi}{3}$, $\gamma=\frac{\pi}{7}$,  and let $\Gamma$ be the subgroup of $PSL_2(\Bbb R)$ generated by rotations $a,b,c$ around the vertices of $\Delta$ by angles $2\alpha,2\beta,2\gamma$ respectively. 

(i) Show that $H/\Gamma$ is naturally homeomorphic to a sphere glued out of two copies of $\Delta$, which can be viewed as an orbifold with three points with nontrivial stabilizers (orders 2,3,7).

(ii) Compute the Euler characteristic $\chi(\Gamma)$. 

(iii) Show that the defining relations for $\Gamma$ are 
$$
a^2=1,\ b^3=1,\ c^7=1,\ abc=1
$$ 
(use an orbifold version of 
van Kampen's theorem). 

(iv) Construct a surjective homomorphism $\phi: \Gamma\to PSL_2(\Bbb F_7)$. 

(v) Show that ${\rm Ker}\phi$ is torsion free and $H/{\rm Ker}\phi$ is a compact Riemann surface $X$
of genus $3$ with an action of $PSL_2(\Bbb F_7)$. Identify $X$ with the Klein quartic $x^3y+y^3z+z^3x=0$ in $\Bbb C\Bbb P^2$. 
\end{exercise}

\section{Matrix integrals and counting planar diagrams}

\subsection{The number of planar gluings}

Let us return to the setting of Section 4. 
Thus, we have a potential 
$$
U(x)=\frac{x^2}{2}-\sum_{j\ge 1}g_j\frac{x^j}{j}
$$
(with $g_j$ being formal parameters), and consider the 
matrix integral 
$$
Z_N(\hbar)=\hbar^{-\frac{N^2}{2}}\int_{{\frak h}_N}e^{-{\rm Tr}U(A)}dA.
$$
Let $\widehat Z_N(\hbar)=Z_N(\hbar/N)$. 
We have seen that 
$$
\lim_{N\to \infty}\frac{\log \widehat Z_N}{N^2}=W_\infty, 
$$
where $W_\infty$ is given by summation over planar fat graphs:
$$
W_\infty=\sum_{\bold n}\prod_i (g_i\hbar^{\frac{i}{2}-1})^{n_i}\sum_{\widetilde{\Gamma}\in
\widetilde{G}_c(\bold n)[0]}\frac{1}
{|{\rm Aut(\widetilde{\Gamma})}|}.
$$
In particular, the coefficient 
of $\prod_i (g_i\hbar^{\frac{i}{2}-1})^{n_i}$ is 
the number of 
(orientation preserving) gluings of a fat graph of genus zero 
out of a collection of fat flowers 
containing $n_i$  $i$-valent flowers for each $i$, 
divided by $\prod_i i^{n_i}n_i!$. 

On the other hand, one can compute $W_\infty$ explicitly 
as a function of $g_i$ by reducing the matrix integral
to an integral over eigenvalues, and then 
using a fundamental fact from the theory of random matrices:
the existence of an asymptotic distribution of eigenvalues
in the limit $N\to \infty$. This approach allows one to obtain 
simple closed formulas for the numbers of planar gluings, 
which are quite nontrivial and 
for which direct combinatorial 
proofs were discovered much later. 

To illustrate this method, we will 
restrict ourselves to the case of the 
potential $U(x)=\frac{x^2}{2}+gx^4$ (so $g_4=-4g$ and other $g_i=0$), 
and set $\hbar=1$. 
Then 
$$
W_\infty=\sum_{n\ge 1}c_n\frac{(-1)^ng^n}{n!},
$$
where
$c_n$ is a number of connected planar gluings of 
a set of $n$ 4-valent flowers. In other words, 
$c_n$ is the number of ways (up to isotopy) to connect 
$n$ ``crosses'' in the 2-sphere so that 
all crosses are connected with each other, all the arms are used, 
and the connecting lines do not intersect. 

\begin{exercise} Check by drawing pictures that 
$c_1=2$, $c_2=36$.
\end{exercise} 

\begin{theorem}\label{BIPZ} (Br\'ezin, Itzykson, Parisi, Zuber, \cite{BIPZ}, 1978).
One has 
$$
c_n=(12)^n \frac{(2n-1)!}{(n+2)!}.
$$
\end{theorem}

The proof of this theorem (with some omissions) is given in the next subsection. 

\subsection{Proof of Theorem \ref{BIPZ}} 
We follow the paper \cite{BIPZ}.
We will assume that $g$ is a positive real number, and 
compute the function $W_\infty(g)$ explicitly. 
The relevant matrix integral has the form 
$$
\widehat Z_N=\int_{{\frak h}_N}e^{-N\Tr(\frac{1}{2}A^2+gA^4)}dA.
$$
Passing to eigenvalues, we get
$$
\widehat Z_N=\frac{J_N(g)}{J_N(0)},
$$
where 
\begin{equation}\label{integrand} 
J_N(g)=\int_{\Bbb R^N}
e^{-N(\frac{1}{2}\sum_i \lambda_i^2+g\sum_i \lambda_i^4)}
\prod_{i<j} (\lambda_i-\lambda_j)^2d\lambda.
\end{equation} 
Thus, $W_\infty(g)=E(g)-E(0)$, 
where $E(g)=\lim_{N\to \infty}N^{-2}\log J_N(g)$. 

\begin{proposition} (Steepest descent principle)
$E(g)$ equals the leading coefficient of the asymptotics 
as $N\to \infty$ of the maximal value of the logarithm of the integrand in \eqref{integrand}. 

\end{proposition}

The proposition says, essentially, that 
the integrand has a sufficiently sharp maximum, 
so that the leading behavior of the integral 
can be computed by the steepest descent formula. 
We note that we cannot apply the steepest descent formula
without explanations, since the integral 
is over a space whose dimension grows as 
the perturbation parameter $1/N$ goes to $0$. 
In other words, it is necessary to do some estimates which we will omit. 
We will just mention that for $g=0$, this result can be derived from 
the explicit evaluation of the integral using Hermite polynomials
(see \S4). For the general case, we refer the reader to the book \cite{De}.

The logarithm of the integrand 
$$
K(\lambda_1,...,\lambda_N):=
-N(\tfrac{1}{2}\sum_i \lambda_i^2+g\sum_i \lambda_i^4)+2\sum_{i<j}\log|\lambda_i-\lambda_j|
$$  
has a unique maximum, because it is concave (check it!). 
This maximum is found by equating the partial derivatives to zero. 
This yields 
\begin{equation}\label{criti}
\sum_{j\ne i}\frac{1}{\lambda_i-\lambda_j}=N(\tfrac{1}{2}\lambda_i+2g\lambda_i^3).
\end{equation}
Let $\lambda_1<\lambda_2<...<\lambda_N$ be the unique (up to permutations) 
solution of this system of equations. 

\begin{proposition}
The normalized counting measures $\frac{1}{N}\sum_i \delta(x-\lambda_i)$ converge 
weakly to a measure $\mu(x)=f(x,g)dx$, where 
$f(x,g)$ is a continuous function supported on a finite interval $[-2a,2a]$ 
and differentiable on the interior of this interval.  
\end{proposition}

For the proof we again refer the reader to \cite{De} (p.132 and later). 
We note that for $g=0$, 
by Wigner's semicircle law, $a=1$ and 
$f(x,0)=\frac{1}{2\pi}\sqrt{4-x^2}$; so $f(x,g)=
\frac{1}{2\pi}\sqrt{4-x^2}+O(g)$.

Now our job will be to find the function $f(x,g)$. 
Passing to the limit in equation \eqref{criti}
(which requires justification that we will omit), we 
get 
$$
\int_{-2a}^{2a}\frac{1}{y-x}f(x,g)dx=\frac{1}{2}y+2gy^3,\ |y|\le 2a
$$
where the integral is understood in the sense of principal value.

This is a linear integral equation on $f(x,g)$, which can be solved in 
a standard way. Namely, one considers the analytic function 
$$
F(y)=\int_{-2a}^{2a}\frac{1}{y-x}f(x,g)dx
$$ 
for $y$ in the complex plane but 
outside of the interval $[-2a,2a]$. 
For $y\in [-2a,2a]$, let $F_+(y)$, $F_-(y)$ denote the limits of $F(y)$ from above 
and below. Then by the Plemelj formula, the integral equation implies 
$$
\frac{1}{2}(F_+(y)+F_-(y))=\frac{1}{2}y+2gy^3.
$$
On the other hand, $F_-(y)=\overline{F_+(y)}$. Hence, 
$$
{\rm Re}F_+(y)={\rm Re}F_-(y)=\frac{1}{2}y+2gy^3.
$$ 

Now set $y:=a(z+z^{-1})$. Then, as $y$ runs through the exterior of $[-2a,2a]$, 
$z$ runs through the exterior of the unit circle. So the function $G(z):=F(y)$ is 
analytic on the outside of the unit circle, with decay at infinity, and 
$$
{\rm Re}G(z)=\frac{1}{2}a(z+z^{-1})+2ga^3(z+z^{-1})^3
$$ 
when $|z|=1$. This implies that $G(z)$ is 
twice the sum of all negative degree terms of this
Laurent polynomial. In other words, we have
$$
G(z)=4ga^3z^{-3}+(a+12ga^3)z^{-1}.
$$
This yields 
$$
F(y)=\frac{1}{2}y+2gy^3
-\left(\frac{1}{2}+4ga^2+2gy^2\right)\sqrt{y^2-4a^2}.
$$
Now $f(y,g)$ is found as the jump of $F$: 
$$
f(y,g)=\frac{1}{\pi}\left(\frac{1}{2}+4ga^2+2gy^2\right)\sqrt{4a^2-y^2}.
$$

It remains to find $a$ in terms of $g$.
We have $yF(y)\to 1$, $y\to \infty$ (as $\int f(x,g)dx=1$),
hence $zG(z)\to 1/a$, $z\to \infty$. 
This yields 
$$
\frac{1}{a}=a+12ga^3,
$$ 
or 
$$
12ga^4+a^2-1=0.
$$
This allows one to determine $a$ 
uniquely: 
$$
a=\biggl(\frac{(1+48g)^{1/2}-1}{24g}\biggr)^{1/2}.
$$

Now let us calculate 
$E(g)$. It follows from the above that 
$$
E(g)=\int_{-2a}^{2a}\int_{-2a}^{2a}\log|x-y|f(x,g)f(y,g)dxdy-
\int_{-2a}^{2a}(\tfrac{1}{2}x^2+gx^4)f(x,g)dx.
$$
On the other hand, let us integrate the integral equation 
defining $f(x,g)$ with respect to $y$ (from $0$ to $u$). Then we get
$$
2\int_{-2a}^{2a}(\log|x-u|-\log|x|)f(x,g)dx=\tfrac{1}{2}u^2+gu^4.
$$
Substituting this into the expression for $E(g)$, we get 
$$
E(g)=\int_{-2a}^{2a}(\log|u|-\tfrac{1}{4}u^2-\tfrac{1}{2}gu^4)f(u,g)du.
$$
Since $f(u,g)$ is known, this integral can be computed. 
In fact, can be expressed via elementary functions, and after calculations we get
$$
E(g)-E(0)=\log a-\frac{1}{24}(a^2-1)(9-a^2).
$$
Substituting here the expression for $a$, after a calculation one finally gets:
$$
E(g)-E(0)=\sum_{k=1}^\infty (-12g)^k\frac{(2k-1)!}{k!(k-2)!}.
$$
This implies the required formula for $c_n$. 

\section{Quantum mechanics}

So far we have considered quantum field theory with 0-dimensional
spacetime (to make a joke, one may say that the dimension of 
the space is $-1$). In this section, we will move closer to 
actual physics: we will consider 1-dimensional spacetime, 
i.e. the dimension of the space is $0$. This does not mean
that we will study motion in a 0-dimensional space
(which would be really a pity) but just
means that we will consider only point-like
quantum objects (particles) and not extended quantum objects 
(fields). In other words, we will be in the realm of quantum
mechanics.

\subsection{The path integral in quantum mechanics}
Let $U(q)$ be a smooth function on the real line (the potential).
We will assume that $U(0)=0$, $U'(0)=0$, 
and $U''(0)=m^2$, where $m>0$. 

\begin{remark} In quantum field theory the parameter $m$ in the
potential is called the {\it mass}\index{mass} parameter. 
To be more precise, in classical mechanics it has the meaning 
of frequency $\omega$ of oscillations. However, in quantum theory
thanks to Einstein frequency is identified with energy
($E=\hbar \omega/2\pi$), while in relativisitic theory 
energy is identified with mass (again thanks to Einstein,
$E=mc^2$). 
\end{remark} 

We want to construct the theory of a quantum particle 
moving in the potential field $U(q)$. 
According to what we discussed before, 
this means that we want to give sense 
to and to evaluate the normalized correlation 
functions
$$
\langle q(t_1) \ldots q(t_n)\rangle:=
\frac{\int q(t_1) \ldots q(t_n)e^{\frac{iS(q)}{\hbar}}Dq}{\int e^{\frac{iS(q)}{\hbar}}Dq}, 
$$
where $S(q)=\int {\mathcal L}(q)dt$, and 
${\mathcal L}(q)=\frac{\dot{q}^2}{2}-U(q)$. 

As we discussed, such integrals cannot be handled
rigorously by means of measure theory if $\hbar$ is a positive number; so we will 
only define these path integrals ``in perturbation theory'', i.e. 
as formal series in $\hbar$. 

Before giving this (fully rigorous) definition, we
will explain the motivation behind it. 
We warn the reader that this explanation is
heuristic and involves steps 
which are mathematically non-rigorous
(or ``formal'' in the language of physicists).

\subsection{Wick rotation}

In Section 1 we discussed path integrals with imaginary exponential
(quantum mechanics), as well as real exponential (Brownian
motion). If $\hbar$ is a number, then the integrals with imaginary exponential 
cannot be defined measure-theoretically. Therefore, 
people study integrals with real exponential (which can be rigorously defined),
and then perform a special analytic continuation procedure called the {\it Wick rotation}.\index{Wick rotation} 

In our formal setting ($\hbar$ is a formal parameter), 
one can actually define the integrals in both the real and 
the imaginary case. Still, the real case is a bit easier, 
and thus the Wick rotation is still useful. Besides, the 
Wick rotation is very important conceptually. Therefore, 
while it is not technically necessary, we 
start with introducing the Wick rotation here. 

Namely, let us 
denote $\langle q(t_1)...q(t_n)\rangle$ by $\G^M_n(t_1,...,t_n)$, and 
``formally'' make a change of variable 
$\tau=it$ in the formula for $\G^M_n(t_1,...,t_n)$.
Let $q(t):=q_*(\tau)$.  
Then, taking into account that $d\tau=idt$, 
$\frac{dq}{dt}=i\frac{dq_*}{d\tau}$, we get 
$$
\G^M_n(t_1,...,t_n)=
\frac{\int q_*(\tau_1) \ldots q_*(\tau_n)e^{-\frac{1}{\hbar}\int (\frac{1}{2}
(\frac{dq_*}{d\tau})^2+U(q_*))d\tau}Dq_*}{\int e^{-\frac{1}{\hbar}\int (\frac{1}{2}
(\frac{dq_*}{d\tau})^2+U(q_*))d\tau}Dq_*}. 
$$
This shows that 
$$
\G^M_n(t_1,...,t_n)=\G^E_n(it_1,...,it_n),
$$
where 
$$
\G^E_n(t_1,...,t_n):=
\frac{\int q(t_1) \ldots q(t_n)e^{-\frac{S_E(q)}{\hbar}}Dq}{\int e^{-\frac{S_E(q)}{\hbar}}Dq}. 
$$
with $S_E(q)=\int {\mathcal L}_E(q)d\tau$, and 
${\mathcal L}_E(q)=\frac{\dot{q}^2}{2}+U(q)$
(i.e. ${\mathcal L}_E$ is obtained from 
${\mathcal L}$ by replacing $U$ with $-U$). 

This manipulation certainly does not make rigorous sense, 
but it motivates the following definition. 

\begin{definition}
The function $\G^M_n(t_1,...,t_n)$ ($t_i\in \Bbb R$) is the analytic continuation of 
the function $\G^E_n(\tau_1,...,\tau_n)$ from the point 
$(t_1,...,t_n)$ to the point 
$(it_1,...,it_n)$ along the path 
$\theta\mapsto e^{i\theta}(t_1,...,t_n)$, $0\le \theta\le \pi/2$. 
\end{definition}

Of course, this definition will only make sense if 
we define the function $\G^E_n(t_1,...,t_n)$ and show that it admits
the required analytic continuation. This will be done below. 

\begin{remark} (On the terminology.)
The function $\G^M_n(t_1,...,t_n)$ is called the {\it Minkowskian}\index{Minkowskian correlation function} (time ordered) correlation function, while the function 
$\G^E_n(t_1,...,t_n)$ is called the {\it Euclidean}\index{Euclidean correlation function} correlation function
(hence the notation). 
This terminology will be explained later, when we consider
relativistic field theory. 

 From now on, we will mostly deal with Euclidean correlation
functions, and therefore will omit the superscript $E$ when
there is no danger of confusion. 
\end{remark}

\subsection{Definition of Euclidean correlation functions}

Now our job is to define the Euclidean correlation functions 
$\G_n(t_1,...,t_n)$. 
Our strategy (which will also be used in field theory) 
will be as follows. Recall that 
if our integrals were finite dimensional then 
by Feynman's theorem the expansion 
of the correlation functions in $\hbar$ would be given
by a sum of amplitudes of Feynman diagrams. 
So, in the infinite dimensional case, 
we will use the sum over Feynman diagrams 
as a {\it definition} of correlation functions. 

More specifically, because of the conditions on $U$
we have an action functional without 
constant and linear terms in $q$, so that 
the correlation function $\G_n(t_1,...,t_n)$ 
should be given by the sum
\begin{equation}\label{defcor}
\G_n(t_1,...,t_n)=
\sum_{\Gamma\in G_{\ge 3}^*(n)}\frac{\hbar^{b(\Gamma)}}{|{\rm Aut}(\Gamma)|}
F_\Gamma(\ell_1, \ldots ,\ell_n),
\end{equation}
where $G_{\ge 3}^*(n)$ is defined in Remark \ref{G3*}.
Thus, we should make sense of 
(=define) the amplitudes $F_\Gamma$ 
in our situation. For this purpose, we 
need to define the following objects. 

1. The space $V$. 

2. The form $B$ on $V$ which defines $B^{-1}$ on $V^*$. 

3. The tensors corresponding to non-quadratic terms in the
action. 

4. The covectors $\ell_i$. 

It is clear how to define these objects naturally. Namely, 
$V$ should be a space of functions 
on $\Bbb R$ with some decay conditions. 
There are many choices for $V$, which do not affect the final
result; for instance, a good choice (which we will make) 
is the space $C^\infty_0(\Bbb R)$ 
of compactly supported smooth functions on $\Bbb R$.
Thus $V^*$ is the space of generalized functions on $\Bbb R$.
Note that $V$ is equipped with 
the inner product $(f,g)=\int_{\Bbb R}f(x)g(x)dx$.

The form $B$, by analogy with the finite dimensional case, 
should be twice the quadratic part of the action. 
In other words, 
$$
B(q,q)=\int(\dot{q}^2+m^2q^2)dt=(Aq,q),
$$ 
where $A$ is the operator 
$$
A=-\frac{d^2}{dt^2}+m^2.
$$ 
This means that $B^{-1}(f,f)=(A^{-1}f,f)$.

The operator $A^{-1}$ is an integral operator, 
with kernel 
$$
K(x,y)=G(x-y),
$$ 
where 
$G(x)$ is the Green's function of $A$, i.e. 
the fundamental (decaying at infinity) solution of the 
differential equation 
$$
(AG)(x)=\delta(x).
$$ 
It is straightforward to find that 
$$
G(x)=\frac{e^{-m|x|}}{2m}.
$$ 
(thus $B^{-1}$ is actually defined not on the whole $V^*$ 
but on a dense subspace of $V^*$). 

\begin{remark} Here we already see the usefulness of the Wick
rotation. Namely, the spectrum of $A$ (interpreted as usual as a 
self-adjoint unbounded operator on $L^2(\Bbb R)$) is $[m^2,+\infty)$, so it is
invertible and the inverse is bounded. However, 
if we did not make a Wick rotation, we would deal
with the operator $A'=-\frac{d^2}{dt^2}-m^2$, whose spectrum is
$[-m^2,+\infty)$, i.e. contains $0$, so this operator 
does not have a bounded inverse. 
\end{remark}

To make sense of the cubic and higher terms in the action as
tensors, consider the decomposition of $U$ in the (asymptotic)
Taylor series at $x=0$: 
$$
U(x)=\frac{m^2x^2}{2}-\sum_{n\ge 3}\frac{g_nx^n}{n!}.
$$ 
This shows that cubic and higher terms in the action
have the form 
$$
B_r(q,q,,...,q)=\int q^r(t)dt.
$$
Thus $B_r(q_1,...,q_r)$ is an element of $(S^rV)^*$ 
given by the generalized function $\delta_{t_1=...=t_r}$
(the delta function of the diagonal). 

Finally, the functionals $\ell_i$ are given by
$\ell_i(q)=q(t_i)$, so $\ell_i=\delta(t-t_i)$. 

This leads to the following Feynman rules
of defining the amplitude of a diagram $\Gamma$. 

1. To the $i$-th external vertex of $\Gamma$ assign 
the number $t_i$. 

2. To each internal vertex $j$ of $\Gamma$, assign
a variable $s_j$. 

3. On each internal edge connecting vertices $j$ and $j'$, write the Green's function 
$G(s_j-s_{j'})$. 

4. On each external edge connecting 
$i$ and $j$ write $G(t_i-s_j)$. 

5. On each external edge connecting 
$i$ and $i'$ write $G(t_i-t_{i'})$. 

6. Let $G_\Gamma(\bold t,\bold s)$ 
be the product of all these functions. 

7. Let $F_\Gamma(\ell_1,...,\ell_n):=\prod_j g_{v(j)}\int G_\Gamma(\bold t,\bold
s)d\bold s$, where $v(j)$ is the valency of $j$. 

We are finally able to give the following definition. 

\begin{definition}
The function $\G_n(t_1,...,t_n)$ is defined by formula 
\eqref{defcor}.
\end{definition}

\begin{remark} Note that the integrals defining $F_\Gamma$ are
convergent since the integrand always decays exponentially at
infinity. It is, however, crucial that we consider only graphs 
without components having no external vertices; 
for example, if $\Gamma$ 
has a single 4-valent vertex connected to itself by 
two loops (Fig.21) then the amplitude integral involves 
$\int_{\Bbb R}G(0)^2ds$, which is obviously divergent. 
\end{remark}

With this definition, the function $\G_n(t_1,...,t_n)$ is a Laurent
series in $\hbar$ whose coefficients are symmetric functions 
of $t_1,...,t_n$ given by linear combinations of explicit (and convergent)
finite dimensional integrals. 
Furthermore, it is easy 
to see that these integrals are in fact computable in
elementary functions, i.e. are (in the region $t_1\ge...\ge t_n$)
linear combinations
of products of functions of the form $t_i^re^{at_i}$. 
This implies the existence of the analytic continuation 
required in the Wick rotation procedure. 

\begin{figure}[htbp]
  \begin{center}
    
    \setlength{\unitlength}{4ex}
    \begin{picture}(10,4)(-5,-2)

      \cbezier(0,0)(-2.5,3)(-2.5,-3)(0,0)
      \cbezier(0,0)(2.5,3)(2.5,-3)(0,0)
      \put(0,0){\circle*{.2}}      

    \end{picture}
    \caption{ }
    \label{fig:21}
  \end{center}
\end{figure}

\begin{remark} As in the finite dimensional case, an alternative setting for making this definition
is to assume that $g_i$ are formal parameters. In this case,
$\hbar$ can be given a numerical value, e.g. $\hbar=1$, and the 
function $\G_n$ will be a well defined power 
series in $g_3,g_4,...$. 
\end{remark}

As an example consider a free massive theory, i.e., a {\it harmonic oscillator}\index{harmonic oscillator}: $U(q)=\frac{m^2q^2}{2}$. In this case, there are no internal vertices, 
hence we get

\begin{proposition} (Wick's theorem) One has 
$\G_n(t_1,...,t_n)=0$ if $n$ is odd, and 
$$
\G_{2k}(t_1,...,t_{2k})=\hbar^{k}\sum_{\sigma\in \Pi_k}\prod_{i\in
  \lbrace{1,...,2k
\rbrace}/\sigma}G(t_i-t_{\sigma(i)}).
$$
\end{proposition}

In particular, $\G_2(t_1,t_2)=\hbar G(t_1-t_2)$. 
In other words, $\G_2(t_1,t_2)$ is (proportional to) 
the Green's function. Motivated by this, physicists 
often refer to all correlation functions of a quantum 
field theory as {\it Green's functions}.\index{Green's functions}

\begin{figure}[htbp]
  \begin{center}
    
    \setlength{\unitlength}{3ex}
    \begin{picture}(8,6)(-4,-1)

      \put(-4,0){\line(1,0){8}}
      \cbezier(0,0)(-3,2.5)(3,2.5)(0,0)
      \put(0,0){\circle*{.2}}      

      \put(-4,-0.5){\makebox(0,0)[l]{$t_1$}}
      \put(4,-0.5){\makebox(0,0)[l]{$t_2$}}

    \end{picture}
    \caption{ }
    \label{fig:22}
  \end{center}
\end{figure}

\begin{example} Consider the potential $U(q)=\frac{m^2q^2}{2}-\frac{gq^4}{24}$,
and set $\hbar=1$. 
In this case, let us calculate 
the 2-point correlation function modulo $g^2$. 
In other words, we have to compute the coefficient of $g$ in this function. 
Thus we have to consider Feynman diagrams
with two external edges and one internal vertex. 
Such a diagram $\Gamma$ is unique: it consists of one edge with a
loop attached in the middle (Fig. 22). This diagram has automorphism group
$\Bbb Z/2$. The amplitude of this diagram is 
$$
F_\Gamma=g\int_{\Bbb R}G(s,t_1)G(s,t_2)G(s,s)ds=
\frac{g}{8m^3}\int_{\Bbb R}e^{-m(|s-t_1|+|s-t_2|)}ds.
$$
Because of symmetry in $t_1$ and $t_2$, we may assume that
$t_1\ge t_2$. Splitting the integral in a sum of 
three integrals, over $(-\infty,t_2],[t_2,t_1]$, and $[t_1,\infty)$,
respectively we get:
$$
F_\Gamma=
\frac{g}{8m^3}\left(2\int_0^\infty e^{-m(2s+|t_1-t_2|)}ds+|t_1-t_2|e^{-m|t_1-t_2|}\right)=
$$
$$
\frac{g}{8m^4}e^{-m|t_1-t_2|}(1+m|t_1-t_2|).
$$
Thus
$$
\G_2(t_1,t_2)=\widetilde{G}(t_1-t_2),
$$
where
$$
\widetilde{G}(t):=\tfrac{1}{2m}e^{-m|t|}+\tfrac{g}{16m^4}e^{-m|t|}(1+
m|t|)+O(g^2).
$$
This expression is called the
{\it 1-loop approximation}\index{1-loop approximation} to the 2-point function, because it comes
from 0-loop and 1-loop Feynman diagrams. 
\end{example} 

\begin{remark} Here we are considering 
quantum mechanics of a single 1-dimensional particle.
However, everything generalizes 
without difficulty to 
the case of an $n$-dimensional particle 
or system of particles (i.e., to path integrals over the
space of vector-valued, rather than scalar, functions of one
variable). 
Indeed, if $q$ takes values in a Euclidean space $V$ 
then the quadratic part of the Lagrangian is 
of the form $\frac{1}{2}(\dot{q}^2-M(q))$, where $M$ is a
positive definite
quadratic form on $V$. Diagonalizing $M$, we may assume
that the quadratic part of the Lagrangian looks like 
$\frac{1}{2}\sum_i(\dot{q_i}^2-m_i^2q_i^2)$, which corresponds to
a system of independent harmonic oscillators. Thus in quantum
theory the propagator will be the diagonal matrix 
with diagonal entries $\frac{e^{-m_i|t-s|}}{2m_i}$, and the correlation
functions can be defined by the usual Feynman diagram procedure. 
\end{remark}

\subsection{Connected correlation functions}

Let $\G^{c}_n(t_1,...,t_n)$ be the {\it connected correlation (or Green) functions}\index{connected correlation functions}, 
defined by the sum of the same amplitudes as $\G_n(t_1,...,t_n)$ 
but taken over connected Feynman diagrams only. It is clear 
that 
$$
\G_n(t_1,...,t_n)=\sum_{\lbrace{1,...,n\rbrace}=S_1\sqcup...\sqcup
  S_k} 
\prod \G^c_{|S_i|}(t_j; j\in S_i).
$$
For example, $\G_2(t_1,t_2)=\G_2^c(t_1,t_2)+\G_1^c(t_1)\G_1^c(t_2)$, 
etc. Thus, to know the correlation functions, it is sufficient 
to know the connected correlation functions. 

\begin{example}\label{harosc} In a free theory ($U=\frac{m^2q^2}{2}$, the harmonic oscillator), all 
connected Green's functions except $\G_2$ vanish. 
\end {example} 

\begin{figure}[htbp]
  \begin{center}
    
    \setlength{\unitlength}{3ex}
    \begin{picture}(9,7)(-4.5,-3.5)

      \put(-4,0){\line(1,0){8}}
      \put(0,-3){\line(0,1){6}}
      \put(0,0){\circle*{.3}}      

      \put(-4.5,0){\makebox(0,0)[l]{$t_1$}}
      \put(4.5,0){\makebox(0,0)[l]{$t_3$}}
      \put(0,3.5){\makebox(0,0)[l]{$t_2$}}
      \put(0,-3.5){\makebox(0,0)[l]{$t_4$}}

    \end{picture}
    \caption{ }
    \label{fig:23}
  \end{center}
\end{figure}

\begin{example} Let us compute the 
connected 4-point function in the theory associated to the quartic
potential $U=\frac{m^2q^2}{2}-\frac{gq^4}{4}$ as above, modulo $g^2$. This means, we should compute the
contribution of connected Feynman diagrams with one internal
vertex and 4 external edges. 
Such a diagram $\Gamma$ is unique -- it is the cross (with one
internal vertex), Fig. 23. This diagram 
has no nontrivial automorphisms. 
Thus, 
$$
\G_4^c(t_1,t_2,t_3,t_4)=g\int_{\Bbb R}
G(t_1-s)G(t_2-s)G(t_3-s)G(t_4-s)ds+O(g^2).
$$
It is elementary to compute this integral; we leave it as an
exercise. 
\end{example} 

\subsection{The clustering property} 
Note that the Green's function $G(t)$ goes to zero 
at infinity. 
This implies the following {\it clustering property}\index{clustering property}
of the correlation functions of the free theory: 
$$
\lim_{z\to \infty} \G_{n}(t_1,...,t_r,t_{r+1}+z,...,t_n+z)=
\G_r(t_1,...,t_r)\G_{n-r}(t_{r+1}...t_n).
$$
Moreover, it is easy to show 
that the same is true in the interacting theory (i.e. with potential) in each degree with
respect to $\hbar$ (check it!). 
The clustering property can be more simply expressed by the equation
$$
\lim_{z\to \infty} \G_{n}^c(t_1,...,t_r,t_{r+1}+z,...,t_n+z)=0.
$$
This property has a physical interpretation: processes distant 
from each other are almost statistically independent. Thus it can be
viewed as a necessary condition of a quantum field theory to be ``physically 
meaningful''. 

\begin{remark} Nevertheless, there exist theories 
(e.g. so called {\it topological quantum field theories}\index{topological quantum field theory}) which 
do not satisfy the
clustering property but are interesting both 
form a physical and mathematical point of view (see Subsection \ref{fqm} below). 
\end{remark} 

\subsection{The partition function}
Let $J(t)dt$ be a compactly supported measure on the real line. 
Consider the ``partition function with external current $J$'', 
which is the formal expression 
$$
Z(J)=\int e^{\frac{-S_E(q) +(J,q)}{\hbar}}Dq.
$$ 
Then we have a formal equality 
$$
\frac{Z(J)}{Z(0)}=\sum_n \frac{\hbar^{-n}}{n!}\int_{\Bbb R^n}
\G_n(t_1,...,t_n)J(t_1)...J(t_n)dt_1...dt_n,
$$
which, as before, we will use as the definition 
of $Z(J)/Z(0)$. So the knowledge of $Z(J)/Z(0)$ 
is equivalent to the knowledge of all the Green's 
functions (in other words, $Z(J)/Z(0)$ is their generating
function). Furthermore, as in the finite dimensional case, 
we have 

\begin{proposition} One has
$$
W(J):=\log \frac{Z(J)}{Z(0)}=
\sum_n \frac{\hbar^{-n}}{n!}\int\G_n^c(t_1,...,t_n)J(t_1)...J(t_n)dt_1...dt_n
$$
(i.e. $W$ is the generating function of connected Green's functions)
\end{proposition}

The proof of this proposition is the same as in the finite
dimensional case. 

\begin{remark} The statement of the proposition is equivalent to the 
relation between usual and connected Green's functions given in
the previous subsection. 
\end{remark} 

\begin{remark} The fact that we can only define
amplitudes of graphs whose all components have at least one
1-valent vertex (see above) means that we actually cannot define
either $Z(0)$ or $Z(J)$ but can only define their ratio 
$Z(J)/Z(0)$. 
\end{remark} 

Like in the finite dimensional case, we 
have an expansion 
$$
W(J)=\hbar^{-1}W_0(J)+W_1(J)+\hbar W_2(J)+...,
$$
where $W_j$ are the $j$-loop contributions
(in particular, $W_0$ is given by a sum over trees).
Furthermore, we have explicit formulas for $W_0$ and $W_1$, 
analogously to the finite dimensional case. 

\begin{proposition} One has 
$$
W_0(J)=-S_E(q_J)+(q_J,J), 
$$
where $q_J$ is the extremal of the functional
$S_E^J(q):=S_E(q)-(q,J)$ which decays at infinity. 
Furthermore, 
$$
W_1(J)=-\frac{1}{2}\log \det L_J,
$$ where $L_J$ is the linear operator on $V$ such that 
$$
d^2S_E^J(q_J)(f_1,f_2)=d^2S_E^0(0)(L_Jf_1,f_2).
$$
\end{proposition}

The proof of this proposition, in particular, 
involves showing that $q_J$ is well defined and that 
$\det L_J$ exists. It is analogous to the proof 
of the same result in the finite dimensional case which is given in 
Subsection \ref{treepar} (to be precise, we gave a proof only in 
 the 0-loop case; but in the 1-loop case, 
the proof is similar). Therefore we will not give this proof; 
rather, we will illustrate the statement by an example. 

\begin{example} Let $U$ be the above quartic potential $\frac{m^2q^2}{2}+\frac{gq^4}{2}$ 
(in which for convenience we change the sign and normalization of the quartic term) and
$J(t)=a\delta(t)$. 
In this case, 
$$
S_E^J(q)=\int(\tfrac{\dot{q}^2}{2}+U(q))dt-aq(0).
$$ 
The Euler-Lagrange equation has the form 
$$
\ddot{q}=m^2q+2gq^3-a\delta(t).
$$
Thus, the function $q_J$ is continuously glued from two 
solutions $q_+,q_-$ of the nonlinear differential equation 
$$
\ddot{q}=m^2q+2gq^3
$$ 
on $(-\infty,0]$ and $[0,\infty)$, 
with jump of derivative at $0$ equal to $-a$. 

The solutions $q_+,q_-$ are required to decay at infinity, so
they must be solutions of zero energy: 
$$
E=\frac{\dot{q_\pm}^2}{2}-U(q_\pm)=0.
$$
Thus, by the standard formula for solutions of Newton's equation,
they are defined by the equality 
$$
t-t_{\pm}=\int\frac{dq}{\sqrt{2U(q)}}=\int \frac{dq}{mq\sqrt{1+\frac{gq^2}
{m^2}}}=
\frac{1}{2m}\log\frac{\sqrt{1+\frac{gq^2}{m^2}}-1}
{\sqrt{1+\frac{gq^2}{m^2}}+1}.
$$
After a calculation one gets 
$$
q_J(t)=\frac{2mg^{-\frac{1}{2}}}{C^{-1}e^{m|t|}-Ce^{-m|t|}},
$$
where $C$ is the solution of the equation 
$$
\frac{C(1+C^2)}{(1-C^2)^2}=\frac{ag^{\frac{1}{2}}}{4m^2}
$$
which is given by a power series in $a$ with zero constant term. 
From this it is elementary (but somewhat lengthy) to compute 
$W_0=-S_E^J(q_J)$. 

Now, the operator $L_J$ is given by the formula 
$$
L_J=1+\frac{gA^{-1}\circ q_J(t)^2}{2},
$$ 
where $A=-\frac{d^2}{dt^2}+m^2$. 
Thus $\det L_J$ makes sense.
Indeed, the operator 
$A^{-1}\circ q_J(t)^2$ is an integral operator 
given by the kernel 
$$
K_J(x,y):=\frac{e^{-m|x-y|}q_J(y)^2}{2m},
$$ 
which decays exponentially at infinity; 
hence the determinant of the operator
$1+\frac{gA^{-1}\circ q_J(t)^2}{2}$ is well defined. 
\end{example}

\begin{remark} In these computations, $g, a$ were formal variables, but the above computations in fact make sense for real numerical values of these variables as long as $ga^2+m^2>0$.  
\end{remark} 

\subsection{1-particle irreducible Green's functions}

Let $\G^{1PI}_n(t_1,...,t_n)$ denote 
{\it 1-particle irreducible Green's functions}\index{1-particle irreducible Green's function}, i.e. 
those defined by the sum of the same amplitudes 
as the usual Green's functions, but 
taken only over 1-particle irreducible 
Feynman graphs. Define also the {\it amputated}\index{amputated 1-particle irreducible Green's function} 1-particle irreducible Green's
function: $\G^{1PIa}_n:=A^{\otimes n}\G^{1PI}_n$
(it is defined by the same sum of amplitudes, except that
instead of $G(t_i-s_j)$ for external edges, we write $\delta(t_i-s_j)$).

Let $S_{\rm eff}(q)$ be the generating function of $\G_n^{1PIa}$ i.e., 
$$
S_{\rm eff}(q)=\sum_n\frac{\hbar^{-n}}{n!}\int
\G^{1PIa}_n(t_1,...,t_n)
q(t_1)...q(t_n)dt_1...dt_n.
$$

\begin{proposition} The function $W(J)=\log(Z(J)/Z(0))$
is the Legendre transform of $S_{\rm eff}(q)$, i.e. it equals 
$-S_{\rm eff}(\widetilde q_J)+(J,\widetilde q_J)$, where $\widetilde q_J$ is the 
extremal of $-S_{\rm eff}(q)+(J,q)$ decaying at infinity. 
\end{proposition}

The proof of this proposition is the same as in the finite
dimensional case. The proposition shows that 
in order to know the Green's functions, 
it ``suffices'' to know 
amputated 1-particle irreducible Green's functions 
(the generating function of usual Green's functions can be reconstructed 
from that for 1PI Green's functions by 
taking the Legendre transform and exponentiation). 
Which is a good news, since there are a lot fewer  
1PI diagrams than general connected diagrams. 

\subsection{Momentum space integration}

We saw that the amplitude of a Feynman diagram 
is given by an integral over the space of dimension equal 
to the number of internal vertices. This is sometimes inconvenient, 
since even for tree diagrams such integrals can be rather 
complicated. However, it turns out that if one passes to Fourier transforms
then Feynman integrals simplify and in particular the number 
of integrations for a connected diagram becomes equal to the number 
of loops (so for tree diagrams we have no integrations at all).

Namely, we will proceed as follows. 
Instead of the time variable $t$ we will consider 
the dual energy variable $E$. A function
$q(t)$ with compact support will be replaced 
by its Fourier transform $\widehat q(E)$. 
Then, by Plancherel's theorem, for real functions
$q_1,q_2$, we have 
$$
(q_1,q_2)=\int_{\Bbb R}q_1(t)q_2(t)dt=\int_{\Bbb R} \widehat q_1(E)\overline{\widehat q_2(E)}dE=
\int_{\Bbb R} \widehat q_1(E)\widehat q_2(-E)dE.
$$
This implies that the propagator 
is given by 
$$
B^{-1}(f,f)=\int_{\Bbb R} \frac{1}{E^2+m^2}\widehat f(E)\widehat f(-E)dE.
$$
The vertex tensors standing at $k$-valent vertices
were $\delta_{s_1=...=s_k}$, so they will be replaced 
by $\delta_{Q_1+...+Q_k=0}$, where $Q_i$ are dual variables to 
$s_i$. 

\begin{remark} (On terminology) Physicists refer 
to the time variables $t_i,s_j$ as {\it position variables},\index{position variables}
and to energy variables $E_i,Q_k$ as {\it momentum variables},\index{momentum variables} since 
in relativistic mechanics (which is the setting we will deal with
when we study field theory) there is no distinction between time and
position and between energy and momentum
(due to the action of the Lorentz group). 
\end{remark}

This shows that the Feynman rules 
``in momentum space'' for a given connected Feynman diagram $\Gamma$ 
with $n$ external vertices are as follows. 

1. Orient the diagram $\Gamma$, so that all external edges 
are oriented inwards.  

2. Assign variables $E_i$ to external edges, and variables $Q_j$ to
internal ones. These variables are subject to 
the linear equations of ``the first Kirchhoff law":
at every internal vertex, the sum of the variables
corresponding to the incoming edges equals the sum of those 
corresponding to the outgoing edges. 
Let $Y(\bold E)$ be the space of solutions $\bold Q$ of these
equations (it depends on $\Gamma$, but we will not write 
the dependence explicitly). It is easy to show that this space is nonempty only 
if $\sum_i E_i=0$, and in that case $\dim Y(\bold E)$ equals the number of
loops of $\Gamma$ (show this!).

3. On each external edge, write $\frac{1}{E_i^2+m^2}$, and on
each internal edge, write $\frac{1}{Q_k^2+m^2}$. 
Let $\phi_\Gamma(\bold E,\bold Q)$ be the product of all these
functions. 

4. Define the {\it momentum space amplitude}\index{momentum space amplitude} of 
$\Gamma$ to be the distribution $\widehat F_\Gamma(\bold E):$ 
$$
\widehat F_\Gamma(E_1,...,E_n)=
\prod_jg_{v(j)}
\int_{Y(\bold E)}\phi_\Gamma(\bold E,\bold Q)d\bold Q\cdot \delta(E_1+...+E_n)d\bold E,
$$
supported on the
hyperplane $\sum_i E_i=0$. 
It is clear that this distribution is independent on the
orientation of $\Gamma$. 

\begin{remark} Here we must specify the 
normalization of the (translation-invariant) Lebesgue measure  $d\bold Q$ 
on the space $Y(\bold E)$. It is defined in such a way that the volume of $Y(\bold E)/Y_{\Bbb Z}(0)$ is $1$, where $Y_{\Bbb  Z}(0)$ is the set of integer elements in $Y(0)$. So if $T\subset \Gamma$ is a spanning tree then in the coordinates $\lbrace Q_e,e\notin T\rbrace$ on $Y(\bold E)$, 
we have $d\bold Q=\prod_{e\notin T}dQ_e$. 
 \end{remark} 

Now we have 

\begin{proposition}
The Fourier transform of the function 
$F_\Gamma(\delta_{t_1},...,\delta_{t_n})$ is 
$\widehat F_\Gamma(E_1,...,E_n)$. 
Hence, the Fourier transform of the connected Green's function
is 
\begin{equation}\label{defcor1}
\widehat\G_n^c(E_1,...,E_n)=
\sum_{\Gamma\in G_{\ge 3}^*(n)}\frac{\hbar^{b(\Gamma)}}{|{\rm Aut}(\Gamma)|}
\widehat F_\Gamma(E_1, \ldots, E_n).
\end{equation}
\end{proposition}

The proof of the proposition is straightforward. 

To illustrate the proposition, consider an example. 

\begin{example}
The connected 4-point function for the quartic potential modulo $g^2$ in
momentum space looks like: 
$$
\widehat \G_4^c(E_1,E_2,E_3,E_4)=
g\prod_{i=1}^4\frac{1}{E_i^2+m^2}\delta(\sum_i E_i)d\bold E+O(g^2).
$$
\end{example}

\begin{figure}[htbp]
  \begin{center}
    
    \setlength{\unitlength}{5ex}
    \begin{picture}(5.5,3)(-0.25,-0.15)

      \cbezier(0,0)(1.25,2.1)(3.75,2.1)(5,0)
      \cbezier(0,2.5)(1.25,0.4)(3.75,0.4)(5,2.5)
      \put(1.25,1.25){\circle*{.1}}      
      \put(3.75,1.25){\circle*{.1}}      
      %\put(0,0){\circle*{.1}}      
      %\put(0,2.5){\circle*{.1}}      

      \put(-0.25,2.75){\makebox(0,0)[c]{$1$}}
      \put(-0.25,-0.25){\makebox(0,0)[c]{$2$}}
      \put(5.25,2.75){\makebox(0,0)[c]{$3$}}
      \put(5.25,-0.25){\makebox(0,0)[c]{$4$}}
      \put(1.25,1.65){\makebox(0,0)[c]{$5$}}
      \put(3.75,1.65){\makebox(0,0)[c]{$6$}}

      \put(0.05,2.05){\makebox(0,0)[c]{$E_1$}}
      \put(0.5,-0.15){\makebox(0,0)[c]{$E_2$}}
      \put(5.1,2.05){\makebox(0,0)[c]{$E_3$}}
      \put(4.65,-0.15){\makebox(0,0)[c]{$E_4$}}
      \put(2.5,2.25){\makebox(0,0)[c]{$Q$}}
      \put(2.5,0.1){\makebox(0,0)[c]{$E_1 + E_2-Q$}}

      \put(0.25,2.25){\vector(1,-1){0.25}}            
      \put(0.25,0.25){\vector(1,1){0.25}}            
      \put(4.75,2.25){\vector(-1,-1){0.25}}            
      \put(4.75,0.25){\vector(-1,1){0.25}}            
      \put(2.25,1.65){\vector(1,0){0.25}}            
      \put(2.25,0.85){\vector(1,0){0.25}}            

      \put(6.5,1){\makebox(0,0)[c]{$\Gamma$.}}

    \end{picture}
    \caption{ }
    \label{fig:24}
  \end{center}
\end{figure}

\begin{example} Let us compute the 1PI 4-point function in the
same problem, but now modulo $g^3$. Thus, in addition to the above, we
need to compute the $g^2$ coefficient, which comes from 1-loop 
diagrams. There are three such diagrams, differing by permutation
of external edges. One of these diagrams is as follows: 
it has external vertices $1,2,3,4$ and internal ones $5,6$
such that $1,2$ are connected to $5$, $3,4$ to $6$, and $5$ and
$6$ are connected by two edges (Fig.24). This diagram has the symmetry
group $\Bbb Z/2$, so its contribution is 
$$
\frac{g^2}{2}\left(\int_{\Bbb R}\frac{dQ}{(Q^2+m^2)((E_1+E_2-Q)^2+m^2)}\right)
\prod_{i=1}^4\frac{1}{E_i^2+m^2}\delta(\sum_i E_i)d\bold E.
$$
The integral inside is easy to compute, 
for example, by residues. This yields
$$
\widehat \G_4^c(E_1,E_2,E_3,E_4)=
$$
$$
g\prod_{i=1}^4\frac{1}{E_i^2+m^2}
\left(1+\frac{\pi
  g}{m}\sum_{i=2}^4\frac{1}{(E_1+E_i)^2+4m^2}\right)\delta(\sum_i E_i)d\bold E+
O(g^3)
$$
(this is symmetric in the $E_1,E_2,E_3,E_4$ since when $\sum_i E_i=0$ then for distinct $i,j,k,\ell$ one has
$(E_i+E_j)^2=(E_k+E_\ell)^2$). 
\end{example}

\subsection{The Wick rotation in momentum space}

To obtain the correlation functions of quantum mechanics, we should, after 
computing them in the Euclidean setting, Wick rotate them back to the 
Minkowski setting. Let us do it at the level of Feynman 
integrals in momentum space. (We could do it in position space as well, but 
it is instructive for the future to do it in momentum space, since 
in higher dimensional field theory which we will discuss later, 
the momentum space representation is more convenient).  

Consider the Euclidean propagator 
$$
\frac{1}{E^2+m^2}=\int_{\Bbb R} G(t)e^{iEt}dt,
$$
where $G$ is the Green's function. 
 When we do analytic continuation back 
to the Minkowski setting, we must replace in the correlation functions 
the time variable $t$ with 
$e^{i\theta}t$, where $\theta$ varies from $0$ to $\frac{\pi}{2}$. 
In particular, the Green's function $G(t)$ must be replaced by 
$G(e^{i\theta}t)$. 
So we must consider 
$$
\int_{\Bbb R} G(e^{i\theta}t)e^{iEt}dt=
e^{-i\theta} \int_{\Bbb R} G(t)e^{ie^{-i\theta}Et}dt=
\frac{e^{-i\theta}}{e^{-2i\theta}E^2+m^2}.
$$
As $\theta\to \frac{\pi}{2}$, this function tends (as a 
distribution) to the function 
$\lim_{\varepsilon\to
0+}\frac{i}{E^2-m^2+i\varepsilon}$.
For brevity the limit sign is usually dropped and this 
distribution is written as $\frac{i}{E^2-m^2+i\varepsilon}$.

We see that in order to compute the correlation functions in
momentum space in the Minkowski setting, we should 
use the same Feynman rules as in the Euclidean setting 
except that the propagator put on the edges should be 
$$
\frac{i}{E^2-m^2+i\varepsilon}.
$$
For instance, the contribution of the diagram in Fig.24 is 
\scriptsize
$$
-\frac{g^2}{2}\left(\int_{\Bbb R}\frac{dQ}{(Q^2-m^2+i\varepsilon)
((E_1+E_2-Q)^2-m^2+i\varepsilon)}\right)
\prod_{j=1}^4\frac{1}{E_j^2-m^2+i\varepsilon}\delta(\sum_i E_j)d\bold E.
$$
\normalsize

\subsection{Quantum mechanics on the circle} 

It is reasonable (at least mathematically) 
to consider Euclidean quantum mechanical 
path integrals in the case when the time axis has been replaced
with a circle of length $L$, i.e. $t\in \Bbb R/L\Bbb Z$ (this corresponds to a Brownian particle in a potential field conditioned to return to the original position in a certain time $L$). 
In this case, the theory is the same, except the 
Green's function $G(t)$ is replaced by the 
periodic solution $G_L(t)$ of the equation $(-\frac{d^2}{dt^2}+m^2)f=\delta(t)$
on the circle. This solution has the form 
\begin{equation}\label{glformula} 
G_L(t)=\sum_{k\in \Bbb
  Z}G(t-kL)=\frac{e^{-m(t-\frac{L}{2})}+e^{-m(\frac{L}{2}-t)}}{2m(e^{\frac{mL}{2}}-e^{-\frac{mL}{2}})},\ 0\le t\le L.
\end{equation} 
We note that in the case of a circle, 
there is no problem with graphs without 
external edges (as integral over the circle of a 
constant function is convergent), and hence one may define not only 
correlation functions (i.e. $Z(J)/Z(0)$), but also $Z(0)$
itself. Namely, let 
$$
U(q)=\frac{m^2q^2}{2}+\sum_{n\ge 3}\frac{g_nq^n}{n!},
$$ 
and let $m^2=m_0^2+g_2$ (where $g_i$ are formal parameters). 
Then we can make sense of the ratio $Z_{m_0,\bold
  g,L}(0)/Z_{m_0,0,L}(0)$
(where $Z_{m,\bold g,L}(0)$ denotes the partition function for the
specified values of parameters; from now on the argument $0$ will
be dropped). 
Indeed, this ratio is defined by the formula 
$$
\frac{Z_{m_0,\bold g,L}}{Z_{m_0,0,L}}=
\sum_{\Gamma\in G_{\ge 2}(0)}\frac{\hbar^{b(\Gamma)}}{|{\rm Aut}(\Gamma)|}
F_\Gamma
$$
(where $G_{\ge 2}(0)$ is the set of Feynman graphs without external vertices 
and all vertices of valency $\ge 2$), which is a well-defined expression.

It is instructive to compute this expression in the case 
$$
g_2=a,\ g_3=g_4=...=0.
$$ 
In this case, we have only 2-valent vertices, so the only
connected Feynman
diagrams are $N$-gons, which are 1-loop. 
Hence, 
$$
\log \frac{Z_{m_0,\bold g,L}}{Z_{m_0,0,L}}=
W_1=-\frac{1}{2}\log \det M,
$$
where 
$$
M=1+a(-\tfrac{d^2}{dt^2}+m_0^2)^{-1}.
$$ 
This determinant may be computed by looking 
at the eigenvalues. Namely, the eigenfunctions
of $-\frac{d^2}{dt^2}+m_0^2$ in the space 
$C^\infty(\Bbb R/L\Bbb Z)$ are $e^{\frac{2\pi int}{L}}$, with eigenvalues
$\frac{4\pi^2n^2}{L^2}+m_0^2$. 
So, 
$$
\det M=\prod_{n\in \Bbb
  Z}\left(1+\frac{a}{\frac{4\pi^2n^2}{L^2}+m_0^2}\right).
$$
Hence, using the Euler product formula 
$$
\sinh(z)=z\prod_{n\ge 1} \left(1+\frac{z^2}{\pi^2n^2}\right),
$$
 we get
$$
\frac{Z_{m_0,\bold g,L}}{Z_{m_0,0,L}}=
\frac{\sinh(\frac{m_0L}{2})}{\sinh(\frac{mL}{2})}.
$$
(Double-check this using summation over Feynman diagrams!)

\begin{remark} More informally speaking, we see that 
the partition function $Z$ for the theory with $U=\frac{m^2q^2}{2}$ has the form 
$\frac{C}{\sinh(\frac{mL}{2})}$, where $C$ is a constant of our choice. 
Our choice from now on will be $C=\frac{1}{2}$; we will see 
later (in Example \ref{benchoi}) why such a choice is preferable. 
\end{remark}

\subsection{The massless case} 
Consider now the massless case, $m=0$. In this case the propagator
should be obtained by inverting the operator $-\frac{d^2}{dt^2}$, 
i.e. it should be the integral operator with kernel 
$G(t-s)$, where $G(t)$ is an even function 
satisfying the differential equation
$$
-G''(t)=\delta(t).
$$ 
There is a 1-parameter family of such
solutions, 
$$
G(t)=-\frac{1}{2}|t|+C.
$$ 

Using this function (for any choice of $C$), 
one may define
the correlation functions of the free theory by the Wick formula. 

Note that the function $G$ does not decay at
infinity. Therefore, this 
theory will not satisfy the clustering property
(i.e. is not ``physically meaningful''). 

We will also have difficulties in defining the corresponding 
interacting theory (i.e. one with a non-quadratic potential), as 
the integrals defining the amplitudes of Feynman diagrams
will diverge. Such divergences are called {\it infrared divergences},\index{infrared divergence}
since they are caused by the failure of the integrand to decay 
at large times (or, in momentum space, its failure to be regular 
at low frequencies).

\subsection{Circle-valued quantum mechanics}
Consider now 
the theory with the same Lagrangian 
in which $q(t)$ takes values in the circle of radius $r$, 
$\Bbb R/2\pi r\Bbb Z$ (the ``sigma-model''). We can do this
at least classically, since the 
Lagrangian $\frac{\dot{q}^2}{2}$ makes sense
in this case. 

Let us define the corresponding quantum theory.
The main difference 
from the line-valued case is that since $q(t)$ is circle-valued, we should  
consider not the usual correlators 
$\langle q(t_1)...q(t_n)\rangle$, but rather correlation functions of 
exponentials 
$\langle e^{\frac{ip_1q(t_1)}{r}}...e^{\frac{ip_nq(t_n)}{r}}\rangle$, where $p_j$ are
integers. They should be defined by 
the path integral
\begin{equation}\label{expon}
\int e^{\frac{ip_1q(t_1)}{r}}...e^{\frac{ip_nq(t_n)}{r}}e^{-\frac{S(q)}{\hbar}}Dq,
\end{equation}
where $S(q):=\frac{1}{2}\int \dot{q}^2dt$
and $\int e^{-\frac{S(q)}{\hbar}}Dq$ is agreed to be $1$. 
Note that it suffices to consider only
the case $\sum_j p_j=0$, otherwise 
the group of translations along the circle 
acts nontrivially on the integrand, 
hence under any reasonable definition 
the integral should be zero. 

Now let us define the integral (\ref{expon}). 
Since the integral is invariant under shifts along the target
circle, we may as well imagine that we are integrating 
over $q:\Bbb R\to \Bbb R$ with $q(0)=0$. 
Now let us use the finite-dimensional analogy.
Following this analogy, by completing the square we would get 
$$
\int e^{\frac{ip_1q(t_1)}{r}}...e^{\frac{ip_nq(t_n)}{r}}e^{-\frac{S(q)}{\hbar}}Dq=
e^{-\frac{\hbar}{2r^2} B^{-1}(\sum_j p_jq(t_j),\sum_j p_jq(t_j))}=
$$
$$
e^{-\frac{\hbar}{2r^2}\sum_{j,\ell} p_\ell p_j G(t_\ell-t_j)}=
e^{\frac{\hbar}{2r^2}\sum_{\ell<j} p_\ell p_j|t_\ell-t_j|},
$$
where $B(q,q):=\int \dot{q}^2dt$. 
Thus, it is natural to define the correlators by the formula
$$
\langle e^{\frac{ip_1q(t_1)}{r}}...e^{\frac{ip_nq(t_k)}{r}}\rangle
=e^{\frac{\hbar}{2r^2}\sum_{\ell<j} p_\ell p_j|t_l-t_j|}.
$$
We note that this theory, unlike the line-valued one, {\it does} 
satisfy the clustering property. Indeed, 
if $\sum p_j=0$ (as we assumed), then 
(assuming $t_1\ge t_2\ge...\ge t_n$), we have
$$
\sum_{\ell<j}p_\ell p_j(t_\ell-t_j)=
\sum_{j=1}^{n-1}(t_j-t_{j+1})(p_{j+1}+...+p_n)(p_1+...+p_j)=
$$
$$
-\sum_j (t_j-t_{j+1})(p_1+...+p_j)^2,
$$ 
so the clustering property follows from the fact that
$(p_1+...+p_j)^2\ge 0$. 

\subsection{Massless quantum mechanics on the circle} 

Consider now the theory with Lagrangian $\frac{\dot{q}^2}{2}$, 
where $q$ is a function on the circle of length $L$. In this case,
according to the Feynman yoga,  
we must invert the operator $-\frac{d^2}{dt^2}$ on the circle $\Bbb
R/L\Bbb Z$, or equivalently solve the differential equation 
$-G''(t)=\delta(t)$. Here we run into trouble: 
the operator $-\frac{d^2}{dt^2}$ is not invertible, since it has an
eigenfunction $1$ with eigenvalue $0$; correspondingly, the
differential equation in question has no solutions, as $\int G''dt$
must be zero, so $-G''(t)$ cannot equal $\delta(t)$
(one may say that the 
quadratic form in the exponential is degenerate,
and therefore the Gaussian integral 
turns out to be meaningless).
This problem can be resolved
by the following technique of ``killing the zero mode''.
Namely, let us invert the operator $-\frac{d^2}{dt^2}$ on the space  
$\lbrace{q\in C^\infty(\Bbb R/L\Bbb Z): 
\int qdt=0\rbrace}$ (this may be interpreted 
as integration over this codimension one subspace, 
on which the quadratic form is non-degenerate). This means that we must
find the solution of the differential equation 
$-G''(t)=\delta(t)-\frac{1}{L}$, such that $\int Gdt=0$. 
Such solution is indeed unique, and it equals 
\begin{equation}\label{greefu} 
G(t)=\frac{(t-\frac{L}{2})^2}{2L}-\frac{L}{24},
\end{equation}
$t\in [0,L]$. Thus, for example 
$\langle q(0)^2\rangle=\frac{L}{12}$. 

Higher correlation functions are defined in the usual way. 
Moreover, one can define the theory with an arbitrary potential 
using the standard procedure with Feynman diagrams.

\subsection{Circle-valued quantum mechanics on the circle} 
Finally, let us consider the circle-valued version of the same
theory. Thus, our integration variable is a map $q: \Bbb R/L\Bbb
Z\to \Bbb R/2\pi r\Bbb Z$. So we have a new feature - 
there are different homotopy classes of maps 
labeled by degree. Let us first 
consider integration over degree zero maps. 
Then we should argue in the same way as in the case $t\in \Bbb
R$, and make the definition
$$
\langle e^{\frac{ip_1q(t_1)}{r}}...e^{\frac{ip_nq(t_n)}{r}}\rangle_0:=
e^{-\frac{\hbar}{2r^2}\sum_{\ell,j}p_\ell p_j G(t_\ell-t_j)},
$$
where $\sum_j p_j=0$.
(Here subscript $0$ stands for degree zero maps).
Assuming that $0\le t_1,...,t_n\le L$, we find  
after a short calculation using \eqref{greefu}:
$$
\langle e^{\frac{ip_1q(t_1)}{r}}...e^{\frac{ip_nq(t_n)}{r}}\rangle_0
=e^{\frac{\hbar}{2r^2}(\sum_{\ell<j} p_\ell p_j |t_\ell-t_j|+\frac{(\sum_j p_jt_j)^2}{L})}
$$
(the second summand disappears as $L\to \infty$, and we recover
the answer on the line). 

It is, however, more natural (as we will see later)
to integrate over all maps $q$, not only degree zero. Namely, let $N$ be an integer. 
Then all maps of degree $N$ have the form 
$q(t)+\frac{2\pi rNt}{L}$, where $q$ is a map of
degree zero. Thus, if we want to integrate over maps of degree
$N$, we should compute the same integral as in degree zero, 
but with shift $q\mapsto q+\frac{2\pi r Nt}{L}$. 
But it is easy to see that this shift results simply in rescaling
of the integrand by the factor $e^{\frac{2\pi iN}{L}\sum_j
p_jt_j-\frac{2\pi^2r^2N^2}{\hbar L}}$. Thus, the integral over all maps
should be defined by the formula 
$$
\langle e^{\frac{ip_1q(t_1)}{r}}...e^{\frac{ip_nq(t_n)}{r}}\rangle=
$$
\begin{equation}\label{thetform}
e^{\frac{\hbar}{2r^2}(\sum_{l<j} p_\ell p_j |t_\ell-t_j|+\frac{(\sum p_jt_j)^2}{L})}
\frac{\sum_{N\in \Bbb Z}e^{\frac{2\pi iN}{L}\sum_j
p_jt_j-\frac{2\pi^2r^2N^2}{\hbar L}}}
{\sum_{N\in \Bbb Z}e^{-\frac{2\pi^2r^2N^2}{\hbar L}}}.
\end{equation}
Introduce the elliptic theta-function 
$$
\theta(u,T):=\sum_{N\in \Bbb Z}e^{2\pi iuN-\pi TN^2}.
$$
Then for $L\ge  t_1\ge...\ge t_n\ge 0$ formula  \eqref{thetform} 
can be rewritten in the form 
\begin{equation}\label{thetaf}
\langle e^{\frac{ip_1q(t_1)}{r}}...e^{\frac{ip_nq(t_n)}{r}}\rangle
=
e^{\frac{\hbar}{2r^2} (\sum_j (t_{j}-t_{j+1})(p_1+...+p_j)^2+\frac{(\sum_j
p_jt_j)^2}{L})}
\frac{\theta(\frac{\sum_j p_jt_j}{L},\frac{2\pi r^2}{\hbar L})}
{\theta(0,\frac{2\pi r^2}{\hbar L})}.
\end{equation}

\begin{exercise} Calculate the 1-particle irreducible 2-point
function for a quantum particle 
with potential $U(q):=\frac{m^2q^2}{2}-\frac{gq^4}{4!}$ modulo $g^3$
in momentum space, for $\hbar=1$. 
(We have done this modulo $g^2$ in position space). 
\end{exercise}

\begin{exercise} Let $U(q):=\frac{m^2q^2}{2}-\frac{gq^3}{3}$. 

(i) Calculate 
the leading term of the 1-point function
${\mathcal G}_1(t)$ (with respect to $g$).

(ii) Calculate the connected 2-point function modulo $g^3$.  
\end{exercise} 

\begin{exercise} 
 Consider the potential $U(x):=\frac{m^2\sinh^2(gx)}{2g^2}$. 
Find a formula for $W_0(J)$ 
(the tree part of $\log(Z(J)/Z(0))$) as explicitly as you can, when $J(t)=a\delta(t)$.
\end{exercise} 

\section{Operator approach to quantum mechanics}

In mechanics and field theory (both classical and quantum), there are two main
languages -- Lagrangian and Hamiltonian. In the classical setting, 
the Lagrangian language is the language of variational calculus
(i.e. one studies extremals of the action functional), while the
Hamiltonian language is that of symplectic geometry and Hamilton
equations. Correspondingly, in the quantum setting, the
Lagrangian language is the language of path integrals, while 
the Hamiltonian language is the language of operators and
Schr\"odinger equation. We have now studied the first one
(at least in perturbation expansion) and are passing to the
second one. 

\subsection{Hamilton's equations in classical mechanics}

We start with recalling the Hamiltonian formalism of classical
mechanics. For more details, we refer the reader to the excellent
book \cite{A}. 

Recall first the Lagrangian description of the motion of a classical particle or system of particles. 
The position of a particle is described by a point $q$ of the 
configuration space $X$, which we will assume to be a manifold. 
The Lagrangian of the system is a (smooth) function ${\mathcal L}: TX\to
\Bbb R$ on the total space of the tangent bundle of $X$. 
Then the action functional is $S(q)=\int \mathcal
L(q,\dot{q})dt$. The trajectories of the particle are the
extremals of $S$. The condition for $q(t)$ to be an extremal of
$S$ is equivalent to the Euler-Lagrange equation
(=the equation of motion), which 
in local coordinates has the form 
$$
\frac{d}{dt}\frac{\partial \mathcal{L}}{\partial \dot{q_i}}=
\frac{\partial \mathcal{L}}{\partial q_i}.
$$
For example, if $X$ is a Riemannian manifold and ${\mathcal L}(q,v)=\frac{v^2}{2}-U(q)$
where $U:X\to \Bbb R$ is a potential function, 
then the Euler-Lagrange equation is the Newton equation
$$
\ddot{q}=-U'(q),
$$
where $\ddot{q}=\nabla_{\dot{q}}\dot q$ 
is the covariant derivative with respect to the Levi-Civita
connection. 

Consider now a system with Lagrangian ${\mathcal L}(q,v)$, whose 
differential with respect to $v$ (for fixed $q$) is a 
diffeomorphism $T_qX\to T^*_qX$. This is definitely true in the
above special case of Riemannian $X$. 

\begin{definition}
The {\it Hamiltonian (or energy function)}\index{Hamiltonian}
of the system with Lagrangian ${\mathcal L}$ 
is the function $H: T^*X\to \Bbb R$, which is the 
Legendre transform of ${\mathcal L}$ along fibers; 
that is, $H(q,p)=pv_0-{\mathcal L(q,v_0)}$, where 
$v_0$ is the (unique) critical point of $pv-{\mathcal L(q,v)}$.
The manifold $T^*X$ is called the {\it phase space
(or space of states)}.\index{phase space} The variable $p$ is called the {\it momentum
variable}.\index{momentum
variables}
\end{definition}

For example, if ${\mathcal L}=\frac{v^2}{2}-U(q)$, 
then $H(q,p)=\frac{p^2}{2}+U(q)$.

\begin{remark} Since Legendre transform is involutive,
we also have that the Lagrangian is the fiberwise 
Legendre transform of the Hamiltonian. 
\end{remark} 

Let $q_i$ be local coordinates on $X$.
This coordinate system defines a coordinate system 
$(q_i,p_i)$ on $T^*X$. We obtain

\begin{proposition}
The equations of motion are equivalent to the 
{\it Hamilton equations}\index{Hamilton equations}
$$
\dot{q_i}=\frac{\partial H}{\partial p_i},\
\dot{p_i}=-\frac{\partial H}{\partial q_i},
$$ 
in the sense that they are obtained from Hamilton's equations by
elimination of $p_i$. 
\end{proposition}

It is useful to write Hamilton's equations in terms of Poisson brackets. 
Recall that the manifold 
$T^*X$ has a canonical symplectic structure 
$\omega=d\alpha$, where $\alpha$ 
is the canonical 1-form on $T^*M$ (called the {\it Liouville form}\index{Liouville form}) constructed as follows:
for any $z\in T_{(q,p)}(T^*X)$, 
$$
\alpha(z)=(p,d\pi(q,p)z),
$$ 
where $\pi: T^*X\to X$ is the
projection. In local coordinates, we have 
$$
\alpha=\sum_i p_idq_i,\ 
\omega=\sum_i dp_i\wedge dq_i.
$$

Now let $(M,\omega)$ be a symplectic manifold (in our case
$M=T^*X$).  
Since $\omega$ is non-degenerate, 
one can define the Poisson bivector 
$\omega^{-1}$, 
which is a section of the bundle $\wedge^2TM$. 
Now, given any two smooth functions 
$f,g$ on $M$, one can define a third function 
-- their {\it Poisson bracket}\index{Poisson bracket}
$$
\lbrace{ f,g\rbrace}=(df\otimes dg,\omega^{-1}).
$$
This operation is skew-symmetric and satisfies the Jacobi identity,
i.e. it is a Lie bracket on $C^\infty(M)$. 
For $M=T^*X$, in local coordinates we have
$$ 
\lbrace{ f,g\rbrace}=\sum_i\biggl(\frac{\partial f}{\partial
  q_i}\frac{\partial g}{\partial p_i}-
\frac{\partial f}{\partial
  p_i}\frac{\partial g}{\partial q_i}\biggr).
$$

This shows that Hamilton's equations can be written in the
following manner in terms of Poisson brackets: 
\begin{equation}\label{heobs}
\frac{d}{dt}f(q(t),p(t))=\lbrace{f,H\rbrace}(q(t),p(t)).
\end{equation}
for any smooth function (``classical observable'') $f\in C^\infty(T^*X)$, or, for shorthand
$$
\frac{df}{dt}=\lbrace f,H\rbrace.
$$ 
In other words, Hamilton's equations say that the 
rate of change of the observed value of 
$f$ equals the observed value of
$\lbrace{f,H\rbrace}$.

%Note that for a given Lagrangian, 
%the unique function $H$ (up to adding a constant) for which 
%equations (\ref{heobs}) are equivalent to the equations of motion is the 
%Hamiltonian. This provides another definition of the Hamiltonian, which 
%does not use the notion of the Legendre transform. 

\subsection{Unbounded self-adjoint operators} 

The rigorous mathematical treatment of quantum mechanics in the Hamiltonian setting
is based on von Neumann's theory of unbounded self-adjoint operators in a Hilbert space.
Let us recall the basics of this theory. 

\subsubsection{Spectral theorem for bounded self-adjoint operators} 
Let $\mathcal H$ be a separable complex Hilbert space with 
inner product $\la,\ra$ (antilinear in the first argument, as is
traditional in quantum physics). We first recall the {\it spectral theorem}\index{spectral theorem} for bounded self-adjoint operators $A: \mathcal H\to \mathcal H$, which generalizes the diagonalization theorem for a Hermitian matrix. 

\begin{theorem}\label{specthe} (von Neumann) Let $A$ be a bounded self-adjoint operator. 
There exists a measure space $(X,\mu)$, 
an essentially bounded measurable function $h: X\to \Bbb R$, and an isometry $\mathcal H\to L^2(X,\mu)$ under which $A$ maps to the operator of multiplication by $h$.  
Moreover, the spectrum $\sigma(A)$ is the set of $\lambda\in \Bbb R$ 
for which $h^{-1}(\lambda-\varepsilon,\lambda+\varepsilon)$ 
is positive for each $\varepsilon>0$, 
and the eigenvalues of $A$ (if they exist) are $\lambda\in \Bbb R$ 
such that $\mu(h^{-1}(\lambda))>0$, with eigenfunctions being indicator functions of 
subsets of $h^{-1}(\lambda)$ of positive measure. 
\end{theorem} 

\subsubsection{Closable and closed operators} 
Now we pass to not necessarily bounded operators. 
Let $\mathcal H'$ be another separable Hilbert space. 
A {\it densely defined linear operator}\index{densely defined linear operator} on $\mathcal H$ is a pair $(A,V)$ where
$V\subset \mathcal H$ is a dense subspace and $A$ 
is a (possibly unbounded) linear operator $V\to \mathcal H'$. The space $V$ is called the {\it domain}\index{domain of a densely defined operator} of $A$; in the notation, we will often suppress it and denote the operator just by $A$. Such an operator $A$ has a graph $\Gamma_A\subset V\times \mathcal H'\subset \mathcal H\times \mathcal H'$. Let $\overline \Gamma_A$ be the closure of $\Gamma_A$ in 
$\mathcal H\times \mathcal H'$. The operator $A$ is said to be {\it closable}\index{closable operator} if $(0,u)\in \overline \Gamma_A$ for $u\in \mathcal H'$ implies $u=0$, i.e., if the first projection 
$p_1:\overline\Gamma_A\to \mathcal H$ is injective. In this case, setting $\overline V:=p_1(\overline \Gamma_A)\subset \mathcal H$, we have $V\subset \overline V$ and obtain an extension of the operator $A: V\to \mathcal H'$ to a densely defined operator $\overline A: \overline V\to \mathcal H'$ which is called the {\it closure}\index{closure of an operator} of $A$. If $A$ is closable and $\overline A=A$, we will say that $A$ is {\it closed}\index{closed operator}; in other words, $A$ is closed iff it has closed graph in $\mathcal H\times \mathcal H'$. Obviously, the closure $\overline A$ 
is closed for any closable $A$. Also, if $A$ is bounded then it is closable, $\overline V=\mathcal H$, and $\overline A: \mathcal H\to \mathcal H'$ is just the continuous (=bounded) extension of $A$.  

In general, however, a densely defined operator need not be closable. For example, 
if $\mathcal H'$ is finite dimensional and $A: V\to \mathcal H'$ is unbounded then there exists a sequence $v_n\in V$ such that $v_n\to 0$ but $||Av_n||\ge 1$. Then the sequence $w_n:=\frac{v_n}{||Av_n||}$ goes to $0$, while $||Aw_n||=1$, so, as the unit sphere in $\mathcal H'$ is compact, passing to a subsequence if needed, we may assume that $Aw_n\to u$ for some $u\in \mathcal H'$ with $||u||=1$. Then 
$(0,u)\in \overline{\Gamma}_A$ and $A$ is not closable. So we see that $A$ is closable iff it is bounded.  

On the other hand, if $\mathcal H'$ is infinite dimensional, then there are important classes of
unbounded closable operators. For example, consider the case $\mathcal H=\mathcal H'$. Let us say that an operator $A: V\to \mathcal H$ is {\it symmetric}\index{symmetric operator} if $\la v,Aw\ra=\la Av,w\ra$ for all $v,w\in V$. We claim that every symmetric operator is closable and its closure is symmetric. Indeed, suppose $(v_n,Av_n)\to (0,u)$ for $u\in\mathcal H$. Fix a sequence $u_k\in V$ such that $u_k\to u$. Then  
$$
\la Au_k,v_n\ra=\la u_k,Av_n\ra\to \la u_k,u\ra,\ n\to \infty.
$$
But the leftmost expression goes to zero, so 
$\la u_k,u\ra=0$ for all $k$, hence $||u||^2=0$ which gives $u=0$, i.e., $A$ is closable. 
Moreover, given $v,w\in \overline V$, there exist 
sequences $v_n\to v,w_n\to w$ in $V$ such that $Av_n\to \overline Av$, $Aw_n\to \overline Aw$, thus 
$$
\la v,\overline Aw\ra=\lim_{n\to \infty}\la v_n,Aw_n\ra=\lim_{n\to \infty}\la Av_n,w_n\ra=\la \overline Av,w\ra,
$$
so $\overline A$ is symmetric. 

\subsubsection{Adjoint operator} 
Closed symmetric operators by themselves are not sufficient for quantum mechanics, however, since such operators cannot, in general, be diagonalized. Instead we need {\it self-adjoint}\index{self-adjoint operator} operators, which are closed symmetric operators satisfying an important additional property. To formulate this property, we first need to define the notion of an adjoint operator.

Let $(A,V)$ be a closed symmetric operator. Denote by $V^\vee$ the space of $u\in \mathcal H$ such that the linear functional $v\mapsto \la u,Av\ra$ 
is bounded on $V$. In this case by the Riesz representation theorem there exists a unique 
vector $w\in \mathcal H$ such that $\la u,Av\ra=\la w,v\ra$, which depends linearly on $u$.
Thus we obtain an operator $A^\dagger: V^\vee\to \mathcal H$. Note that 
$V^\vee\supset V$ and $A^\dagger$ is an extension of $A$ to $V^\vee$, so 
$(A^\dagger,V^\vee)$ is a densely defined operator  called the {\it adjoint operator}\index{adjoint operator} of $(A,V)$. Furthermore, this operator is closed: if $(u_n,A^\dagger u_n)\to (u,w)$ 
then for $v\in V$, 
$$
\la A^\dagger u_n,v\ra=\la u_n,Av\ra\to \la u,Av\ra,\ n\to \infty,
$$
and at the same time $\la A^\dagger u_n,v\ra\to \la w,v\ra$, so 
$\la u,Av\ra=\la w,v\ra$, hence $u\in V^\vee$ and $w=A^\dagger u$. 

However, we will see that the operator $A^\dagger$ fails to be symmetric, in general. 
So we may consider the skew-Hermitian form 
$$
B(v,w):=(A^\dagger v,w)-(v,A^\dagger w)
$$
on $V^\vee$ that measures its failure to be symmetric, called the {\it boundary form}\index{boundary form}
(it is called this way because in examples it corresponds to boundary terms arising from integration by parts). 
By definition, $V\subset {\rm Ker}B$ (in fact, one can show that $V={\rm Ker}B$, but we don't need this fact). It is easy to see that closed symmetric extensions of $A$ correspond 
to isotropic closed subspaces $V\subset L\subset V^\vee$ 
with respect to the form $B$; namely, the extension of $A$ to $L$ is defined to be  
the restriction of $A^\dagger|_L$. Moreover, the adjoint operator to such an extension 
$(A^\dagger,L)$ is $(A^\dagger, L^\perp)$, where $L^\perp$ is the 
orthogonal complement of $L$ in $V^\vee$ with respect to $B$. 

\subsubsection{Self-adjoint operators} 
Let us say that a closed symmetric operator $(A,V)$ is {\it self-adjoint}\index{self-adjoint operator} if 
$V^\vee=V$, i.e., $A^\dagger=A$. We see that {\it self-adjoint extensions}\index{self-adjoint extension}
of $A$ correspond to {\it Lagrangian} subspaces $L$, i.e., those for which 
$L=L^\perp$. Note that such extensions/subspaces may or may not exist: 
the necessary and sufficient condition for existence of self-adjoint extensions (or Lagrangian subspaces) is that the signature $(n_+,n_-)$ of the Hermitian form $iB$ satisfies the equation $n_+=n_-$ (i.e., the so-called {\it deficiency indices}\index{deficiency index} $n_\pm \in \Bbb Z_{\ge 0}\cup\infty$ of $A$ are equal). However, in quantum mechanical models they usually exist and correspond to various spatial boundary conditions. 

We say that a symmetric operator $(A,V)$ is {\it essentially self-adjoint}\index{essentially self-adjoint operator}
if the closure $(\overline A,\overline V)$ is self-adjoint. Thus an essentially self-adjoint operator 
has a unique self-adjoint extension, so having such an operator is basically as good as having a self-adjoint one. This notion is convenient, for instance, when we do not want to describe explicitly the space $\overline V$.  

The importance of unbounded self-adjoint operators consists in the fact that von Neumann's spectral theorem extends naturally to them. Namely, define the {\it spectrum}\index{spectrum of a self-adjoint operator} $\sigma(A,V)$ of a self-adjoint operator $(A,V)$ to be the subset of $\lambda\in \Bbb C$ for which the operator 
$A-\lambda: V\to \mathcal H$ fails to be surjective. Then we have 

\begin{theorem} Theorem \ref{specthe} except for the statement that $h$ is essentially bounded holds for not necessarily bounded self-adjoint operators. Moreover, the domain $V$ of $A$ in its spectral theorem realization is the space of $g\in L^2(X,\mu)$ such that $hg\in L^2(X,\mu)$. 
\end{theorem} 

If the measure $\mu$ is concentrated on a countable set 
(i.e., we may take $X=\Bbb N$ with $\mu(j)=1$ for $j\in \Bbb N$) 
then $\mathcal H$ has a basis consisting of eigenfunctions, and vice versa; 
in this case one says that the spectrum of $A$ is {\it purely point spectrum}.\index{purely point spectrum} This happens, for example, when $A$ is compact (the Hilbert-Schmidt theorem). 
The other extreme is {\it purely continuous spectrum},\index{purely continuous spectrum} when there are no eigenvalues
(i.e., in the spectral theorem realization, all points of $X$ have zero measure). 
The spectral theorem implies that any self-adjoint operator can be uniquely written as 
an orthogonal direct sum of two self-adjoint operators with purely point and purely continuous spectrum, respectively. 

The spectral theorem also implies the following corollary. 

\begin{corollary} Let $(A,V)$ be a self-adjoint operator. Then there exists a unique 1-parameter group of unitary operators $U(t)=e^{iAt}: \mathcal H\to \mathcal H$ strongly continuous in $t$ which preserve $V$ and commute with $A$, such that for all $v\in V$ the function $t\mapsto U(t)v$ is differentiable and 
$$
\frac{d}{dt}(U(t)v)=iAU(t)v.
$$ 
\end{corollary} 

\begin{proof} Using the spectral theorem realization where $A$ is the operator of multiplication by $h: X\to \Bbb R$, we may define $U(t)$ 
as the operator of multiplication by $e^{ith}$. 
\end{proof} 

In fact, the converse also holds: every strongly continuous 1-parameter group 
$U(t)$ (i.e., a unitary representation of the Lie group $\Bbb R$ on $\mathcal H$) arises uniquely (up to isometry) from a self-adjoint operator. 

\begin{remark} The spectral theorem implies that if $(A,V)$ is a self-adjoint operator and 
$Av=\lambda v$ for some nonzero $v\in V$ then $\lambda\in \Bbb R$. On the contrary,
if $(A,V)$ is only symmetric and not self-adjoint, von Neumann showed that the set of eigenvalues 
of $A$ on $V$ is either the (open) upper-half plane $\Bbb C_+$ (if $n_+>0,n_-=0$), or 
the lower half-plane $\Bbb C_-$, (if $n_->0,n_+=0$), or contains both (if $n_+,n_->0$). 
\end{remark} 

\subsubsection{Examples} 

\begin{example} Consider the symmetric operator $P:=-i\frac{d}{dx}$ on 
$\mathcal H=L^2(S)$, where $S:=\Bbb R/2\pi \Bbb Z$ (the momentum operator on the circle). This operator is symmetric on the space $V:=C^\infty(S)$, and one can show that it is moreover essentially self-adjoint on this space (check it!). The corresponding space 
$\overline V$ is the {\it Sobolev space}\index{Sobolev space} $H^1(S)$ of functions 
$f\in L^2(S)$ with $f'\in L^2(S)$ in the sense of distributions (note that such functions 
are continuous). The spectrum of the corresponding self-adjoint operator 
is purely point and equals $\Bbb Z$, with 
eigenfunctions $e^{inx}$, i.e., $Pe^{inx}=ne^{inx}$. 
Thus the spectral realization of $A$ is 
on $\ell_2(\Bbb Z)$ with counting measure 
on which $P$ acts by multiplication by the function $n$ 
(i.e., this realization reduces simply to the Fourier 
expansion of functions on $S$). Similarly, the energy operator 
$P^2=-\frac{d^2}{dx^2}$ is essentially self-adjoint 
on the same space but with smaller domain of the closure - the Sobolev space 
$H^2(S)$ of functions $f\in L^2(S)$ such that $f''\in L^2(S)$. 
Its spectrum in $\Bbb Z_{\ge 0}$ with the same eigenfunctions: 
$P^2e^{inx}=n^2e^{inx}$. 
\end{example} 

\begin{example} The next example is more interesting, and prototypical for the theory of self-adjoint extensions. Namely consider the same momentum operator $P:=-i\frac{d}{dx}$, 
but now acting on the dense subspace $V\subset L^2[0,2\pi]$ of smooth functions 
with vanishing derivatives of all orders on both ends of the interval. In this case, 
$P$ is {\bf not} essentially self-adjoint: the space $\overline V$ is the space of functions $f\in H^1[0,2\pi]$ 
with $f(0)=f(2\pi)=0$, while $V^\vee=H^1[0,2\pi]$ with 
$$
B(f,g)=i(\overline{f(2\pi)}g(2\pi)-\overline{f(0)}g(0)). 
$$
So on the quotient $V^\vee/\overline V=\Bbb C^2$ we have 
$$
B((a,b),(a,b))=i(|b|^2-|a|^2), 
$$
where $a=f(0)$, $b=f(2\pi)$. 
Thus a Lagrangian subspace of $V^\vee$ is given by points $b\in \Bbb C$ with $|b|=1$; namely, it is the space $L_b$ of functions $f\in H^1[0,2\pi]$ with 
$f(2\pi)=bf(0)$. The spectrum of the corresponding self-adjoint operator is again purely point, 
so we should look for eigenfunctions in the space $L_b$. Thus we get eigenfunctions 
$e^{i(n+s)x}$ where $b=e^{2\pi is}$. So the set of eigenvalues is $\Bbb Z+s$, and we see that the spectrum depends on the choice of the self-adjoint extension. 

Observe also that any complex number $\lambda$ is the eigenvalue of the symmetric (non-self-adjoint!) operator $P^\dagger$ on $V^\vee$, with eigenvector $e^{i\lambda x}$. 
\end{example} 

\begin{example} Now consider the same momentum operator $P:=-i\frac{d}{dx}$ but acting on the space $V=C_0^\infty(\Bbb R)$ of compactly supported smooth functions, a subspace of $\mathcal H=L^2(\Bbb R)$. In this case $P$ is essentially self-adjoint, with $\overline V=V^\vee$ 
being the subspace of $H^1(\Bbb R)$ of $f\in L^2(\Bbb R)$ such that $f'\in L^2(\Bbb R)$. 
The spectral theorem realization of $P$ is on $L^2(\Bbb R)$ as the operator of multiplication by $x$, which is given by the Fourier transform. Thus the spectrum of this operator is purely continuous and constitutes the whole real line $\Bbb R$. Similarly, the operator $P^2=-\frac{d^2}{dx^2}$ is also essentially self-adjoint on $V$, and its self-adjoint extension 
has purely continuous spectrum $\Bbb R_{\ge 0}$. 
\end{example} 

\begin{example} And yet again, take $P:=-i\frac{d}{dx}$, but now on the subspace $V$ of $\mathcal H=L^2(\Bbb R_{\ge 0})$ of compactly supported smooth functions with vanishing derivatives at $0$. This operator is not essentially self-adjoint: the space $\overline V$ 
is the space of $f\in H^1(\Bbb R_{\ge 0})$ with $f(0)=0$, while $V^\vee$ is the whole 
$H^1(\Bbb R_{\ge 0})$. Thus the space $V^\vee/\overline V$ is 1-dimensional with form $B$ 
given by $B(f,g)=-i\overline{f(0)}g(0)$, and so there are no self-adjoint extensions 
(the deficiency indices are not equal: $n_+=1,n_-=0$). 

Let us find the eigenvalues of $P$ on $V$. The eigenvector 
with eigenvalue $\lambda$ is $e^{i\lambda x}$, and it belongs 
to $V$ iff $\lambda\in \Bbb C_+$. Thus the set of eigenvalues of 
$P$ is $\Bbb C_+$. 
\end{example} 

\begin{example}\label{harosc1} Let $A=-\frac{1}{2}\frac{d^2}{dx^2}+\frac{1}{2}x^2$ with $V=C^\infty_0(\Bbb R)$ (quantum harmonic oscillator). Then $A$ is essentially self-adjoint and $\overline A$ has pure point spectrum $n+\frac{1}{2}$, $n\in \Bbb N$, with eigenvectors $H_n(x)e^{\frac{x^2}{2}}$, where $H_n$ are Hermite polynomials (Theorem \ref{Hermite}). 
\end{example} 

\begin{remark} More generally, it is known that if $U(x)$ is a piecewise continuous potential on $\Bbb R$ 
which tends to $+\infty$ at $\pm \infty$ then the operator $A:=-\frac{1}{2}\frac{d^2}{dx^2}+U(x)$
is essentially self-adjoint on $V=C^\infty_0(\Bbb R)$ and has pure point spectrum, with eigenvalues $E_0< E_1\le E_2\le ...$ (it is shown in Lemma \ref{uniqueeig} below that $E_0$ is always a simple eigenvalue 
and the corresponding eigenvector is a positive function). 
\end{remark} 

\begin{example} Let $M$ be a compact Riemannian manifold with boundary $\partial M$, $\mathcal H=L^2(M)$, and $A=\Delta$ be the Laplace operator on $M$ acting on the space $V$ of smooth functions on $M$ vanishing with all derivatives on the boundary. In this case
$\overline V$ is the space of functions in the Sobolev space $H^2(M)$ (functions $f\in L^2(M)$ such that $\Delta f\in L^2(M)$) which vanish with first normal derivative on $\partial M$, and $V^\vee=H^2(M)$. By Stokes' formula 
$$
\int_M (u\Delta v-v\Delta u)dx=\int_{\partial M}(u\partial_{\bold n} v-v\partial_{\bold n}u)d\sigma,
$$
where $\bold n$ denotes the normal derivative to $\partial M$, so we have 
$$
B(f,g)=i\int_{\partial M}(\overline f\partial_{\bold n} g-g\partial_{\bold n} \overline f)d\sigma.
$$
So if $\partial M=0$, the operator $A$ is essentially self-adjoint and has a unique self-adjoint extension, while if $\partial M\ne 0$, it is not and 
there are many self-adjoint extensions corresponding to various boundary conditions on $\partial M$. The most common ones are the Dirichlet boundary condition $f=0$ 
and Neumann boundary condition $\partial_{\bold n} f=0$. Of course, the spectra associated to these conditions (which are always purely point) are completely different.  

The simplest example with non-trivial boundary is $M=[0,\pi]$, in which case we have 
$\dim V^\vee/\overline V=4$ and 
$$
B(f,g)=i(\overline {f}g'-\overline{f'}g)|_0^{\pi}.
$$
For the Dirichlet boundary conditions $f(0)=f(\pi)=0$ we get eigenbasis 
$\sin nx$ with eigenvalues $-n^2$, $n\in \Bbb Z_{\ge 1}$, while for the Neumann boundary conditions $f'(0)=f'(\pi)=0$ we get eigenbasis $\cos nx$ also with eigenvalues $-n^2$ but now for $n\in \Bbb Z_{\ge 0}$. 

Let us consider the mixed boundary condition: 
$$
f(0)=0,\ f'(\pi)-af(\pi)=0
$$ 
for some real number $a$. Then the eigenfunctions are 
$\sin \lambda x$ where 
$$
\lambda \cos \pi \lambda = a \sin \pi\lambda.
$$
Thus the eigenvalues are $-\lambda^2$ where $\lambda$ runs over solutions 
of the equation 
$$
\lambda\ {\rm cotan} \pi\lambda=a.
$$
For example, in the limit $a\to \infty$ we recover the answer for the Dirichlet boundary condition. 
\end{example} 

\begin{exercise} Let $H=-\frac{1}{2}\frac{d^2}{dx^2}+a\chi_{[-1,1]}(x)$ where $\chi$ is the indicator function and $a\in \Bbb R$, and let it be defined on $V=C^\infty_0(\Bbb R)\subset \mathcal H=L^2(\Bbb R)$. Show that $H$ is essentially self-adjoint and find the spectrum and eigenvalues of its self-adjoint extension (consider separately the cases $a\ge 0$ and $a<0$). 

{\bf Hint.} As explained above, the spectrum consists of $E\in \Bbb R$ for which 
$H-E: \overline V\to \mathcal H$ is not surjective. So try to solve the equation 
$$
(H-E)u=f
$$
for $f\in \mathcal H$ as 
$$
f(x)=\int_{\Bbb R} G(x,y)dy, 
$$
where $G(x,y)$ is the fundamental solution of the equation 
$$
(H-E)f=\delta(x-y).
$$
You should get that there are no eigenfunctions for $a\ge 0$ (purely continuous spectrum), while for $a<0$ the spectrum is mixed: there is continuous spectrum and also some eigenfunctions with negative eigenvalues; they are called {\bf bound states}.  
\end{exercise} 

\subsection{Hamiltonians in quantum mechanics} 

The yoga of quantization says that 
to quantize classical mechanics 
on a manifold $X$, we need to replace 
the classical space of states $T^*X$ by 
the quantum space of states -- the
Hilbert space 
${\mathcal H}=L^2(X)$ on square integrable complex
half-densities on $X$ (or, more precisely, the corresponding
projective space). Further, we need to replace 
classical observables, i.e. (sufficiently nice) 
real functions $f\in C^\infty(T^*X)$, by 
quantum observables $\widehat f$, which are 
(unbounded, densely defined) operators 
on ${\mathcal H}$, not commuting with each
other in general. Then the (expected) value of an observable $A$ in
a state $\psi\in \mathcal H$ of unit norm is, by definition, 
$\la\psi,A\psi\ra$ (provided that it is well defined). 

The operators $\widehat f$ should linearly depend on $f$. Moreover, they should 
depend on a positive real parameter $\hbar$ called the 
{\it Planck constant}\index{Planck constant}, and satisfy the following 
relation:
$$
[\widehat{f},\widehat{g}]=i\hbar\widehat{\lbrace{f,g\rbrace}}+O(\hbar^2),\
\hbar\to 0.
$$
Since the role of Poisson brackets of functions is played in quantum mechanics
by commutators of operators, 
this relation expresses the condition that classical mechanics 
should be the limit of quantum mechanics as $\hbar\to 0$.\footnote{Note that the assignment 
$f\mapsto \widehat f$ cannot possibly satisfy the identity $\widehat{fg}=\widehat f\widehat g$ 
since the product of functions is commutative but the product of operators is not.}  

We must immediately disappoint the reader by confessing that 
there is no canonical choice of the quantization map
$f\mapsto \widehat f$. Nevertheless, there are some standard choices of
$\widehat f$ for particular $f$, which we will now discuss. 

Let us restrict ourselves to the situation 
$X=\Bbb R$, so on the phase space we have coordinates $q$
(position) and $p$ (momentum). In this case we can naturally think of half-densities as functions and 
there are the following standard 
conventions. 

1. $\widehat f=f(q)$ (multiplication operator by $f(q)$) when 
$f$ is independent of $p$. 

2. $\widehat{p^m}\to (-i\hbar \frac{d}{dq})^m$. 

(Note that these conventions satisfy our condition,
since $[\widehat q,\widehat p]=i\hbar$, while $\lbrace{q,p\rbrace}=1$.)

\begin{example} For the classical Hamiltonian $H=\frac{p^2}{2}+U(q)$ 
considered above, the quantization will be the {\it Schr\"odinger operator}\index{Schr\"odinger operator}
$$
\widehat H=-\frac{\hbar^2}{2}\frac{d^2}{dq^2}+U(q).
$$ 
\end{example} 

\begin{remark} The extension of these conventions to other
functions is not unique. However, such an extension will not be
used, so we will not specify it. 
\end{remark} 

Now let us see what the quantum analog of Hamilton's equations
should be. In accordance with the outlined quantization yoga, 
Poisson brackets should be replaced in quantum theory by
commutators (with coefficient $(i\hbar)^{-1}=-i/\hbar$). Thus, 
Hamilton's equations should be replaced by 
the equation
$$
\frac{d}{dt}\langle\psi(t),A\psi(t)\rangle=\langle\psi(t),\tfrac{[A,\widehat
  H]}{i\hbar}\psi(t)\rangle=
-\tfrac{i}{\hbar}
\langle\psi(t),[A,\widehat H]\psi(t)\rangle, 
$$
where $\langle,\rangle$ is the Hermitian form on ${\mathcal H}$ and $\widehat H$ is some quantization of the classical 
Hamiltonian $H$. Since this equation must hold for any $A$, 
it is equivalent to the {\it Schr\"odinger equation}\index{Schr\"odinger equation}
$$
\dot{\psi}=-\frac{i}{\hbar}\widehat H\psi
$$
up to changing $\psi$ by a time-dependent phase factor (check it!).  
Thus, the quantum analog of the Hamilton equations is the
Schr\"odinger equation. 

\begin{remark} This ``derivation'' of the Schr\"odinger equation 
is definitely not a mathematical argument. It is merely a reasoning 
aimed to motivate a definition.
\end{remark} 

To solve the initial value problem for the 
Schr\"odinger equation, we need to make sense of 
the Hamiltonian $\widehat H$ as an unbounded self-adjoint operator 
on $\mathcal H$ in the sense of von Neumann, which in practice boils down to 
giving spatial boundary conditions for $\psi$, in addition to the initial value. 
The general solution of the Schr\"odinger equation then has 
the form 
$$
\psi(t)=e^{-\frac{it\widehat H}{\hbar}}\psi(0),
$$
where $e^{-\frac{it\widehat H}{\hbar}}$ is the 1-parameter group 
of unitary operators attached to the self-adjoint operator 
$\widehat H$, which exists thanks to von Neumann's spectral theorem. 
Therefore, for any quantum observable $A$ it is reasonable to 
define a new observable 
$$
A(t):=e^{\frac{it\widehat H}{\hbar}}A(0)e^{-\frac{it\widehat
  H}{\hbar}}
 $$ 
 (such that to observe $A(t)$ is the same as to 
evolve for time $t$ and then observe $A=A(0)$). The observable $A(t)$
satisfies the equation 
$$
A'(t)=-\frac{i}{\hbar}[A(t),\widehat H]
$$
called the {\it operator Schr\"odinger equation}\index{operator Schr\"odinger equation}, and we have 
$$
\langle\psi(t),A\psi(t)\rangle=\langle\psi(0),A(t)\psi(0)\rangle.
$$ 
The two sides of this equation represent two pictures of quantum
mechanics: Schr\"odinger's
(states change in time, observables don't) and Heisenberg's 
(observables change in time, states don't). The equation expresses the equivalence of the two pictures.

\subsection{Feynman-Kac formula}

Let us consider a 1-dimensional particle with 
potential $U(q)$. 
Let us assume that $U\ge 0$ and $U(q)\to \infty$ as $|q|\to
\infty$. In this case, the operator 
$\widehat H=-\frac{\hbar^2}{2}\frac{d^2}{dq^2}+U(q)$
is essentially self-adjoint on Schwartz functions, 
positive definite, and its spectrum is purely point. 

\begin{lemma}\label{uniqueeig} There is a unique eigenvector 
$\Omega$ of $\widehat H$ with smallest eigenvalue given by a positive
function with norm $1$. 
\end{lemma} 

\begin{proof} An eigenvector $\Omega$ of $\widehat H$ with 
smallest eigenvalue $\lambda$ minimizes the ``energy" functional 
$$
E(\phi):=\langle\phi,\widehat H\phi\rangle=\int_{\Bbb R} (\tfrac{\hbar^2}{2}\phi'(q)^2+U(q)\phi(q)^2)dq
$$
on the space of real $C^1$-functions $\phi: \Bbb R\to \Bbb R$ 
with $\int_{\Bbb R}\phi(t)^2dt=1$. Suppose that $\Omega(a)=0$, then the equation $\widehat H\Omega=\lambda\Omega$ implies 
$\Omega'(a)\ne 0$. But $E(\Omega)=E(|\Omega|)$, so, since $\Omega'(a)\ne 0$, 
 this value can be reduced by smoothing out $\Omega$ in a small neighborhood of $a$ and then normalizing it to have unit norm, a contradiction. 
This also implies that $\lambda$ is a simple eigenvalue, hence $\Omega$ is unique. 
\end{proof} 

\begin{remark} The vector $\Omega$ is called the {\it ground state}\index{ground state}, or {\it vacuum
state}\index{vacuum state}, since it has lowest energy, and physicists often shift 
the Hamiltonian by a constant so that the energy of this state
is zero (i.e. ``there is no matter"). 
\end{remark}

The correlation functions in 
the Hamiltonian setting are defined by the formula 
$$
\G_n^{\rm Ham}(t_1,...,t_n):=\langle\Omega,q(t_1)...q(t_n)\Omega\rangle
$$
where $q(t)$ is the operator quantizing the observable ``coordinate of the particle at the time $t$".

\begin{remark} Physicists usually write the inner product
$\langle v,Aw\rangle$ as $\langle v|A|w\rangle$. In particular, 
$\Omega$ is written as $\langle 0|$ or $|0 \rangle$ (the so-called Dirac {\it bra-ket notation}\index{bra-ket notation}).  
\end{remark}

\begin{theorem}\label{FK} (Feynman-Kac formula) 
If $t_1\ge ...\ge t_n$ then the function 
$\G_n^{\rm Ham}$ admits an asymptotic expansion in $\hbar$
(near $\hbar=0$), which
coincides with the path integral correlation function 
$\G_n^M$ constructed above. Equivalently, 
the Wick rotated function $\G_n^{\rm Ham}(-it_1,...,-it_n)$
equals $\G_n^E(t_1,...,t_n)$.  
\end{theorem}

This theorem plays a central role in quantum mechanics, 
and we will prove it below. Before we do so, let us formulate an
analog of this theorem for ``quantum mechanics on the circle''. 

Let $\G_{n,L}(t_1,...,t_n)$ denote the correlation 
function on the circle of length $L$ 
(for $0\le t_n\le...\le t_1\le L$), and let $Z_{L}$ 
be the partition function on the circle of length $L$, 
defined from (Euclidean) path integrals. 
Also, let 
$$
Z_{L}^{\rm Ham}=\Tr(e^{-\frac{L\widehat H}{\hbar}}),
$$  
and 
$$
\G_{n,L}^{\rm Ham}(-it_1,...,-it_n)=\frac{\Tr(q(-it_n)...q(-it_1)e^{-\frac{L\widehat
    H}{\hbar}})}{\Tr(e^{-\frac{L\widehat H}{\hbar}})}.
$$

\begin{theorem}\label{FK1} (Feynman-Kac formula on the circle) 
The functions $Z_{L}^{\rm Ham}$,  
$\G_{n,L}^{\rm Ham}$ admit asymptotic expansions in $\hbar$, which
coincide with the functions $Z_{L}$ and $\G_{n,L}$ computed   
from path integrals.
\end{theorem}

Note that Theorem \ref{FK} is obtained from Theorem \ref{FK1}
by sending $L$ to infinity. Thus, it is sufficient to prove
Theorem \ref{FK1}. 

\begin{remark} As we mentioned before, the function 
$\G_n^E$ can be defined by means of the Wiener integral, 
and the equality 
$$
\G_n^{\rm Ham}(-it_1,...,-it_n)=
\G_n^E(t_1,...,t_n)
$$ 
actually holds for numerical 
values of $\hbar$, and not just in the sense 
of power series expansions. 
The same applies to the equalities $Z_L^{\rm Ham}=Z_L$, 
$\G_{n,L}^{\rm Ham}=\G_{n,L}$. However, these
results are technically more complicated (as they require 
non-trivial analytic input) and thus are beyond the scope 
of these notes. 
\end{remark}

\begin{example}\label{benchoi} Consider the case of the quadratic potential.
By renormalizing variables, we can assume that $\hbar=m=1$, 
so $U=\frac{q^2}{2}$. In this case we know that 
$Z_L=\frac{1}{2\sinh({L\over 2})}$. On the other hand, 
$\widehat H$ is the Hamiltonian of the 
quantum harmonic oscillator:
$$
\widehat H=-\frac{1}{2}\frac{d^2}{dq^2}+\frac{q^2}{2}.
$$
The eigenvectors of this operator are 
$H_n(x)e^{-\frac{x^2}{2}}$, where $H_n$ are 
the Hermite polynomials ($k\ge 0$), and the eigenvalues are
$n+\frac{1}{2}$ (see Theorem \ref{Hermite}). Hence, 
$$
Z_L^{\rm Ham}=e^{-{L\over 2}}+e^{-{3L\over 2}}+...=\frac{1}{e^{L\over 2}-e^{-{L\over 2}}}=Z_L,
$$
as expected from the Feynman-Kac formula. 
(This shows the benefit of the choice $C={1\over 2}$ in the
normalization of $Z_L$). 
\end{example} 

\subsection{Proof of the Feynman-Kac formula in the free case (harmonic oscillator)}\label{harosc3}

Consider again the quadratic Hamiltonian
$\widehat H=-\frac{1}{2}\frac{d^2}{dq^2}+\frac{q^2}{2}$ of the quantum Harmonic oscillator.
Note that it can be written in the form
$$
\widehat H=a^\dagger a+\tfrac{1}{2},
$$
where $a=\frac{1}{\sqrt{2}}(\frac{d}{dq}+q)$,
$a^\dagger=\frac{1}{\sqrt{2}}
(-\frac{d}{dq}+q)$.  
The operators $a,a^\dagger$ define a representation 
of the Heisenberg Lie algebra on (a dense subspace of) the Hilbert space $\mathcal H$: 
$$
[a,a^\dagger]=1.
$$
Thus the eigenvectors of $\widehat H$ are $(a^\dagger)^n\Omega$
where $\Omega=e^{-\frac{q^2}{2}}$ is the lowest eigenvector
and the corresponding eigenvalues are $n+\frac{1}{2}$, $n\in \Bbb Z_{\ge 0}$
(as we already saw before in Theorem \ref{Hermite}). 

\begin{remark} The operators $a$ and $a^\dagger$ 
are called the {\it annihilation and creation operators}\index{annihilation operator}\index{creation operator}, since 
$a\Omega=0$, while all eigenvectors of $\widehat H$ can be
``created'' from $\Omega$ by action of powers of $a^\dagger$. 
\end{remark} 

Now, we have 
$$
q(0)=q=\frac{1}{\sqrt{2}}(a+a^\dagger).
$$
Since $[a^\dagger a,a]=-a$, $[a^\dagger a,a^\dagger]=a^\dagger$, 
we have 
$$
q(t) =\frac{1}{\sqrt{2}}e^{ita^\dagger a}(a+a^\dagger)
e^{-ita^\dagger a}=
\frac{1}{\sqrt{2}}(e^{-it}a+e^{it}a^\dagger)
$$
This shows that 
$$
\G_{n,L}^{\rm Ham}(-it_1,...,-it_n)=2^{-\frac{n}{2}}\frac{\Tr(\prod_{j=1}^n
(e^{t_j}a^\dagger+e^{-t_j}a)e^{-L(a^\dagger a+\frac{1}{2})})}
{\Tr(e^{-L(a^\dagger a+\frac{1}{2})})}.
$$
Now we can easily prove Theorem \ref{FK1}. 
Indeed, let us move the terms $e^{t_1}a^\dagger$ 
and $e^{-t_1}a$ around the trace (using the cyclic property of
the trace). This will yield, after a short calculation, using \eqref{glformula} :
$$
\G_{n,L}^{\rm Ham}(-it_1,...,-it_n)=
$$
$$
\sum_{j=2}^n \tfrac{1}{2}\G_{n-2,L}^{\rm Ham}(-it_2,...,-it_{j-1},-it_{j+1},...,-it_n)
\left(\frac{e^{t_1-t_j}}{e^L-1}-\frac{e^{t_j-t_1}}{e^{-L}-1}\right)=
$$
$$
\sum_{j=2}^n \G_{n-2,L}^{\rm Ham}(-it_2,...,-it_{j-1},-it_{j+1},...,-it_n)G_L(t_1-t_j).
$$
This implies the theorem by induction in $n$. 

\begin{remark} 1. In the quadratic case there is no formal expansions and
the Feynman-Kac formula holds as an equality between usual
functions. 

2. Note that the equality $\frac{e^{t-s}}{e^L-1}-\frac{e^{s-t}}{e^{-L}-1}=G_L(t-s)$ used above 
holds only if $t\ge s$. In fact, the matrix coefficient $\langle\Omega, q(t_1)...q(t_n)\Omega\rangle$ 
is not symmetric in $t_j$, as the operators $q(t_j)$ do not commute. Thus the Feynman-Kac formula 
only holds if $t_1\ge...\ge t_n$. For this reason the correlation function $\mathcal G_{n}^M$ 
is called time-ordered - it corresponds to the matrix coefficient where the operators 
$q(t_j)$ are ordered chronologically. 
\end{remark} 

\subsection{Proof of the Feynman-Kac formula (general case)}

Now we consider an arbitrary potential 
$U(q):=\frac{m^2q^2}{2}-V(q)$, where 
$$
V(q)=\sum_{k\ge 3}\frac{g_kq^k}{k!}. 
$$
For simplicity we will assume that  the coefficients $g_j$ are formal parameters
and $\hbar=1$ (the latter condition does not cause a loss of generality, as 
this situation can be achieved by rescaling).
Let us first consider the case of partition function. 
We have 
$$
Z_L^{\rm Ham}=\Tr(e^{-L\widehat H})=\Tr(e^{-L(\widehat H_0-V)}),
$$ 
where $\widehat H_0=-\frac{1}{2}\frac{d^2}{dq^2}+\frac{1}{2}m^2q^2$
is the free (=quadratic) part of the Hamiltonian. Since $g_j$ are
formal parameters, we have a series expansion
\scriptsize
\begin{equation}\label{serexp}
e^{-L(\widehat H_0-V)}=e^{-L\widehat H_0}+\\
\sum_{N\ge 1}\int_{L\ge s_1\ge...\ge s_N\ge 0}
e^{-(L-s_1)\widehat H_0}Ve^{-(s_1-s_2)\widehat H_0}V...e^{-(s_{n-1}-s_n)\widehat
  H_0}Ve^{-s_n\widehat H_0}d\bold s
\end{equation}
\normalsize This follows from the general fact that in the (completed) free
algebra with generators $A,B$, one has 
\begin{equation}\label{eAB}
e^{A+B}=e^A+\sum_{N\ge 1}\int_{1\ge s_1\ge...\ge s_N\ge 0}
e^{(1-s_1)A}Be^{(s_1-s_2)A}B...e^{(s_{N-1}-s_N)A}Be^{s_NA}d\bold s
\end{equation}
(check this identity!). 

Equation \eqref{serexp} shows that 
$$ 
Z_L^{\rm Ham}=
$$
$$
\sum_{N\ge 0}\sum_{j_1,...,j_N=3}^\infty
\frac{g_{j_1}...g_{j_N}}{j_1!...j_N!}\int_{1\ge s_1\ge...\ge s_N\ge 0}
\Tr(q_0(-is_1)^{j_1}...q_0(-is_N)^{j_N}e^{-L\widehat H_0})d\bold s,
$$
where $q_0(t)$ is the operator $q(t)$ in the free theory
associated to the potential $\frac{m^2q^2}{2}$. 

Since the Feynman-Kac formula for the free theory has already been
proved, we know that the trace on the right hand side can be evaluated 
as a sum over matchings. To see what exactly is obtained, 
let us collect the terms corresponding to all 
permutations of $j_1,...,j_N$ together.
This means that the summation variables will be the numbers
$i_3,i_4,..$ of occurences of $3,4,..$ among $j_1,...,j_N$.  
Further, to every factor $q_0(-is)^j$
will be assigned a $j$-valent vertex, with a variable 
$s$ attached to it, and it is easy to see that 
$Z_L^{\rm Ham}$ equals the sum 
over all ways of connecting the vertices (i.e. Feynman diagrams $\Gamma$)
of integrals 
$$
\int_{0\le s_1,...,s_N\le L}\prod_{v-w}
G_L(s_v-s_w)d\bold s, 
$$
multiplied by the coefficients $\frac{\prod_k g_k^{i_k}}{|{\rm Aut}\Gamma|}$. 
Thus, $Z_L^{\rm Ham}=Z_L$, as desired. 

Now let us consider correlation functions.
Thus we have to compute 
$$
\Tr(e^{-(L-t_1)\widehat H}qe^{-(t_1-t_2)\widehat H}q...qe^{-t_n\widehat
    H}).
$$
Explanding each exponential inside the trace as above, 
we will clearly get the same Feynman diagram sum, 
except that the Feynman diagrams will contain $n$ external vertices 
marked by variables $t_1,...,t_n$. 
This implies that $\G_{n,L}^{\rm Ham}=\G_{n,L}$, and we are
done. 

\subsection{The massless case}
Consider now the massless case, $m=0$, in the Hamiltonian setting. 
For maps $q:\Bbb R\to \Bbb R$, we have 
${\mathcal H}=L^2(\Bbb R)$, and 
$\widehat H=-\frac{\hbar^2}{2}\frac{d^2}{dq^2}$. 
This operator has continuous spectrum, and 
there is no lowest eigenvector $\Omega$
(more precisely, there is a lowest eigenvector $\Omega=1$, 
but it is not in $L^2$), which means that we cannot define 
the correlation functions in the usual way, 
i.e. as $\langle\Omega,q(t_1)...q(t_n)\Omega\rangle$. 
(This is the reflection, in the Hamiltonian setting, 
of the difficulties related to the growth of the Green's
function at infinity, i.e., infrared divergences, 
which we encountered in the Lagrangian setting).

Consider now the case $q: \Bbb R\to S^1=\Bbb R/2\pi r\Bbb Z$. 
In this case, we have the same Hamiltonian but acting in the
space ${\mathcal H}:=L^2(S^1)$. The eigenvectors of this operator
are $e^{\frac{iNq}{r}}$, with eigenvalues $\hbar^2 \frac{N^2}{2r^2}$. 
In particular, the lowest eigenvector is 
$\Omega=1$. Thus the Hamiltonian correlation
functions (in the Euclidean setting, for $t_1\ge...\ge t_n$) are
$$
\la\Omega,e^{\frac{t_1\widehat H}{\hbar}}e^{\frac{ip_1q}{r}}e^{\frac{(t_2-t_1)\widehat H}{\hbar}}...
e^{\frac{ip_nq}{r}}e^{-\frac{t_n\widehat H}{\hbar}}\Omega\ra=
$$
$$
e^{\frac{\hbar}{2r^2}\sum_j (t_{j}-t_{j+1})(p_1+...+p_{j})^2},
$$
which is equal to 
the correlation function in the Lagrangian setting. 
Thus the Feynman-Kac formula holds. 

Now we pass to the case of circle-valued quantum mechanics on the circle.
In this case, we have 
$$
\Tr(e^{-\frac{L\widehat H}{\hbar}})=\sum_{N\in \Bbb Z} e^{-\frac{N^2L\hbar}{2r^2}}
$$ 
and
$$
\Tr(e^{\frac{t_1\widehat H}{\hbar}}e^{\frac{ip_1q}{r}}e^{\frac{(t_2-t_1)\widehat H}{\hbar}}...
e^{\frac{ip_nq}{r}}e^{\frac{(L-t_n)\widehat H}{\hbar}})=
\sum_{N\in \Bbb Z}e^{\frac{\hbar}{2r^2}\sum_{j=0}^{n} (t_j-t_{j+1})(N-p_1-...-p_{j})^2},
$$
where $t_{n+1}:=L$, $t_0:=0$. Simplifying this expression, we obtain
$$
e^{\frac{\hbar}{2r^2}\sum_j (t_j-t_{j+1})(p_1+...+p_j)^2}
\sum_{N\in \Bbb Z} e^{-\frac{\hbar}{2r^2} (LN^2+2N\sum_j p_jt_j)}=
$$
$$
e^{\frac{\hbar}{2r^2}\sum_j (t_j-t_{j+1})(p_1+...+p_j)^2}
\theta(\tfrac{\hbar}{2\pi ir^2} \sum_j p_jt_j,\tfrac{L\hbar}{2\pi r^2}).
$$
Comparing with (\ref{thetaf}), 
we see that the Feynman-Kac formula reduces to  
the modular invariance of the theta-function:
$$
\theta(\tfrac{u}{iT},\tfrac{1}{T})=\sqrt{T} e^{\frac{\pi u^2}{T}}\theta(u,T)
$$
with $T=\frac{2\pi r^2}{\hbar L}$ (which follows from the Poisson summation formula applied to the Gaussian).

Note that the Feynman-Kac formula in this example 
would have been false if in the Lagrangian setting we had ignored 
the topologically nontrivial maps.
Thus we may say 
that the Feynman-Kac formula ``sees topology''.
This ability of the Feynman-Kac formula to ``see 
topology'' (in much more complex situations)
 lies at the foundation of many interrelations between
geometry and quantum field theory. 

\begin{remark} It should be noted that the contributions of topologically
nontrivial maps from the source circle to the target circle
are, strictly speaking, beyond
our usual setting of perturbation theory, since they are exponentially small in $\hbar$.
To be specific, the contribution from maps of degree $N$
mostly comes from those maps which are close to the minimal action
map $q_N(t)=\frac{2\pi tNr}{L}$, so it is of the
order $e^{-\frac{2\pi^2N^2r^2}{L\hbar}}$.
The maps $q_N(t)$ are the simplest examples of ``instantons''
-- nonconstant solutions of the classical equations of motion, 
which have finite action (and are nontrivial in the
topological sense). Exponentially small contributions to the
path integral coming from integration over neighborhoods of
instantons are called
``instanton corrections to the perturbation series''. 
\end{remark} 

\begin{remark}\label{givesen} 
This calculation allows us to give sense to the partition function $Z(L)$ 
of the line-valued massless quantum mechanics on the circle. To this end, 
we just need to look at the asymptotics $r\to \infty$ of the partition function 
$$
Z(r,L)=\theta(0,\tfrac{\hbar L}{2\pi r^2})=r\sqrt{\tfrac{2\pi}{\hbar L}}\theta(0,\tfrac{2\pi r^2}{\hbar L}). 
$$
Since $\theta(0,T)\to 1$ as $T\to \infty$, for the leading coefficient of the asymptotics we have (up to numerical scaling, 
which we are free to choose):
$$
Z(L)\sim \frac{1}{\sqrt{\hbar L}}. 
$$

Note however that in this case we cannot write 
$Z(L)={\rm Tr}(e^{-\frac{L\widehat H}{\hbar}})$ since this operator is not trace class. 
Also the vector $\Omega=1$ is not normalizable. Thus this theory is somewhat ill-defined, as already
mentioned above. 
\end{remark} 

\subsection{Spectrum of the Schr\"odinger operator for a piecewise constant periodic potential}

In this subsection we demonstrate the behavior of the spectrum of a 1-dimensional Schr\"odinger operator on the example  of a  piecewise constant periodic potential, when the eigenvalues and eigenfunctions can be computed fairly explicitly. 

We consider the Schr\"odinger operator on the circle $\Bbb R/2\pi \Bbb Z$ 
given by $H:=-\frac{\hbar^2}{2}\partial^2+U(x)$, where $U$ is a piecewise continuous 
$2\pi$-periodic potential. Clearly, without loss of generality we may assume that $\int_0^{2\pi}U(x)dx=0$, otherwise we can shift $U(x)$ by a constant. 
By a standard result in analysis (the theory of Sturm-Liouville operators), 
the operator $H$ has discrete spectrum, i.e., eigenvalues $E_0<E_1\le E_2\le...$ going to $+\infty$ with the corresponding eigenfunctions $\Psi_0,\Psi_1,\Psi_2,...$. For example, if 
$U=0$ then $E_0=0$ and $E_{2m-1}=E_{2m}=\frac{\hbar^2 m^2}{2}$ for $m>0$, with eigenfunctions $\Psi_0=1, \Psi_{2m-1}=\sin mx, \Psi_{2m}=\cos mx$. 

Consider now the simplest non-trivial example -- the piecewise constant potential 
\begin{equation}\label{pcon}
U(x)=\begin{cases} Mb,\ 0\le x<a\\ -Ma,a\le x<2\pi\end{cases}
\end{equation}
where $a,b,M>0$, $a+b=2\pi$. 

For every $p\in \Bbb R$, we have a basis $f_p,g_p$ of solutions of the equation $H\Psi=E\Psi$
on $[p,\infty]$ such that $f_p(p)=g_p'(p)=1,g_p(p)=f_p'(p)=0$. For example, 
$$
f_0(x)=\cos \sqrt{\tfrac{2}{\hbar^2}(E-Mb)}x,\ g_0(x)=\frac{\sin \sqrt{\frac{2}{\hbar^2}(E-Mb)}x}{\sqrt{\frac{2}{\hbar^2}(E-Mb)}}
$$
 for $0\le x<a$ and  
$$
f_a(x)=\cos \sqrt{\tfrac{2}{\hbar^2}(E+Ma)}(x-a),\ g_a(x)=\frac{\sin \sqrt{\frac{2}{\hbar^2}(E+Ma)}(x-a)}{\sqrt{\frac{2}{\hbar^2}(E+Ma)}(x-a)}
$$ 
for $a\le x<2\pi$. 
Thus the monodromy matrices along the intervals $[0,a]$, $[a,2\pi]$ in these bases are
$$
A:=\begin{pmatrix} \cos \sqrt{\frac{2}{\hbar^2}(E-Mb)} a & \frac{\sin \sqrt{\frac{2}{\hbar^2}(E-Mb)}a}{\sqrt{\frac{2}{\hbar^2}(E-Mb)}} \\
-\sqrt{\frac{2}{\hbar^2}(E-Mb)}\sin \sqrt{\frac{2}{\hbar^2}(E-Mb)} a& \cos \sqrt{\frac{2}{\hbar^2}(E-Mb)} a\end{pmatrix},
$$
$$
B:=\begin{pmatrix} \cos \sqrt{\frac{2}{\hbar^2}(E+Ma)} b & \frac{\sin \sqrt{\frac{2}{\hbar^2}(E+Ma)}b}{\sqrt{\frac{2}{\hbar^2}(E+Ma)}} \\
-\sqrt{\frac{2}{\hbar^2}(E+Ma)}\sin \sqrt{\frac{2}{\hbar^2}(E+Ma)} b& \cos \sqrt{\frac{2}{\hbar^2}(E+Ma)} b\end{pmatrix}.
$$
The condition for a periodic solution is that the matrix $AB$ (monodromy around the circle) has an eigenvalue $1$. Since $\det A=\det B=1$, in this case the second eigenvalue of $AB$ is also $1$ 
(generically this matrix is a unipotent Jordan block), so the condition is ${\rm Tr}(AB)=2$, which gives 
\scriptsize
\begin{equation}\label{eigeneq}
\cos \sqrt{\tfrac{2}{\hbar^2}(E-Mb)} a \cos \sqrt{\tfrac{2}{\hbar^2}(E+Ma)} b-\tfrac{E+M\frac{a-b}{2}}{\sqrt{(E-Mb)(E+Ma)}}\sin \sqrt{\tfrac{2}{\hbar^2}(E-Mb)} a \sin \sqrt{\tfrac{2}{\hbar^2}(E+Ma)} b=1.
\end{equation}
\normalsize
Thus the eigenvalues of $H$ are the solutions $E$ of \eqref{eigeneq}. 

If $a<E<b$ then $\sqrt{\frac{2}{\hbar^2}(E-Mb)}$ is imaginary, so \eqref{eigeneq} can be written in terms of real parameters as 
\scriptsize
\begin{equation}\label{eigeneq1}
\cosh \sqrt{\tfrac{2}{\hbar^2}(Mb-E)} a \cos \sqrt{\tfrac{2}{\hbar^2}(E+Ma)} b-\tfrac{E+M\frac{a-b}{2}}{\sqrt{(Mb-E)(E+Ma)}}\sinh \sqrt{\tfrac{2}{\hbar^2}(Mb-E)} a \sin \sqrt{\tfrac{2}{\hbar^2}(E+Ma)} b=1.
\end{equation}
\normalsize

As mentioned above, if $M=0$, then for each $n\ge 1$ the operator $H$ double eigenvalue 
$\frac{1}{2}\hbar^2 n^2$. We would like to see what happens to this eigenvalue for large $n$
as we turn on $M$ and keep the product $\hbar n$ in a bounded interval $[C^{-1},C]$ (so $\hbar\to 0$).

Let us rewrite \eqref{eigeneq} in the form 
$$
1-\cos \left(\sqrt{\tfrac{2}{\hbar^2}(E-Mb)}a+\sqrt{\tfrac{2}{\hbar^2}(E+Ma) }b\right)=
$$
$$
\left(1-\tfrac{E+M\frac{a-b}{2}}{\sqrt{(E-Mb)(E+Ma)}}\right)\sin \sqrt{\tfrac{2}{\hbar^2}(E-Mb)} a \sin \sqrt{\tfrac{2}{\hbar^2}(E+Ma)} b
$$
and look for solutions 
$$
E=\frac{1}{2}(\hbar^2 n^2+\varepsilon),
$$
where $|\varepsilon|\ll \frac{1}{n}$. We have 
$$
\sqrt{\tfrac{2}{\hbar^2}(E-Mb)}=\sqrt{n^2+\tfrac{\varepsilon-2Mb}{\hbar^2}}=n\left(1+\tfrac{\varepsilon-2Mb}{2\hbar^2 n^2}-\tfrac{(\varepsilon-2Mb)^2}{8\hbar^4 n^4}...\right),
$$
$$
\sqrt{\tfrac{2}{\hbar^2}(E+Ma)}=\sqrt{n^2+\tfrac{\varepsilon+2Ma}{\hbar^2}}=n\left(1+\tfrac{\varepsilon+2Ma}{2\hbar^2 n^2}-\tfrac{(\varepsilon+2Ma)^2}{8\hbar^4 n^4}...\right),
$$
so 
$$
\sqrt{\tfrac{2}{\hbar^2}(E-Mb)}a+\sqrt{\tfrac{2}{\hbar^2}(E+Ma) }b=2\pi n\left(1+\tfrac{\varepsilon}{2\hbar^2 n^2}-\tfrac{M^2ab}{2\hbar^4n^4}+...\right).
$$
Thus the left hand side of the above equation has the form 
$$
LHS=\frac{\pi^2}{2}\left(\frac{\varepsilon}{\hbar^2 n}-\frac{M^2ab}{2\hbar^4n^3}\right)^2+...
$$
We also have
$$
1-\frac{E+M\frac{a-b}{2}}{\sqrt{(E-Mb)(E+Ma)}}=-\frac{\pi^2M^2}{2\hbar^4n^4}+...
$$
So we get 
$$
RHS=\frac{\pi^2M^2}{2\hbar^4n^4}\sin^2 na+...
$$
Thus we obtain 
$$
\frac{\varepsilon}{\hbar^2 n}-\frac{M^2ab}{2\hbar^4n^3}=\pm \frac{M|\sin na|}{\hbar^2n^2}, 
$$
which yields 
$$
\varepsilon=\frac{M^2ab}{2\hbar^2n^2}\pm \frac{M|\sin na|}{n}, 
$$
We see that the double eigenvalue $\Lambda_n=\frac{\hbar^2 n^2}{2}$, $n>0$ for $M=0$ bifurcates into two eigenvalues 
\begin{equation}\label{bifur}
\Lambda_n^\pm(M)=\Lambda_n+\frac{M^2a(2\pi -a)}{8\Lambda_n}\pm \frac{M|\sin na|}{2n}+o((M+\tfrac{1}{n})^2),\ M\to 0,n\to \infty.
\end{equation}

\subsection{WKB approximation and the Weyl law} 

The goal of this subsection is to explain how to compute semiclassical asymptotics of eigenvalues and eigenfunctions of quantum hamiltonians. This method is called the {\bf Wentzel-Kramers-Brillouin (WKB) approximation},\index{WKB approximation} named after the authors of three separate papers which introduced it independently in 1926. 

We start with a general discussion of WKB approximation for linear ODE. Suppose we have an equation
\begin{equation}\label{odeeq}
\hbar\frac{dF}{dx}=AF
\end{equation}
for a vector-function of one variable $F(x)\in \Bbb C^n$, where 
$A(x)\in {\rm Mat}_n(\Bbb C)$ is a matrix-valued function (smooth on a certain interval $I\subset \Bbb R$).
We would like to understand the asymptotic behavior of solutions 
of this equation as $\hbar\to 0$. To this end, assume for simplicity that $A(x)$ has simple spectrum for generic $x$, and let $v_1(x),...,v_n(x)$ 
be its column eigenvectors with eigenvalues $\lambda_1(x),...,\lambda_n(x)$, and $v_1^*(x),...,v_n^*(x)$ the dual basis of row eigenvectors. Let us now look for solutions of \eqref{odeeq} 
in the form 
$$
F(x)=e^{\frac{\phi(x)}{\hbar}}(\psi_0(x)+\hbar \psi_1(x)+\hbar^2 \psi_2(x)...),
$$
where $\psi_0(x)\ne 0$ and the series in parentheses is formal. Substituting, we get 
$$
(\hbar\partial_x+\phi'-A)(\psi_0+\hbar \psi_1+\hbar^2 \psi_2+...)=0,
$$
which in degree $0$ with respect to $\hbar$ yields the equation
$$
A\psi_0=\phi'\psi_0.
$$
Thus $\phi'=\lambda_j$ is an eigenvalue of $A$, so 
$$
\phi(x)=\int \lambda_j(x)dx,\ \psi_0(x)=f(x)v_j(x),
$$ 
where $f$ is a scalar function. 

Further, in degree $1$ in $\hbar$ we obtain the equation 
$$\psi_0'=(A-\lambda)\psi_1,$$
i.e., 
$$
f'v_j+fv_j'=(A-\lambda)\psi_1. 
$$
For this to have a solution $\psi_1$, we need 
$(v_j^*,f'v_j+fv_j')=0$, i.e., 
$$
f'=-(v_j^*,v_j')f.
$$
Thus 
$$
f(x)=\exp\left(-\int (v_j^*,v_j')dx\right). 
$$
Now we can recursively solve for $\psi_1,\psi_2,...$. This leads to the following result. 

\begin{theorem}
There is a unique, up to scaling, basis of formal solutions of equation \eqref{odeeq} of the form
$$
F_j(x)=\exp\left(\frac{\int \lambda_j(x)dx}{\hbar}\right)\left(\exp\left(-\int (v_j^*(x),v_j'(x))dx\right)v_j(x)+O(\hbar)\right).
$$
\end{theorem} 

Let us now apply this theorem to the stationary Schr\"odinger equation 
\begin{equation}\label{staschr}
(-\tfrac{\hbar^2}{2}\partial_x^2+U(x))\Psi=E\Psi.
\end{equation} 
Set $p(x):=\sqrt{2(E-U(x))}$, then \eqref{staschr} takes the form 
$$
\hbar^2\partial_x^2\Psi=-p^2\Psi.
$$
This can be written as the system of equations 
$$
\hbar \partial_x \binom{\Psi}{\hbar \Psi'}=\begin{pmatrix} 0 & 1 \\ -p^2 & 0\end{pmatrix}\binom{\Psi}{\hbar \Psi'}.
$$
Thus we have equation \eqref{odeeq} with $A=\begin{pmatrix} 0 & 1 \\ -p^2 & 0\end{pmatrix}$. So we have 
$$
\lambda_1=ip,\ \lambda_2=-ip
$$
and we may take 
$$ 
v_1=\binom{1}{ip},\ v_2=\binom{1}{-ip},
$$
so that 
$$
v_1^*=\frac{1}{2}(1,-ip^{-1}),\ v_2^*=\frac{1}{2}(1,ip^{-1}).
$$ 
Thus we obtain the following formal solutions of \eqref{staschr}:
$$
\Psi_\pm=\exp\left(\pm \frac{i\int pdx}{\hbar}\right)\left(\exp\left(-\frac{1}{2}\int p^{-1}p'dx\right)+O(\hbar)\right)=
$$
$$
p^{-\frac{1}{2}}\exp\left(\pm\frac{i\int pdx}{\hbar}\right)(1+O(\hbar)).
$$
We get

\begin{theorem}\label{localWKB} (local WKB approximation) Equation \eqref{staschr} has a basis of formal solutions 
$$
\Psi_\pm(x)=(2(E-U(x))^{-\frac{1}{4}}\exp\left(\pm \frac{i\int \sqrt{2(E-U(x))}dx}{\hbar}\right)(1+O(\hbar)).
$$
\end{theorem} 

The WKB approximation can also be used to find asymptotic distribution of eigenvalues of a Schr\"odinger operator when it has discrete spectrum. Let us explain, somewhat informally, how this works. 

As an example, consider the stationary Schr\"odinger equation 
\eqref{staschr} on the circle $\Bbb R/2\pi \Bbb Z$ 
with piecewise continuous $2\pi$-periodic potential $U(x)$. 
We would like to write an asymptotic formula for the $n$-th eigenvalue $E_n(\hbar)$
of the operator $H=-\frac{1}{2}\hbar^2\partial^2+U(x)$ when $n\sim \frac{A}{\hbar}$ 
for a given constant $A$. This is equivalent to determining the number $\nu(E)$ of eigenvalues of $H$ satisfying the inequality $\Lambda\le E$ for a given constant $E$. 

To this end, we will use Theorem \ref{localWKB}. Assume first that 
$$
E>\sup U(x).
$$  
The periodicity condition for the solutions $\Psi_\pm$ in Theorem \ref{localWKB} (called the {\it quantization condition}\index{quantization condition} in quantum mechanics) in the zeroth approximation is that 
\begin{equation}\label{qcond}
\int_0^{2\pi} \sqrt{2(E-U(x))}dx=2\pi n\hbar,\ n\in \Bbb Z_{\ge 0}.
\end{equation} 
It follows that if \ref{qcond} holds then the number of eigenvalues 
of $H$ which are $\le E$ is about $2n$. So we get 

\begin{proposition}\label{weyllaw}
$$
\nu(E)\sim \frac{A(E)}{\hbar},\ \hbar\to 0,\text{ where } A(E):=\frac{1}{\pi}\int_0^{2\pi} \sqrt{2(E-U(x))}dx.
$$
Thus for sufficiently large $A$, we have 
$$
E_{[\frac{A}{\hbar}]}(\hbar)\sim E(A),
$$ 
where 
$E(A)$ is the solution of the equation 
$$
A=\frac{1}{\pi}\int_0^{2\pi}\sqrt{2(E-U(x))}dx.
$$
\end{proposition}  

Note that $A(E)$ is the area of the region in the classical phase space $T^*S^1=S^1\times \Bbb R$ defined by the inequality 
$$
H_{\rm cl}\le E,
$$
where $H_{\rm cl}:=\frac{1}{2}p^2+U(x)$ is the corresponding classical hamiltonian. Moreover, one can show that with this definition of $A(E)$, the formula 
$$
\nu(E)\sim \frac{A(E)}{\hbar}
$$
in fact holds in a much larger generality, whenever $H$ has discrete spectrum (namely, for the operator $-\frac{1}{2}\hbar^2\Delta+U(x)$ on any compact Riemannian manifold, or even on a non-compact one when one has $U(x)\to +\infty$ as $x\to \infty$). This formula is known as the {\it Weyl law}.\index{Weyl law}

\begin{exercise} Prove the Weyl law on the circle for $E\le \sup U(x)$. 
\end{exercise} 

Finally, let $U(x):=MU_0(x)$ where $U_0$ is a fixed potential and consider the asymptotics of eigenvalues for small $M$, assuming that $\hbar \ll M$ (i.e., $\frac{1}{n}\ll M$). In this case we can write equation \eqref{qcond} as 
\begin{equation}\label{qcond1}
\sqrt{2E}\int_0^{2\pi} \left(1-\frac{MU_0(x)}{2E}-\frac{M^2U_0(x)^2}{8E^2}+o(M^2)\right)dx=2\pi n\hbar.
\end{equation} 
As before, we assume without loss of generality that $\int_0^{2\pi} U_0(x)dx=0$. Let $I:=\frac{1}{2\pi}\int_0^{2\pi}U_0(x)^2dx$. Then we obtain
\begin{equation}\label{qcond2}
\sqrt{2E}=n\hbar+\tfrac{M^2I}{2(2 E)^\frac{3}{2}}+o(M^2)=n\hbar(1+\tfrac{M^2I}{2n^4\hbar^4}+...)
\end{equation} 
It follows that 
$$
E=\tfrac{1}{2}n^2\hbar^2(1+\tfrac{M^2I}{n^4\hbar^4}+...)=\Lambda_n+\frac{M^2I}{8\Lambda_n}+...
$$
This gives the first correction of the eigenvalue $\Lambda_n:=\tfrac{1}{2}n^2\hbar^2$ as we turn on $M$. 

For example, if $U(x)$ is given by \eqref{pcon} then $I=a(2\pi-a)$ and we recover
the asymptotics \eqref{bifur} without the last (bifurcation) term (which is negligible compared to 
 $\frac{M^2a(2\pi-a)}{8\Lambda_n}$ in the range $\frac{1}{n}\ll M$).

\section{Fermionic integrals}

\subsection{Bosons and fermions}

In physics there exist two kinds of particles -- bosons 
and fermions. So far we have dealt with bosons only, but
many important particles are fermions: e.g., 
electron, proton, etc. Thus it is important to adapt our
techniques to the fermionic case. 

In quantum theory, the difference between 
bosons and fermions is as follows: if the 
space of states of a single particle 
is ${\mathcal H}$ then the space of states 
of the system of $k$ such particles
is $S^k{\mathcal H}$ for bosons and $\Lambda^k{\mathcal H}$ 
for fermions. In particular, in the fermionic case, if $\dim \mathcal H=n$ 
then the space of states of $\ge n+1$ identical particles is zero, 
which is the {\it Pauli exclusion principle}\index{Pauli exclusion principle} (leading, for instance, to the fact 
that the number of electrons in an atom at the $m$-th energy level is bounded by $2m^2$).  
In classical theory, 
this means that the space of states 
of a bosonic particle is a usual real vector space (or, 
more generally, a manifold), while for a fermionic particle
it is an {\it odd vector space}.\index{odd vector space} Mathematically ``odd'' means that 
the algebra of smooth functions on this space (i.e. the algebra of classical 
observables) is an {\it exterior} algebra  
(unlike the case of a usual, {\it even} space, for which 
the algebra of polynomial functions is a {\it symmetric} algebra).

More generally, one may consider systems of classical particles
or fields some of which are bosonic and some fermionic. 
In this case, the space of states will be a supervector space,
i.e. the direct sum of an even and an odd space (or, more
generally, a supermanifold -- a notion we will define below).

When such a theory is quantized using the path integral approach,
one has to integrate functions over supermanifolds. 
Thus, we should learn to integrate over supermanifolds
and then generalize to this case our Feynman diagram techniques.
This is what we do in this section. 

\subsection{Supervector spaces}

Let $k$ be a field of characteristic zero. 
A {\it supervector space}\index{supervector space} (or shortly, superspace) over $k$ is just 
a $\Bbb Z/2$-graded vector space: $V=V_0\oplus V_1$.
If $V_0=k^n$ and $V_1=k^m$ then $V$ is denoted by $k^{n|m}$. 
The notions of a linear operator, direct sum, tensor product,
dual space for supervector spaces 
are defined in the same way as for $\Bbb Z/2$-graded vector
spaces. In other words, the tensor category of
supervector spaces is the same as that of $\Bbb Z/2$-graded
vector spaces. 

However, the notions of a supervector space 
and a $\Bbb Z/2$-graded vector space are {\it not} the same. 
The difference is as follows. The category of vector
(and hence $\Bbb Z/2$-graded vector) spaces has a
{\it symmetric structure}\index{symmetric structure}, which is the standard isomorphism 
$V\otimes W\to W\otimes V$ (given by $v\otimes w\to w\otimes
v$). This isomorphism allows one to define symmetric powers 
$S^iV$, exterior powers $\Lambda^iV$, etc. 
For supervector spaces, there is also a symmetry
$V\otimes W\to W\otimes V$, but it is defined 
differently. Namely, $v\otimes w$ goes to 
$(-1)^{ij}w\otimes v$, $v\in V_i, w\in V_j$ ($i,j\in \lbrace{0,1\rbrace}$). 
In other words, it is the same as usual except that if $v,w$ are
both odd then $v\otimes w\mapsto -w\otimes v$. As a result, 
we can define the superspaces 
$S^iV$ and $\Lambda^iV$ for a superspace $V$, but 
they are not the same as the symmetric and exterior powers 
in the usual sense. For example, if $V$ is purely odd
($V=V_1$), then $S^iV$ is the $i$-th exterior  power of $V$, and
$\Lambda^iV$ is the $i$-th symmetric power of $V$
(purely even for even $i$ and purely odd for odd $i$). 
Thus in general for $V=V_0\oplus V_1$, we 
have the following expressions for the symmetric algebra 
$SV:=\oplus_{i\ge 0}S^iV$ and exterior algebra $\Lambda V:=\oplus_{i\ge 0}\Lambda^i V$: 
$$
SV=SV_0\otimes \Lambda V_1,\ \Lambda V=\Lambda V_0\otimes SV_1.
$$

For a superspace $V$, let $\Pi V$ be the same space with opposite
parity, i.e. $(\Pi V)_j=V_{1-j}$, $j=0,1$. Then
we have
$$
S^iV=\Pi^i(\Lambda^i\Pi V),\ \Lambda^iV=\Pi^i(S^i\Pi V).
$$

Let $V=V_0\oplus V_1$ be a finite dimensional superspace. 
Define the algebra of polynomial functions on $V$, ${\mathcal O}(V)$, 
to be the algebra $SV^*$ (where symmetric powers 
are taken in the supersense). Thus, ${\mathcal O}(V)=SV_0^*\otimes
\Lambda V_1^*$, where $V_0$ and $V_1$ are regarded as usual
spaces. More explicitly, if 
$x_1,...,x_n$ are linear coordinates on $V_0$, and
$\xi_1,...,\xi_m$ are linear coordinates on $V_1$, 
then ${\mathcal O}(V)=k[x_1,...,x_n,\xi_1,...,\xi_m]$, 
with defining relations 
$$
x_ix_j=x_jx_i,\ x_i\xi_r=\xi_rx_i,\ \xi_r\xi_s=-\xi_s\xi_r
$$
(in particular, $\xi_r^2=0$). Note that this algebra is itself a (generally, infinite
dimensional) supervector space, and is commutative in the
supersense. Also, if $V,W$ are two superspaces, then ${\mathcal
  O}(V\oplus W)={\mathcal O}(V)\otimes {\mathcal O}(W)$, 
where the tensor product of algebras is understood 
in the supersense, i.e. 
$$
(a\otimes b)(c\otimes d)=
(-1)^{p(b)p(c)}(ac\otimes bd),
$$ 
where $p(x)$ is the parity of $x$. 

\subsection{Supermanifolds} 

Now assume that $k=\Bbb R$. Then 
by analogy with the above 
for any supervector space $V$ 
we can define the algebra 
of smooth functions, $C^\infty(V):=C^\infty(V_0)\otimes \Lambda
V_1^*$. In fact, this is a special case of the following more
general setting. 

\begin{definition} A {\it supermanifold}\index{supermanifold}
$M$ is a usual manifold $M_0$ 
with a sheaf $C^\infty_{M}$ of 
$\Bbb Z/2\Bbb Z$ graded algebras
(called the {\it structure sheaf}\index{structure sheaf}), which is locally isomorphic to 
$C^\infty_{M_0}\otimes \Lambda (\xi_1,...,\xi_m)$.
\end{definition}

The manifold $M_0$ is called the {\it reduced manifold}\index{reduced manifold} of $M$. 
The dimension of $M$ is the pair of integers $\dim M_0|m$.  

For example, a supervector space $V$ is a supermanifold 
of dimension $\dim V_0|\dim V_1$. 
Another (more general) example of a supermanifold is 
a superdomain $U:=U_0\times V_1$, i.e. a domain $U_0\subset V_0$ together 
with the sheaf $C^\infty_{U_0}\otimes \Lambda V_1^*$. 
Moreover, the definition of a supermanifold implies that 
any supermanifold is ``locally isomorphic'' to a superdomain.

Let $M$ be a supermanifold. An {\it open set} $U$ in $M$ is 
the supermanifold ($U_0$, $C^\infty_M|_{U_0}$),
where $U_0$ is an open subset in $M_0$.

By the definition, supermanifolds form a category $\mathcal
{S}$. Let us describe explicitly 
morphisms in this category, i.e. maps $F:M\to N$ 
between supermanifolds $M$ and $N$. By the definition, it
suffices to assume that $M,N$ are superdomains,
with global coordinates $x_1,...,x_n,\xi_1,...,\xi_m$, and 
$y_1,...,y_p,\eta_1,...,\eta_q$, respectively
(here $x_i$, $y_i$ are even variables, 
and $\xi_i,\eta_i$ are odd variables). Then 
the map $F$ is defined by the formulas:
$$
y_i=f_{0,i}(x_1,...,x_n)+f_{2,i}^{j_1j_2}(x_1,...,x_n)\xi_{j_1}\xi_{j_2}+...,
$$
$$
\eta_i=a_{1,i}^j(x_1,...,x_n)\xi_j+a_{3,i}^{j_1j_2j_3}(x_1,...,x_n)
\xi_{j_1}\xi_{j_2}\xi_{j_3}+...,
$$
where $f_{0,i},f_{2,i}^{j_1j_2},...,a_{1,i}^j,a_{3,i}^{j_1j_2j_3},...$ are 
usual smooth functions, and we assume summation over repeated
indices. These formulas, determine $F$ completely, 
since for any $g\in C^\infty(N)$ one can find $g\circ F\in C^\infty(M)$ by Taylor's
formula. 
For example, if $M=N=\Bbb R^{1|2}$,
$F(x,\xi_1,\xi_2)=(x+\xi_1\xi_2,\xi_1,\xi_2)$, 
and $g=g(x)$, then 
$$
g\circ F(x,\xi_1,\xi_2)=
g(x+\xi_1\xi_2)=g(x)+g'(x)\xi_1\xi_2.
$$ 

%\begin{remark} For this reason, we consider only 
%$C^\infty$ (and not $C^r$) functions on supermanifolds. 
%Indeed, if for example $g(x)$ is a $C^r$ function of one variable which is 
%not differentiable $r+1$ times, then 
%the expression $g(x+\sum_{i=1}^{r+1}\xi_{2i-1}\xi_{2i})$
%will not be defined, because the coefficient of 
%$\xi_1,...,\xi_{2r+2}$ in this expression 
%should be $g^{(r+1)}(x)$, but this derivative does not exist.
%\end{remark}

\subsection{Supermanifolds and vector bundles} 

Let $M_0$ be a manifold, and $E$ be a real vector bundle on $M_0$. 
Then we can define the supermanifold $M:={\rm Tot}(\Pi E)$, the total
space of $E$ with changed parity. Namely, the reduced manifold 
of $M$ is $M_0$, and the structure sheaf $C^\infty_M$ is the sheaf of
sections of $\Lambda
E^*$. This defines a functor $S: {\mathcal {B}}\to {\mathcal
  {S}}$, from the category of manifolds with vector bundles 
to the category of supermanifolds. We also have a functor 
$S_*$ in the opposite direction: namely, 
$S_*(M)$ is the manifold $M_0$ with the vector bundle 
$(R/R^2)^*$, where $R$ is the nilpotent radical of $C^\infty_M$.

The following proposition (whose proof we leave as an exercise) 
gives a classification of
supermanifolds. 

\begin{proposition}\label{batch} (i) $S_*\circ S={\rm Id}$; 

(ii) $S\circ S_*={\rm Id}$ on isomorphism classes of objects. 
\end{proposition}

The usefulness of this proposition is limited by the fact that, 
as one can see from the above description of maps between 
supermanifolds, $S\circ S_*$ is {\it not} the identity on morphisms
(e.g. it maps the automorphism 
$x\to x+\xi_1\xi_2$ of $\Bbb R^{1|2}$ to ${\rm Id}$),
and hence, $S$ is not an equivalence of categories. 
In fact, the category of supermanifolds is not equivalent to the
category of manifolds with vector bundles 
(namely, the category of supermanifolds ``has more morphisms'').

\begin{remark} 1. The relationship between these two categories is
quite similar to the relationship between the categories of 
(finite dimensional) filtered and graded vector spaces,
respectively (namely, for them we also have functors $S$, $S_*$
with the same properties -- check it!). Therefore in 
supergeometry, it is better to avoid 
using realizations of supermanifolds as $S(M_0,E)$, similarly to how
in linear algebra it is better to avoid choosing a splitting of a
filtered space. 

2. In the definition of a supermanifold one can replace the real exterior algebra  $\Lambda(\xi_1,...,\xi_m)$ 
with the complexified exterior algebra $\Lambda_{\Bbb C}(\xi_1,...,\xi_m)$. This gives a notion of a $\Bbb C$-supermanifold, which generalizes the notion of an ordinary smooth manifold with the sheaf of complex-valued (as opposed to real-valued) smooth functions. Similarly to Proposition \ref{batch}, isomorphism classes of $\Bbb C$-supermanifolds with reduced submanifolds $M_0$ are in bijection with isomorphism classes of {\it complex} vector bundles on $M_0$, so they are more general (as not every complex vector bundle is the complexification of a real one). Otherwise, the theory of $\Bbb C$-supermanifolds (which does actually arise in quantum field theory, see Remark \ref{realcom} below) is completely parallel to the theory of usual supermanifolds. 

One may also similarly define complex analytic and algebraic supermanifolds, but this is a different story which we will not discuss here. 
\end{remark}

\subsection{Supertrace and superdeterminant (Berezinian)}

Before proceeding further, we need to generalize 
to the supercase the basic notions of linear algebra, such as trace and determinant of a matrix.

Let $R:=R_0\oplus R_1$ be a supercommutative $\Bbb C$-algebra. Fix two nonnegative integers $m,n$.  
Let ${\rm Mat}_{n|m}(R)$ be the algebra of $n+m$ by $n+m$ matrices over $R$ which have 
the block decomposition 
$$
A=\begin{pmatrix} A_{00} & A_{01}\\ A_{10} & A_{11}\end{pmatrix}
$$ 
so that $A_{00}$ is $n$ by $n$, $A_{11}$ is $m$ by $m$,  
and $A_{00},A_{11}$ 
have even entries (i.e., in $R_0$), 
while $A_{01},A_{10}$ have odd entries (i.e., in $R_1$). 
We would like to define the {\it supertrace}\index{supertrace} of $A$ as a
linear function 
$$
{\rm sTr}(A)=\sum_{i,j=1}^{n+m} \lambda_{ij}a_{ij},\ \lambda_{ij}\in \Bbb Z, 
$$
so that ${\rm sTr}\begin{pmatrix} 1& 0\\ 0& 0\end{pmatrix}=n$ and 
${\rm sTr}(AB)={\rm sTr}(BA)$ 
for any $R$ and $A,B\in {\rm Mat}_{n|m}(R)$. 
Thus we must have 
${\rm sTr}(A)={\rm Tr}(A_{00})+\varepsilon {\rm Tr}(A_{11})$
for some $\varepsilon\in \Bbb Z$, and taking all blocks 
of $A,B$ except $A_{01},B_{10}$ to be zero, we get $\varepsilon =-1$. 
So the supertrace of $A$ has to be defined by the formula 
$$
{\rm sTr}(A)={\rm Tr}(A_{00})-{\rm Tr}(A_{11}). 
$$

Now let us generalize to the supercase the definition of 
determinant. For a finite dimensional algebra $R$ and $C\in {\rm Mat}_{n|m}(\Bbb R)$ 
we would like to have 
\begin{equation}\label{expone}
{\rm sdet}(e^C)=e^{{\rm sTr}C}=e^{{\rm Tr}(C_{00})-{\rm Tr}(C_{11})},
\end{equation}
which generalizes the usual property of trace and determinant. 
So in the case of a block-diagonal matrix $C=C_{00}\oplus C_{11}$ we get 
$$
{\rm sdet}(e^C)=\frac{\det(e^{C_{00}})}{\det(e^{C_{11}})}.
$$
Thus if $A=A_{00}\oplus A_{11}$ is block-diagonal, we must have 
$$
{\rm sdet}A=\frac{\det A_{00}}{\det A_{11}}.
$$
This shows that we cannot hope that the superdeterminant will be a polynomial in the entries of $A$ -- 
it has to be a rational function defined only on some open subset. In fact, if we want to have the usual property 
${\rm sdet}(AB)={\rm sdet}(A){\rm sdet}(B)$ then there is just one possibility. Indeed, 
suppose that 
$$
A=\begin{pmatrix} 1 & b\\ 0& 1\end{pmatrix} 
\begin{pmatrix} a_+ & 0\\ 0& a_-\end{pmatrix}
\begin{pmatrix} 1 & 0\\ c& 1\end{pmatrix}=\begin{pmatrix} a_++ba_-c & ba_-\\ a_-c& a_-\end{pmatrix}.
$$
By \eqref{expone}, we must have 
$$
{\rm sdet} \begin{pmatrix} 1 & b\\ 0& 1\end{pmatrix}={\rm sdet} \begin{pmatrix} 1 & 0\\ c& 1\end{pmatrix}=1,
$$
hence 
$$
{\rm sdet}(A)=\frac{\det a_+}{\det a_-}.
$$
In other words, the superdeterminant has to be defined by the formula  
$$
{\rm sdet}(A)=\frac{\det(A_{00}-A_{01}A_{11}^{-1}A_{10})}
{\det(A_{11})}
$$
provided that $A_{11}$ is invertible; otherwise the superdeterminant is not defined.  

This function is also called the {\it Berezinian}\index{Berezinian} of $A$ and denoted ${\rm Ber}(A)$. 
So for $m=0$ one has ${\rm Ber}(A)=\det(A)$, and for
$n=0$ one has ${\rm Ber}(A)=(\det A)^{-1}$.

\begin{remark} Recall for comparison that if $A$ is a purely even block matrix 
then 
$$
\det(A)=\det(A_{00}-A_{01}A_{11}^{-1}A_{10})\det(A_{11}).
$$
\end{remark} 

\begin{proposition} (i) For any $A,B\in {\rm Mat}_{n|m}(R)$ 
with $A_{11},B_{11}$ invertible, we have 
$$
{\rm Ber}(AB)={\rm Ber}(A){\rm Ber}(B).
$$

(ii) If $R$ is finite dimensional and $A(t)\in {\rm Mat}_{n|m}(R)$ 
is a $C^1$-function near $0$ with $A(0)$ invertible then 
$$
\tfrac{d}{dt}|_{t=0}{\rm Ber}(A(t))={\rm sTr}(A'(0)A(0)^{-1}){\rm Ber}(A(0)).
$$

(iii) If $R$ is finite dimensional then for any $C\in {\rm Mat}_{n|m}(R)$ 
we have 
$$
{\rm Ber}(e^C)=e^{{\rm sTr}C}.
$$
\end{proposition} 

\begin{proof} (i) From the triangular factorization,
it is clear that it suffices to consider the case 
$$
A=\begin{pmatrix} 1 & 0\\ X &1\end{pmatrix},\ 
B=\begin{pmatrix} 1 & Y\\ 0 &1\end{pmatrix},
$$
where $X,Y$ are matrices with odd elements, so that 
$$
AB=\begin{pmatrix} 1 & Y\\ X &1+XY\end{pmatrix}.
$$
Then the required identity
is
$$
\det(1-Y(1+XY)^{-1}X)=\det(1+XY).
$$
To prove this identity, recall that $X: V_0\to V_1\otimes R$ and $Y: V_1\to V_0\otimes R$. 
We have \scriptsize
$$
\det(1-Y(1+XY)^{-1}X)=\sum_{k\ge 0}(-1)^k {\rm Tr}(Y(1+XY)^{-1}X|_{\Lambda^k V_0})=
$$
$$
=\sum_{k\ge 0}(-1)^k{\rm sTr}(Y(1+XY)^{-1}|_{\Lambda^kV_1}\circ X|_{\Lambda^kV_0})=\sum_{k\ge 0}(-1)^k{\rm sTr}(XY(1+XY)^{-1}|_{\Lambda^kV_1})
$$
$$
\sum_{k\ge 0}{\rm Tr}(XY(1+XY)^{-1}|_{S^k\Pi V_1})=\det(1-XY(1+XY)^{-1})^{-1}=\det(1+XY). 
$$ \normalsize

(ii) By (i) we may replace $A(t)$ by $A(t)A(0)^{-1}$, so it suffices to consider the case $A(0)=1$, where
the statement easily follows from the definition. 

(iii) Consider the function $f(t):={\rm Ber}(e^{Ct})$. By (ii) it satisfies the differential equation 
$f'(t)={\rm sTr}(C)f(t)$ with $f(0)=1$. Thus $f(t)=e^{{\rm sTr}(C)t}$, and the statement follows by setting $t=1$. 
\end{proof} 
 
\subsection{Integration on superdomains}

We would now like to develop integration theory on supermanifolds. 
Before doing so, let us recall how it is done for usual manifolds. 
In this case, one proceeds as follows. 

1. Define integration of compactly supported (say, smooth) functions on 
a domain in $\Bbb R^n$. 

2. Find the transformation formula for the integral under change of 
coordinates (i.e. discover the factor $|J|$, where $J$ is the Jacobian). 

3. Define a {\it density}\index{density} on a manifold to be a quantity which is locally 
the same as a function, but multiplies by $|J|$ under coordinate change
(unlike true functions, which don't multiply by anything). 
Then define integral of compactly supported densities on the manifold
using partitions of unity. The independence of the integral on the choices
is guaranteed by the change of variable formula and the definition of a density. 

We will now realize this program for supermanifolds. 
We start with defining integration over superdomains. 

Let $V=V_0\oplus V_1$ be a supervector space.
The {\it Berezinian} of $V$ is the line $\Lambda^{\rm top}V_0^*
\otimes \Lambda^{\rm top}V_1$ (where $V_0,V_1$ are treated as usual spaces). 
Suppose that $V$ is equipped with 
a nonzero element $dv$ of its Berezinian
(called a {\it supervolume element})\index{supervolume element}. 

Let $U_0$ be an open set in $V_0$, and $f\in C^\infty(U_0)\otimes
\Lambda V_1^*$ be a compactly supported smooth function on
the superdomain $U:=U_0\times V_1$ (i.e. 
$f=\sum f_i\otimes \omega_i$, $f_i\in C^\infty(U_0)$, $\omega_i\in
\Lambda V_1^*$, and $f_i$ are
compactly supported). Let 
$dv_0,dv_1$ be volume forms on $V_0,V_1$ such that
$dv=dv_0/dv_1$. 

\begin{definition} The integral $\int_U f(v)dv$ 
is $\int_{U_0}(f(v),(dv_1)^{-1})dv_0$. 
\end{definition}

It is clear that this quantity depends only on $dv$ and not on
$dv_0$ and $dv_1$ separately. 

Thus, $\int_U f(v)dv$ is defined as the integral of the suitably
normalized top coefficient of $f$ (expanded with respect to some 
homogeneous basis of $\Lambda V_1^*$). To write it in
coordinates, let $x_1,...,x_n,\xi_1,...,\xi_m$ be a linear system of
coordinates on $V$ such that 
$dv=\frac{dx_1...dx_n}{d\xi_1...d\xi_m}$
(such coordinate systems will be called unimodular with respect
to $dv$). Then 
$\int_U f(v)dv$ equals $\int_{U_0} f_{\rm top}(x_1,...,x_n)dx_1...dx_n$, where 
$f_{\rm top}$ is the coefficient of $\xi_1...\xi_m$ in 
the expansion of $f$.

\subsection{Berezin's change of variable formula}

Let $V$ be a vector space, $f\in \Lambda V^*$, $v\in V$. 
Denote by $\frac{\partial f}{\partial v}$ the result of 
contraction of $f$ with $v$. 

Let $U,U'$ be superdomains, and $F: U\to U'$ be a morphism. 
As explained above, given linear coordinates $x_1,...,x_n,\xi_1,...,\xi_m$ on $U$
and $y_1,...,y_p,\eta_1,...,\eta_q$ on $U'$, we can describe $F$
by expressing $y_i$ and $\eta_j$ as functions of $x_i$ and
$\xi_j$. Define the {\it Berezin matrix}\index{Berezin matrix} of $F$, $A:=DF(x,\xi)$ by the formulas:
$$
A_{00}=(\tfrac{\partial y_i}{\partial x_k}),\
A_{01}=(\tfrac{\partial y_i}{\partial \xi_\ell}),\
A_{10}=(\tfrac{\partial \eta_j}{\partial x_k}),\
A_{11}=(\tfrac{\partial \eta_j}{\partial \xi_\ell}).
$$
Clearly, this is a superanalog of the Jacobi matrix. 

The main theorem of supercalculus is the following theorem.

\begin{theorem}\label{Ber} (Berezin)
Let $g$ be a smooth function with compact support 
on $U'$, and $F: U\to U'$ be an isomorphism. 
Let $dv,dv'$ be supervolume elements on $U,U'$. 
Then 
$$
\int_{U'} g(v')dv'=\int_U g(F(v))|{\rm Ber}(DF(v))|dv,
$$ 
where the Berezinian is computed with respect to unimodular
coordinate systems. 
\end{theorem}

Here if 
$f(\xi)=a+$terms containing $\xi_j$, $a\in \Bbb R$, $a\ne 0$ then 
by definition $|f(\xi)|:=f(\xi)$ is $a>0$ and $|f(\xi)|:=-f(\xi)$ if $a<0$.

\begin{proof} The chain rule of the usual calculus extends
verbatim to supercalculus. Thus, since  
${\rm Ber}(AB)={\rm Ber}(A){\rm Ber
 }(B)$, if we know the statement for two isomorphisms 
$F_1: U_2\to U_1$ and $F_2: U_3\to U_2$, then we know it for the composition 
$F_1\circ F_2$. 

Let
$F(x_1,...,x_n,\xi_1,...,\xi_m)=(x_1',...,x_n',\xi_1',...,\xi_m')$.
We see that it suffices 
to consider the following cases. 

1. $x_i'$ depend only on $x_k$, $k=1,...,n$, and $\xi_j'=\xi_j$. 

2. $x_i'=x_i+z_i$, where $z_i$ lie in the ideal generated by
$\xi_j$, and $\xi_j'=\xi_j$. 

3. $x_i'=x_i$. 

Indeed, it is clear that any isomorphism $F$ is a composition of isomorphisms 
of types $1,2,3$.

In case 1, the statement of the theorem follows from 
the usual change of variable formula. Thus it suffices to
consider cases 2 and 3. 

In case 2, it is sufficient to consider the case 
when only one coordinate is changed by $F$, i.e. 
$x_1'=x_1+z$, and $x_i'=x_i$ for $i\ge 2$. 
In this case we have to show that the integral of 
$$
g(x_1+z,x_2,...,x_n,\xi)(1+\tfrac{\partial z}{\partial
  x_1})-g(x_1,x_2,...,x_n,\xi)
$$
 is zero. 
But this follows easily upon expansion in powers of $z$, since all
the terms are manifestly total derivatives with respect to $x_1$. 

In case 3, we can also assume $\xi_j'=\xi_j$, $j\ge 2$,
and a similar (actually, even simpler) argument 
proves the result.  
\end{proof}

\subsection{Integration on supermanifolds}

Now we will define densities on supermanifolds. 
Let $M$ be a supermanifold, and $\lbrace{U_{\alpha}\rbrace}$ be an
open cover of $M$ together with isomorphisms 
$f_{\alpha}: U_\alpha\to U_{\alpha}'$ , where $U_{\alpha}'$ is a 
superdomain in $\Bbb
R^{n|m}$. Let $g_{\alpha\beta}: f_\beta (U_\alpha\cap U_\beta)\to 
f_\alpha(U_\alpha\cap U_\beta)$ be the transition map 
$f_\alpha f_\beta^{-1}$. 
Then a density $s$ on $M$ is a choice of an element 
$s_\alpha\in C^\infty_M(U_\alpha)$ for each $\alpha$,
such that on $U_\alpha\cap U_\beta$ one has 
$s_\beta(z)=s_\alpha(z) |{\rm Ber}(g_{\alpha\beta})(f_\beta(z))|$.

\begin{remark} It is clear that a density on $M$ is a global section of 
a certain sheaf on $M$, called the sheaf of densities. 
\end{remark}

Now, for any (compactly supported) density 
$\omega$ on $M$, the integral 
$\int_M \omega$ is well defined. Namely, it is defined as in
usual calculus: one uses a partition of unity $\phi_\alpha$ 
such that ${\rm Supp}\phi_\alpha\subset (U_\alpha)_0$ are compact subsets,
and sets $\int_M\omega:=\sum_\alpha \int_M\phi_\alpha\omega$
(where the summands can be defined using $f_\alpha$).
Berezin's theorem guarantees then that the final answer 
will be independent on the choices made. 

\subsection{Gaussian integrals in an odd space}

Now let us generalize to the odd case the theory of Gaussian integrals, 
which was, in the even case, the basis for the path integral approach to quantum 
mechanics and field theory. 

Recall first the notion of {\it Pfaffian}.\index{Pfaffian} 
Let $A$ be a skew-symmetric matrix of even size. 
Then the determinant of $A$ is the square of a 
polynomial in the entries of $A$. This polynomial 
is determined by this condition up to sign. 
The sign is usually fixed by requiring that 
the polynomial should be $1$ 
for the direct sum of matrices $\begin{pmatrix}
0&1\\ -1&0\end{pmatrix}$. With this convention, this polynomial 
is called the Pfaffian of $A$ and denoted ${\rm Pf}A$. 
The Pfaffian obviously has the property ${\rm Pf}(X^TAX)=
{\rm Pf}(A)\det(X)$ for any matrix $X$. 
  
Let now $V$ be a $2m$-dimensional vector space with a volume element $dv$, and $B$ 
a skew-symmetric bilinear form on $V$. We define the Pfaffian ${\rm Pf}B$ of $B$ 
to be the Pfaffian of the matrix of $B$ in any unimodular basis
(by the above transformation formula, it does not depend on the choice of the basis). 
It is easy to see (by reducing $B$ to the canonical form) that 
$$
\frac{\Lambda^mB}{m!}={\rm Pf}(B)dv. 
$$
In terms of matrices, this translates into the following 
(well known) formula for the Pfaffian of a skew symmetric matrix of size $2m$:
$$
{\rm Pf}(A)=\sum_{\sigma\in \Pi_m}\varepsilon_\sigma
\prod_{i\in \lbrace{1,...,2m\rbrace}, 
i<\sigma(i)}a_{i\sigma(i)},
$$
where $\Pi_m$ is the set of matchings of $\lbrace{1,...,2m\rbrace}$,
and $\varepsilon_\sigma$ is the sign of the permutation 
sending $1,...,2m$ to $i_1,\sigma(i_1),...,i_m,\sigma(i_m)$
(where $i_r<\sigma(i_r)$ for all $r$). 
For example, for $m=2$ (i.e. a 4 by 4 matrix), 
$$
{\rm Pf}(A)=a_{12}a_{34}+a_{14}a_{23}-a_{13}a_{24}.
$$

Now consider an odd vector space $V$ of dimension $2m$ with a volume element $d\xi$. 
Let $B$ be a symmetric bilinear form on $V$ (i.e. 
a skewsymmetric form on $\Pi V$). 
Let $\xi_1,...,\xi_{2m}$ be unimodular linear coordinates on $V$ (i.e. 
$d\xi=d\xi_1\wedge...\wedge d\xi_{2m}$). So if $\xi=(\xi_1,...,\xi_{2m})$
then $B(\xi,\xi)=\sum_{i,j}b_{ij}\xi_i\xi_j$, where $b_{ij}$ is a skewsymmetric matrix.  

\begin{proposition}\label{pfaff}
$$
\int_V e^{\frac{1}{2}B(\xi,\xi)}(d\xi)^{-1}={\rm Pf}(B).
$$
\end{proposition}

\begin{proof}
The integral equals $\frac{1}{m!}\frac{\wedge^m B}{d\xi}$, which is precisely ${\rm Pf}(B)$.
\end{proof}

This formula has the following important special case.  
Let $Y$ be a finite dimensional odd vector space, 
and $V=Y\oplus Y^*$. The space $Y$ has a canonical volume element
$dv=dydy^*$, defined as follows: if 
$e_1,...,e_m$ is a basis of $Y$ and $e_1^*,...,e_m^*$ 
is the dual basis of $Y^*$ then $dydy^*=e_1\wedge e_1^*\wedge...\wedge e_n\wedge e_n^*$. 

Let $A: Y\to Y$ be a linear operator. Then we can define an even smooth function 
$S$ on the odd space $Y$ as follows: $S(y,y^*)=(Ay,y^*)$. 
More explicitly, if $\xi_i$ are coordinates on $Y$ corresponding 
to the basis $e_i$, and $\eta_i$ the dual system of coordinates on $Y^*$, 
then 
$$
S(\xi_1,...,\xi_m,\eta_1,...,\eta_m)=\sum_{i,j} a_{ij}\xi_j\eta_i,
$$
where $(a_{ij})$ is the matrix of $A$ in the basis $e_i$. 

\begin{proposition}\label{detp}
$$
\int_V e^S(dv)^{-1}=(-1)^{\frac{n(n-1)}{2}}\det A.
$$
\end{proposition} 

\begin{proof} We have $S(y,y_*)=\frac{1}{2}B((y,y_*),(y,y_*))$, 
where $B$ is the skewsymmetric form on $\Pi V$ given by the formula 
$$
B((y,y^*),(w,w^*))=(Ay,w^*)-(Aw,y^*).
$$ 
It is easy to see that ${\rm Pf}(B)=(-1)^{\frac{n(n-1)}{2}}\det(A)$, so 
Proposition \ref{detp} follows from Proposition \ref{pfaff}.

Another proof can be obtained by direct evaluation of the top coefficient. 
\end{proof}

\subsection{The Wick formula in the odd case}

Let $V$ be a $2m$-dimensional odd space with a volume form $d\xi$, 
and $B\in S^2V^*$ a non-degenerate 
form (symmetric in the supersense and antisymmetric in the usual sense).
Let $\lambda_1,...,\lambda_n$ be linear functions on $V$. 
Then $\lambda_1,...,\lambda_n$ can be regarded as odd smooth functions on the superspace $V$. 
\begin{theorem}\label{wickodd}
$$
\int_V \lambda_1(\xi)...\lambda_n(\xi)e^{-\frac{1}{2}B(\xi,\xi)}(d\xi)^{-1}=
{\rm Pf}(-B){\rm Pf}(B^{-1}(\lambda_i,\lambda_j)).
$$
(By definition, this is zero if $n$ is odd).
In other words, we have: 
$$
\int_V \lambda_1(\xi)...\lambda_n(\xi)e^{-\frac{1}{2}B(\xi,\xi)}(d\xi)^{-1}=
$$
$$
{\rm Pf}(-B)\sum_{\sigma\in \Pi_m}\varepsilon_\sigma
\prod_{i\in \lbrace{1,...,2m\rbrace}, 
i<\sigma(i)}B^{-1}(\lambda_i,\lambda_{\sigma(i)}).
$$
\end{theorem}

\begin{proof}
We prove the second formula. 
Choose a basis $e_i$ of $V$ with respect to which 
the form $B$ is standard: $B(e_j,e_l)=1$
if $j=2i-1,l=2i$,  
and $B(e_j,e_l)=0$ for other pairs $j<l$.
Since both sides of the formula are polylinear 
with respect to $\lambda_1,...,\lambda_n$, 
it suffices to check it if $\lambda_1=e_{i_1}^*$,...,
$\lambda_n=e_{i_n}^*$. This is easily done by direct 
computation (in the sum on the right hand side, only one term may
be nonzero). 
\end{proof} 

\begin{exercise} Let  $Y=\Bbb R^{\frac{n(n+1)}{2}+\frac{m(m-1)}{2}|mn}$ 
be the real superspace of matrices $$
A=\begin{pmatrix} A_{00} & A_{01}\\ A_{10} & A_{11}\end{pmatrix}
$$ 
(where $A_{00}$ is $n$ by $n$ and $A_{11}$ is $m$ by $m$) which are symmetric in the supersense, i.e., $A_{00}$ is symmetric, $A_{11}$ is skew-symmetric, and $A_{01}^T=A_{10}$. Let  
$Y_+\subset Y$ be the superdomain of those matrices for which 
$A_{00}>0$. 
Let $dA$ be a supervolume element on $Y$.
Let $f$ be a compactly supported smooth function on $Y_+$. 
Show that
$$
\int_{Y_+\times \Bbb R^{n|m}}
f(A)e^{-x^TA_{00}x-2x^TA_{01}\xi-\xi^TA_{11}\xi}dAdx(d\xi)^{-1}
=
$$
$$
=C\int_{Y_+}f(A){\rm Ber}(A)^{-1/2}dA.
$$
($C$ is a constant). What is $C$?
\end{exercise}

\begin{exercise} Prove the Amitsur-Levitzki identity: 
if $X_1,...,X_{2n}$ are n by n matrices 
over a commutative ring, then 
$$
\sum_{\sigma\in S_{2n}}
(-1)^\sigma X_{\sigma(1)}...X_{\sigma(2n)}=0.
$$  

{\bf Hint.} (a) Show that for any n by n matrix $X$ with anticommuting 
entries, $X^{2n}=0$ (namely, show that traces of $X^{2k}$ vanish for all positive 
$k$, then use the Cayley-Hamilton theorem for $X^2$). 

(b) Apply this to $X=\sum_{i=1}^{2n} X_i\xi_i$, where $\xi_i$ are anticommuting variables.
\end{exercise}  

\section{Quantum mechanics for fermions}

\subsection{Feynman calculus in the supercase}

Wick's theorem allows us to extend Feynman calculus to the supercase. 
Namely, let $$V=V_0\oplus V_1$$ be a finite dimensional real superspace
with a supervolume element $dv=dv_0(dv_1)^{-1}$, 
equipped with a symmetric non-degenerate form $B=B_0\oplus B_1$ ($B_0>0$).
Let 
$$
S(v)=\frac{1}{2}B(v,v)-\sum_{r\ge 3}\frac{B_r(v,v,...,v)}{r!}
$$ 
be an even function on $V$ (the action). Note that $B_r$, $r\ge 3$ can contain 
mixed terms involving both odd and even variables, e.g. 
$x\xi_1\xi_2$ (the so called ``Yukawa term''). 
We will consider the integral
$$
I(\hbar)=\int_V \ell_1(v_0)...\ell_n(v_0)\lambda_1(v_1)...\lambda_p(v_1)e^{-\frac{S(v)}{\hbar}}dv,
$$   
where $v_0,v_1$ are the even and odd components of $v$. 
Then this integral has an expansion in $\hbar$ written in terms of Feynman diagrams. 
Since $v$ has both odd and even part, these diagrams 
will contain ``odd'' and ``even'' edges (which are usually depicted by 
straight and wiggly lines, respectively). 
More precisely, let us write 
$$
B_r(v,v,...,v)=\sum_{s=0}^r\begin{pmatrix}r\\ s
\end{pmatrix}
B_{s,r-s}(v_1,...,v_1,v_0,...,v_0),
$$
where $B_{s,r-s}$ has homogeneity degree $s$ with respect to $v_1$ and $r-s$ with 
respect to $v_0$ (i.e. it will be nonzero only for even $s$). 
Then to each term $B_{s,r-s}$ we assign an $(s,r-s)$-valent flower, 
i.e. a flower with $s$ odd and $r-s$ even outgoing edges, and for the set 
of odd outgoing edges, specify which orderings are even.   
Then, given an arrangement of flowers, for every matching $\sigma$ 
of outgoing edges, we can define an amplitude $\Bbb F(\sigma)$ by contracting 
the tensors $B_{s,r-s}$ (and being careful with the signs). 
It is easy to check that all matchings giving the same graph 
will contribute to $I(\hbar)$ with the same sign, and thus we have 
almost the same formula 
as in the bosonic case: 
$$
I(\hbar)=(2\pi)^{\frac{\dim V_0}{2}}\hbar^{\frac{\dim V_0-\dim V_1}{2}}
\frac{{\rm Pf}(-B_1)}{\sqrt{\det B_0}}
\sum_{\Gamma}\frac{\hbar^{b(\Gamma)}}{|{\rm Aut}(\Gamma)|}\Bbb F_\Gamma (\ell_1,...,\ell_n,\lambda_1,...,\lambda_p),
$$
where the summation is taken over graphs with $n$ even and $p$ odd outgoing edges. 

\begin{remark} More precisely, 
we can define the sign $\varepsilon_\sigma$ of a matching $\sigma$ 
as follows: label outgoing edges by $1,2,...$, starting from the fisrt flower, then second, 
etc., so that the labeling is even on each flower. 
Then write the labels in a sequence, enumerating (in any order) the pairs 
defined by $\sigma$ (the element with the smaller of the two labels goes first).
The sign $\varepsilon_\sigma$ 
is by definition the sign 
of this ordering (as a permutation of $1,2,...$).
Then $\Bbb F_\Gamma$ is $\Bbb F(\sigma)$ for any matching $\sigma$ yielding 
$\Gamma$ which is {\it positive}, i.e. such that $\varepsilon_\sigma=1$. 
For a negative matching, $\Bbb F_\Gamma=-\Bbb F(\sigma)$. 
\end{remark}

%\subsection{Fermionic loops and the corresponding minus signs}
In most (but not all) situations considered in physics, the action 
is quadratic in the fermionic variables, i.e. 
$$
S(v)=S_b(v_0)-\tfrac{1}{2}S_f(v_0)(v_1,v_1),
$$ 
where 
$S_f(v_0)$ is a skew-symmetric bilinear form on $\Pi V_1$. 
In this case, using fermionic Wick's theorem, 
we can perform exact 
integration with respect to $v_1$, and reduce $I(\hbar)$ to a purely bosonic integral.
For example, if we have only  $\ell_i$ and no $\lambda_i$, then
$$
I(\hbar)=\hbar^{-\frac{\dim V_1}{2}}\int_{V_0}\ell_1(v_0)...\ell_n(v_0)e^{-\frac{S_b(v_0)}{\hbar}}
{\rm Pf}(S_f(v_0))dv_0.
$$
In this situation, all vertices which have odd outgoing edges, will have only two 
of them, and therefore in any Feynman diagram 
with even outgoing edges, odd lines form nonintersecting simple curves, called {\it fermionic 
loops}\index{fermionic loop} (in fact, the last formula is nothing but the result of regarding 
these loops as a new kind of vertices -- convince yourself of this!). 
In this case, there is the following simple way of assigning signs to 
Feynman diagrams. For each vertex with two odd outgoing edges, 
we orient the first edge inward and the second one outward. 
We allow only connections (matchings) 
that preserve orientations (so the fermionic loops become 
oriented). Then the sign is $(-1)^r$, where $r$ is the number of fermionic loops (i.e. 
each fermionic loop contributes a minus sign). This follows from the fact that an even cycle
is an odd permutation.  

\subsection{Fermionic quantum mechanics} \label{fqm}

Let us now pass from finite dimensional fermionic integrals 
to quantum mechanics, i.e. 
integrals over fermionic functions of one (even) real 
variable $t$. 
 
Let us first discuss fermionic classical mechanics, in the
Lagrangian setting. Its difference with the bosonic case is 
that the ``trajectory'' of the particle is described by 
an {\it odd-valued} function \index{odd-valued function}of one variable, i.e. 
$\psi: \Bbb R\to \Pi V$, where $V$ is a vector space. 
Mathematically this means that the space of fields (=trajectories) 
is an odd vector space $\Pi C^\infty(\Bbb R,V)$. 
A Lagrangian ${\mathcal L}(\psi)$ is a local expression in such a field
(i.e. a polynomial in $\psi,\dot{\psi},...$),
and an action is the integral $S=\int_{\Bbb R}{\mathcal L}dt$. This means that 
the action is an element of the space $\Lambda (C^\infty_0(\Bbb R,V)^*)$.

Consider for example the theory of a single scalar-valued free fermion $\psi(t)$. 
By definition, 
the Lagrangian for such a theory is 
$$
{\mathcal L}=\frac{1}{2}
\psi\dot{\psi},
$$ 
i.e. the action 
is 
$$
S=
\frac{1}{2}\int \psi\dot{\psi} dt.
$$
This Lagrangian is 
the odd analog of the Lagrangian of a free particle, $\frac{\dot{q}^2}{2}$. 

\begin{remark} Note that $\psi\dot{\psi}\ne \frac{d}{dt}(\frac{\psi^2}{2})=0$,
since $\psi\dot{\psi}=-\dot{\psi}\psi$, so this Lagrangian is
``reasonable''. On the other hand, the same Lagrangian would be
unreasonable in the bosonic case, as it
would be a total derivative,
and hence the action would be zero. 
Finally, note that it would be equally unreasonable 
to use in the fermionic case the usual bosonic Lagrangian 
$\frac{1}{2}(\dot{q}^2-m^2q^2)$; it would
identically vanish if $q$ were odd-valued. 
\end{remark} 

The Lagrangian ${\mathcal L}$ is invariant under the group 
of reparametrizations ${\rm Diff}_+(\Bbb R)$, and 
the Euler-Lagrange equation for this Lagrangian is 
$$
\dot{\psi}=0
$$
(i.e. no dynamics). Theories with such properties are called 
{\it topological quantum field theories.}\index{topological quantum field theory}

Let us now turn to quantum theory in the Lagrangian setting,
i.e. the theory given 
by the Feynman integral $\int \psi(t_1)...\psi(t_n)e^{\frac{iS(\psi)}{\hbar}}D\psi$. 
In the bosonic case, we ``integrated'' such expressions over 
the space $C^\infty_0(\Bbb R)$. 
This integration did not make immediate sense because of difficulties
with measure theory in infinite dimensions. So we had to make sense of this integration
in terms of $\hbar$-expansion, using Wick's formula and 
Feynman diagrams. In the fermionic case, 
the situation is analogous. Namely, now we must integrate functions over 
$\Pi C^\infty_0(\Bbb R)$, which are elements of $\Lambda {\mathcal D}(\Bbb R)$, 
where ${\mathcal D}(\Bbb R)$ is the space of distributions on $\Bbb R$.  
Although in the fermionic case we don't need measure theory (as
integration is 
completely
algebraic), we still have trouble defining the integral: recall that by definition 
the integral should be the top coefficient of the integrand as the element of 
$\Lambda {\mathcal D}(\Bbb R)$, which makes no sense since 
in the exterior algebra of an infinite dimensional space there is no top component.
Thus we have to use the same strategy as in the bosonic case, i.e. 
Feynman diagrams. 

Let us, for instance, define the quantum theory for a free scalar
valued fermion, 
i.e one described by the Lagrangian ${\mathcal L}=\frac{1}{2}\psi\dot{\psi}$. 
According to the yoga we used in the bosonic case, 
the two-point function of this theory $\langle\psi(t_1)\psi(t_2)\rangle$ 
should be the function $G(t_1-t_2)$, where $G$ is the solution 
of the differential equation 
$$
\frac{dG}{dt}=i\delta(t).
$$
(the factor $i$ comes from the exponent 
in the Feynman integral; note that in the fermionic case it does {\it not} go away under Wick rotation).

The general solution of this equation
 has the form 
 $$
 G(t)=\frac{i}{2}{\rm sign}(t)+C.
 $$
Because of the fermionic nature of the field $\psi(t)$, 
it is natural to impose the 
requirement that $G(-t)=-G(t)$, i.e. that the  
correlation functions are antisymmetric;
this singles out the solution $G(t)=\frac{i}{2}{\rm sign}(t)$
(we also see 
from this condition that we should set $G(0)=0$).
As usual, the $2n$-point correlation functions are defined
by the Wick formula. That is, for distinct $t_j$,
$$
\langle\psi(t_1)...\psi(t_{2n})\rangle=(-1)^\sigma(2n-1)!!(\tfrac{i}{2})^n,
$$ 
where $\sigma$ is the permutation that orders $t_j$ in the decreasing order.
If at least two points coincide, the correlation function is zero.

Thus we see that the correlation functions are invariant under ${\rm Diff}_+(\Bbb R)$. 
In other words, using physical terminology, we have a 
{\it topological quantum field theory}.\index{topological quantum field theory} 

Note that the correlation functions in the Euclidian setting 
for this model are the same as in the Minkowski setting, since 
they are (piecewise) constant in $t_j$. 
In particular, they don't decay at infinity, and hence our theory
does not have the clustering property. 

We have considered the theory of a massless fermionic field. 
Consider now the massive case. This means, we want to add 
to the Lagrangian a quadratic term in $\psi$ which 
does not contain derivatives. If we have only one field
$\psi$, the only choice for such term is $\psi^2$, which is zero.
So in the massive case we must have at least two fields. 
Let us therefore consider the theory of two fermionic fields 
$\psi_1,\psi_2$ with (Euclidean) Lagrangian 
$$
{\mathcal
L}=\frac{1}{2}(\psi_1\dot{\psi}_1+\psi_2\dot{\psi_2}-m\psi_1\psi_2),
$$
where $m>0$ is a mass parameter. 
The Green's function for this model satisfies the
differential equation 
$$
\frac{dG}{dt}-MG=i\delta(t),
$$
where $M=\begin{pmatrix} 0&m\\ -m&0\end{pmatrix}$
and $G$ is a 2 by 2 matrix-valued function. The general solution of this equation 
is 
$$
G(t)=\begin{cases}e^{Mt}Q_-,\ t<0\\ e^{Mt}Q_+,\ t>0\end{cases}  
$$
where $Q_+-Q_-=i$. Now, we want the Wick rotated Green's function 
$G(-it)$ to have the clustering property. Thus we want  
$$
\lim_{t\to +\infty}e^{-iMt}Q_+=0,\ \lim_{t\to -\infty}e^{-iMt}Q_-=0.
$$
This implies that $Q_+=iP_+$, $Q_-=-iP_-$, where $P_{\pm}$ are the orthogonal projectors 
to the eigenspaces of $iM$ with eigenvalues $\pm m$ (and $G(0)=0$).

\begin{remark} It is easy to generalize this analysis to 
the situation when $\psi$ takes values in a positive definite inner
product space $V$, and $M:V\to V$ is a skewsymmetric operator,
since such a situation is a direct sum of the situations
considered above. 
\end{remark} 

In the case when $M$ is non-degenerate, one can define 
the corresponding theory with interactions, i.e. 
with higher than quadratic terms in $\psi$. Namely, one 
defines the correlators as sums of amplitudes of appropriate
Feynman diagrams. We leave it to the reader to work out 
this definition, by analogy with the finite dimensional case
which we have discussed above. 

\subsection{Super Hilbert spaces}

The space of states of a quantum system is a Hilbert space. 
As we plan to do Hamiltonian quantum mechanics 
for fermions, we must define a superanalog of this notion. 

Suppose ${\mathcal H}={\mathcal H}_0\oplus {\mathcal H}_1$ 
is a $\Bbb Z/2$-graded complex vector space. 

\begin{definition} (i) A Hermitian form on ${\mathcal H}$ is 
an even sesquilinear form $\la,\ra$, such that 
$\la x,y\ra=\overline{\la y,x\ra}$ for even $x,y$, and 
$\la x,y\ra=-\overline{\la y,x\ra}$ for odd $x,y$.

(ii) A Hermitian form is positive definite if $\la x,x\ra>0$ for even
$x\ne 0$, and $-i\la x,x\ra>0$ for odd $x\ne 0$.
A super Hilbert space is a superspace with a positive definite
Hermitian form $\la,\ra$, which is complete in the corresponding norm. 

(iii) Let $\mathcal H$ be a super Hilbert space, 
and $T: {\mathcal H_0}\oplus \Pi {\mathcal H_1}\to 
{\mathcal H_0}\oplus \Pi {\mathcal H_1}$ be a homogeneous linear operator
between the underlying purely even spaces.
The Hermitian adjoint operator $T^\dagger$ is defined by 
the equation $\la x,T^\dagger y\ra=(-1)^{p(x)p(T)}\la Tx,y\ra$, 
where $p$ denotes the parity. 
\end{definition}

\subsection{The Hamiltonian setting for fermionic quantum
mechanics} 

Let us now discuss what should be the Hamiltonian picture
for the theory of a free fermion. 
More precisely, let $V$ be a positive definite 
finite dimensional real inner product space, 
and consider the Lagrangian 
$$
{\mathcal L}=\frac{1}{2}((\psi,\dot{\psi})-(\psi,M\psi)),
$$ 
where $\psi: \Bbb R\to \Pi V$, and $M:V\to V$ is a skew-symmetric
operator. 

To understand what the Hamiltonian picture should be, let us
compare with the bosonic case. Namely, consider the Lagrangian
$$
{\mathcal L}_b=\frac{1}{2}(\dot{q}^2-m^2q^2),
$$ 
where $q:\Bbb R\to V$. 
In this case, 
the classical space of states is 
$$
Y:=T^*V=V\oplus
V^*.
$$ 
The equations of motion are Newton's equations
$$
\ddot{q}=-m^2q,
$$ 
which can be reduced to Hamilton's equations 
$$
\dot{q}=p,\ \dot{p}=-m^2q.
$$
The algebra of classical observables is 
$C^\infty(Y)$, with Poisson bracket defined by 
$\lbrace{a,b\rbrace}=(a,b)$, $a,b\in Y^*$, where $(,)$ is the
form on $Y^*$ inverse to the natural symplectic form on $Y$. 
The hamiltonian $H$ is determined (up to adding a constant) by the condition 
that the equations of motion are $\dot{f}=\lbrace{f,H\rbrace}$; 
in this case it is $H=\frac{1}{2}(p^2+m^2q^2)$. 

The situation in the fermionic case is analogous, with some
important differences which we will explain below. Namely, 
it is easy to compute that the equation of motion 
(i.e. the Euler-Lagrange equation) is 
$$
\dot{\psi}=M\psi.
$$ 
The main difference with the bosonic case is that 
this equation is of first and not of second order,
so the space of classical states is just $\Pi V$
(no momentum or velocity variables are introduced). Hence
the algebra of classical observables is 
$C^\infty(\Pi V)=\Lambda V^*$. To define 
a Poisson bracket on this algebra, recall that 
$\Pi V$ has a natural ``symplectic structure'', defined 
by the {\it symmetric} form $(,)$ on $V$. 
Thus we can define a Poisson bracket on 
$\Lambda V^*$ by the same formula as above:
$\lbrace{a,b\rbrace}=(a,b)$ when $a,b\in V^*$. More precisely, 
$\lbrace{,\rbrace}$ is a unique skew symmetric (in the supersense) 
bilinear operation on $\Lambda
V^*$ which restricts to $(a,b)$ for $a,b\in V^*$, 
and is a derivation with respect to each variable:
$$
\lbrace{a,bc\rbrace}=\lbrace{a,b\rbrace}c+(-1)^{p(a)p(b)}b\lbrace{a,c\rbrace},
$$
where $p(a)$ denotes the parity of $a$. 

Now it is easy to see what should play the role of the
Hamiltonian. More precisely, the definition 
with Legendre transform is not valid in our situation, since the
Legendre transform
was done with respect to the velocity variables, which we don't
have in the fermionic case. On the other hand, 
as we discussed in Section 8, in the bosonic case the equation of motion 
$$\dot{f}=\lbrace{f,H\rbrace}$$ determines $H$ uniquely, up to a
constant. The situation is the same in the fermionic case. 
Namely, by looking at the equation 
of motion $\dot{\psi}=M\psi$, 
it is easy to see that the Hamiltonian equals $$H=\frac{1}{2}(\psi,M\psi).$$
In particular, if $M=0$ (massless case), the Hamiltonian is zero 
(a characteristic property of topological field theories).  

Now let us turn to quantum theory. 
In the bosonic case the algebra of quantum observables 
is a noncommutative deformation of the algebra $C^\infty(Y)$
in which the relation $\lbrace{a,b\rbrace}=(a,b)$ is replaced
with its quantum analog $$ab-ba=i(a,b)$$ (up to the Planck constant
factor which here we will set to $1$). In particular, the subalgebra
of polynomial observables is the Weyl algebra $W(Y)$, generated by 
$Y^*$ with this defining relation. By analogy with this, we should
define the algebra of quantum observables in the fermionic case 
to be generated by $V^*$ with the relation $$ab+ba=i(a,b)$$ 
(it deforms the relation $ab+ba=0$ which defines
$\Lambda V^*$). So we recall the following definition. 

\begin{definition} Let $V$ be a vector space over a field $k$ with
a symmetric bilinear form $Q$. The {\it Clifford algebra}\index{Clifford algebra} ${\rm Cl}(V,Q)$ is 
generated by $V$ with defining relations $ab+ba=Q(a,b)$, $a,b\in
V$. 
\end{definition}

We see that the algebra of
quantum observables should be ${\rm Cl}(V_{\Bbb C}^*,i(,))$. 
Note that like in the classical case, this algebra is
naturally $\Bbb Z/2$ graded, so that 
we have even and odd quantum observables.

Now let us see what should be the Hilbert space of quantum
states. In the bosonic case it was $L^2(V)$, 
which is, by the well known Stone-von Neumann theorem,
the unique irreducible unitary representation
of $W(Y)$. By analogy with this, in the fermionic case the
Hilbert space of states should be an irreducible
unitary representation
of ${\rm Cl}(V^*_{\Bbb C})$ on a supervector space ${\mathcal H}$. 

The structure of the Clifford algebra ${\rm Cl}(V^*_{\Bbb C})$ is well known. Namely,
consider separately the cases when $\dim V$ is odd and even. 

In the even case, $\dim V=2d$, ${\rm Cl}(V^*_{\Bbb C})$ is simple (i.e., isomorphic to a matrix algebra), 
and has a unique irreducible representation
${\mathcal H}$, of dimension $2^d$. This representation is constructed as follows:
choose a decomposition $V_{\Bbb C}=L\oplus L^*$, where $L,L^*$
are Lagrangian subspaces; then ${\mathcal H}=\Lambda L$, where
$L\subset V^*_\Bbb C$ acts by multiplication and $L^*$ by differentiation
(multiplied by $-i$). 
The structure of the superspace on ${\mathcal H}$ is the standard
one on the exterior algebra. 

In the odd case, $\dim V=2d+1$, choose a decomposition 
$$
V_{\Bbb C}=
L\oplus L^*\oplus K,
$$ 
where $L,L^*$ are maximal isotropic,
and $K$ is a non-degenerate 
1-dimensional subspace orthogonal to $L$ and
$L^*$. Let ${\mathcal H}=\Lambda(L\oplus K)$, where 
$L,K$ act by multiplication and $L^*$ by ($-i$ times) differentiation.  
This is a representation of ${\rm Cl}(V^*_{\Bbb C})$ with a ${\Bbb Z/2}$ grading. 
This representation is not irreducible, 
and decomposes in a direct sum 
of two non-isomorphic irreducible representations
${\mathcal H}_+\oplus {\mathcal
H}_-$ (this is related to the fact 
that the  Clifford algebra
for odd $\dim V$ is not simple but is a direct product of two simple, i.e. matrix, 
algebras). However, this decomposition is not consistent with the
${\Bbb Z/2}$-grading, and therefore as superrepresentation, 
${\mathcal H}$ is irreducible. 

Now, it is easy 
to show that both in the odd and in the even case the space
${\mathcal H}$ carries a unique 
up to scaling Hermitian form, such that
$V^*\subset V^*_{\Bbb C}$ acts by self-adjoint operators.
This form is positive definite.  
So the situation is similar to the
bosonic case for any $\dim V$.

Let us now see
which operator on ${\mathcal H}$ should play the role of the Hamiltonian of the system. 
The most natural choice is to define 
the quantum Hamiltonian to be 
the obvious quantization of the classical Hamiltonian $H=\frac{1}{2}(\psi,M\psi)$. 
Namely, if $\varepsilon_i$ is an orthonormal basis of $V^*$ and $a_{ij}$ is the matrix
of $M$ in this basis, then one sets 
$$
\widehat{H}=\frac{1}{2}\sum_{i,j} a_{ij}\varepsilon_i\varepsilon_j.
$$ 
To compute this operator more explicitly,  
we will assume 
(without loss of generality) 
that the decomposition of $V_{\Bbb C}$ that we
chose is stable under $M$. Let 
$\xi_j$ be an eigenbasis of $M$ in $L$ with eigenvalues $im_j$ where $m_j\ge 0$, 
and $\partial_j$ be differentiations along the vectors of this basis.
Then 
$$
\widehat H=\sum_j m_j(\xi_j\partial_j-\partial_j\xi_j)=
\sum_j m_j(2\xi_j\partial_j-1).
$$

This shows that if $\dim V$ is even then the partition function on the circle 
of length $L$ for our theory is 
$$
Z={\rm sTr}(e^{-L\widehat H})=\prod_j (e^{m_jL}-e^{-m_jL}).
$$
If the dimension of $V$ is odd then the partition function is
zero. 

Now we would like to consider the fermionic analog of the
Feynman-Kac formula. For simplicity consider the fully massive
case, when $\dim V$ is even and $m_j\ne 0$ (i.e. $M$ is
non-degenerate). In this case, we have a unique up to scaling 
lowest eigenvector of $\widehat H$, namely $\Omega=1$. 

Let $\psi(0)\in V\otimes {\rm End}({\mathcal H})$ be the element 
corresponding to the action map $V^*\to {\rm End}({\mathcal H})$ (the {\it Clifford multiplication}\index{Clifford multiplication}), and
$\psi(t)=e^{it\widehat H}\psi(0)e^{-it\widehat H}$. Also,  
denote by $\langle\psi(t_1)...\psi(t_n)\rangle$, $t_1\ge...\ge t_n$, the
correlation function for the free theory in the Lagrangian setting,
taking values in $V^{\otimes n}$ (so in this expression 
$\psi(t_j)$ is a formal symbol and not an operator). 

\begin{theorem} (Feynman-Kac formula)
(i) For the free theory on the line we have 
$$
\langle\psi(t_1)...\psi(t_n)\rangle=
\la\Omega,\psi(t_1)...\psi(t_n)\Omega\ra.
$$ 

(ii) For the free theory on the circle of length $L$ we have 
$$
\langle\psi(t_1)...\psi(t_n)\rangle=
\frac{{\rm sTr}(\psi(t_1)...\psi(t_n)e^{-L\widehat H})}{{\rm sTr}(e^{-L\widehat H})}.
$$
\end{theorem}
 
\begin{exercise} Prove this theorem. (The proof is analogous to
Theorem \ref{FK} in the free case).
\end{exercise} 

It should now be straightforward 
for the reader to formulate and prove the Feynman-Kac formula 
for an interacting (i.e., not necessarily free) quantum-mechanical model which includes both bosonic and
fermionic massive fields. We leave this as an instructive
exercise. 

\begin{exercise} (i) Consider quantum mechanics with Yukawa coupling. 
That is, we have a scalar boson $\phi(t)$ and 
two fermions $\psi_1(t),\psi_2(t)$, and the Euclidean Lagragian is 
$$
{\mathcal L}=\frac{1}{2}(\dot{\phi}^2+m^2\phi^2+
\psi_1\dot{\psi_1}+\psi_2\dot{\psi_2}-\mu \psi_1\psi_2)+
g\phi\psi_1\psi_2. 
$$
Compute the 2-point function $\langle\phi(t)\phi(0)\rangle$ modulo $g^3$
(in the Euclidean setting). 

{\bf Hint.} The correction to the free theory answer is given by one Feynman 
diagram. Remember about automorphism groups and the minus sign corresponding 
to fermionic loops. 

(ii) In the same theory, compute the two-point function 
$\langle\psi_1(t)\psi_1(0)\rangle$ modulo $g^3$
(in the Euclidean setting). Does the corresponding diagram 
have non-trivial automorphisms?
\end{exercise} 

\section{Free field theories in higher dimensions}

\subsection{Minkowski and Euclidean space}
Now we pass from quantum mechanics 
to quantum field theory in dimensions $d\ge 1$. 
As we explained above, we have two main settings. 

{\bf 1. Minkowski space.} Fields are functions 
on a spacetime $V=V_M$, which is a real inner product space of 
signature $(1,d-1)$. This is where physical processes 
actually ``take place''. The symmetry group of $V$, 
$G=SO(1,d-1)$, is called the {\it Lorentz group}\index{Lorentz group}; 
it is the group of transformations of spacetime in special
relativity. Therefore, field theories in 
Minkowski space which are in an appropriate sense 
``compatible'' with the action of $G$ are called {\it relativistic}.\index{relativistic field theory}

Recall some standard facts and definitions. 
The {\it light cone}\index{light cone} in $V$ is the cone 
described by the equation $|\bold v|^2=0$, where 
$|\bold v|^2:=(\bold v,\bold v)$. 
Vectors belonging to the light cone are called {\it lightlike}.\index{lightlike vector} 
The light cone divides the space $V$ into  
{\it spacelike vectors} $|\bold v|^2<0$ (outside the cone),
 and {\it timelike vectors} $|\bold v|^2>0$ (inside the cone).\index{spacelike vector}\index{timelike vector}
We will choose one of the two components of 
the interior of the cone and call it positive; it will be denoted
by $V_+$. The opposite (negative) component is denoted by $V_-$. 
The group of $g\in SO(V)=SO(1,d-1)$ which preserve $V_+$ is
denoted by $SO_+(1,d-1)$; it is the connected component of the
identity of the group $SO(1,d-1)$ (which has two connected components).

Often (e.g. when doing Hamiltonian field theory) it is necessary to split 
$V$ in an orthogonal direct sum $V=V_s\oplus \Bbb R$ of space and time. 
In this decomposition, 
the space $V_s$ is required to be spacelike (i.e. negative
definite), which implies that the time axis ${\Bbb R}$ has to be timelike
(positive definite). Note that such a splitting is not unique, 
and that fixing it breaks the Lorentz symmetry 
$SO_+(1,d-1)$ down to the usual rotation group $SO(d-1)$.  

To do explicit calculations, one further chooses 
Cartesian coordinates $x_1,...,x_{d-1}$ on $V_s$ and 
$t$ on the time axis $\Bbb R$, so that $\bold v=(t,x_1,...,x_{d-1})$. In these coordinates
the inner product takes the form 
$$
|\bold v|^2=c^2t^2-\sum_{j=1}^{d-1}x_j^2
$$
where $c$ is the speed of light. This explains the origin of the term ``light cone" -- it consists 
of worldlines of free photons (particles of light) traveling in space in some direction at speed $c$. 
To simplify notation, we will chose units of measurement so that $c=1$. 

{\bf 2. Euclidean space.} Fields are functions on a spacetime $V_E$,
which is a positive definite inner product space.
It plays an auxiliary role and has no direct physical meaning, although path integrals 
computed in this space are similar to expectation values in
statistical mechanics. 

The two settings are related by the ``Wick rotation''.
Namely the Euclidean space $V_E$ corresponding to the Minkowski
space $V_M$ is 
the real subspace in $(V_M)_{\Bbb C}$ 
consisting of vectors $(it,x_1,...,x_{d-1})$, where $t$ and $x_j$
are real. In other words, to pass to the Euclidean
space, one needs to make a change of variable $t\mapsto it$. 
Note that under this change, the standard metric on the Minkowski space,
$dt^2-\sum_j dx_j^2$ goes into a negative definite metric 
$-dt^2-\sum_j dx_j^2$. However, the minus sign is traditionally 
dropped and one considers instead the positive metric 
$dt^2+\sum_j dx_j^2$ on $V_E$.  

\subsection{Free scalar boson}

Consider the theory of a free scalar bosonic field 
$\phi$ of mass $m$. The procedure of quantization of this theory 
in the Lagrangian setting is a straightforward generalization 
from the case of quantum mechanics. 
Namely, the Lagrangian for this theory in Minkowski space is
$$
{\mathcal L}=\frac{1}{2}((d\phi)^2-m^2\phi^2),
$$ 
and the Euler-Lagrange 
equation is the {\it Klein-Gordon equation}\index{Klein-Gordon equation}
$$
(\square+m^2)\phi=0,
$$ 
where $\square$ is the D'Alembertian (wave operator), 
$$
\square:=\frac{\partial^2}{\partial t^2}-
\sum_j \frac{\partial^2}{\partial x_j^2}.
$$ 
Thus to define the corresponding 
quantum theory, we should invert the operator $\square+m^2$. 
This operator is essentially self-adjoint on compactly supported smooth functions and thus defines 
a self-adjoint operator, but as in the quantum mechanics case, it is not invertible
-- its spectrum is the whole $\Bbb R$, as can be easily seen by taking the Fourier transform. 
So as before, it is best to proceed using the
Wick rotation. 

After the Wick rotation (i.e. the transformation $t\mapsto it$), we 
arrive at the Euclidean Lagrangian
$$
{\mathcal L}_E=\frac{1}{2}((d\phi)^2+m^2\phi^2),
$$ 
and the Euler-Lagrange equation is the Euclidean Klein-Gordon equation
$$
(-\Delta+m^2)\phi=0.
$$ 
So to define the quantum theory,
i.e. the path integral 
$$
\int \phi(x_1)...\phi(x_n)e^{-S(\phi)}D\phi
$$
where $S=\int {\mathcal L}$, we now need to invert the self-adjoint operator $A=-\Delta+m^2$ (initially defined as an essentially self-adjoint operator on smooth compactly supported functions), whose spectrum is $[m^2,\infty)$, so it is invertible when $m>0$. 
The operator $A^{-1}$ is an integral operator whose Schwartz kernel
is $G(x-y)$, where $G(x)$ is the Green's function, i.e. the fundamental
solution of the Klein-Gordon equation: 
$$
-\Delta G+m^2G=\delta.
$$

To solve this equation, note that the solution is rotationally
invariant. Therefore, 
outside of the origin, 
$G(x)=g(|x|)$, where $g$ is a function on $(0,\infty)$
such that 
$$
-g''-\frac{d-1}{r}g'+m^2g=0
$$
(where the left hand side is the radial part of the
operator $A$). This is a version of the Bessel equation. 
If $m>0$, the two basic solutions are $r^{\frac{2-d}{2}}
J_{\pm\frac{2-d}{2}}(imr)$, where $J$ is the Bessel function. (Actually, these functions are elementary
for odd $d$). Since we want $G$ to decay at infinity (clustering property), we should pick the unique up to scaling linear combination which decays at infinity, namely, 
\begin{equation}\label{bess}
g=Cr^{\frac{2-d}{2}}(J_{\frac{2-d}{2}}(imr)+i^d
J_{-\frac{2-d}{2}}(imr)),\ d\ne 2. 
\end{equation}
For $d=2$, 
this expression is zero, and one should 
instead take the limit of the right hand side 
divided by $d-2$ as $d\to 2$.
The normalizing constant can be found from the condition that
$AG=\delta$. 

\begin{remark} It is easy to check that for $d=1$ this function
equals the familiar Green's function for quantum mechanics, 
$\frac{e^{-mr}}{2m}$.
\end{remark} 

If $m=0$ (massless case), the basis of solutions is: $1,r$ for $d=1$, $1, \log r$
for $d=2$, and $1, r^{2-d}$ for $d>2$. Thus, if $d\le 2$, we
don't have a decaying solution and thus the corresponding quantum theory
will be deficient: it will not satisfy the clustering property. 
On the other hand, for $d>2$ we have a unique up to scaling 
decaying solution $g=Cr^{1-d}$. The normalizing constant is found as in
the massive case. 

The higher correlation functions are found from the 2-point
function via the Wick formula, as usual. 

We should now note a fundamental difference between quantum
mechanics and quantum field theory in $d>1$ dimensions. 
This difference comes from the fact that while for $d=1$, the Green's
function $G(x)$ is continuous at $x=0$, for $d>1$ it is
singular at $x=0$. Namely, $G(x)$ behaves like $C|x|^{2-d}$ 
as $x\to 0$ for $d>2$, and as $C\log|x|$ as $d=2$. Thus for $d>1$,
unlike the case $d=1$, the path integral 
$$
\int \phi(x_1)...\phi(x_n)e^{-S(\phi)}D\phi
$$
(as defined above) makes sense only if $x_i\ne x_j$. 
In other words, this path integral should be regarded not as a
function but rather as a distribution. Luckily, there is a
canonical way to do it, since the Green's function $G(x)$ is
locally $L^1$. 

Now we can Wick rotate this theory back
into the Minkowski space. It is clear that the Green's function 
will then turn into 
$$
G_M(x)=g(\sqrt{-|x|^2-i\varepsilon}),
$$
which involves Bessel functions of both real and imaginary argument (depending on whether 
$x$ is timelike or spacelike) and has a singularity 
on the light cone $|x|^2=0$. In particular, it is easy to check that 
$G_M(x)$ is real-valued for spacelike $x$, while for 
timelike $x$ it is not. The function $G_M(x)$ satisfies the equation 
$$
(\square+m^2)G_M=i\delta.
$$
The higher correlation functions, as before, are determined 
from this by the Wick formula.  

Actually, it is more convenient to describe this theory 
``in momentum space'', where the Green's function can be written 
more explicitly. 
Namely, the Fourier transform $\widehat G(p)$ of the
distribution $G(x)$ is a solution of the equation 
$$
p^2\widehat G+m^2\widehat G=1,
$$
obtained by Fourier transforming the differential equation for
$G$. Thus, 
$$
\widehat G(p)=\frac{1}{p^2+m^2},
$$
as in the quantum mechanics case. Therefore, like in quantum
mechanics, the Wick rotation 
produces the distribution 
$$
\widehat G_M(p)=\frac{i}{p^2-m^2+i\varepsilon},
$$
which is the Fourier transform of $G_M(x)$. 

\subsection{Spinors}\label{spinorssec1}

To consider field theory for fermions, we must generalize to the
case of $d>1$ the basic fermionic Lagrangian
$\frac{1}{2}\psi\frac{d\psi}{dt}$. To do this, we must replace 
$\frac{d}{dt}$ by some differential operator on $V$.
This operator should be of first order, 
since in fermionic quantum mechanics it was important that the
equations of motion are first order equations. 
Clearly, it is impossible to define such an operator 
so that the Lagrangian is $SO_+(V)$-invariant, if $\psi$ is
a scalar-valued (odd) function on $V$. 
Thus, a fermionic field in field theory of dimension
$d>1$ cannot be scalar-valued, but rather must take values 
in a real representation $S$ of $SO_+(V)$, 
such that there exists a nonzero intertwining operator
$V\to {\rm Sym}^2S^*$. This property is satisfied by {\it spinor
representations}.\index{spinor representation} They are indeed basic in fermionic field theory,
and we will now briefly discuss them (for 
more detail see ``Spinors'' by P.Deligne, 
in ``QFT and string theory: a course for mathematicians''). 

First consider the complex case. Let $V$ be a complex inner
product space of dimension $d>1$. Let 
${\rm Cl}(V)$ be the Clifford algebra of $V$, defined by the relation 
$\xi\eta+\eta\xi=2(\xi,\eta)$, $\xi,\eta\in V$.
As we discussed, for even $d$ it is simple and has a unique 
irreducible representation $S$ of dimension $2^{\frac{d}{2}}$, 
while for odd $d$ it has two
such representations $S',S''$ of dimension $2^{\frac{d-1}{2}}$. 
It is easy to show that 
the space ${\rm Cl}_2(V)$ of quadratic elements of ${\rm Cl}(V)$ 
(i.e. the subspace spanned 
elements of the form $\xi\eta-\eta\xi,\ \xi,\eta\in V$) is closed 
under bracket, and constitutes the Lie algebra ${\frak o}(V)$. 
Thus ${\frak o}(V)$ acts on $S$ (respectively, $S',S''$). 
This action does not integrate to an action of $SO(V)$, 
but integrates to an action of its double cover ${\rm Spin}(V)$.

If $d$ is even, the representation $S$ of ${\rm Spin}(V)$ is
not irreducible. Namely, recall that $S$ is the exterior algebra
of a Lagrangian subspace of $V$. Thus it splits in a direct sum
$S=S_+\oplus S_-$ (odd and even elements). The subspaces
$S_+,S_-$ are subrepresentations of $S$, which are irreducible. 
They are called the {\it half-spinor representations}.\index{half-spinor representation} The half-spinor representations are interchanged by 
the adjoint action of $O(V)$ on ${\rm Spin}(V)$ ($SO(V)$ clearly acts trivially, so this is, in fact, and action 
of $O(V)/SO(V)=\Bbb Z/2$ on the set of irreducible representations of $SO(V)$). Note that in contrast, 
for odd $d$ we have $O(V)=SO(V)\times \Bbb Z/2$, so the $\Bbb Z/2$ acts
on representations of ${\rm Spin}(V)$ trivially.

If $d$ is odd, the representations $S'$ and $S''$ of ${\rm
Spin}(V)$ are irreducible and isomorphic. Any of them will be
denoted by $S$ and called the {\it spinor representation}.\index{spinor representation}
Thus, we have the spinor representation $S$ for both odd and even
$d$, but for even $d$ it is reducible.

An important structure attached to the spinor representation 
$S$ is the intertwining operator $\Gamma: V\to {\rm End}S$ called {\it Clifford multiplication},\index{Clifford multiplication}
given by the action of $V\subset {\rm Cl}(V)$ in $S$, which we already encountered above. This intertwiner
allows us to define the {\it Dirac operator}\index{Dirac operator}
\begin{equation}\label{dirop}
\bold{D}=\sum_i \Gamma_i \frac{\partial}{\partial x_i}
\end{equation}
where $x_i$ are coordinates on $V$ associated to an orthornormal basis $e_i$,
and $\Gamma_i=\Gamma(e_i)$. This operator acts on 
functions from $V$ to $S$, and $\bold D^2=\Delta$, so 
$\bold D$ is a square root of the Laplacian. 
The matrices $\Gamma_i$ are called {\it $\Gamma$-matrices}.\index{$\Gamma$-matrices}

Note that for even $d$, one has 
$\Gamma(v): S_\pm\to S_\mp$, so $\bold D$ acts from
functions with values in $S_\pm$ to functions with values in
$S_\mp$. 

By a {\it polyspinor representation}\index{polyspinor representation} of ${\rm Spin}(V)$ we will mean 
any linear combination of $S_+,S_-$ for even $d$, and any
multiple of $S$ for odd $d$. For even $d$ and a polyspinor representation 
$Y=Y_+\otimes S_+\oplus Y_-\otimes S_-$ (i.e., $Y_{\pm}={\rm Hom}(S_\pm,Y)$)
where $Y_+,Y_-$ are vector spaces, set 
$Y':=Y_+\otimes S_-\oplus Y_-\otimes S_+$, while for odd $d$ and $Y=Y_0\otimes S$ 
we set $Y':=Y$; thus $Y\mapsto Y'$ is an endofunctor on the category of polyspinor representations.  
Then for every polyspinor
representation $Y$ and $v\in V$ we have the Clifford multiplication operator 
$\Gamma(v): Y\to Y'$. 

Now assume that $V$ is a real inner product space 
with Minkowski metric. In this case we can define the group ${\rm
Spin}_+(V)$ to be the preimage 
of $SO_+(V)$ under the map ${\rm Spin}(V_{\Bbb C})\to SO(V_{\Bbb
C})$. It is a double cover of $SO_+(V)$ (if $d=2$, this double cover
is disconnenced and actually a direct product by $\Bbb Z/2$). 

By a {\it real polyspinor representation}\index{real polyspinor representation} of ${\rm Spin}_+(V)$ we will mean 
a real representation $Y$ of this group such that $Y_{\Bbb C}$ is a polyspinor representation of ${\rm Spin}(V_{\Bbb C})$. 

\begin{remark}\label{anyons} 
Note that in all dimensions except $d=2$, the group ${\rm Spin}(d)$ 
is the universal cover of $SO(d)$, which means that spins of all particles 
are either integers or half-integers. On the other hand, the 
universal cover of $SO(2)$ is not ${\rm Spin}(2)$, but rather 
$\Bbb R$. This creates in two dimensions a possibility of particles 
whose spin is any positive real number. Such particles are called 
{\it anyons}\index{anyon} (particles of any spin), 
and we will see how they appear in 2-dimensional conformal field theory. 
\end{remark} 

\subsection{Fermionic Lagrangians}\label{ferlag} 

Now let us consider Lagrangians for a spinor field $\psi$ with values in
a polyspinor representation $Y$. 
Note that in even dimensions such fields 
are split into fields valued in $S_+$ and $S_-$, respectively. 
Such spinors are called {\it chiral}.\index{chiral spinor} 

As the Lagrangian is supposed to be
real in the Minkowski setting, we will require in that case that
$Y$ be real. First of all, let us see what we need in order to write 
the ``kinetic term'' $(\psi,\bold D\psi)$. 
Clearly, to define such a term 
(so that the corresponding term in the action does not reduce to zero via integration by parts), we need an invariant non-degenerate pairing $(,)$
between $Y$ and $Y'$ (i.e., an isomorphism of representations $Y'\cong Y^*$) such that for any $v\in V$, the bilinear form
$(x,\Gamma(v)y)$ on $Y$ is symmetric. 

Let us find for which $Y$ this is possible (for complex $V$). The behavior of Spin
groups depends on $d$ modulo $8$ ({\it real Bott periodicity}\index{real Bott periodicity}). Thus we will list the answers 
labeling them by $d$ mod $8$ (they are easily extracted from the
tables given in Deligne's text). First we summarize properties of 
spin representations. 

\vskip .05in

0. $S_\pm$ orthogonal.

1. $S$ orthogonal, $S\otimes S\to V$ symmetric.

2. $S_+^*=S_-$, $S_\pm\otimes S_\pm\to V$ symmetric.

3. $S$ symplectic, $S\otimes S\to V$ symmetric. 

4. $S_\pm$ symplectic. 

5. $S$ symplectic, $S\otimes S\to V$ antisymmetric. 

6. $S_+=S_-^*$, $S_\pm\otimes S_\pm \to V$ antisymmetric.

7. $S$ orthogonal, $S\otimes S\to V$ antisymmetric. 

\vskip .05in 
Thus the possibilities for the kinetic term are: 
\vskip .05in 

0. $n(S_+\oplus S_-)$; (,) gives a perfect pairing between $Y_+$ and $Y_-$.

1. $nS$; (,) gives a symmetric inner product on $Y_0$.

2. $nS_+\oplus kS_-$; (,) gives symmetric inner products on $Y_\pm$. 

3. $nS$; (,) gives a symmetric inner product on $Y_0$.

4. $n(S_+\oplus S_-)$; (,) gives a perfect pairing between $Y_+$ and $Y_-$.

5. $2nS$; (,) gives a skew-symmetric inner product on $Y_0$.

6. $2nS_+\oplus 2kS_-$;  (,) gives skew-symmetric inner products on $Y_\pm$.

7. $2nS$; (,) gives a skew-symmetric inner product on $Y_0$.

\vskip .05in

Let us now find when we can also add a mass term. 
Recall that the mass term has the form $(\psi,M\psi)$, 
so it corresponds to an invariant skew-symmetric operator 
$M:Y\to Y^*\cong Y'$ (note that by definition, $\Gamma_i$ commute with $M$). Let us list those $Y$ from the above list 
for which such a non-degenerate operator exists. 

\vskip .05in

0. $2n(S_+\oplus S_-)$; $M_\pm: Y_\pm\to Y_\mp$ are skew-symmetric under $(,)$. 

1. $2nS$; $M: Y_0\to Y_0$ is skew-symmetric under $(,)$.

2. $n(S_+\oplus S_-)$; $M_\pm: Y_\pm\cong Y_\mp$ satisfy $M_+^*=-M_-$ under $(,)$.

3. $nS$; $M: Y_0\to Y_0$ is symmetric under $(,)$.

4. $n(S_+\oplus S_-)$; $M_\pm: Y_\pm\to Y_\mp$ are symmetric under $(,)$.

5. $2nS$; $M: Y_0\to Y_0$ is symmetric under $(,)$.

6. $2n(S_+\oplus S_-)$; $M_\pm: Y_\pm\cong Y_\mp$ satisfy $M_+^*=-M_-$ under $(,)$.

7. $2nS$; $M: Y_0\to Y_0$ is skew-symmetric under $(,)$.

\vskip .05in

To pass to the real Minkowski space (in both massless and massive
case), one should put the additional requirement that $Y$ should be a real
representation.

We note that upon Wick rotation to Minkowski space, 
it may turn out that a real spinor representation $Y$ will turn into a
complex representation which has no real structure. Namely, this happens 
for massless spinors that take values in $S_\pm$ if $d=2$ mod 8. 
These representations have a real structure for Minkowskian $V$ 
(i.e. for ${\rm Spin}_+(1,d-1)$), but no real structure 
for Euclidean $V$ (i.e. for ${\rm Spin}(d)$). 
This is quite obvious, for example, when $d=2$ (check!). 
 
\begin{remark}\label{realcom} 
One may think that this causes a problem in quantum field theory,
where we would be puzzled what to integrate over -- real or
complex space. However, the problem in fact does not arise, since
we have to integrate over fermions, and integration over fermions
 (say, in the finite dimensional case)
is purely algebraic and does not make a distinction between real
and complex. 
\end{remark} 

\subsection{Free fermions}

Let us now consider a free theory for a spinor field $\psi: V\to \Pi
Y$, where $Y$ is a polyspinor representation, defined by a Lagrangian 
$$
{\mathcal L}=\frac{1}{2}(\psi,(\bold D-M)\psi),
$$ 
where $M$ is allowed to be degenerate (we assume that $Y$ is such
that this expression makes sense). The equation of motion in
Minkowski space is 
$$
\bold D\psi=M\psi.
$$ 
Thus, to define the corresponding
quantum theory, we need to invert the operator $\bold D-M$. 
As usual, this cannot be done because of a singularity, 
and it is best to use the Wick rotation. 

The Wick rotation produces the Euclidean Lagrangian
$$
{\mathcal L}=\frac{1}{2}(\psi,(\bold D_E+M)\psi)
$$
(note that the $i$ in the kinetic term is hidden in the definition of the
Euclidean Dirac operator). 
We invert $\bold D_E+M$ to obtain the Euclidean Green's
function. 
To do this, it is convenient to go to momentum space,
i.e. perform a Fourier transform. Namely, 
after Fourier transform $\bold D_E$ turns into the operator $i\bold
p$, where $\bold p=\sum_j p_j\Gamma_j$, and $p_j$ are the operators 
of multiplication by the momentum coordinates $p_j$. Thus, the Green's
function (i.e. the 2-point function) 
$G(x)\in {\rm Hom}(Y^*,Y)$ is the Fourier transform of
the matrix-valued function 
$\frac{1}{i\bold p+M}$. 

In the Euclidean case the group ${\rm Spin}(V)$ is compact 
and the spinor representations carry natural positive invariant Hermitian forms. 
So in this case without loss of generality we may consider polyspinor representations equipped with such positive forms, and on every polyspinor representation such a form is unique up to isomorphism. Let $$
M^\dagger: Y^*\to Y
$$ 
be the Hermitian adjoint operator 
to $M$. Then the reality condition is that $M$ is Hermitian: $M^\dagger=M$. 
Thus
$$
(-i\bold p+M)(i\bold p+M)=p^2+M^2
$$ 
so that 
$$
\widehat G(p)=(p^2+M^2)^{-1}(-i\bold p+M).
$$
This shows that $G(x)$ is expressed through the Green's function
in the bosonic case by differentiations (how?).   
After Wick rotation back to the Minkowski space, we get 
$$
\widehat G_M(p)=(p^2-M^2+i\varepsilon)^{-1}(\bold
p+iM).
$$
Finally, the higher
correlation functions, as usual, are found from the Wick formula. 

\subsection{Hamiltonian formalism of classical field theory} \label{hfcft}

Let us now develop the hamiltonian approach to QFT, extending the hamiltonian formalism of quantum mechanics. We start with classical field theory, extending the hamiltonian formalism of classical mechanics. 
As in the Lagrangian setting, this can be done by formalizing the idea that field theory is mechanics of a continuum of particles occupying each point of the space $\Bbb R^{d-1}$. 

Namely, consider a free scalar bosonic field $\phi(x)$ on a Minkowski space $\Bbb R^d$. 
As we have discussed, its Lagrangian is 
$\mathcal L=\frac{1}{2}((d\phi)^2-m^2\phi^2)$ 
and the equation of motion is the Klein-Gordon equation 
$$
\phi_{tt}-\Delta_s\phi+m^2\phi=0,
$$ 
where $\Delta_s$ is the spacial Laplacian. This is a second order equation with respect to $t$, 
so the initial value problem for this equation has the form 
$$
\phi(0,x)=q(x),\ \phi_t(0,x)=p(x)
$$
(there is a standard explicit formula for solution of this problem, expressing it via the fundamental solution of 
the Klein-Gordon equation). Thus it is natural to introduce the phase space 
$$
Y:=T^*C^\infty_0(\Bbb R^{d-1}):=C^\infty_0(\Bbb R^{d-1})\oplus C^\infty_0(\Bbb R^{d-1})
$$
of pairs $(q,p)$ of smooth functions with compact support, on which the dynamics 
of the Klein-Gordon equation takes place (note that the space $C^\infty_0(\Bbb R^{d-1})$ is invariant under this dynamics since the speed of wave propagation is finite, namely equals $1$). Note that the phase space is an infinite dimensional symplectic space with constant symplectic form 
$$
\omega((q_1,p_1),(q_2,p_2))=\int_{\Bbb R^{d-1}}(p_1(x)q_2(x)-p_2(x)q_1(x))dx.
$$ 
Also for any point $x\in \Bbb R^{d-1}$ we have the local linear functionals 
$$
(q,p)\mapsto q(x),\ (q,p)\mapsto p(x)
$$ 
which we will denote by $\phi(x)$ and $\phi_t(x)$, respectively. 
From these functionals we can make other linear functionals: for example, given 
$\rho\in C^\infty_0(\Bbb R^{d-1})$, we can define the functionals 
$$
\phi(\rho)(q,p):=\int_{\Bbb R^{d-1}}q(x)\rho(x)dx,\ \phi_t(\rho)(q,p):=\int_{\Bbb R^{d-1}}p(x)\rho(x)dx.
$$
The Poisson bracket between such functionals can be computed 
by the formulas 
$$
\lbrace \phi(\rho_1),\phi(\rho_2)\rbrace=0,\ 
\lbrace \phi_t(\rho_1),\phi_t(\rho_2)\rbrace=0,
$$
$$
\lbrace \phi(\rho_1),\phi_t(\rho_2)\rbrace=\int_{\Bbb R^{d-1}}\rho_1(x)\rho_2(x)dx.
$$
This can be written as a {\it field-theoretic Poisson bracket}:\index{field-theoretic Poisson bracket}
$$
\lbrace \phi(x),\phi(y)\rbrace=0,\ 
\lbrace \phi_t(x),\phi_t(y)\rbrace=0,\
\lbrace \phi(x),\phi_t(y)\rbrace=\delta(x-y);
$$
then the previous formulas can be recovered by integrating both sides against $\rho_1(x)\rho_2(y)$. 
In other words, the linear local functionals $\phi(x)$ and $\phi_t(x)$ should be thought of not as smooth functions on $Y$ depending on a point $x\in \Bbb R^{d-1}$ but rather as distributions 
on $\Bbb R^{d-1}$ with values in smooth functions on $Y$. 

Similarly, one may consider non-linear polynomial local functionals, given by differential polynomials
$P(\phi,\phi_t)$ evaluated at a point $x$, such as $\phi^n,\phi_t^2,(d_s\phi)^2, \phi^2 (d_s\phi)^2$ (where $d_s$ is the spatial differential), etc., and even non-polynomial ones depending on finitely many derivatives of $\phi$, 
such as $e^\phi (d_s\phi)^2, \cos \phi$, and so on. 
They are called local because they depend only on the derivatives of 
$\phi$ at a single point $x$. Each of them is a distribution on $\Bbb R^{d-1}$ with values in smooth functions on $Y$, and can be applied to any density $\rho(x)$ to produce a smooth function on $Y$. 
Poisson brackets of such functionals are computed using the chain rule, the Leibniz rule, and 
the fact that taking Poisson brackets commutes with differentiation by $x$.
So given two local functionals $P$ and $Q$, we obtain 
$$
\lbrace P(\phi)(x),Q(\phi)(y)\rbrace=\sum_\alpha \lbrace P,Q\rbrace_\alpha(\phi)(x)\partial^\alpha_x\delta(x-y)
$$
for some local functionals $\lbrace P,Q\rbrace_\alpha$, where $\partial^\alpha$ are monomials in the derivatives. For example, for $d=2$
$$
\lbrace \phi_{tx}(u)\phi_t(u), \tfrac{1}{3}\phi^3(v)\rbrace=
$$
$$
-(\phi_{tx}(u)\phi^2(u)+2\phi_{tx}(u)\phi(u)\phi_x(u))\delta(u-v)-\phi_{t}(u)\phi^2(u)\delta'(u-v).
$$
This Poisson bracket can of course be extended to products of local functionals at different points using the Leibniz rule. 

The Hamiltonian of the theory is then given by integrating a local functional against the constant density: 
$$
H(\phi)=\frac{1}{2}\int_{\Bbb R^{d-1}}(\phi_t^2+(d_s\phi)^2+m^2\phi^2)dx. 
$$ 
Namely, it is determined (up to a constant) by the condition that the Hamilton equation 
$$
F_t=\lbrace F,H\rbrace
$$ 
for local functionals $\phi,\phi_t$ is equivalent to the Klein-Gordon equation. 

The Hamiltonian dynamics allows us to define the local functionals 
not just at a point $x\in \Bbb R^{d-1}$ but actually at any point $(t,x)\in \Bbb R^d$. 
When we do, by definition we get $\phi_t(t,x)=\frac{d}{dt}\phi(t,x)$ and 
the local functional $\phi(t,x)$ becomes a solution 
of the Klein-Gordon equation: 
$$
\phi_{tt}-\Delta_s\phi+m^2\phi=0.
$$
This can be used to compute the Poisson brackets: for example, we see that 
$$
\lbrace\phi(t_1,x_1),\phi(t_2,x_2)\rbrace=\bold G(t_2-t_1,x_2-x_1)
$$
where $\bold G(t,x)$ solves the Klein-Gordon equation with initial conditions 
$$
\bold G(0,x)=0,\ \bold G_t(0,x)=\delta(x). 
$$
To find it, take the Fourier transform. 
Then we get a distribution $\widehat {\bold G}$ supported on the two-sheeted hyperboloid 
$X_m$ given by the equation $E^2=p^2+m^2$, of the form 
$$
\widehat {\bold G}(E,p)=f_+(p)\delta_{X_m^+}+f_-(p)\delta_{X_m^-},
$$
where $X_m^\pm$ are the sheets of $X_m$. 
Moreover, the initial conditions give (up to appropriate normalization)
$$
\int_{\Bbb R} \widehat {\bold G}(E,p)dE=0,\ \int_{\Bbb R} \widehat {\bold G}(E,p)EdE=1,
$$
which yields 
$$
f_+(p)+f_-(p)=0,\ \sqrt{p^2+m^2}(f_+(p)-f_-(p))=1.
$$
Thus $f_+=-f_-=\frac{1}{2\sqrt{p^2+m^2}}$ and we have 
$$
\widehat {\bold G}(E,p)=\frac{1}{2\sqrt{p^2+m^2}}(\delta_{X_m^+}-\delta_{X_m^-}).
$$
Now $\bold G$ can be found by taking the inverse Fourier transform
(it expresses via the Bessel functions). 

Note that since the speed of wave propagation is $1$, this distribution $\bold G$ is supported on the solid light cone, 
so $\lbrace\phi(t_1,x_1),\phi(t_2,x_2)\rbrace=0$ if the points $(t_1,x_1)$ and $(t_2,x_2)$ 
are {\it spacelike separated},\index{spacelike separated points} meaning that the vector $(t_1-t_2,x_1-x_2)$ is spacelike. 
This property is called {\it space locality},\index{space locality} a mathematical expression of {\it causality}\index{causality} in special relativity. 

\begin{remark} A part of this analysis extends straightforwardly to the case of non-free theories, 
for example the $\phi^4$-{\it theory},\index{$\phi^4$-theory} having the Lagrangian 
$$
\mathcal L=\frac{1}{2}((d\phi)^2-m^2\phi^2)-\frac{g}{4}\phi^4.
$$
In this case the Klein-Gordon equation is replaced by its non-linear deformation
$$
\phi_{tt}-\Delta_s\phi+m^2\phi+g\phi^3=0,
$$
so there is a nontrivial issue of existence of solutions of the initial value problem 
for this non-linear PDE. However, this issue is irrelevant 
if we just want to consider Poisson brackets of local functionals 
on $\Bbb R^{d-1}$ or its formal neighborhood, since 
then the computations are purely formal (algebraic). 
\end{remark} 

An important fact is that this structure is invariant under the {\it Poincar\'e group}\index{Poincar\'e group}
$\bold P:=SO_+(V)\ltimes V$ generated by Minkowski rotations and translations,
where $V=\Bbb R^d$ is the spacetime (the semidirect product of the {\it Lorentz group}\index{Lorentz group} $SO_+(V)$ and the group of translations $V$). This follows from the fact that 
the Lagrangian of the theory is relativistically invariant. 
Namely, for $g\in \bold P$ given by 
$$
g(t,x)=(at+bx+c,\alpha t+\beta x+\gamma)
$$
we have 
$$
(\phi g)(x)(q,p)=\phi(bx+c,\beta x+\gamma)
$$
and 
$$
(\phi_tg)(x)(q,p)=(\partial_\alpha \phi)(bx+c,\beta x+\gamma)+a\phi_t (bx+c,\beta x+\gamma).
$$
where $\phi(t,x)$ is the solution of the Klein-Gordon equation with 
initial conditions $(q(x),p(x))$. 

In particular, note that the {\it Galileo subgroup}\index{Galileo group} $SO(\Bbb R^{d-1})\ltimes \Bbb R^{d-1}$ acts 
by manifest geometric symmetries, while time translations act by the Hamiltonian flow. 

Finally, note that this discussion extends in a straighforward way to theories including fermions. 
In this case, as in fermionic classical mechanics, we get a field theoretic {\it super}-Poisson bracket on classical fields,\index{super-Poisson bracket} which is symmetric rather than skew-symmetric if both fields are odd. Also,  since odd fields take values in polyspinor representations, the Poincar\'e group should be replaced by its double cover 
$\widetilde{\bold P}:={\rm Spin}_+(V)\ltimes V$. We leave the details to the reader. 

\subsection{Hamiltonian formalism of QFT: the Wightman axioms}

To quantize this picture, we need to define a Hilbert space $\mathcal H$ 
and lift classical observables (local functionals and their integrals) to (densely defined) operators 
on $\mathcal H$, notably lift the classical hamiltonian $H$ to a quantum hamiltonian 
$\widehat H$ depending on the Planck constant $\hbar$ which should be a self-adjoint (in general, unbounded) operator on $\mathcal H$. Moreover, this should be done in such a way that commutators 
vanish at $\hbar=0$ and in first order in $\hbar$ recover the Poisson bracket. We should also 
have a unitary representation of the double cover $\widetilde{\bold P}$ 
of the Poincar\'e group on the space $\mathcal H$ 
such that the 1-parameter subgroup of time translations acts by the quantum dynamics 
1-parameter group $e^{-it\widehat H}$. This generalization of Hamiltonian quantum mechanics 
can be accomplished by means of so called {\it Wightman axioms},\index{Wightman axioms} which we now describe. 

First of all, for the quantum theory to have good properties, we want the energy to be bounded below. Thus we introduce the following definition. Let us fix an orthogonal decomposition $V= \Bbb R\oplus V_s$ 
into space and time and consider the self-adjoint operator 
$$
\widehat H_\pi:=i\tfrac{d}{dt}|_{t=0}\pi(t,0).
$$

\begin{definition} A unitary representation $\pi: \widetilde {\bold P}\to {\rm Aut}\mathcal H$ 
is said to be {\it positive energy}\index{positive energy representation} if the spectrum of $\widehat H_\pi$ is bounded below. 
\end{definition} 

Note that every unitary representation $\pi$ of $V$ has a spectrum $\sigma(\pi)$, which is a closed subset of $V^*\cong V$; namely, $\sigma(\pi)$ is the set of characters of $V$ that occur (discretely or continuously) in $\pi$ (i.e., the smallest set containing the support of the Fourier transform of the distribution $\la w_1,\pi(v)w_2\ra$, $v\in V$, for any 
$w_1,w_2\in \mathcal H$). 

\begin{lemma} Suppose $\dim V\ge 2$. Then $\pi$ is positive energy if and only if $\sigma(\pi)$ is contained in the positive part of the solid light cone, 
$\overline V_+$. 
\end{lemma} 

\begin{proof} By definition, $\pi$ is of positive energy iff the orthogonal projection of $\sigma(\pi)$ onto the dual of the time axis is bounded below. Since $\sigma(\pi)$ is invariant under $SO_+(V)$, this implies the statement (an 
$SO_+(V)$-orbit on $V$ has bounded below projection iff it is contained in $\overline V_+$).
\end{proof} 

Note that this is false for $d=1$ (quantum mechanics), where the hamiltonian can be shifted by a constant without any effect on the theory. But the latter is longer so in quantum field theory on a Minkowski space of dimension $>1$. 

We are now ready to give Wightman's definition of a QFT. Let $\mathcal S=\mathcal S(V)$ 
be the Schwartz space of $V$. 

\begin{definition} A {\it Wightman QFT}\index{Wightman QFT} on a Minkowski space $V$ entails the following data: 

1. A finite dimensional real super-representation $R=R_0\oplus R_1$ of ${\rm Spin}_+(V)$ (the {\it field space}).\index{field space}

2. A super Hilbert space $\mathcal H=\mathcal H_0\oplus \mathcal H_1$ carrying a positive energy 
unitary representation $\pi: \widetilde{\bold P}\to {\rm Aut}\mathcal H$ of the double cover of the Poincar\'e group, $\widetilde{\bold P}={\rm Spin}_+(V)\ltimes V$. 

3. A dense $\widetilde{\bold P}$-stable subspace $\mathcal D\subset \mathcal H$.  

4. A $\widetilde{\bold P}$-invariant unit vector $\Omega\in \mathcal D$ called the {\it vacuum vector}.\index{vacuum vector}

5. A $\widetilde{\bold P}$-invariant even linear map: $\mathcal S\otimes R^*\to {\rm End}\mathcal D$ 
called the {\it field map}.\index{field map}

This data is subject to the following axioms.

{\bf A1.} If $f$ is real then $\phi(f)$ is Hermitian symmetric (in the supersense). 

{\bf A2.} $\phi$ is weakly continuous, i.e. for every $w_1,w_2\in \mathcal D$, 
the functional $\mathcal S\otimes R^*\to \Bbb C$ defined by $f\mapsto \la w_1,\phi(f)w_2\ra$ is continuous. 

{\bf A3.} $\mathcal D$ is spanned (algebraically) by vectors $\phi(f_1)...\phi(f_n)\Omega$. 

{\bf A4.} {\it Space locality}:\index{space locality} If $f_1,f_2$ have spacelike separated supports, i.e., for any $v_1\in {\rm supp}f_1$, $v_2\in {\rm supp}f_2$ we have $|v_1-v_2|^2<0$, then 
$$
[\phi(f_1),\phi(f_2)]=0
$$ 
(with commutator understood in the supersense). 
\end{definition} 

In addition, if $\mathcal H^{\widetilde{\bold P}}=\Bbb C\Omega$, one says that 
we have a Wightman QFT with a {\it unique vacuum}. 

We will also always assume that our QFT is {\it nondegenerate}, i.e., 
for every irreducible subrepresentation $E\subset R^*_j$, $j=0,1$, one has
$\phi|_{\mathcal S\otimes E}\ne 0$; otherwise we can simply remove this subrepresentation 
without any effect on the theory.

A fundamental fact about Wightman QFT is the following theorem, which we will not prove here. Let $\zeta$ be the 
generator of the kernel of the map ${\rm Spin}_+(V)\to {\rm SO}_+(V)$, so $\zeta^2=1$. 

\begin{theorem} (The spin-statistics theorem) If $E\subset R_j^*$ is a subrepresentation then $\zeta|_E=(-1)^j$. 
\end{theorem} 

In other words, there is a relationship between the {\it spin}\index{spin of a quantum field} (mod integers) of a quantum field (essentially, the eigenvalue of $\zeta$) and its {\it statistics}\index{statistics of a quantum field}, i.e., whether it is bosonic (even) or fermionic (odd). Namely, the theorem says that all bosonic fields must have $\zeta=1$ (integer spin) and all fermionic fields must have $\zeta=-1$ (half-integer spin). 

\begin{remark} We will see that the theory of free bosons and fermions can be naturally formulated as a Wightman QFT. Moreover, this is also the case for a number of non-free theories, which is the subject of 
a difficult area of mathematical physics called {\it constructive field theory}. Still, most theories 
that physicists really care about are either not known to be Wightman QFT, or simply fail to be ones for various 
reasons (perturbative theories, low energy effective theories, non-unitary theories, Euclidean theories, theories living on compact manifolds, etc.) Thus we will view Wightman axioms just as one (somewhat limited) rigorous model for our mathematical understanding of QFT. 
\end{remark} 

\subsection{Wightman functions} 

\begin{proposition} 
In a Wightman QFT on a Minkowski space $V$, for every $n\ge 1$ there exists a unique tempered distribution $W_n$ 
on $V^n$ valued in $R^{*\otimes n}$ such that 
$$
W_n(f_1\boxtimes...\boxtimes f_n)=\la \Omega,\phi(f_1)...\phi(f_n)\Omega\ra.
$$
\end{proposition} 

We leave the proof of this proposition as an exercise. 

We will therefore think of $W_n$ as a (generalized) function on $V^{\otimes n}$ valued in $R^{*\otimes n}$, denoted 
$W_n(x_1,...,x_n)$, so that 
$$
W_n(f_1\boxtimes...\boxtimes f_n)=\int_{V^n}W_n(x_1,...,x_n)f_1(x_1)...f_n(x_n)dx_1...dx_n
$$
where the product on the right hand side involves contraction of corresponding copies of $R$ and $R^*$.
Thus, given $u_1,...,u_n\in R$, we have the scalar-valued distribution
$$
W_n^{u_1,...,u_n}(x_1,...,x_n):=(W_n(x_1,...,x_n),u_1\otimes...\otimes u_n).
$$
 
In other words, we may define an operator-valued distribution $\phi(x)$ 
such that
$$
\phi(f)=\int_V \phi(x)f(x)dx;
$$
then 
$$
W_n(x_1,...,x_n)=\la \Omega,\phi(x_1)...\phi(x_n)\Omega\ra.
$$

\begin{definition}  The generalized functions $W_n(x_1,...,x_n)$ 
are called the {\it Wightman (correlation) functions}\index{Wightman functions} of the Wightman QFT.  
\end{definition}

Note that Wightman functions completely determine the Wightman QFT as follows. Let $\widetilde{\mathcal D}:=T(\mathcal S\otimes R)$ (the tensor algebra), so it is spanned by elements 
$f_1\otimes f_2\otimes...\otimes f_n$, $f_i\in \mathcal S\otimes R$. 
Define the inner product on 
$\widetilde{\mathcal D}$ by 
$$
\la f_1\otimes...\otimes f_n, g_1\otimes...\otimes g_m\ra:=(-1)^{\sum_{i<j}p(f_i)p(f_j)}W_{n+m}(\overline f_n\boxtimes...\boxtimes\overline f_1\boxtimes g_1\boxtimes...\boxtimes g_m).
$$
It is easy to see that this inner product is well defined, and 
$$
\la f_1\otimes...\otimes f_n, g_1\otimes...\otimes g_m\ra=\la \phi(f_1)...\phi(f_n)\Omega, \phi(g_1)...\phi(g_m)\Omega\ra
$$
(where $f_i$ are purely odd or purely even). 
Thus the inner product $\la,\ra$ on $\widetilde{\mathcal D}$ is nonnegative definite, 
the Hilbert space $\mathcal H$ can be recovered as the completion of $\widetilde{\mathcal D}$ with respect to $\la,\ra$, 
and $\mathcal D$ is the image of $\widetilde{\mathcal D}$ in $\mathcal H$ (note that 
the map $\widetilde{\mathcal D}\to \mathcal H$ need not be injective). 
Moreover, the vector $\Omega$ is the image of $1\in \widetilde{\mathcal D}$ in $\mathcal D$, and the representation $\pi$ is obtained by extending the action of $\widetilde{\bold P}$ on $\mathcal D$ (which descends from $\widetilde {\mathcal D}$) by continuity. 

So we can ask: what conditions should Wightman functions satisfy to define a Wightman QFT? 
Let us list some necessary conditions, which follow from the above discussion. 
To this end, denote by 
$$
W: T(\mathcal S\otimes R)\to \Bbb C
$$ 
the natural liner map and by $*: T(\mathcal S\otimes R)\to T(\mathcal S\otimes R)$ the antilinear map given by 
$(f_1\otimes...\otimes f_n)^*=(-1)^{\sum_{i<j}p(f_i)p(f_j)}\overline f_n\otimes...\otimes \overline f_1$.  

\begin{proposition}\label{cond15} The Wightman functions $W_n$ of a Wightman QFT satisfy the following properties. 

1. $W_n$ are $\widetilde{\bold P}$-invariant. 

2. Positive energy: the Fourier transform of $W_n$ is supported on the set of $(p_1,...,p_n)\in V^n$ 
such that $\sum_i p_i=0$ and $p_{i+1}-p_i\in \overline V_+$. 

3. $W_n(f^*)=\overline{W_n(f)}$. 

4. Space locality: 
$$
W_n^{u_1,...,u_n}(x_1,...,x_i,x_{i+1},...,x_n)=(-1)^{p(u_i)p(u_{i+1})}W_n(x_1,...,x_{i+1},x_{i},...,x_n)
$$
 if $|x_i-x_{i+1}|^2<0$. 

5. Positivity: $W(f^*\otimes f)\ge 0$ for any $f\in T(\mathcal S\otimes R)$. 
\end{proposition} 

\begin{proof} (1) follows from the invariance of the vacuum vector and the field map. (3) follows from the fact that for real $f$, $\phi(f)$ is hermitian symmetric. (4) follows from the space locality axiom. (5) follows from positivity 
of the inner product on $\mathcal H$. So it remains to prove (2). Let us do so for $n=2$, the general proof is similar. 

By translation invariance we have 
$$
W_2(v_1,v_2)=\Bbb W(v)
$$ 
where $v=v_2-v_1$. Thus our job is to show that the Fourier transform of $\Bbb W$ is supported on $\overline V_+$. We have 
$$
\Bbb W(v)=\la \Omega,\phi(0)\phi(v)\Omega\ra=
$$
$$
=\la \Omega,\phi(0)\pi(v)\phi(0)\pi(-v)\Omega\ra=
\la \phi(0)\Omega,\pi(v)\phi(0)\Omega\ra.
$$
So the statement follows from the fact that every character of $V$ which occurs in $\mathcal H$ belongs to 
$\overline V_+$. 
\end{proof} 

In fact, it turns out that these necessary conditions are also sufficient, and we have the following theorem, which can be proved by following the above reconstruction procedure (but we will not give a proof): 

\begin{theorem} If a collection of distributions $W_n$ satisfies conditions (1)-(5) of Proposition \ref{cond15} then they define a Wightman QFT. 
\end{theorem} 

\begin{remark} The 1-point Wightman function $W_1(x)=\la \Omega,\phi(x)\Omega\ra$ 
is a constant $c$ by translation invariance, i.e. it is an element of $R^*$, 
and by invariance under rotations it is in $(R^*)^{{\rm Spin}_+(V)}$.
Thus we may (and will) assume without loss of generality that $c=0$ (otherwise 
we can replace $\phi(x)$ by $\phi(x)-c$). So we may assume without loss of generality that 
$W_1=0$. 
\end{remark} 

\begin{remark}\label{ft} The positivity property for the 2-point function can be written as 
$$
\int_{V^2}\Bbb W(x_2-x_1)\overline{f(x_1)}f(x_2)dx_1dx_2\ge 0,
$$
where $\Bbb W(x)=W_2(0,x)$. Thus, taking Fourier transforms, we have 
$$
\int_{V}\widehat{\Bbb W}(p)\overline {\widehat f(p)}\widehat f(p)dp\ge 0.
$$
This shows that $\widehat{\Bbb W}(p)dp$ is a measure concentrated 
on $\overline V_+$ and valued in nonnegative hermitian forms on $R_{\Bbb C}$. 
\end{remark} 

\subsection{The mass spectrum of a Wightman QFT} 

Let $\mathcal H^{(1)}\subset \mathcal H$ be the closure of the span 
of vectors $\phi(x)\Omega$, $x\in V$. It is called the space of {\it $1$-particle states},\index{$1$-particle state} and it is 
clearly a $\widetilde{\bold P}$-subrepresentation of $\mathcal H$. The {\it mass spectrum}\index{mass spectrum of a QFT}
of the theory is determined by the structure of this representation. So we need to discuss the 
representation theory of $\widetilde{\bold P}$. 

Since $\widetilde{\bold P}$ is a semidirect product, its irreducible unitary representations are unitarily induced. 
Namely, let $\mathcal O$ be an orbit of ${\rm Spin}_+(V)$ on $V$ 
and $\rho$ be an irreducible unitary representation 
of the stabilizer $\widetilde{\bold P}_0$ of a point $v_0\in \mathcal O$.   
Then $\rho$ defines an equivariant Hilbert bundle 
on $\mathcal O$ with total space 
$(\widetilde{\bold P}\times \mathcal \rho)/\widetilde{\bold P}_0$ where $\widetilde{\bold P}_0$ acts diagonally. 
Thus we can consider the space $\mathcal H_{\mathcal O,\rho}$ of square integrable half-densities 
on $\mathcal O$ with values in this bundle. This space carries a unitary representation 
of $\widetilde{\bold P}$. A theorem of Mackey then says that this unitary 
representation is irreducible, and all irreducible unitary representations of 
$\widetilde{\bold P}$ are obtained uniquely in this way. For example, if $\mathcal O=\lbrace 0\rbrace$, 
then $\mathcal H_{0,\rho}$ is just a unitary irreducible representation of ${\rm Spin}_+(V)$. 

Now we are ready to discuss the structure of the representation $\mathcal H^{(1)}$. 
By taking Fourier transforms (see Remark \ref{ft}), we see that if $\mathcal H_{\mathcal O,\rho}$ occurs in 
$\mathcal H^{(1)}$ then $\rho$ needs to be finite dimensional. For example, for 
$d\ge 3$ and $\mathcal O=\langle 0\rangle$ the only choice is the trivial representation, as the group
${\rm Spin}_+(V)$ is a connected semisimple non-compact Lie group. Moreover, if the theory has a unique vacuum then the trivial representation occurs in $\mathcal H$ discretely with multiplicity $1$, as the span of the vacuum vector $\Omega$. As $\mathcal H^{(1)}$ is orthogonal to $\Omega$ (since $W_1=0$), 
we see that the trivial representation does not occur in $\mathcal H^{(1)}$. 

Let us now consider what happens with other orbits. By the positive energy condition, the only orbits that can occur are 
$X_m^+$ defined by $E=\sqrt{p^2+m^2}$, $E>0$ (where for $d=2$ the set $X_0^+$ falls into two 
orbits $X_0^{++}$ and $X_0^{+-}$ defined by $p=\pm E>0$). For $m>0$ this is the upper sheet of a 2-sheeted hyperboloid and for $m=0$ it is the upper part of the light cone (which is a union of two orbits for $d=2$). 

In the case $m>0$, we may take $v_0=(m,0)$, 
then $\widetilde{\bold P}_0={\rm Spin}(d-1)$, so $\rho$ is a (necessarily finite dimensional) unitary representation of this compact Lie group. Physicists say that this representation corresponds to a {\it massive particle of mass $m$ and 
type $\rho$}.\index{massive particle} Particles arising in physically relevant quantum field theories are usually {\it scalars}\index{scalar particle}
($\rho=\Bbb C$), {\it spinors}\index{spinor particle} ($\rho$ is a spinor representation of ${\rm Spin}(d-1)$) and {\it vectors}\index{vector particle} ($\rho=\Bbb C^{d-1}$ is the vector representation ${\rm Spin}(d-1)$). Note that by the spin-statistics theorem, 
scalars and vectors are bosons and spinors are fermions. 

If $m=0$, $d\ge 3$, then we can take $v_0=(1,1,0,...,0)$, and the stabilizer 
is the non-reductive Lie group ${\rm Spin}(d-2)\ltimes \Bbb R^{d-2}$. Since $\rho$ is finite dimensional, 
$\Bbb R^{d-2}$ has to act trivially, so $\rho$ 
is an irreducible representation of the compact Lie group ${\rm Spin}(d-2)$. 
Physicists say that this representation corresponds to a {\it massless particle of 
type $\rho$}.\index{massless particle} The classification of massless particles is the same 
as for massive ones; however, note that since for massless particles $\rho$ is a representation of 
${\rm Spin}(d-2)$ rather than ${\rm Spin}(d-1)$, they in general have fewer components than massive ones; for example, a massless vector has one fewer component than a massive one. 

If $m=0,d=2$ then there are two choices for $v_0$: $(1,1)$ and $(1,-1)$. They have trivial stabilizer, so $\rho=\Bbb C$. 
Thus we have two types of massless particles: {\it right-moving}\index{right-moving particle} and {\it left-moving},\index{left-moving particle} corresponding to the two choices 
of $v_0$. These particles are called this way since the corresponding operators $\phi(x)$ satisfy the conditions 
$\phi(t,x)=\phi(0,x-t)$, $\phi(t,x)=\phi(0,x+t)$, respectively, which classically would be right-moving and left-moving waves.  

 The set $M$ of numbers $m$ corresponding to representations $\mathcal H_{X_m^+,\rho}$ (or $\mathcal H_{X_0^{+\pm},\rho}$ for $d=2$) occurring in $\mathcal H^{(1)}$ is called the {\it mass spectrum}\index{mass spectrum of a QFT} of the theory. 
One says that the theory has a {\it mass gap}\index{mass gap} when ${\rm inf} M=m>0$. In this case the spectrum of 
$\widehat H$ is $\lbrace 0\rbrace\cup [m,+\infty]$, so there is a gap between $0$ and $m$. 
To find the mass spectrum, it suffices to look at the function $\widehat{\Bbb W}$: 
the mass spectrum is just the intersection of its support with the time axis (this follows from Remark \ref{ft}). 

\subsection{Free theory of a scalar boson} 

Let us now construct a Wightman QFT corresponding to a scalar boson of mass $m>0$. 
Recall that in the Lagrangian setting we had a 2-point function 
$G_M(x_2-x_1)$, where $G_M(x)$ is a distribution satisfying the Klein-Gordon equation
$$
(\square+m^2)G_M=i\delta.
$$  
So at first sight for the corresponding Wightman QFT we want to have 
$\Bbb W(x)=G_M(x)$, so that the Lagrangian and Hamiltonian approach agree.  
However, the function $G_M(x)$ is even, while for $\Bbb W(x)$ we are supposed to have 
$\Bbb W(-x)=\overline{\Bbb W(x)}$, so our equality needs to be relaxed. 
 In fact, the correct condition is that the identity $\Bbb W(x)=G_M(x)$ only needs to hold
when $x$ is spacelike or when $x\in \overline V_+$. 
When $x\in \overline V_-$, we should rather have $\Bbb W(x)=\overline{G_M(x)}$. In other 
words, 
$$
G_M(x_2-x_1)=W_2^T(x_1,x_2)
$$ 
is the so-called {\it time ordered}\index{time-ordered 2-point function} 2-point function, i.e. one obtained 
from $W_2(x_1,x_2)$ when $x_1,x_2$ are put in the chronological order (where in the spacelike separated case 
the order does not matter due to space locality). 

We claim that with this definition the function $\Bbb W(x)$ satisfies the Klein-Gordon equation 
$$
(\square+m^2)\Bbb W=0
$$
on the nose (without the delta-function on the right hand side). Indeed, we have ${\rm Re}\Bbb W(x)={\rm Re}G(x)$, which satisfies the Klein-Gordon equation, so it remains to show that ${\rm Im}\Bbb W(x)$ 
satisfies it as well. But it is easy to see that 
$(\square+m^2){\rm Im}\Bbb W(x)$ is a distribution supported 
at the origin of homogeneity degree $-d$, so it is a multiple of $\delta$. Since ${\rm Im}\Bbb W(x)$ is an odd function, this distribution must be zero, as claimed. 

Also, since $G_M(x)$ is real for spacelike $x$, we get $\Bbb W(-x)=\overline{\Bbb W(x)}$. 
Thus the Fourier transform $\widehat {\Bbb W}(p)$ 
is real valued, supported on the hyperboloid $X_m$ and invariant under $SO_+(V)$. 
It follows that  
$$
\widehat {\Bbb W}(p)=c_+\delta_{X_m^+}+c_-\delta_{X_m^-}.
$$ 
where $c_\pm \in \Bbb R$. 
but in fact it can be shown that only $\delta_{X_m^+}$ occurs (this follows from the exponential decay of the Euclidean 2-point correlation function at infinity). Thus 
$$
\widehat {\Bbb W}(p)=c\delta_{X_m^+}. 
$$
In fact, one can show that $c=2\pi$. 

Similarly, we define higher $W_n$ for $n>2$ by the Wick formula, and 
this analysis implies after some work that these functions define a Wightman QFT. 

In this case, $\mathcal H^{(1)}=L^2(X_m^+)$, so we have a single particle of mass $m$. 

The theory of a free massless scalar, as well as massive and massless spinor is defined similarly. 

\subsection{Normal ordering, composite operators and operator product expansion in a free QFT}\label{compop}

In classical field theory, given a classical scalar field $\phi(x)$, we may consider arbitrary 
polynomials and even any smooth functions of $\phi$. The same is true for quantum mechanics, where $\phi(t)$ is a self-adjoint (possibly unbounded) operator on the Hilbert space $\mathcal H$ of quantum states, so using 
its spectral decomposition, we may define functions of $\phi$. However, in quantum field theory in $d+1$ dimensions with $d\ge 1$ the situation is more complicated. Indeed, in this case $\phi(x)$ is not a usual operator-valued function of $x$, but rather a generalized one -- an operator-valued distribution, and we know
that for singular distributions, such as $\delta(x)$, we cannot even define the square 
$\delta(x)^2$.   

Indeed, let $\phi(x)$ be a quantum scalar boson. Then the 2-point correlation function
$$
\langle\phi(x)\phi(y)\rangle=\langle \Omega,\phi(x)\phi(y)\Omega\rangle=G(x-y)
$$
blows up when $|x-y|^2=0$ (so in Euclidean signature, when $x=y$), so the operator $\phi^2(x)$ 
cannot possibly be well defined. 

Thus, if we want to quantize the classical field $\phi^2(x)$, we need to regularize the 
corresponding operator product. This can be done by a standard regularization procedure 
called the {\it normally ordered product}.\label{Normally ordered product} 

For example, in Euclidean signature, the operator product $\phi(x)\phi(y)$ is well defined 
when $x\ne y$: indeed, by Wick's formula 
$$
\langle \phi(x)\phi(y)\phi(z_1)...\phi(z_k)\rangle=
$$
$$
G(x-y)\langle \phi(z_1)...\phi(z_k)\rangle
+\sum_{i\ne j}G(x-z_i)G(y-z_j)\langle \phi(z_1)...\widehat \phi(z_i)...\widehat \phi(z_j)...\phi(z_k)\rangle,
$$
where the hat indicates omissions (here $x,y,z_1,...,z_k$ are distinct). Now, when $x\to y$, the first summand in this formula blows up while the second one does not. So it is natural to define the {\it normally ordered product} $\colon \phi(x)\phi(y)\colon$ just by throwing away the singular terms, i.e. by the condition that its correlation function with $\phi(z_1)...\phi(z_k)$ is 
 $$
\langle \colon\phi(x)\phi(y)\colon\phi(z_1)...\phi(z_k)\rangle=
\sum_{i\ne j}G(x-z_i)G(y-z_j)\langle \phi(z_1)...\widehat \phi(z_i)...\widehat \phi(z_j)...\phi(z_k)\rangle.
$$
This is equivalent to just saying that 
$$
\colon\phi(x)\phi(y)\colon=\phi(x)\phi(y)-G(x-y).
$$
Note that while $\phi(x)\phi(y)$ blows up when $x=y$, the normally ordered product 
$\colon \phi(x)\phi(y)\colon$ does not: 
$$
\langle \colon\phi^2(x)\colon\phi(z_1)...\phi(z_k)\rangle=
\sum_{i\ne j}G(x-z_i)G(x-z_j)\langle \phi(z_1)...\widehat \phi(z_i)...\widehat \phi(z_j)...\phi(z_k)\rangle.
$$
This defines a {\it composite operator}\index{Composite operator} $\colon\phi^2(x)\colon$, which is 
a well defined operator-valued distribution. 

Similarly one may define the normally ordered product $\colon \phi(x_1)...\phi(x_m)\colon$ of any number of factors, by removing all the singular terms from the correlators. For example,
$$
\colon \phi(x)\phi(y)\phi(z)\colon=\phi(x)\phi(y)\phi(z)-G(x-y)\phi(z)-G(y-z)\phi(x)-G(z-x)\phi(y).
$$
Such a product is well defined for all values of $x_1,...,x_k$ and is commutative (independent of ordering of factors) and associative. We can also differentiate 
by $x_j$ any number of times, to define the normally ordered product of arbitrary derivatives of $\phi$. Evaluating such products on the diagonal (when all points are the same), we obtain 
composite operators attached to any differential monomials (hence polynomials) with respect to $\phi$, such as $\colon\phi^3(x)\colon$, $:\phi_{x_i}\phi_{x_j}:$, etc. 

\begin{exercise} Derive a formula for the correlation function of several composite operators 
(evaluated at different points) in the theory of the scalar boson. 
\end{exercise} 

In particular, we can now consider the product of two composite operators, e.g. 
$\colon\phi^2(x)\colon \phi(y)$. Of course, this has a singularity 
at $x=y$, and an important problem is to understand the nature of this singularity. 
This is achieved by the procedure called the {\it operator product expansion},\index{Operator product expansion} which replaces the non-existent multiplication of composite operators. 

To explain this procedure, consider first the simplest example of operator product: 
$$
\phi(x)\phi(y)=G(x-y)+\colon \phi(x)\phi(y)\colon.
$$
Using Taylor's formula, this can be rewritten so that the right hand side only contains $\phi(y)$ and no $\phi(x)$:
$$
\phi(x)\phi(y)=G(x-y)+\sum_{\bold n}\frac{(x-y)^{\bold n}}{\bold n!}\colon \partial^{\bold n}\phi(y)\cdot \phi(y)\colon,
$$
where $\bold n:=(n_1,...,n_{d+1})$, $(x-y)^\bold n:=\prod_i (x_i-y_i)^{n_i}$, 
$\partial^{\bold n}:=\prod_i \partial_{x_i}^{n_i}$, and $\bold n!:=\prod_i n_i!$. 
In this sum, all terms except the first one are regular (i.e., continuous) at $x=y$. 

Let us now try to write down a similar expansion for a more complicated example of operator product, 
$\colon\phi^2(x)\colon \phi(y)$. We have 
$$
\langle \colon\phi^2(x)\colon \phi(y)\phi(z_1)...\phi(z_k)\rangle=2G(x-y)\langle \phi(x)\phi(z_1)...\widehat{\phi(z_j)}...\phi(z_k)\rangle+
$$
$$
\sum_{j,m,n\text{ distinct}}G(x-z_j)G(x-z_m)G(y-z_n)\langle \phi(z_1)...\widehat{\phi(z_j)}...\widehat{\phi(z_m)}...\widehat{\phi(z_n)}...\phi(z_k)\rangle.
$$
 Thus we get 
$$
\colon\phi^2(x)\colon \phi(y)=2G(x-y)\phi(y)+\colon \phi^2(x)\phi(y)\colon. 
$$
As before, using Taylor's formula, this can be rewritten so that the right hand side only contains $\phi(y)$ and no $\phi(x)$:
$$
\colon\phi^2(x)\colon \phi(y)=2G(x-y)\phi(y)+\sum_{\bold n,\bold m}\frac{(x-y)^{\bold n+\bold m}}{\bold n!\bold m!}\colon \partial^{\bold n}\phi(y)\cdot \partial^{\bold m}\phi(y)\cdot \phi(y)\colon.
$$
And again, all terms except the first one are regular at $x=y$.

As a final example, consider the product $\colon\phi^2(x)\colon \cdot \colon\phi^2(y)\colon$. 
A similar computation yields 
$$
\colon\phi^2(x)\colon \cdot \colon\phi^2(y)\colon =
2G^2(x-y)+4G(x-y)\colon \phi(x)\phi(y)\colon+\colon \phi^2(x)\phi^2(y)\colon, 
$$
and as before we can expand this to remove $\phi(x)$ using Taylor's formula. Namely, 
expanding the second summand, we get 
$$
\colon\phi^2(x)\colon \cdot \colon\phi^2(y)\colon =
$$
$$
2G^2(x-y)+4G(x-y)\sum_{\bold n}\frac{(x-y)^{\bold n}}{\bold n!}\colon \partial^{\bold n}\phi(y)\cdot \phi(y)\colon+\colon \phi^2(x)\phi^2(y)\colon, 
$$
and the last summand can be expanded similarly. We now see that there are many singular terms: $G(x)$ behaves as $|x|^{1-d}$ for $d>1$ and as $\log|x|$ for $d=1$, so 
the singular terms are the ones with $|\bold n|\le d-1$, where $|\bold n|:=\sum_i n_i$. 
For example, for $d=2$ for the massless boson we have 
$$
\colon\phi^2(x)\colon \cdot \colon\phi^2(y)\colon =
$$
$$
\frac{2}{|x-y|^2}+\frac{4}{|x-y|}\colon \phi^2(y)\colon+\sum_{j=1}^3\frac{4}{|x-y|}(x_j-y_j)\colon \partial_{x_i}\phi(y)\cdot \phi(y)\colon+\text{ regular}.
$$

Yet we see that the number of singular terms is finite. 
In fact, it is not hard to prove the following proposition (see \cite{QFS}, vol 1, p.449). 

\begin{proposition}\label{ope} Let $A,B$ be two composite operators in the theory of 
scalar boson. Then there exist a unique collection of functions 
$F_j(y)$ and composite operators $C_j(y)$ such that we have an asymptotic expansion 
$$
A(x)B(y)\sim \sum_j F_j(x-y)C_j(y),\ x\to y 
$$
such that for every $N$ we have $|F_j(z)|=O(|z|^N)$, $z\to 0$, for all but finitely many $j$. 
In particular, there are finitely many singular terms (not continuous at $x=y$). 
\end{proposition}  

The expansion of Proposition \ref{ope} is called the {\it operator product expansion}. It is not hard to show that it exists in any free quantum field theory. 
 
\subsection{Symmetries in quantum field theory} 

In studying any physical system, it is crucial to find all its symmetries and use them to their full potential. For example, the equations of motion of a particle in a rotationally symmetric potential field can be fully solved 
by utilizing the rotational symmetry (see \cite{A}). 

The most fundamental fact about symmetries in classical or quantum mechanics is that for any 1-parameter group of symmetries of the system there is an (essentially unique) observable responsible for this symmetry, which is conserved in this system; i.e., every 1-parameter symmetry corresponds to a conservation law, and vice versa. This statement is called {\it Noether's theorem}\index{Noether's theorem}. 

Let us first explain the precise meaning of Noether's theorem in the setting of classical mechanics. Suppose we have a system with phase space a symplectic manifold $(M,\omega)$ (typically $M=T^*X$, where $X$ is the configuration space, and $\omega=d\alpha$ is the differential of the Liouville form) and hamiltonian $H\in C^\infty(M)$. Let $g^t$ be a 1-parameter group of symmetries of this system, i.e., of symplectic diffeomorphisms of $M$ which preserve $H$. Let $v:=\frac{d}{dt}|_{t=0}g^t$ be the vector field generating the flow $g^t$. Then we have $L_v \omega=0$ (i.e., $v$ is a symplectic vector field, so $\omega_v:=\omega(v,?)$ is a closed 1-form), and $L_v H=0$. 
Let us assume that $M$ is simply connected (for example, we can restrict ourselves to a neighborhood of a point in $M$ or $X$). In this case $\omega_v$ is exact, so there exists $Q\in C^\infty(M)$ (unique up to adding a constant) such that $\omega_v=dQ$. Then for any observable $F\in C^\infty(M)$ we have $L_vF=\lbrace Q,F\rbrace$. 
Moreover $\lbrace Q,H\rbrace=L_vH=0$. The observable $Q$ is thus conserved under the hamiltonian flow and is the conservation law corresponding to the 1-parameter group $g^t$. It is called (especially in the setting of field theory) the {\it Noether (conserved) charge} of the symmetry.\index{Noether charge}\index{Conserved charge}

A trivial example of this is the hamiltonian flow $h^t$ defined by the hamiltonian $H$ itself, i.e., the time translation symmetry; in this case $Q=H$, so the corresponding conserved quantity is $H$ (the energy). Other examples include the momenta $p_1,...,p_n$ 
which corresponds to translation symmetry (for $X=\Bbb R^n$) 
and {\it angular momenta}\index{angular momentum} $M_{kj}:=x_kp_j-x_jp_k$ corresponding to rotational symmetries around the codimension $2$ hyperplanes $x_k=x_j=0$. 

More generally, suppose $G$ is a Lie group acting (on the right) by symmetries of the system. Let $\g={\rm Lie}G$ be the Lie algebra of $G$. Any element 
$y\in \g$ gives rise to a 1-parameter subgroup $e^{ty}\in G$, so 
defines a conserved quantity $Q_y$ such that 
$$
\lbrace Q_y,F\rbrace=y\cdot F:=\tfrac{d}{dt}|_{t=0} e^{ty}\cdot F
$$
for each $F\in C^\infty(M)$, where $(g\cdot F)(m)=F(mg)$, $m\in M$. More precisely, $Q_y$ is defined only up to adding a constant, so let us fix some linear assignment $y\mapsto Q_y$. 

Moreover, it is clear that for $y,z\in \g$ 
$$
\lbrace Q_y,Q_z\rbrace =Q_{[y,z]}+C(y,z),
$$
where $C(y,z)$ is a skew-symmetric bilinear form on $\g$ which arises because $Q_y$ is uniquely determined by $y$ only up to adding a constant.
Furthermore, by the Jacobi identity, the form $C$ is a 2-cocycle :
$$
C([x,y],z)+C([y,z],x)+C([z,x],y)=0.
$$
It follows that the assignment $y\mapsto Q_y$ is almost a homomorphism $\g\to C^\infty(M)$, but not quite: rather, 
it defines a homomorphism 
$$
\mu: \widehat \g\to C^\infty(M),
$$ 
where 
$\widehat \g:=\g\oplus \Bbb R$ is a 1-dimensional central extension 
of $\g$ with commutator 
$$
[(y,a),(z,b)]=([y,z],C(y,z)).
$$
Namely, $\mu(y,a)=Q_y+a$.

The map $\mu$ may be viewed as an element of $C^\infty(M)\otimes \widehat\g^*$, i.e., geometrically as a $C^\infty$-map 
$$
\mu: M\to \widehat\g^*.
$$ 
This map is called the {\it moment map}\index{moment map} and plays a fundamental role in symplectic geometry. 

The following example shows that the cohomology class of $C$ may be nonzero, which means that we may not be able to choose $Q_y$ to make $C=0$. 

\begin{example} The group $\Bbb R^{2n}$ acts on $M=T^*\Bbb R^n$ 
(with trivial hamiltonian $H=0$) by translations. 
So we have $\g=\Bbb R^{2n}$ and $C(y,z)=\omega(y,z)$. 
Thus $\widehat \g$ is the {\it Heisenberg Lie algebra}\index{Heisenberg Lie algebra} $\Bbb R^{2n}\oplus \Bbb R$ with commutation relations 
 $$
[(y,a),(z,b)]=([y,z],C(y,z)),
$$
which is a non-trivial central extension of $\g$. 
\end{example} 

However, in many examples the cohomology class $[C]\in H^2(\g)$ is, in fact, zero, i.e., $\widehat \g=\g\oplus \Bbb R$ as Lie algebras. For instance, this is automatically so if $H^2(\g)=0$ (e.g., if $G$ is a compact Lie group). In this case, we may choose $Q_y$ so that $C=0$, and we have a moment map 
$$
\mu: M\to \g^*.
$$
For example, for translation symmetries of the free particle, 
$\mu$ is the momentum $\bold p$ of the particle, which explains the terminology ``moment map". 

A similar discussion applies to classical field theory, using the formalism of Subsection \ref{hfcft}. Namely, in this case, 
the Noether charge is given by the integral over the space of a certain local field called {\it Noether current}\index{Noether current}. 

For example, consider the free massive boson $\phi$ on the spacetime $\Bbb R^d\times \Bbb R$. The Hamiltonian is 
$$
H=\frac{1}{2}\int_{X} (\phi_t^2+|d_x\phi|^2+m^2\phi^2)dx.
$$
Thus $H=\int_{\Bbb R^d} Jdx$ where 
\begin{equation}\label{kertime}
J=\frac{1}{2}(\phi_t^2+|d_x\phi|^2+m^2\phi^2)=\frac{1}{2}(:\phi_t^2:+\sum_{j=1}^d:\phi_{x_j}^2:+m^2:\phi^2:)
\end{equation} 
is the Noether current associated to the time translation symmetry. 

Similarly, the Noether current for the spacial translation 
in the $i$-th coordinate is 
\begin{equation}\label{kerspace}
J_k=\phi_t \phi_{x_k}. 
\end{equation} 
Indeed, using the formulas of Subsection \ref{hfcft}, we have 
$$
\lbrace J_k(x),\phi(y)\rbrace=-\phi_{x_k}(x)\delta(x-y),\ 
\lbrace J_k(x),\phi_t(y)\rbrace=\phi_t(x)\delta_{x_k}(x-y).
$$
Thus defining the charge 
$$
P_k=\int_{\Bbb R^d} J_k(x)dx,
$$ 
using integration by parts, we get 
$$
\lbrace P_k,\phi(y)\rbrace=-\phi_{x_k}(y),\ \lbrace P_k,\phi_t(y)\rbrace=
-\phi_{tx_k}(y),
$$
as needed. 

Furthermore, this discussion extends to quantum theory, with observables replaced by operators as usual. Namely, 
in this case, we have a unitary projective representation $\pi: G\to {\rm Aut}(\mathcal H)$ of the Lie group $G$ of symmetries on the Hilbert space $\mathcal H$ of quantum states of the system, so that 
$[\pi(g),\widehat H]=0$, where $\widehat H: \mathcal H\to \mathcal H$ is the hamiltonian
(an unbounded self-adjoint operator). The quantum Noether charges 
corresponding to these symmetries simply define the corresponding Lie algebra representation $\pi_*: \g\to {\rm End}(\mathcal S)$, where 
$\mathcal S$ is a certain dense subspace of $\mathcal H$ (of smooth vectors) on which all the operators $\pi_*(y)$ are defined. For instance, in quantum mechanics, like in classical one, the time translation corresponds to the Hamiltonian $\widehat H$, the spacial translations to 
the momentum operators $\widehat p_j:=-i\hbar \partial_{x_j}$, 
and rotations around $x_k=x_j=0$ to the angular momentum operators\index{Angular momentum operator}
$$
\widehat M_{kj}:=-i\hbar (x_k\partial_j-x_j\partial_k).
$$

Finally, in quantum field theory, by analogy with classical one, a quantum Noether charge is an operator of the form 
$$
Q=\int_{\Bbb R^d}J(x)dx, 
$$
where $J(x)$ is a quantum local operator called the {\it quantum Noether current}. For example, in the case of a free massive boson, the 
currents $J(x)$ and $J_k(x)$ for time and space translations 
are given by the same formulas \eqref{kertime},\eqref{kerspace}, 
but now with $\phi(x,t)$ being the quantum field corresponding to the massive boson (say, in the setting of Wightman axioms) rather than the classical field, and with normal ordered product 
instead of the usual product: 
$$
J=\frac{1}{2}(\colon\phi_t^2\colon+\sum_{j=1}^d\colon\phi_{x_j}^2\colon+m^2\colon\phi^2\colon),
$$
$$
J_k=\colon\phi_t \phi_{x_k}\colon,
$$
and the corresponding charges, as in the classical case, are given by integration 
of the current over the space. For example, for the free boson
\begin{equation}\label{quanham}
\widehat H=\int_{\Bbb R^d} J(x)dx
\end{equation} 
is the quantum hamiltonian, and 
$$
\widehat P_k:=\int_{\Bbb R^d} J_k(x)dx 
$$
are the quantum momentum operators. 

\subsection{Field theories on manifolds} \label{manif} 

As already mentioned above, an important feature of classical and quantum field theory is the possibility to consider them not just on a Euclidean or Minkowskian space, but more generally on Riemannian and Lorentzian manifolds. The main examples are theories 
on $X\times \Bbb R$, where $X$ is a Riemannian $d$-dimensional space manifold with metric $g_{ij}dx^i dx^j$ (Einstein summation) and $\Bbb R$ is the time line, with Lorentzian metric 
$$
|dx|^2:=(dt)^2-g_{ij}dx^i dx^j,
$$ 
and Euclidean theories on a Riemannian $d+1$-dimensional spacetime manifold $M$. 

Here we will consider only classical field theories on manifolds. 
These theories can then be quantized using either Lagrangian or Hamiltonian approach, but we will not discuss this, except in some examples. The story is parallel to the case of flat space considered above, but we should make sure that the kinetic term and other terms in the Lagrangian are defined canonically (i.e., do not depend on the choice of coordinates). For simplicity consider the Euclidean case (in the Lorentzian case the story is similar). We restrict ourselves to reviewing the most common types of classical fields in such theories, as well as the corresponding kinetic and other terms in their Lagrangians.
A more complete discussion can be found in \cite{QFS}. 

1. {\bf Scalar (bosonic) fields.} In the simplest case a scalar field is just a real function on $M$ (real scalar), but one can also consider scalars valued in a finite dimensional real vector space with a positive inner product (for example, $\Bbb C$, for complex scalars) or, more generally, valued in a real vector bundle on $M$. The kinetic term for a scalar $\phi: M\to E$ valued in a vector space $E\cong E^*$ with inner product is $|d\phi(x)|^2$, the squared norm of the vector $d\phi\in T_{\phi(x)}M\otimes E$ with respect to the inner products on $T_{\phi(x)}M$ and $E$. Thus if 
this vector has components $(d\phi)_{ij}$ in orthonormal bases then 
$$
|d\phi|^2=\sum_{i,j}(d\phi)_{ij}^2.
$$
More generally, if $E$ is a vector bundle on $M$ then we need to fix 
an inner product on $E$ (i.e., $E$ should be an orthogonal bundle) 
and also a connection $A$ preserving this inner product, which gives rise to the covariant derivative operator $\nabla_A$; if $E$ is trivialized on a local chart $U\subset M$ then $A$ becomes a 1-form on $U$ with values in $\mathfrak{o}(E)$ and we have $\nabla_A=d+A$. In this case, an $E$-valued scalar field $\phi$ is a section of $E$ over $M$, and the kinetic term is $|\nabla_A \phi|^2$, which in local trivialization has the form $|d\phi+A\phi|^2$. 

Note that for a scalar field $\phi$, we can always add to the kinetic term a mass term $m^2|\phi|^2$, where $m^2$ is a real number. More generally, we can add a mass term $(\phi, Q\phi)$, where $Q$ is a self-adjoint endomorphism of $E$. 

2. {\bf Spinor (fermionic) fields.} Spinor fields can be defined on a spin manifold $M$, i.e., an oriented manifold equipped with a spin structure (a lift of the tangent bundle from $SO(n)$ to ${\rm Spin}(n)$). For such a manifold, we have the canonically defined {\it spin bundle}\index{spin bundle} $S_M$, which is the associated bundle to the above ${\rm Spin}(n)$ bundle via the spin representation ${\rm Spin}(n)\to {\rm Aut}(S)$. 
This bundle carries a natural inner product and a connection induced by the Levi-Civita connection of $M$ that preserves this inner product. Moreover, as explained in Subsection \ref{spinorssec1}, in even dimensions we have $S=S_+\oplus S_-$, where $S_+,S_-$ are irreducible representations of ${\rm Spin}(n)$, so we have $S_M=S_{M+}\oplus S_{M-}$, an orthogonal decomposition of $S_M$ into two subbundles. 

Spinor fields, in the most basic case, are sections of the Spin bundle $S_M$. The sections of $S_{M+}$ and $S_{M-}$, as noted in Subsection  
\ref{ferlag}, are called {\it chiral spinors}. 

The possible kinetic and mass terms for spinors on the flat space are described in Subsection \ref{ferlag}, and the story on the curved manifold is similar. The only new feature is that we have to define the Dirac operator $\bold D$ for a spinor field on an arbitrary spin manifold. 
To this end, all we have to do is replace ordinary partial derivatives 
in formula \eqref{dirop} by the covariant ones with respect to the Levi-Civita connection: 
\begin{equation}\label{dirop1}
\bold{D}=\sum_i \Gamma_i \nabla^{LC}_i.
\end{equation}

More generally, similar to the scalar field case, we may consider spinors valued in a vector bundle $E$ with an inner product and an orthogonal connection $A$, i.e., sections of the bundle $S_M\otimes E$. This bundle carries a tensor product connection $\nabla^{\rm total}=\nabla^{LC}\otimes \nabla_A$, and the Dirac operator is defined by 
the formula 
$$
\bold{D}=\sum_i \Gamma_i \nabla^{\rm total}_i.
$$

3. {\bf Gauge fields.}\index{Gauge field} Let $G$ be a compact Lie group $\g={\rm Lie}G$ equipped with a positive invariant inner product. 
Gauge fields are connections $A$ on principal $G$-bundles
$E$ on $M$, so in local trivialization $A$ is a 1-form on $M$
with values in $\g$ and the covariant derivative with respect to $A$ looks like $\nabla_A=d+A$. The connection 
$A$ has curvature $F_A$ (called {\it field strength}\index{field strength} in physical terminology), which is a 2-form on $M$ with values in the adjoint bundle ${\rm ad}E$. In local trivialization the curvature of $A$ is the Maurer-Cartan 
form 
$$
F_A=dA+\frac{1}{2}[A,A].
$$
In particular, if $G$ is abelian then we just have $F_A=dA$. 
The kinetic term for a gauge field $A$ is $|F_A|^2$, where the squared norm is taken with respect to the inner product on $(\wedge^2 T^*M\otimes {\rm ad}E)_x$ induced by the inner products on $T_xM$ and $\g$ (note that this does not depend on the identification of Lie algebras $({\rm ad}E)_x\cong \g$ since the form on $\g$ is invariant). 

It makes sense to fix the topological type of the 
$C^\infty$-bundle $E$ (which does not change under deformations)
and consider the space ${\rm Conn}(E)$ of all connections $A$ 
on $E$. If $A_1,A_2\in {\rm Conn}E$ then $\nabla_{A_1}-\nabla_{A_2}\in \Omega^1(M)\otimes {\rm ad} E$, so 
${\rm Conn}(E)$ is an affine space with underlying vector space
$\Omega^1(M)\otimes {\rm ad} E$. Moreover, this space carries a natural right affine linear action of the gauge group $\mathcal G_E=C^\infty(M,E)$, which in local trivialization looks like 
$$
A^g=g^{-1} dg+g^{-1}Ag.
$$
The configuration space of a classical gauge theory is then 
$$
\mathcal M:=\sqcup_{\text{ topological types $E$}}{\rm Conn}(E)/\mathcal G_E,
$$ 
so the phase space is 
the cotangent bundle $T^*\mathcal M$.

\section{Perturbative expansion for interacting QFT} 

\subsection{General strategy of quantization}
We now pass to non-free field theories defined by the action 
$S(\phi):=\int \mathcal L(\phi)dx$ in Minkowski space $V\cong \Bbb R^d$, where $\mathcal L(\phi)$ is a local Poincar\'e-invariant Lagrangian. The general strategy of quantization of such theories is as follows. 

{\bf Step 1.}  Write down the Euclidean path integral correlators for the theory: 
$$
\la \phi(x_1)...\phi(x_n)\ra=\int \phi(x_1)...\phi(x_n)e^{-\frac{S_E(\phi)}{\hbar}}D\phi.
$$
Compute the corresponding formal expansion in $\hbar$ using the Feynman rules 
(as we have done in the case of quantum  mechanics, $d=1$). 

{\bf Step 2.} Perform Borel summation of this formal series, to obtain actual functions defined for small enough $\hbar>0$.  

{\bf Step 3.} Perform the Wick rotation of these functions to Minkowski space to obtain Wightman correlation functions $W_n$. 

{\bf Step 4.} Use the functions $W_n$ to define a Wightman QFT, i.e., extract 
the Hilbert space $\mathcal H$, the representation $\pi$ of the (double cover of the) Poincar\'e group on $\mathcal H$, 
the vacuum vector $\Omega$ and the field map $\phi$. 

All these steps are non-trivial, and while Step 1 can be performed fully rigorously, starting from Step 2 a rigorous implementation is only known for a handful of theories treated in constructive field theory (and for many Lagrangians the ultimate Wightman QFT, in fact, does not exist). For most physically interesting theories, doing these steps rigorously is still an open problem. In this section, we will only discuss Step 1. 

\subsection{The $\phi^3$ theory} 

As a running example, we will use the theory of a scalar boson $\phi$ with Euclidean Lagrangian 
$$
\mathcal L_E(\phi):=\frac{1}{2}((d\phi)^2+m^2\phi^2)+\frac{g}{6}\phi^3,
$$
which we will call the {\it $\phi^3$-theory}.\index{$\phi^3$-theory} This theory is a deformation of the theory of free scalar boson 
obtained by adding a single interaction term $\frac{g}{6}\phi^3$, 
which in Feynman calculus corresponds to a 3-valent vertex. 
Physically this vertex corresponds to an interaction in which two particles 
collide and transform into a third one. 

We will set $\hbar=1$ and consider the formal expansion in powers of $g$ 
(which is equivalent to Step 1 by rescaling $\phi$). 

Let us compute the 1-loop correction to the 2-point correlation function of the free theory
$$
\widehat G_0(p)=\frac{1}{p^2+m^2}
$$ 
in the momentum 
space presentation. It is easy to see that this correction is given by a single Feynman diagram 

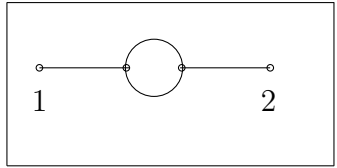
\setlength{\unitlength}{0.435cm}
    \begin{picture}(6,12)(7,-1)

 \put(18,4){\line(1,0){10}}
      \put(18,9){\line(1,0){10}}
      \put(18,4){\line(0,1){5}}
      \put(28,4){\line(0,1){5}}
      %%% make graph
      \put(19,7){\circle*{0.2}}
      \put(19,7){\line(1,0){2.75}}
      \put(21.65,7){\circle*{0.2}}
      \put(22.5,7){\circle{1.75}}
      \put(19,6){\makebox(0,0)[c]{$1$}}
      \put(23.3,7){\line(1,0){2.75}}
      \put(26,6){\makebox(0,0)[c]{$2$}}
      \put(23.35,7){\circle*{0.2}}
      \put(26.05,7){\circle*{0.2}}
      \put(23.5,5){\makebox(0,0)[c]{}}
      \put(27,8){\makebox(0,0)[c]{}}

    \end{picture}

The amplitude of this Feynman diagram is  
$$
A(p)=\frac{g^2}{2(p^2+m^2)^2}\int_V \frac{dq}{(q^2+m^2)((p-q)^2+m^2)}.
$$
If $d<4$, this integral is convergent and can be computed explicitly. 
To this end, we may use the following lemma from multivariable calculus, which is known in physics 
literature as the {\it Feynman famous formula}:\index{Feynman famous formula}

\begin{lemma} Let $\Delta_n$ be the $n-1$-dimensional simplex  
defined in $\Bbb R^n$ by the equation
$$
y_1+...+y_n=1,
$$ 
and $dy$ be the Lebesgue measure 
on $\Delta_n$ of volume $1$. Then for positive numbers $a_1,...,a_n$ we have 
$$
\int_{\Delta_n}\frac{dy}{(a_1y_1+...+a_ny_n)^n}=\frac{1}{a_1...a_n}.
$$
\end{lemma} 

\begin{proof} We have 
$$
\frac{1}{(a_1y_1+...+a_ny_n)^n}=\frac{1}{(n-1)!}\int_0^\infty t^{n-1}e^{-(a_1y_1+...+a_ny_n)t}dt.
$$
So we get 
$$
\int_{\Delta_n}\frac{dy}{(a_1y_1+...+a_ny_n)^n}=\frac{1}{(n-1)!}\int_{\Delta_n}\int_0^\infty t^{n-1}e^{-(a_1y_1+...+a_ny_n)t}dtdy
$$
$$
=\frac{1}{(n-1)!}\int_{t\Delta_n}\int_0^\infty e^{-a_1z_1+...+a_nz_n}dtdz
$$
$$
=\int_{z_1,...,z_n\ge 0}e^{-a_1z_1+...+a_nz_n}dz=\prod_{j=1}^n\int_0^\infty e^{-a_jz_j}dz_j=\frac{1}{a_1...a_n}.
$$
\end{proof} 

Applying the Feynman famous formula to our integral and making a change of variable 
$q\mapsto q+(1-y)p$, we have 
$$
\int_V \frac{dq}{(q^2+m^2)((p-q)^2+m^2)}=\int_0^1\int_V \frac{dq}{((1-y)q^2+y(p-q)^2+m^2)^2}dy=
$$
$$
\int_0^1\int_V \frac{dq}{(q^2+M^2(y,p))^2}dy,
$$
where 
$$
M^2(y,p):=y(1-y)p^2+m^2.
$$
Now, using spherical coordinates
$$
\int_V \frac{dq}{(q^2+M^2)^2}=C_d\int_0^\infty \frac{r^{d-1}dr}{(r^2+M^2)^2},
$$
where $C_d$ is the area of the unit sphere in $\Bbb R^d$.
Thus for $d=2$
$$
\int_V \frac{dq}{(q^2+M^2)^2}=2\pi\int_0^\infty \frac{rdr}{(r^2+M^2)^2}=\pi \int_0^\infty\frac{ds}{(s+M^2)^2}=\frac{\pi}{M^2}.
$$
It follows that 
$$
\int_V \frac{dq}{(q^2+m^2)((p-q)^2+m^2)}=\pi\int_0^1 \frac{dy}{y(1-y)p^2+m^2}
$$
$$
=\frac{2\pi}{p^2\sqrt{\frac{4 m^2}{p^2} + 1}}{\rm arccotanh}\sqrt{\frac{4 m^2}{p^2} + 1}.
$$
The case $d=3$ can be computed similarly. 

However, for $d\ge 4$ we encounter our first difficulty: the integral diverges (as the integrand behaves at infinity as $|q|^{-4}$). More specifically, for a cutoff $\Lambda>0$,  
define 
$$
A_\Lambda(p):=\frac{g^2}{2(p^2+m^2)^2}\int_{|q|\le \Lambda} \frac{dq}{(q^2+m^2)((p-q)^2+m^2)},
$$
the integral over the ball in $V$ of radius $\Lambda$. 
Then 
$$
A_\Lambda(p)\sim \pi^2\frac{g^2}{(p^2+m^2)^2}\log(\tfrac{\Lambda}{m}),\ \Lambda\to \infty
$$
for $d=4$ and 
$$
A_\Lambda(p)\sim C_d\frac{g^2}{2(d-4)(p^2+m^2)^2}\Lambda^{d-4},\ \Lambda\to \infty
$$
if $d>4$. A way to remedy this difficulty is to add a $\Lambda$-dependent term in the Lagrangian, called a {\it counterterm},\index{counterterm} which blows up as $\Lambda\to \infty$ but which will cancel this divergence, in the sense that when integration is performed over the ball $|q|\le \Lambda$ then the integral has a finite limit as $\Lambda\to \infty$. 

For example, consider $d=4$. In this case modulo $g^3$ the momentum space 2-point function computed 
with cutoff $\Lambda$ looks like 
$$
\widehat G_{\Lambda,m^2}(p)=\frac{1}{p^2+m^2}+\pi^2\frac{g^2}{(p^2+m^2)^2}\log(\tfrac{\Lambda}{m})+...
$$
(here we explicitly indicate dependence of $\widehat G$ on $m^2$ since we are about to vary it). 
Let us try to fix the divergence by replacing the parameter $m^2$ by $m^2+Kg^2\log(\tfrac{\Lambda}{m})$ for a constant $K$. 
So we have 
$$
\widehat G_{\Lambda,m^2+Kg^2\log(\tfrac{\Lambda}{m})}(p)=\frac{1}{p^2+m^2+Kg^2\log(\tfrac{\Lambda}{m})}+\pi^2\frac{g^2}{(p^2+m^2)^2}\log(\tfrac{\Lambda}{m})+...
$$
$$
=\frac{1}{p^2+m^2}+(\pi^2-K)\frac{g^2}{(p^2+m^2)^2}\log(\tfrac{\Lambda}{m})+...
$$
where we ignore terms of order higher than $g^2$. Thus to cancel the divergence, we should take 
$K=\pi^2$, i.e., replace the Lagrangian with 
$$
\mathcal L_{E,\Lambda}:=\frac{1}{2}((d\phi)^2+(m^2+\pi^2g^2\log(\tfrac{\Lambda}{m}))\phi^2)+\frac{g}{6}\phi^3.
$$
For this Lagrangian, if integration is performed with cutoff $\Lambda$, then 
the 2-point function modulo $g^2$ will have a finite limit as 
$\Lambda\to \infty$, given by 
$$
\widehat G(p)=\frac{1}{p^2+m^2}+\frac{g^2}{2(p^2+m^2)^2}I(p),
$$
where 
$$
I(p)=\lim_{\Lambda\to \infty} \left(\int_{\Bbb R^4} \frac{dq}{(q^2+m^2)((p-q)^2+m^2)}-2\pi^2\log(\tfrac{\Lambda}{m})\right).
$$
This limit is easy to compute using the Feynman famous formula. Namely, computing similarly to the $d<4$ case, we 
get 
$$
I(p)=\int_0^1 I(p,y)dy,\ I(p,y):=\lim_{\Lambda\to \infty} \left(\int_0^\Lambda \frac{r^3dr}{(r^2+M^2(y,p))^2}-2\pi^2\log (\tfrac{\Lambda}{m})\right).
$$
So 
$$
I(p,y)=2\pi^2(\log m-\frac{1}{2}(1+\log(y(1-y)p^2+m^2))).
$$
Thus 
$$
I(p)=2\pi^2\left(\frac{1}{2}+\sqrt{\frac{4 m^2}{p^2} + 1}\cdot {\rm arccotanh}\sqrt{\frac{4 m^2}{p^2} + 1}\right).
$$

For $d>4$ the calculation becomes more elaborate. 
Namely, while for $d=5$ we have 
$$
A_\Lambda(p)\sim C_5\frac{g^2}{2(p^2+m^2)^2}\Lambda+O(1), \lambda\to \infty,
$$
so the procedure is the same, with mass parameter modification
$m^2\mapsto m^2+K\Lambda$, already for $d=6$ we will have to take a deeper expansion of 
the divergent integral: 
$$
A_\Lambda(p)\sim \frac{g^2}{4(p^2+m^2)^2}(C_6\Lambda^2+Cp^2\log(\tfrac{\Lambda}{m})+O(1)),\ \Lambda\to \infty
$$
We can cancel the most singular term $C_6\Lambda^2$ by mass  modification 
$m^2\mapsto m^2+K\Lambda^2$, but after that we will still have logarithmic divergence, 
$Cp^2\log(\tfrac{\Lambda}{m})$, which depends on $p$. To kill this divergence, we must modify 
the coefficient of $\frac{1}{2}(d\phi)^2$ in the Lagrangian by a counterterm, changing it from $1$
to $1+C'g^2\log(\tfrac{\Lambda}{m})$ for an appropriate constant $C'$. Also we find that the 1-loop correction 
to the 3-point function is logarithmically divergent:
the corresponding contribution in momentum presentation is, up to scaling,  
$$
\frac{g^3}{\prod_{j=1}^3(p_j^2+m^2)}J(p_1,p_2,p_3)\delta(p_1+p_2+p_3),
$$
where 
$$
J(p_1,p_2,p_3)=\int_V \frac{dq}{(q^2+m^2)((q-p_1)^2+m^2)((q-p_1-p_2)^2+m^2)}
$$
for $p_1+p_2+p_3=0$, which is divergent and behaves like $\log \Lambda$ when computed over the ball of radius 
$\Lambda$. So to kill this divergence, we must change the coefficient of $\frac{1}{6}\phi^3$ 
in the Lagrangian by a counterterm, changing it from $g$ to $g+C''g^3\log(\tfrac{\Lambda}{m})$. 

We are starting to see the main idea of {\it renormalization theory},\index{renormalization theory} which allows us to regularize divergent integrals coming from Feynman diagrams in all orders of perturbation series. This idea is that the coefficients 
of the Lagrangian are actually {\bf not} meaningful physical quantities, --- they are just mathematical parameters depending on the scale (cutoff) $\Lambda$ at which we are doing the computation, and may blow up when $\Lambda\to \infty$ (called the {\it ultraviolet limit},\index{ultraviolet limit} as $\Lambda$ has the meaning of frequency of oscillation). Rather, the meaningful quantity is the answer, the correlation functions $\la \phi(x_1)...\phi(x_n)\ra$ (or their Fourier transforms, if we work in the momentum realization). This answer depends on some parameters, which are the actual 
parameters of the theory. So the coefficients in the Lagrangian must be adjusted in such a way that 
the answer has a finite limit as $\Lambda\to \infty$. The specific answer we will get will depend on the adjustment procedure, but in good cases (called renormalizable) will lie in a nice universal family (often, but not always depending 
on finitely many parameters). 

\subsection{Super-renormalizable, renormalizable, and non-renormalizable theories} 
Let us discuss this more systematically. Consider a theory of a scalar boson with a general Lagrangian. Given a 
Feynman diagram $\Gamma$, we have the corresponding Feynman integral $I_\Gamma$ in momentum space realization, which is an integral of a rational volume form over a real vector space. We can define the {\it superficial degree of divergence}\index{superficial degree of divergence} $D(\Gamma)$ to be the degree of the numerator of this form (where the differentials of coordinates have degree $1$) minus the degree of its denominator. It is clear that if $D(\Gamma)\ge 0$ then the integral diverges. Note that the converse is false: if $d(\Gamma)<0$, the integral may still diverge. 

Let us compute $D(\Gamma)$. The degree of the denominator is easy to compute: 
it is just $2e(\Gamma)$ where $e(\Gamma)$ is the number of internal edges of $\Gamma$ 
(indeed, every edge contributes a propagator, which is the inverse of a quadratic function). 
On the other hand, the number of integrations over $V$ is the number of loops, i.e, 
$d(e(\Gamma)-v(\Gamma)+1)$, where $v(\Gamma)$ is the number of internal vertices. 
Finally, the terms in the Lagrangian containing derivatives of $\phi$ contribute the number of such derivatives 
to the degree of the numerator. It follows that 
$$
D(\Gamma)=(d-2)e(\Gamma)-dv(\Gamma)+d+N,
$$
where $N$ is the total number of derivatives in vertex monomials. In particular, when there are no derivatives, we have 
$$
D(\Gamma)=(d-2)e(\Gamma)-dv(\Gamma)+d.
$$
This shows that we may compute $D(\Gamma)$ as a sum of contributions over vertices, defining
the degree $D(\Phi)$ of a differential monomial $\Phi$ standing at a fully internal vertex 
(one whose all edges are internal) as the contribution of this vertex to $D(\Gamma)$. 
Indeed, every $\Phi$ contributes 
$$
D(\Phi)=\frac{d-2}{2}e(\Phi)-d+N_\Phi,
$$ 
where $e(\Phi)$ is the number of edges of $\Phi$ (i.e., its degree with respect to $\phi$) 
and $N_\Phi$ is the number of derivatives in $\Phi$. 

We see that a more natural invariant is 
$$
[\Phi]:=D(\Phi)+d,
$$
as it is multiplicative:
$$
[\Phi_1\Phi_2]=[\Phi_1][\Phi_2].
$$
This is not surprising since $\Phi$ comes with a volume factor $dx$, 
so $D(\Phi)$ is actually the scaling dimension of $\Phi dx$; thus 
to get the scaling dimension of $\Phi$, we need to add 
$d$ (as the scaling dimension of $dx$ is $-d$). 
This motivates 

\begin{definition} The number $[\Phi]$ is called the {\it classical scaling dimension}\index{classical scaling dimension} of the differential monomial $\Phi$.
\end{definition} 

Thus for a Feynman diagram $\Gamma$ we have 

\begin{equation}\label{degfo} 
D(\Gamma)=d-\frac{k(d-2)}{2}+\sum_\Phi D(\Phi),
\end{equation}
where $k$ is the number of external vertices of $\Gamma$.

For example, for $\Phi=\phi^n$ we get 
$$
D(\phi^n)=\tfrac{n}{2}(d-2)-d=(\tfrac{n}{2}-1)d-n,
$$
Each derivative adds a $1$ to the degree, so for instance 
$$
D(\phi^{n-2}(d\phi)^2)=(\tfrac{n}{2}-1)(d-2).
$$
So for the 1-loop Feynman diagram $\Gamma$ for the $k$-point function (a cycle with $k$ legs), we have 
$$
D(\Gamma)=d-\frac{k}{2}(d-2)+kD(\phi^3)=d-\frac{k}{2}(d-2)+\frac{k}{2}(d-6)=d-2k. 
$$

\begin{definition} Let $\Phi$ be a differential monomial in $\phi$. We will say 
that $\Phi$ is {\it super-renormalizable}\index{super-renormalizable theory} if 
$D(\Phi)<0$, {\it renormalizable}\index{renormalizable theory} (or critical)\index{critical theory} if $D(\Phi)=0$, and 
{\it non-renormalizable}\index{non-renormalizable theory} if $D(\Phi)>0$. 
\end{definition} 

Thus super-renormalizable terms improve convergence, renormalizable ones do not affect it, and 
non-renormalizable ones worsen it. 

\begin{example} 1. The kinetic term $(d\phi)^2$ has $D=0$, so is renormalizable;  
in fact, this is so by definition in any QFT. 
Note that this can be used to easily 
compute the classical scaling dimensions of monomials. Namely, 
we have $[(d\phi)^2]=D((d\phi)^2)+d=d$, so 
$2[\phi]+2=d$, i.e. $[\phi]=\frac{d-2}{2}$. 
Using multiplicativity, we now immediately compute 
$[\Phi]$ for any $\Phi$. 

2. The mass term $\phi^2$ has $D=-2$, so it is super-renormalizable. 
The term $\phi^3$ has $D=\frac{1}{2}d-3$, so it is super-renormalizable 
for $d<6$, renormalizable for $d=6$ and non-renormalizable for $d>6$. 
\end{example} 

\begin{definition} A Lagrangian is called 

$\bullet$ super-renormalizable if all its terms except the kinetic term are super-renormalizable; 

$\bullet$ renormalizable (or critical) if all its terms are at worst renormalizable and there is at least one renormalizable non-kinetic (i.e., interacting) term; 

$\bullet$ non-renormalizable if it contains non-renormalizable terms. 
\end{definition} 
 
 Clearly, every Lagrangian is of exactly one of these three types. 
 
\begin{proposition} 
(i) If a Lagrangian is super-renormalizable then the degree of superficial divergence of the corresponding Feynman diagrams is bounded above, and there are finitely many superficially divergent diagrams with any given number of external edges; moreover, if $d>2$ then there are finitely many superficially divergent diagrams altogether. 

(ii) If a Lagrangian is renormalizable, then there are infinitely many superficially divergent diagrams with a fixed number of external edges, but the degree of superficial divergence of these diagrams is still bounded above. 

(iii) If a Lagrangian is non-renormalizable, then the degree 
of superficial divergence of diagrams with a fixed number of external edges is unbounded above. 
\end{proposition} 

\begin{proof} This is clear from formula \eqref{degfo}.
\end{proof} 

This means that for a non-renormalizable Lagrangian, regularization of divergent 
integrals will definitely get out of control. Namely, if we want to regularize diagrams 
with unbounded above degree of superficial divergence, then we will have to introduce counterterms with 
unlimited number of derivatives, and our renormalized Lagrangian will no longer depend on a finite number of derivatives of $\phi$. 

On the other hand, if the Lagrangian is renormalizable, then for $d>2$ there are only finitely many 
terms that we will need to modify in the renormalization procedure; namely, these are the possible super-renormalizable and renormalizable terms in the Lagrangian. The fact that this procedure works to all orders of perturbation theory is a rather non-trivial fact which we will not prove here; but the result is 
a finite-parametric family of perturbative QFT. 

In two dimensions, there is an additional feature - 
there are infinitely many (super)renormalizable terms in the Lagrangian; but they all have at most two derivatives. 

Finally, in the super-renormalizable case the renormalization procedure is completed in finitely many steps.  

\subsection{Critical dimensions of some important QFT}\label{critthe} 

For interacting QFT defined by Lagrangians, the theory is only (super-)renormalizable in small dimensions,
and becomes non-renormalizable when dimension grows. If a theory is renormalizable in some dimension $d$
and non-renormalizable for bigger dimensions, we say that $d$ is the {\it critical dimension}\index{critical dimension}
of the theory. 

\subsubsection{Scalar bosons} 
For example, since $D(\phi^n)=(\tfrac{n}{2}-1)d-n$, for a scalar boson, 
a term $\phi^n$ is (super-)renormalizable  iff 
$d\le \frac{2n}{n-2}$. So in a (super-)renormalizable theory, 
the term $\phi^3$ can be present only for $d\le 6$, $\phi^4$ 
only for $d\le 4$, $\phi^5$ and $\phi^6$ only for $d\le 3$. 
Also, since $D(\phi^{n-2}(d\phi)^2)=(\tfrac{n}{2}-1)(d-2)$, such terms with $n>2$ cannot 
be present in a (super-)renormalizable theory unless $d=2$. With more derivatives things get even worse. 
So we obtain 

\begin{proposition} For the scalar bosonic field $\phi$, the most general (super-)renormalizable 
non-quadratic Poincar\'e-invariant Lagrangian is (up to scaling): 

$\bullet$ $d>6$: none;

$\bullet$: $d=5,6$: $\mathcal L=\frac{1}{2}(d \phi)^2+P_3(\phi)$;

$\bullet$: $d=4$: $\mathcal L=\frac{1}{2}(d \phi)^2+P_4(\phi)$;

$\bullet$: $d=3$: $\mathcal L=\frac{1}{2}(d \phi)^2+P_6(\phi)$;

$\bullet$: $d=2$: $\mathcal L=\frac{1}{2}g(\phi)(d \phi)^2+U(\phi)$,

where $P_m$ is a polynomial of degree $m$, and $U$ and $g$ are arbitrary 
(real analytic) functions. 
\end{proposition} 

Note that without loss of generality, one may assume that $P_m$ are missing 
the constant and linear terms. Thus the number of parameters for the theory 
with Lagrangian $\frac{1}{2}(d \phi)^2+P_m(\phi)$ is $m-1$ (the coefficients of $P_m$). 

\subsubsection{Fermions} 

Recall that for a fermionic field $\psi$ the kinetic term 
looks like $(\psi,\bold D\psi)$. This implies that 
$$
2[\psi]+1=d,
$$
i.e., 
$$
[\psi]=\frac{d-1}{2},
$$
which is always positive. So for mass terms $(\psi,M\psi)$ we have $D=-1$ and they are super-renomalizable. 
Beyond quadratic, we see that the only possibly (super-)renormalizable terms in $\psi$ for $d\ge 2$
are of the general shape $\psi^{2k}$, and 
$$
D(\psi^{2k})=2k[\psi]-d=k(d-1)-d=(k-1)(d-1)-1.
$$
The only case when this is (super-)renormalizable is $d=2$ and $k=2$, i.e., 
the term $\psi^4$, in which case $D=0$ (critical). Such terms indeed occur in the so-called {\it Gross-Neveu model}.\index{Gross-Neveu model}

For $d>2$, any fermionic term in a renormalizable Lagrangian must therefore be quadratic in the fermions. 
But it can contain other (bosonic) fields as factors. For example, 
$[\phi^n\psi^2]=n\frac{d-2}{2}+d-1$, so 
$$
D(\phi^n\psi^2)=n\frac{d-2}{2}-1.
$$
This shows that in 3 dimensions we can have a term 
$\phi\psi^2$ (Yukawa interaction) and $\phi^2\psi^2$, while in $4$ dimensions we can have 
only the Yukawa term $\phi\psi^2$, and for $d>4$ there are no possible (super-)renromalizable terms. 

\subsubsection{Gauge theory} 
A similar result holds when $\phi$ is vector-valued, i.e., has any number of components. This allows us to treat another important example, which is {\it gauge theory}.\index{gauge theory} 

Recall from Subsection \ref{manif} that to define a gauge theory, we fix a compact Lie group $G$ (for example, $U(n)$) and the field is a connection $\nabla$ on a principal $G$-bundle $P$ on $V$. Since all such bundles are trivial, we may think of $\nabla$ as a 1-form $A$ with values in $\g={\rm Lie}G$; i.e. $\nabla_A=d+A$. The curvature of $\nabla_A$ is given by the formula 
$$
F_A=dA+\frac{1}{2}[A,A],
$$
and the Lagrangian of the pure gauge theory is 
$$
\mathcal L:=\int_V |F_A|^2dx.
$$
As mentioned in Subsection \ref{manif}, 
he subtlety here is that $A$ is only considered up to gauge transformations
$\nabla_A\mapsto g^{-1}\nabla_Ag$, i.e., 
$A\mapsto g^{-1}dg+g^{-1}Ag$, where $g: V\to G$ 
is a smooth function with prescribed behavior at infinity, 
but this is irrelevant for the discussion of critical dimension. 

If $G$ is abelian (e.g. $G=U(1)$) then the Lagrangian is quadratic and this theory is free (this is 
the quantum electrodynamics without matter, i.e., quantization of Maxwell equations). 
This theory satisfies Wightman axioms in all dimensions, and its Wightman functions 
can be explicitly computed similarly to the case of scalar boson.
  
However, if $G$ is non-abelian (e.g. $G=SU(2)$ for weak interactions and $G=SU(3)$ 
for strong interactions in the standard model) then the Lagrangian is not quadratic and the equations of motion are not linear (they are the {\it Yang-Mills equations}).\index{Yang-Mills equations} Treating $A$ as a (vector-valued) boson, 
we see that the non-quadratic terms in the Lagrangian 
are of schematic form $A^2dA$ and $A^4$. The degrees of these terms are 
$\frac{1}{2}(d-4)$ and $d-4$, so we see that this theory is critical in dimension $4$ 
(the physical case!) and super-renormalizable in lower dimensions, but non-renormalizable for $d>4$. 
Note that the fact that we have a vector boson rather than 
a collection of scalar bosons (under the action of $\bold P$) does not matter for the dimension count. 

Note also that in $d\le 4$ dimensions we can also consider 
renormalizable Lagrangians with terms $(\nabla_A\phi)^2$ or $(\psi,\bold D_A\psi)$, where $\phi$ 
is a scalar and $\psi$ a spinor with values in the associated bundle $P\times_G \rho$, where $\rho$ is a finite dimensional 
representation of $G$ (it is easy to check that all occurring monomials have $D\le 0$). Such terms do occur in 
the standard model; the simplest case is $(\psi,\bold D_A\psi)$ where 
$A$ is a $U(1)$-connection and $\psi$ is a spinor valued in the 
tautological representation of $U(1)$, corresponding to an electron. 

\subsubsection{$\sigma$-model} 

The $\sigma$-{\it model}\index{$\sigma$-model} is a theory of a scalar boson taking values in a Riemannian manifold $M$. Thus the field is a map $\phi: V\to M$, and the Lagrangian is 
$\mathcal L=\frac{1}{2}(d\phi)^2$, which in local coordinates has the form 
$$
\mathcal L=\frac{1}{2}\sum_{i,j=1}^{\dim M}g_{ij}(\phi)d\phi^id\phi^j,
$$
where $g_{ij}$ is the Riemannian metric on $M$. 
We may also add a potential $U(\phi)$, where $U$ is a smooth function on $M$. 
By the above computations, this Lagrangian for a non-constant metric is renormalizable 
only in dimension $d=2$, but in this case $g_{ij}$ and $U$ can be arbitrary. 

\subsubsection{Gravity}

The theory of gravity (general relativity) is a theory of a bosonic field $h(x)$ 
taking values in symmetric tensors $S^2V^*$; i.e., the Minkowskian metric 
on $V$ is perturbed by setting 
$g=g_0+h$, where $g_0$ is the standard 
Minkowskian metric. The Lagrangian 
of general relativity is 
$$
\mathcal L=R(g)
$$
where $R$ is the scalar curvature of the metric $g$. Since curvature is expressed in terms of second derivatives of the metric, up to scaling this can be schematically written in terms of $h$ as  
$$
\mathcal L=(dh)^2+...
$$
where the dots stand for terms having at most two derivatives in $h$. Thus the general shape of this Lagrangian (for the purposes of computing classical scaling dimensions) is the same as for the $\sigma$-model; so this theory is only renormalizable in two dimensions. This is one of the main reasons why it has not yet been possible to incorporate gravity into the standard model, which lives in $4$ spacetime dimensions. 

\begin{remark} 
We have seen in Subsection \ref{compop} 
that even in a free quantum field theory, the composite operators 
like $\phi^2(x)$ are not automatically defined, and require a normal ordering procedure to regularize them. This is all the more so in an interacting QFT.  

It turns out that the normal ordering procedure, composite operators, and operator product expansion in a critical perturbative QFT can be defined analogously to the free case, using renormalization theory. We will not discuss it here and refer the reader to \cite{QFS}, vol. 1, p. 452. 
\end{remark}

\section{Two-dimensional conformal field theory} 

\subsection{Classical free massless scalar in two dimensions} \label{mlsb}
Consider a free massless scalar boson $\phi$ on $\Bbb R^2$ with Lagrangian $\mathcal L=\frac{1}{2}(d\phi)^2$. 
In this case the local functional $\phi(t,x)$ satisfies the 2-dimensional wave (=string) equation
$$
\phi_{tt}-\phi_{xx}=0, 
$$
so it splits into a sum of two functionals 
$$
\phi=\tfrac{1}{\sqrt{2}}\phi_L+\tfrac{1}{\sqrt{2}}\phi_R,
$$
where 
$$
\phi_L(t,x)=\psi_L(x+t),\ \phi_R(t,x)=\psi_R(x-t),
$$
which for obvious reasons are called the {\it left-mover and right-mover}.\index{left-mover}\index{right-mover}
In other words, we have 
$$
(\partial_t-\partial_x)\phi_L=0,\ (\partial_t+\partial_x)\phi_R=0.
$$
So we get 
$$
\phi_x+\phi_t=\sqrt{2}\psi_L'(x+t),\ \phi_x-\phi_t=\sqrt{2}\psi_R'(x-t). 
$$
So the Poisson bracket of $\psi_L',\psi_R'$ is given by 
$$
\lbrace \psi_L'(x),\psi_L'(y)\rbrace=\delta'(x-y),\ \lbrace \psi_R'(x),\psi_R'(y)\rbrace=-\delta'(x-y),
$$
$$
\lbrace\psi_L'(x),\psi_R'(y)\rbrace=0.
$$

Thus upon Wick rotation, which replaces $t$ with $it$ and makes 
$\phi$ complex-valued, setting $u:=x+it$, we have
$$
\overline\partial_u \phi_L=0,\ \partial_u \phi_R=0,
$$
i.e., $\phi_L=\psi_L(u)$ is holomorphic and $\phi_R=\psi_R(\overline u)$ is antiholomorphic.  

Now consider the case when $x$ runs over the circle $\Bbb R/2\pi \Bbb Z$, with Lebesgue measure normalized to have volume $1$. Then, if we still want to have a decomposition of $\phi$ into a left-mover and a right-mover, we should ``kill the zero mode" by requiring that $\int_0^{2\pi}\phi(t,x)dx=0$ (otherwise we have a solution $\phi(t,x)=t$ of the string equation which cannot be written as a sum of a left-moving and right-moving periodic wave). Then we may introduce the coordinate $z=e^{iu}$ which takes values in $\Bbb C^\times$, and $\phi_L,\phi_R$ become holomorphic, respectively antiholomorphic fields on $\Bbb C^\times$, which we'll denote by 
$\varphi$, $\varphi^*$. So we have Laurent expansions 
$$
\varphi(z)=\sum_{n\in \Bbb Z}\varphi_n z^{-n},\ \varphi^*(\overline z)=\sum_{n\in \Bbb Z}\varphi^*_n \overline z^{-n},
$$
with $\varphi_0=\varphi_0^*=0$. 
When $z$ is on the unit circle, these are just the Fourier expansions 
of $\phi_L(0,x)$, $\phi_R(0,x)$, and for 
$$
a(z)=\sum_{n\in \Bbb Z}a_nz^{-n-1}:=i\partial_z \varphi(z),\ a^*(\overline z)=\sum_{n\in \Bbb Z}a_n^*\overline z^{-n-1}:=-i\overline\partial_z\varphi^*(\overline z),
$$
where $a_0=a_0^*=0$, 
we have 
$$
za=\partial_u\phi_L=\psi_L'(u),\ \overline z a^*=\overline \partial_u\phi_R=\psi_R'(\overline u).
$$
Thus for $z=e^{iu},w=e^{iv}$ we get 
\begin{equation}\label{ftpb}
\lbrace za(z),wa(w)\rbrace=\delta'(u-v).
\end{equation}
Note that  
$$
\delta'(u-v)=i\sum_{n\in \Bbb Z}nz^{n}w^{-n}.
$$ 
So setting 
$$
\delta(w-z):=\sum_{n\in \Bbb Z}z^nw^{-n-1}
$$
(Fourier expansion of the distribution $\delta(w-z)$ on $(S^1)^2$, where 
$|z|=|w|=1$), we can write \eqref{ftpb} as 
\begin{equation}\label{ftpb1}
\lbrace a(z),a(w)\rbrace=-i\delta'(w-z).
\end{equation} 
In components, this takes the form
$$
\lbrace \sum_{m\in \Bbb Z}a_m z^{-m},\sum_{n\in \Bbb Z}a_{-n} w^{n}\rbrace=-i\sum_{n\in \Bbb Z}nz^{-n}w^n.
$$
Thus we get 
\begin{equation}\label{mlsb1}
\lbrace a_n, a_m\rbrace=-in\delta_{n,-m}.
\end{equation} 
Similarly, 
\begin{equation}\label{mlsb2}
\lbrace a_n^*,  a_m^*\rbrace=in\delta_{n,-m},
\end{equation}
and 
\begin{equation}\label{mlsb3}
\lbrace a_n, a_m^*\rbrace=0,
\end{equation} 
which in terms of generating functions can be written as 
\begin{equation}\label{ftpb2}
\lbrace a^*(z),a^*(w)\rbrace=i\delta'(w-z), \lbrace a(z),a^*(w)\rbrace=0.
\end{equation} 
Finally, let us write down the hamiltonian of the theory in terms of the Fourier (=Laurent) modes $a_n$. 
Recall that in the original notation it has the form 
$$
H=\frac{1}{2}\int_{\Bbb R/2\pi \Bbb Z} (\phi_t^2+\phi_x^2)dx.
$$
Thus we have 
$$
H=\frac{1}{4}\int_{\Bbb R/2\pi \Bbb Z} ((\overline z a^*-z a)^2+(\overline z a^*+za)^2)dx=\frac{1}{2}\int_{\Bbb R/2\pi \Bbb Z}(\overline z^2a^{*2}+z^2a^2)dx,
$$
i.e., 
\begin{equation}\label{classham}
H=\sum_{n>0}(a_{-n}a_n+a_{-n}^*a_n^*).
\end{equation}
It satisfies the relations 
$$
\lbrace a_m,H\rbrace=-ima_m,\ \lbrace a_m^*,H\rbrace=ima_m^*.
$$

\subsection{Free quantum massless scalar on 
$\Bbb R\times \Bbb R/2\pi \Bbb Z$ with killed zero mode}

Consider now the free QFT of a massless scalar boson $\phi$ on $\Bbb R\times \Bbb R/2\pi \Bbb Z$ with 
Minkowskian metric $dt^2-dx^2$, with killed zero mode, i.e., a quantization of 
the classical field theory described in Subsection \ref{mlsb}. Since this is not a theory 
on a vector space, it won't satisfy Wightman axioms. However, we can naturally quantize 
the commutation relations \eqref{mlsb1},\eqref{mlsb2},\eqref{mlsb3} (with $\hbar=1$), by replacing them with 
$$
[a_n,a_m]=n\delta_{n,-m},\ [a_n^*,a_m^*]=-n\delta_{n,-m},\ [a_n,a_m^*]=0.
$$
In other words, for $a(z)=\sum_{n\in \Bbb Z}a_nz^{-n-1}$, $a^*(\overline z)=\sum_{n\in \Bbb Z}a_n^*\overline z^{-n-1}$ we have 
$$
[a(z),a(w)]=\delta'(w-z),\ [a^*(z),a^*(w)]=-\delta'(w-z),\ [a(z),a^*(w)]=0,
$$
which quantize equations \eqref{ftpb1},\eqref{ftpb2} (this is a field-theoretic generalization of 
the analysis of Subsection \ref{harosc3}, with an infinite sequence of harmonic oscillators 
labeled by positive integers). Thus we see that the Euclidean space-locality property is satisfied. 

This shows that we have an infinite system of independent harmonic oscillators. 
To restate this algebraically, consider the infinite dimensional Heisenberg Lie algebra 
$\mathcal A$ with basis $a_n,n\ne 0$ and $K$ (central) with commutation relations 
$$
[a_n,a_m]=n\delta_{n,-m}K.
$$
Then we see that some dense subspace of the Hilbert space $\mathcal H$ 
of our theory should carry a pair of commuting actions of 
$\mathcal A$ (by left-movers and right-movers), 
with $K$ acting by $1$ and $-1$, respectively (we'll denote the second copy of $\mathcal A$ by $\mathcal A^*$). 

Let us now describe the Hilbert space $\mathcal H$. 
Note that the Lie algebra $\mathcal A$ has 
an irreducible {\it Fock representation}\index{Fock representation}
$\mathcal F$ generated by $\Omega$ with
defining relations 
$$
a_n\Omega=0,\ n>0,\quad K\Omega=\Omega.
$$ 
As a vector space, $\mathcal F$ is the {\it Fock space}\index{Fock space}
$$
\mathcal F=\Bbb C[X_1,X_2,...]
$$ 
(with $\Omega=1$), on which the operators 
$a_{-n}$ for $n>0$ act by multiplication by $X_n$ and $a_n$ 
act by $n\frac{\partial}{\partial X_n}$. 

Now, the hamiltonian of the system (which we rescale for convenience by a factor of $2$) should satisfy the commutation relations 
$$
[\widehat H,a_n]=-na_n,\ [\widehat H,a_n^*]=na_n^*.
$$
Thus we see that if we want the spectrum of $\widehat H$ to be bounded below and if $\Omega\in \mathcal H$ 
is the lowest eigenvector of $\widehat H$ then we must have 
$$
a_n\Omega=0,\ a_{-n}^*\Omega=0
$$
for $n>0$. But in this case the space $\mathcal D$ generated from 
$\Omega$ by the action of $a_n,a_n^*$ has to be the irreducible representation 
$\mathcal F\otimes \mathcal F^*$ of the Lie algebra $\mathcal A\oplus \mathcal A^*$, where 
$\mathcal F^*:=\Bbb C[X_1^*,X_2^*,...]$ with 
$a_n^*$ acting by multiplication by $X_n^*$ and $a_{-n}^*\mapsto n\frac{\partial}{\partial X_n^*}$ for $n>0$. 

Thus the space $\mathcal D$ is the tensor product of polynomial algebras $\Bbb C[X_j]$ 
and $\Bbb C[X_j^*]$. Each of the algebras $\Bbb C[X_j]$ carries a positive inner product 
with $X_j^n$ being an orthogonal basis and $||X_j^n||^2=j^nn!$, and similarly for $\Bbb C[X_j^*]$. 
This yields a positive inner product $\la,\ra$ 
on $\mathcal F,\mathcal F^*$ and $\mathcal D$, with respect to which $a_i^\dagger=a_{-i}$ and $a_i^{*\dagger}=a_{-i}^*$. 
The Hilbert space $\mathcal H$ is the completion of $\mathcal D$ with respect to $\la,\ra$.

This implies that the quantum Hamiltonian has to be given by the formula 
$$
\widehat H=\sum_{n>0}(a_{-n}a_n+a_n^*a_{-n}^*)+C
$$
obtained by quantizing the classical hamiltonian \eqref{classham} (note that the annihilation operators are written on the right 
to make sure the infinite sum makes sense). We may write $\widehat H$ as the sum of left-moving and right-moving parts: 
$$
\widehat H=\widehat H_L+\widehat H_R,
$$
where 
$$
\widehat H_L:=\sum_{n>0}a_{-n}a_n+\tfrac{C}{2},\ \widehat H_L:=\sum_{n>0}a_{n}^*a_{-n}^*+\tfrac{C}{2}.
$$

\subsection{$\zeta$-function regularization} 
At the moment it is not clear what the right value of $C$ should be. To answer this question, recall that the hamiltonian of a single harmonic oscillator is $z\partial_z+\frac{1}{2}$ acting on $\Bbb C[z]$. 
This suggests that the formula for $\widehat H$ should be 
$$
\widehat H=\sum_{n>0}(a_{-n}a_n+a_n^*a_{-n}^*+n)=\frac{1}{2}\sum_{n\ne 0} (a_{-n}a_n+a_n^*a_{-n}^*),
$$
which is a more symmetric and natural formula for quantization of $H$. This formula, however, does not make sense, 
since the series 
$$
1+2+3+...
$$ 
is divergent. We may, however, regularize it using 
{\it $\zeta$-function regularization}.\index{$\zeta$-function regularization} 

Namely, recall that the Riemann $\zeta$-function
is defined by the formula 
$$
\zeta(s)=\sum_{n=1}^\infty n^{-s}.
$$
It is well known that this function extends meromorphically to the entire complex plane with a unique (simple) pole 
at $s=1$ and satisfies the functional equation, which says that the function $\pi^{-\frac{s}{2}}\Gamma(\tfrac{s}{2})\zeta(s)$
is symmetric under the change $s\mapsto 1-s$:
$$
\zeta(1-s)=\pi^{\frac{1}{2}-s}\frac{\Gamma(\tfrac{s}{2})}{\Gamma(\tfrac{1-s}{2})}\zeta(s).
$$
Now, it is natural to define 
$$
C=1+2+3+...:=\zeta(-1).
$$
But the functional equation for $s=2$ implies 
that 
$$
\zeta(-1)=\frac{\pi^{-\frac{3}{2}}}{\Gamma(-\frac{1}{2})}\zeta(2)=-\frac{\pi^{-\frac{3}{2}}}{2\pi^{\frac{1}{2}}}\frac{\pi^2}{6}=-\frac{1}{12}.
$$
So from this point of view it is natural to set 
$$
C:=-\frac{1}{12}
$$

\begin{remark} Recall that for integer ${\rm g}\ge 1$
$$
\zeta(2{\rm g})=(-1)^{{\rm g}+1}2^{2{\rm g}-1}\frac{B_{2{\rm g}}}{(2{\rm g})!}\pi^{2{\rm g}}.
$$
So the functional equation for $\zeta$ implies that 
$$
\zeta(1-2{\rm g})=\pi^{\frac{1}{2}-2{\rm g}}\frac{\Gamma({\rm g})}{\Gamma(\frac{1}{2}-{\rm g})}\cdot (-1)^{{\rm g}+1}2^{2{\rm g}-1}\frac{B_{2{\rm g}}}{(2{\rm g})!}\pi^{2{\rm g}}=
-\frac{B_{2{\rm g}}}{2{\rm g}}.
$$
Thus the Harer-Zagier theorem can be interpreted as the statement that 
the Euler characteristic of the moduli space of curves of genus ${\rm g}$ is $\zeta(1-2{\rm g})$. 
In particular, for ${\rm g}=1$ we get the Euler characteristic of $SL_2(\Bbb Z)$, which is $-\frac{1}{12}$. 
\end{remark} 

\subsection{Modularity of the partition function} 
The value $-1/12$ turns out indeed to be the most natural value of $C$. To see this, let us 
return to Euclidean signature $|du|^2=dt^2+dx^2$ and put our theory on the 
complex torus $E=E_\tau=\Bbb C^\times/q^{\Bbb Z}\cong \Bbb R/2\pi T\Bbb Z\times \Bbb R/2\pi \Bbb Z$, where 
$$
T>0,\ \tau=iT,\ q=e^{2\pi i\tau}=e^{-2\pi T}\in (0,1).
$$
In this case, as we know from quantum mechanics, we should consider 
the partition function 
$$
Z(\tau):={\rm Tr}(e^{2\pi i\tau \widehat H}).
$$
Note that 
$$
\widehat H(P\otimes Q)=(\deg P+\deg Q+C)P\otimes Q,
$$
where $P\in \mathcal F$ and $Q\in \mathcal F^*$, and the degree is given by 
$$
\deg(X_n)=\deg(X_n^*)=n.
$$ 
Thus we have 
$$
Z(\tau)=\frac{e^{2\pi i\tau(C+\frac{1}{12})}}{\eta(\tau)^2},
$$
where 
$$
\eta(\tau):=q^{\frac{1}{24}}\prod_{n=1}^{\infty}(1-q^n)
$$
is the Dedekind $\eta$-function.
Now recall that $\eta(\tau)$ is a modular form of weight $\frac{1}{2}$, namely, 
$$
\eta(-\tfrac{1}{\tau})=\sqrt{-i\tau}\cdot\eta(\tau).
$$
So the partition function $Z$ has a nice modular property for a unique value of $C$, which is 
exactly $-\frac{1}{12}$. 
 
Let us explain why we should expect $Z(\tau)$ to have a modular property. 

For this, note that the Lagrangian of the theory 
$$
\mathcal L(\phi)=\frac{1}{4\pi}\int_E (d\phi)^2 =\frac{1}{4\pi}\int_E d\phi\wedge *d\phi
$$ 
is {\it conformally invariant},\index{conformal invariance} as it is written purely in terms of the Hodge *-operator which depends only on the conformal 
structure on $E$. The same is true for the equation of motion, which is the Laplace's equation $\Delta \phi=0$. 
In other words, our classical field theory is {\it conformal}.\index{conformal classical field theory} Thus we could hope that the corresponding quantum theory 
is conformal as well. This should mean that $Z(-\tfrac{1}{\tau})=Z(\tau)$, since the complex tori $E_{-\frac{1}{\tau}}$ and $E_\tau$ are conformally equivalent. 

This said, we note that this modular property is only satisfied up to a linear factor in $\tau$: 
in fact, we have 
$$
Z(-\tfrac{1}{\tau})=-i\tau Z(\tau). 
$$
This is because we have killed the zero mode, which we should, in fact, 
have included (after all, the space cycle in the torus $E_\tau$ is not $SL_2(\Bbb Z)$-invariant, hence neither is 
the condition that the integral of $\phi$ over this cycle vanishes). This is done in the next subsection. 

\subsection{Including the zero mode} 
 The zero mode corresponds to the periodic solutions 
$\phi(t,x)=\alpha+\mu t$ of the string equation ($\alpha,\mu\in \Bbb R$), which for nonzero $\mu$ cannot be split into 
a left-moving and right-moving periodic wave. So putting back the zero mode 
corresponds to replacing the Hilbert space $\mathcal H$ with $\mathcal H_{\rm full}:=\mathcal H\otimes L^2(\Bbb R)$, where $L^2(\Bbb R)$ is the Hilbert space of a quantum-mechanical free massless particle, and the Hamiltonian $\widehat H$ by  
$$
\widehat H_{\rm full}:=\widehat H+\widehat \mu^2,
$$
where $\widehat\mu$ is the quantum momentum operator for this quantum mechanical particle, acting on $L^2(\Bbb R)$ by multiplication by the momentum $\mu$. Thus we may write 
$$
\mathcal H_{\rm full}=\int_{\Bbb R}\mathcal H_\mu d\mu,
$$
with $\mathcal H_\mu=\mathcal F_\mu\otimes \mathcal F_\mu^*$
where $\mathcal F_\mu=\mathcal F$ but with $a_0=\mu$ instead of $a_0=0$, and similarly 
$\mathcal F_\mu^*=\mathcal F^*$ but with $a_0^*=\mu$ instead of $a_0^*=0$. 
Then we still have 
$$
\widehat H=\widehat H_L+\widehat H_R,
$$
where 
$$
\widehat H_L=\tfrac{1}{2}a_0^2+\sum_{n>0}a_{-n}a_n-\tfrac{1}{24},\ \widehat H_R=\tfrac{1}{2}a_0^{*2}+\sum_{n>0}a_{n}^*a_{-n}^*-\tfrac{1}{24}.
$$

According to Remark \ref{givesen}, the partition function of such a particle when time runs over $\Bbb R/L\Bbb Z$
is, up to scaling, $L^{-\frac{1}{2}}$. Thus the full partition function 
should be 
$$
\mathcal Z(\tau)=(-i\tau)^{-\frac{1}{2}}Z(\tau). 
$$
And then we have the genuine modular property: 
$$
\mathcal Z(-\tfrac{1}{\tau})=\mathcal Z(\tau).
$$

We note that the function $\mathcal Z(\tau)$ has a natural extension 
to arbitrary $\tau\in \Bbb C_+$  (not necessarily purely imaginary), which 
is just the path integral over a ``non-rectangular" complex torus $E_\tau$. 
To explain this, note that we have a natural action of 
the translation group $\Bbb R/2\pi \Bbb Z$ 
on our spacetime, hence we should expect its action 
on the Hilbert space $\mathcal H$. 
The infinitesimal generator $D$ of this group 
should satisfy the commutation relations 
$$
[D,a_n]=na_n,\ [D,a_n^*]=na_n^*
$$ 
(which differs from the corresponding relations for $\widehat H$ by the sign in the first relation). 
As $D\Omega=0$, it follows that 
$$
D(P\otimes Q)=(\deg P-\deg Q)P\otimes Q,
$$
i.e., 
$$
D=\widehat H_L-\widehat H_R.
$$ 
Let $s\in \Bbb R$ and $\tau:=iT+s$. Then a twisted version of the Feynman-Kac formula implies that 
given $s\in \Bbb R$, we have   
$$
Z(\tau)={\rm Tr}(e^{-2\pi T\widehat H}e^{2\pi isD})=|q|^{-\frac{1}{12}}{\rm Tr}(q^{\widehat H_L} \overline q^{\widehat H_R}),
$$
where $q=e^{-2\pi(T+is)}=e^{2\pi i\tau}$. Thus we still have 
$$
Z(\tau)=\frac{1}{|\eta(\tau)|^{2}}.
$$
Hence 
$$
\mathcal Z(\tau)=\frac{1}{\sqrt{{\rm Im}\tau} |\eta(\tau)|^2},
$$
which is a (real analytic) modular function for $SL(2,\Bbb Z)$, i.e., invariant under 
$\tau\mapsto \frac{a\tau+b}{c\tau+d}$ for $a,b,c,d\in \Bbb Z$, $ad-bc=1$. 
Thus here we have a genuine quantum conformal symmetry (as the moduli of complex tori $E_\tau$ 
is exactly $\Bbb C_+/SL_2(\Bbb Z)$). Indeed, this function is obviously symmetric 
under $\tau\mapsto \tau+1$, and we've seen that it is invariant under $\tau\mapsto -1/\tau$, but these two transformations generate $SL_2(\Bbb Z)$.  

\subsection{Correlation functions on the cylinder and torus} 
We may also consider correlation functions of the quantum fields $a$ and $a^*$. 
They are computed separately in $\mathcal F$ and $\mathcal F^*$ and 
can be easily found using representation theory. For example, 
we have $a_{n}a_{-n}\Omega=n\Omega$ for $n>0$, so the 2-point function is given by 
$$
\la \Omega, a(z)a(w)\Omega\ra=\sum_{n=1}^\infty nz^{-n-1}w^{n-1}=\frac{1}{(z-w)^2}.
$$
More precisely, the series converges only for $|w|<|z|$, but the function analytically continues 
to all $z\ne w$. Since our theory is free, the higher correlation functions are given by Wick's formula:

\begin{proposition}\label{corfu} We have 
$$
\la \Omega,a(z_1)....a(z_{2k})\Omega\ra=\sum_{\sigma\in \Pi_{2k}}\frac{1}{\prod_{j\in [1,2k]/\sigma}(z_j-z_{\sigma(j)})^2},
$$
and the $2k+1$-point correlation functions are zero. 
\end{proposition} 

We note that since $\mathcal F$ is generated by $\Omega$ as an $\mathcal A$-module, 
these functions determine $a(z)$ as a local operator (=quantum field). More generally, they determine the operators 
$a(z_1)...a(z_r)$ when $z_i\ne z_j$, which are symmetric in $z_1,..,z_r$ due to space locality. However, these operators are not well defined (have poles) on the diagonals $z_i=z_j$. 
 
\begin{exercise}
Give a direct algebraic proof of Proposition \ref{corfu}.
\end{exercise}  

\begin{exercise}
Compute the normalized 2-point correlation function of the quantum field $\widetilde a(z):=za(z)$ on the torus $E:=\Bbb R/2\pi T\Bbb Z\times \Bbb R/2\pi \Bbb Z$ in terms of theta functions. 

{\bf Hint.} This correlation function is given by 
$$
\frac{\la \widetilde a(z)\widetilde a(w)\ra_E}{\la \emptyset\ra_E}={\rm Tr}_{\mathcal F}(\widetilde a(z)\widetilde a(w)e^{-2\pi T\widehat H_L}).
$$
\end{exercise} 

\subsection{Infinitesimal conformal symmetry: the Virasoro algebra} 

We have already pointed out that the theory of a free massless scalar in two dimensions is classically conformally invariant and saw some manifestations of the fact that this invariance survives at the quantum level (modular invariance of the partition function on the torus). However, to study conformal symmetry systematically, 
we need to consider {\it infinitesimal conformal symmetry},\index{infinitesimal conformal symmetry} given by ``infinitesimal conformal mappings", i.e., holomorphic 
vector fields on $\Bbb C^\times$. 

For simplicity we consider polynomial vector 
fields $P(z)\partial_z$ where $P$ is a Laurent polynomial 
(this is sufficient since polynomial fields are dense in all holomorphic vector fields in an appropriate topology). 
Such vector fields form a Lie algebra called the {\it Witt algebra}\label{Witt algebra} (or {\it centerless Virasoro algebra}
in the physics literature), and we'll denote it by $W$. A convenient basis of $W$ is $\lbrace L_n=-z^{n+1}\partial_z, n\in \Bbb Z\rbrace$ which satisfies the commutation relations 
$$
[L_n,L_m]=(n-m)L_{m+n},\ m,n\in \Bbb Z.
$$

The Lie algebra $W$ acts by symmetries of the classical field theory of a free massless scalar, since 
its Lagrangian is conformally invariant. In fact, importantly, this action is only $\Bbb R$-linear and not $\Bbb C$-linear, which is a good thing - this means that we have an action of the complexification $W_{\Bbb C}=W\oplus W^*$, 
where $W^*$ is the Lie algebra of antiholomorphic vector fields; in other words, 
we have two commuting actions of $W$. 

If our theory is quantum-mechanically conformally invariant, then the Lie algebra $W\oplus W^*$ 
should act on the space $\mathcal D$ in a way compatible with the action of $\mathcal A\oplus \mathcal A^*$, i.e., 
so that 
$$
[L_n,a(z)]=z^{n+1}a'(z)+(n+1)z^na(z), 
$$
$$
[L_n^*,a^*(\overline z)]=\overline z^{n+1}a^{*\prime}(\overline z)+(n+1)\overline z^na_*(\overline z),
$$
$$
[L_n^*,a(z)]=[L_n,a^*(\overline z)]=0,
$$
or in components 
$$
[L_n,a_m]=-ma_{m+n},\ [L_m^*,a_n^*]=-ma_{m+n}^*,\ [L_n,a_m^*]=[L_n^*,a_m]=0.
$$
Is there such an action? To figure this out, first note that the operators $L_0, L_0^*$ satisfy the same commutation relations with $a,a^*$ as $\widehat H_L,-\widehat H_R$ respectively. Since $\mathcal D$ is an irreducible $\mathcal A\oplus \mathcal A^*$-module, this means that by Schur's lemma we must have
$$
L_0=\widehat H_L+C_L,\ L_0^*=-\widehat H_R+C_R
$$
for some constants $C_L,C_R$. This shows that $L_n$ has to shift 
the grading in $\mathcal F$ by $n$, and similarly for $L_n^*$ and $\mathcal F^*$. 

Now by analogy with the formula 
$$
L_0=\sum_{k\ge 1} a_{-k}a_k+{\rm const},
$$ 
define for $n\ne 0$
\begin{equation}\label{formvir}
L_n:=\frac{1}{2}\sum_{k\in \Bbb Z}a_{-k}a_{k+n}.
\end{equation} 
It is easy to check that this operator on $\mathcal F$ (and hence on
$\mathcal D=\mathcal F\otimes \mathcal F^*$) is well defined, and satisfies the desired commutation relations
$$
[L_n,a(z)]=-z^{n+1}a'(z)+(n+1)z^na(z),\ [L_n,a^*(z)]=0.
$$
Again using irreducibility of $\mathcal D$ and Schur's lemma, we
see that if the desired action of $W$ exists at all, then $L_n$ {\bf must} be given by formula \eqref{formvir} 
(note that here we can't add a constant since $L_n$ must shift the degree). So it remains to check 
if the constructed operators satisfy the commutation relations of $W$. 

First assume $n\ne -m$. In this case using the Jacobi identity, we see that 
the operator $[L_n,L_m]-(n-m)L_{m+n}$ commutes with $a,a^*$, so again by Schur's lemma 
it must be a constant; however, since it shifts degree, we get the desired relation 
$$
[L_n,L_m]-(n-m)L_{m+n}=0.
$$
So it remains to consider the case $n=-m>0$. In this case the same argument shows that 
$$
[L_n,L_{-n}]-2nL_0=C(n), 
$$
where $C(n)\in \Bbb C$, and we have an action of $W$ if 
$C(n)=0$ for all $n$. So let us compute $C(n)$. To this end, note that the eigenvalue by which 
$[L_n,L_{-n}]$ acts on $\Omega$ is $2nC_L+C(n)$. So it suffices to compute 
this eigenvalue, i.e., the vector 
$L_nL_{-n}\Omega$. 

In terms of the polynomial realization, we have 
$$
L_{-n}\Omega=\tfrac{1}{2}\sum_{0<j<n}X_jX_{n-j}.
$$
Thus 
$$
L_nL_{-n}\Omega=\tfrac{1}{4}\sum_{0<j<n}j(n-j)\tfrac{\partial^2}{\partial X_j\partial X_{n-j}}\sum_{0<j<n}X_jX_{n-j}=
\tfrac{1}{2}\sum_{0<j<n}j(n-j)=\frac{n^3-n}{12}.
$$
So 
$$
C(n)=\frac{n^3-n}{12}.
$$
Thus we see that we almost have an action of $W$, but not quite - no matter how we choose $C_L$, the cubic term in $n$ will be present (a quantum anomaly)! Instead, we have a {\it projective} representation of $W$, which is, in fact, a representation of a {\it central extension}\index{central extension} of $W$. Such projective actions are, in fact, common in quantum mechanics, since quantum states correspond not to actual unit vectors in the space of states, but rather to vectors up to a phase factor, on which (as well as on quantum observables) there is a genuine action of the symmetry group. Prototypical examples of this are the {\it Heisenberg uncertainty relation} $[\widehat p,\widehat x]=-i\hbar$, when the classical 2-dimensional group (or Lie algebra) of translations of the phase plane is replaced in quantum theory 
by the 3-dimensional Heisenberg group (Lie algebra), and the 
{\it phenomenon of spin}, when the classical rotational symmetry group $SO(3)$ 
is replaced in quantum theory by its double cover $SU(2)$. 

This motivates the following definition. 

\begin{definition} The {\it Virasoro algebra}\index{Virasoro algebra} is the 1-dimensional central extension 
of the Witt algebra $W$ with basis $L_n,n\in \Bbb Z$ and 
$C$ (a central element) with commutation relations 
$$
[L_n,L_m]=(n-m)L_{m+n}+\frac{n^3-n}{12}\delta_{n,-m}C.
$$
\end{definition} 

Thus we have a 1-dimensional central ideal $\Bbb CC\subset {\rm Vir}$ spanned by $C$, and 
${\rm Vir}/\Bbb CC\cong W$. 

So we obtain 

\begin{theorem} The formulas 
$$
L_0=\sum_{k\ge 1}a_{-k}a_k,\ L_n=\tfrac{1}{2}\sum_{k\in \Bbb Z}a_{-k}a_{k+n},n\ne 0
$$
define an action of ${\rm Vir}$ on $\mathcal F$ with $C$ acting by $1$. 
\end{theorem} 

It is easy to check that the same theorem holds more generally on the space $\mathcal F_\mu$ where $a_0=\mu$.
The only change is that $L_0$ acquires an additional summand $\tfrac{1}{2}\mu^2$:
$$
L_0=\tfrac{1}{2}\mu^2+\sum_{k\ge 1}a_{-k}a_k. 
$$

If $C$ acts on a representation $\Bbb V$ of ${\rm Vir}$ by a scalar $c$ (as it will, for instance, on every irreducible representation) then one says that $\Bbb V$ has {\it central charge $c$}.\index{central charge}  Thus $\mathcal F_\mu$ 
is a representation of ${\rm Vir}$ of central charge $c=1$. 

Similarly, the formulas 
$$
L_0^*=-\tfrac{1}{2}\mu^2-\sum_{k\ge 1}a_{k}^*a_{-k}^*,\ L_n^*=-\tfrac{1}{2}\sum_{k\in \Bbb Z}a_{k}^*a_{-k+n}^*,n\ne 0
$$
define an action of ${\rm Vir}$ on $\mathcal F_\mu^*$ with the central element $C^*$ acting by $-1$ (i.e., of central charge $c=-1$).

Thus we obtain two commuting projective actions of $W$ on the space $\mathcal D=\mathcal F\otimes \mathcal F^*$
which define usual linear actions only for the central extension ${\rm Vir}$ 
of $W$. Still, the corresponding adjoint action of $W$ on quantum observables 
is a genuine linear action, so this quantum field theory is {\it conformal}.\index{conformal quantum field theory}

We note that the Virasoro action preserves the positive Hermitian form on $\mathcal F_\mu$ in the sense that 
$$
L_n^\dagger=L_{-n}.
$$
Thus $\mathcal F_\mu$ is a {\it positive energy unitary representation of} ${\rm Vir}$ (positive energy means that 
$L_0$ is diagonalizable with spectrum bounded below). 

More generally, we may consider the theory of $\ell$ massless scalars $\phi_1,...,\phi_\ell$. 
In this case $\mathcal D=\mathcal F^{\otimes \ell}\otimes \mathcal F^{*\otimes \ell}$, 
and $\mathcal F^{\otimes \ell}$ is a positive energy unitary ${\rm Vir}$-module with central charge $c=\ell$ (the tensor product of 
$\ell$ copies of $\mathcal F$). 

\begin{exercise} 1. Show that ${\rm Vir}$ is a non-trivial central extension of $W$ (i.e., not isomorphic to $W\oplus \Bbb C$ as a Lie algebra).

2. Show that ${\rm Vir}$ is a universal central extension of $W$, i.e., every non-trivial central extension of $W$ by $\Bbb C$ 
is isomorphic to ${\rm Vir}$.
\end{exercise} 

\subsection{Normal ordering, composite operators and operator product expansion in conformal field theory} 

Let us now summarize the theory of normal ordering, composite operators and operator product expansion from Subsection \ref{compop} in the case of conformal field theory, for the running example of a quantum massless scalar boson. We have seen that the operator product $a(z)a(w)$ is well defined only if $w\ne z$ and has a pole when $w=z$, leading to the local operator $a(z)^2$ not being well defined. So let us expand this operator product in a Laurent series near $w=z$ and identify the singular part involving negative powers of $w-z$. 
For this purpose consider the difference 
$$
:a(z)a(w):=a(z)a(w)-\frac{1}{(z-w)^2}.
$$
The formula for the correlation functions for $a(z)$ implies that 
$$
\la \Omega,a(z_1)...a(z_{i-1}) :a(z_i)a(z_{i+1}): a(z_{i+2})...a(z_n)\Omega\ra=
$$
$$
\sum_{\sigma\in \Pi_{2k}: \sigma(i)\ne i+1}\frac{1}{\prod_{j\in \Pi_{2k}/\sigma}(z_j-z_{\sigma(j)})^2}.
$$
Note that this function is regular at $z_i=z_{i+1}$, hence the operator 
$:a(z)a(w):$ is regular at $z=w$, i.e., defined for all $z,w\in \Bbb C^\times$. 
This operator is called the {\it normally ordered product}\index{normally ordered product} of $a(z)$ and $a(w)$. 
In particular, although the square $a(z)^2$ is not defined, we have a well defined 
normally ordered square $:a(z)^2:$. 

In terms of Laurent coefficients, 
$$
:a(z)a(w):=\sum_{m,n\in \Bbb Z}:a_na_m:z^{-n-1}w^{-m-1},
$$
where $:a_na_m:=a_na_m$ if $m\ge n$ 
and $:a_na_m:=a_ma_n$ if $m<n$ (normal ordering of modes). 
Of course, this ordering only matters if $m+n=0$. 
In particular, we see that   
$$
\tfrac{1}{2}:a(z)^2:=T(z):=\sum_{n\in \Bbb Z}L_nz^{-n-2},
$$
the generating function of the Virasoro modes $L_n$. 
This operator is called the {\it (quantum) energy-momentum tensor}.\index{energy-momentum tensor}

Thus we see that the Virasoro modes $L_n$ may be viewed as Noether charges
for the corresponding infinitesimal conformal symmetries, in the holomorphic sector of the theory. 
The corresponding Noether currents are $z^{n+1}T(z)$, as 
$$
L_n=\frac{1}{2\pi i}\oint z^{n+1}T(z)dz.
$$
The Noether charges for the full theory are then $L_n+ \overline L_n$, with currents $z^{n+1}T(z)+\overline{z^{n+1}T(z)}$. 
In particular, the Hamiltonian $H$, up to adding a constant, is 
$L_0+\overline{L_0}$, which agrees with formula \eqref{quanham}.

Similarly, we may define the normal ordered products of more than two factors,  
$:a(z_1)....a(z_n):$. This can be done by induction in $n$. Namely, we have 
\begin{equation}\label{indu}
:a(z_0)a(z_1)...a(z_n):=a(z_0):a(z_1)...a(z_n):-\sum_{k\in [1,n]}\frac{:\prod_{j\ne k}a(z_j):}{(z_0-z_k)^2}
\end{equation}
It is easy to see that 
the operator $:a(z_1)....a(z_n):$ has no singularities and is well defined for all values $z_1,...,z_n\in \Bbb C^\times$. 
Thus for every $r_1,..,r_n$ we have the operator 
$$
:a^{(r_1)}(z_1)...a^{(r_n)}(z_n):=\partial_{z_1}^{r_1}...\partial_{z_n}^{r_n}:a(z_1)....a(z_n):
$$
Setting $z_1=...=z_n$, we can then define the local operator $:P(a)(z):$
for any differential polynomial $P$ in $a(z)$. This local operator, called a {\it composite operator},\index{composite operator} is a quantization 
of the corresponding local functional $P(a)(z)$ in classical field theory. 

\begin{exercise} (The state-operator correspondence) 
Show that the map $P\mapsto P(a)(z)\Omega|_{z=0}$ 
is well defined and gives an isomorphism between the space 
$\mathcal V$ of (polynomial) local operators and the Fock space $\mathcal F$. 
\end{exercise} 

More generally, repeatedly using \eqref{indu}, 
we have 
\scriptsize
$$
:a(z_1)...a(z_n):\cdot :a(w_1)...a(w_m):=
\sum_{I\subset [1,n], J\subset [1,m], s: I\cong J}
\frac{:\prod_{i\notin I}a(z_i)\prod_{j\notin J}a(w_j):}{\prod_{i\in I}(z_i-w_{s(i)})^2}.
$$
\normalsize So setting $z_i=z,w_j=w$, we obtain 
$$
:a(z)^n::a(w)^m:=\sum_{k=0}^{\min(m,n)}k!\binom{n}{k}\binom{m}{k}\frac{:a(z)^{n-k}a(w)^{m-k}:}{(z-w)^{2k}}.
$$
E.g. for $n=m=1$ we get the familiar identity
$$
a(z)a(w)=\frac{1}{(z-w)^2}+:a(z)a(w):=\frac{1}{(z-w)^2}+\text{regular terms}.
$$

More generally, for $n=1$ and any $m$ we get 
$$
a(z):a(w)^m:=\frac{m:a^{m-1}(w):}{(z-w)^2}+:a(z)a(w)^m:
$$
$$
=\frac{m:a^{m-1}(w):}{(z-w)^2}+\text{regular terms}.
$$
For $m=2$ this can be written as 
$$
a(z)T(w)=\frac{a(w)}{(z-w)^2}+\text{regular terms},
$$
which encodes the commutation relations between $a_i$ and $L_j$. 

For $n=2$, $m=2$ we get 
$$
:a(z)^2::a(w)^2:=\frac{2}{(z-w)^4}+\frac{4:a(z)a(w):}{(z-w)^2}+:a(z)^2a(w)^2:=
$$
$$
\frac{2}{(z-w)^4}+\frac{4:a(w)^2:}{(z-w)^2}+
\frac{4:a(w)a'(w)}{z-w}+\text{regular terms}.
$$
This can also be written as 
$$
T(z)T(w)=\frac{1}{2(z-w)^4}+\frac{2T(w)}{(z-w)^2}+\frac{T'(w)}{z-w}+\text{regular terms},
$$ 
which encodes the commutation relations between $L_i$. More generally, at central charge $c$ 
this relation would look like 
$$
T(z)T(w)=\frac{c}{2(z-w)^4}+\frac{2T(w)}{(z-w)^2}+\frac{T'(w)}{z-w}+\text{regular terms}.
$$ 

These are the simplest examples of the {\it operator product expansion}.\index{operator product expansion} In fact, we have the following theorem, whose 
proof we will leave to the reader: 

\begin{theorem}\label{opethm}
For any local operators $P,Q\in \mathcal V$, there exist a unique finite sequence of 
local operators $R_{1},...,R_{N}\in \mathcal V$ such that 
$$
P(a)(z)Q(a)(w)=\sum_{j=1}^{N}R_j(a)(w)(z-w)^{-j}+\text{\rm regular terms},
$$
where $(z-w)^{-j}:=\sum_{k\ge 0} \binom{k+j-1}{j-1}z^{-j-k}w^k$. 
\end{theorem} 

Note that the space locality property implies that $Q(a)(w)P(a)(z)$ is given by the same formula, but with $(z-w)^{-j}$ expanded in the opposite direction, i.e., 
$(z-w)^{-j}:=-\sum_{k<0} \binom{k+j-1}{j-1}z^{-j-k}w^k$. 
Thus. we have 
$$
[P(a)(z),Q(a)(w)]=\sum_{j=1}^{N}\frac{1}{(j-1)!}R_j(a)(w)\delta^{(j-1)}(w-z).
$$
Thus Theorem \ref{opethm} gives us information about commutators between the modes 
of $P$ and $Q$. For example, as we have seen above, 
$$
[a(z),a(w)]=\delta'(w-z), 
$$
and also 
$$
[a(z),T(w)]=a(w)\delta'(w-z),
$$
$$
[T(z),T(w)]=\frac{c}{12}\delta'''(w-z)+2T(w)\delta'(w-z)+T'(w)\delta(w-z),
$$
where in our example $c=1$. 

Moreover, it is clear that one can uniquely continue the expansion of Theorem \ref{opethm} to also include terms of nonnegative degree; namely, we simply need to expand the regular terms into a Taylor series with respect to $z-w$ for fixed $w$. For example, we have 
an asymptotic expansion 
$$
a(z)a(w)\sim \frac{1}{(z-w)^2}+\sum_{k=0}^\infty :a^{(k)}(w)a(w):\frac{(z-w)^k}{k!}
$$
So in general we have 
$$
P(a)(z)Q(a)(w)\sim \sum_{j=-\infty}^{N}R_j(a)(w)(z-w)^{-j}.
$$

This formula is called the {\it operator product expansion} of the product of $P$ and $Q$. 
The operator product expansion satisfies certain axioms, 
which means that it defines on the space $\mathcal V\cong \mathcal F$ an algebraic structure called a 
{\it vertex algebra}\index{vertex algebra} (which we will not discuss here, however). 

\subsection{Vertex operators} 

{\it Vertex operators}\index{vertex operator} are obtained by quantizing the local functional $e^{i\lambda \varphi(z)}$, where 
$$
\varphi(z)=-i\int a(z)dz=-i(a_0\log z+\sum_{n\ne 0}\tfrac{a_{-n}}{n}z^{n}+a_0^\vee)
$$
and $a_0^\vee$ is a constant of integration (dual variable to $a_0$). In other words, we have 
$$
e^{i\lambda \varphi(z)}=e^{\lambda \int a(z)dz}=e^{\lambda(a_0\log z+\sum_{n\ne 0}\frac{a_{-n}}{n}z^{n})}e^{\lambda a_0^\vee}.
$$
A natural quantization of this functional is the operator  
$$
X(\lambda,z):=:e^{\lambda(a_0\log z+\sum_{n\ne 0}\frac{a_{-n}}{n}z^n)}:e^{\lambda a_0^\vee}=
$$
$$
=e^{\lambda\sum_{n>0}\frac{a_{-n}}{n}z^n}e^{-\lambda\sum_{n>0}\frac{a_{n}}{n}z^{-n}}
z^{\lambda\mu} e^{\lambda\partial_\mu},
$$
which, due to the last factor, acts from $\mathcal F_\mu$ to $\mathcal F_{\mu+\lambda}$ 
by $X_0(\lambda,z)z^{\lambda\mu}$, where 
$$
X_0(\lambda,z):=e^{\lambda\sum_{n>0}\frac{a_{-n}}{n}z^n}e^{-\lambda\sum_{n>0}\frac{a_{n}}{n}z^{-n}}.
$$
Here we work over the group algebra of $\Bbb C$ with basis $z^\alpha,\alpha\in \Bbb C$. 

Now note that if $[A,B]$ commutes with $A,B$ then by the Campbell-Hausdorff formula 
$$
e^A e^B=e^B e^A e^{[A,B]}, 
$$
and that 
$$
[\sum_{n>0}\frac{a_{n}}{n}z^{-n},\sum_{n>0}\frac{a_{-n}}{n}w^{n}]=\sum_{n>0}\frac{z^{-n}w^{n}}{n}=-\log(1-\tfrac{w}{z}).
$$
Thus 
$$
X_0(\lambda,z)X_0(\nu,w)=(1-\tfrac{w}{z})^{\lambda \nu}:X_0(\lambda,z)X_0(\nu,w):
$$
for $|w|<|z|$. So we get
$$
X(\lambda,z)X(\nu,w)=(z-w)^{\lambda \nu}:X(\lambda,z)X(\nu,w):
$$
for $|w|<|z|$, where the normal ordering puts $\partial_\mu$ to the right of $\mu$. 
More generally, we see that 
$$
X(\lambda_1,z_1)...X(\lambda_n,z_n)=\prod_{1\le j<k\le n}(z_j-z_k)^{\lambda_j \lambda_k} :X(\lambda_1,z_1)...X(\lambda_n,z_n):
$$
for $|z_1|>...>|z_n|$. 
In particular, denoting the highest weight vector of $\mathcal F_\mu$ by $\Omega_\mu$, we have  
$$
\la \Omega_{\mu+\lambda}, X(\lambda_1,z_1)...X(\lambda_n,z_n)\Omega_\mu \ra=
\prod_{j=1}^n z_j^{\lambda_j\mu}\prod_{1\le j<k\le n}(z_j-z_k)^{\lambda_j \lambda_k}.
$$
for $|z_1|>...>|z_n|$.

We see that this correlation function admits analytic continuation to the complement of the diagonals $z_i\ne z_j$, but this continuation is not, in general, single valued. In other words, the fields $X(\lambda,z)$ in general do not satisfy space locality. Instead, we have 
\begin{equation}\label{commrel}
X(\lambda,z)X(\nu,w)=e^{\pi i\lambda\nu}X(\nu,w)X(\lambda,z),
\end{equation} 
which is understood in the sense of analytic continuation along a path where $v:=w/z$ passes from the region $|v|<1$ to the region $|v|>1$ along positive reals, avoiding the point $v=1$ from above. 
In particular, 
$$
X(\lambda,z)X(\lambda,w)=e^{\pi i\lambda^2}X(\lambda,w)X(\lambda,z),
$$
i.e., $X(\lambda,z)$ has ``statistics $\lambda^2/2$" (where 
statistics $\alpha\in \Bbb R/ \Bbb Z$ means 
that switching the order produces a phase factor 
$e^{2\pi i\alpha}$; e.g. statistics 
$0$ corresponds to bosons and statistics $1/2$ to 
fermions).

Note that if we apply commutation relation \eqref{commrel} 
twice, we obtain a multiplier 
$e^{2\pi i\lambda\nu}$, which corresponds to the fact that the operator product $X(\lambda,z)X(\nu,w)$ is multivalued in general. 

This is an example of appearance of a {\it braiding}\index{braiding} 
in conformal field theory. Namely, relation \eqref{commrel} 
is called {\it braided space-locality}\index{braided space-locality} (or {\it braided commutativity}),\index{braided commutativity} since it 
can be viewed as commutativity in a suitable braided monoidal category. 

Note also that 
$$
X'(\lambda,z)=\lambda :a(z)X(\lambda,z):
$$
where $X':=\partial_z X$, and 
$$
[a_n,X(\lambda,z)]=\lambda z^nX(\lambda,z).
$$
Hence 
$$
[L_n,X(\lambda,z)]=z^{n+1}X'(\lambda,z)+\frac{\lambda^2}{2}
(n+1)z^nX(\lambda,z),
$$
which implies that $X(\lambda,z)$ has spin $\lambda^2/2$. 
Thus we have the spin-statistics property for $X(\lambda,z)$, which generalizes the usual one: spin modulo $\Bbb Z$ equals statistics. 

As noted in Remark \ref{anyons}, such quantum fields are called ``anyons" (as they can have any spin and statistics) and can exist only in two dimensions. The most general 
spin-statistics property for these anyons 
says that if $X,Y$ are anyons of spins $s_X,s_Y\ge 0$ then 
$$
X(z)Y(w)=e^{2\pi i \sqrt{s_Xs_Y}}Y(w)X(z).
$$
In particular, we see that if $\lambda^2\in \Bbb Z$ is odd then $X(n\lambda,z)$ 
behave like fermions for odd $n$ and like bosons for even $n$ with respect to each other (i.e., the corresponding operators $X(n_1\lambda,z)$ and 
$X(n_2\lambda,z)$ commute if $n_1n_2$ is even 
and anticommute if $n_1n_2$ is odd), while for even $\lambda^2$ they all behave like bosons (i.e., the operators commute). 

\subsection{The circle-valued theory} 

Now consider the theory of a massless scalar on $\Bbb C^\times$ 
with values in the circle $\Bbb R/2\pi r\Bbb Z$. This theory is the same 
as the line-valued one, except for the zero mode, which 
entails the following circle-valued solutions of the string equation: 
$$
\phi(t,x)=\alpha+\mu t+Nrx,
$$
where $\alpha\in \Bbb R/2\pi r\Bbb Z$, $\mu\in \Bbb R$, and 
$N$ is an integer (the winding number). The space of such solutions 
is a disjoint union of cylinders $T^*S^1$ labeled by values of $N$. 
So in quantum theory we get the Hilbert space
$$
\mathcal H_r^\circ=\bigoplus_{N,\ell\in \Bbb Z}\mathcal H_r^\circ(N,\ell), 
$$
where 
$\mathcal H_r^\circ(N,\ell)$ is the completion of 
$\mathcal F_{\frac{1}{\sqrt{2}}(\ell r^{-1}+Nr)}\otimes \mathcal F_{\frac{1}{\sqrt{2}}(\ell r^{-1}-Nr)}^*$. Thus we obtain the following formula for the partition function on the torus $E_\tau$: 
$$
\mathcal Z_r^\circ(\tau)=|\eta(\tau)|^{-2}\vartheta_r(\tau,\overline \tau),
$$
where
$$
\vartheta_r(\tau,\overline \tau):=\sum_{\ell,N\in \Bbb Z}e^{\frac{1}{2}\pi i\tau (\ell r^{-1}+Nr)^2-\frac{1}{2}\pi i\overline \tau (\ell r^{-1}-Nr)^2}=
$$
$$
\sum_{\ell,N\in \Bbb Z}e^{-\pi (\ell^2 r^{-2}+N^2r^2){\rm Im}\tau+2\pi i\ell N {\rm Re}\tau}.
$$
This shows an interesting duality $\mathcal Z_r^\circ(\tau)=\mathcal Z_{r^{-1}}^\circ(\tau)$; in fact, 
we see that the whole theory with parameter $r$ is equivalent to the one with parameter $r^{-1}$. 
This duality is called {\it T-duality},\index{$T$-duality} and it plays an important role in string theory. 

Also we note that $\vartheta_r$ is a real modular form of weight $1$:
$$
\vartheta_r(-\tfrac{1}{\tau},-\tfrac{1}{\overline \tau})=|\tau| \vartheta_r(\tau,\overline \tau),
$$
which leads to modular invariance of the function $\mathcal Z_r(\tau)$, as expected in a conformal field theory. 
To see this, it is enough to note that in the exponential we have a quadratic form on $\Bbb Z^2$ with matrix 
$$
Q(\tau)=\begin{pmatrix} r^2{\rm Im\tau} & -i{\rm Re}\tau\\  -i{\rm Re \tau} & r^{-2}{\rm Im \tau}\end{pmatrix} 
$$
So 
$$
Q(\tau)^{-1}=|\tau|^{-2}\begin{pmatrix} r^{-2}{\rm Im\tau} & i{\rm Re}\tau\\  i{\rm Re \tau} & r^{2}{\rm Im \tau}\end{pmatrix}=
\begin{pmatrix} r^{-2}{\rm Im\tau'} & -i{\rm Re}\tau'\\  -i{\rm Re \tau'} & r^{2}{\rm Im \tau'}\end{pmatrix}=SQ(\tau')S,
$$
where $\tau':=-\frac{1}{\tau}$ and $S=\begin{pmatrix} 0& 1\\ 1& 0\end{pmatrix}$. Thus the result follows from the Poisson summation formula. 

We see that if $r^2=\frac{p}{q}\in \Bbb Q$ (in lowest terms), 
this conformal field theory has a special property called {\it rationality}:\index{rational conformal field theory}  the Hilbert space $\mathcal H_r^\circ$ is the completion of a finite sum 
of ``sectors" $\oplus_{i=1}^n \mathcal V_i\otimes \mathcal V_i^*$, where 
the left-moving fields act on $\mathcal V_i$ and 
right-moving ones in $\mathcal V_i^*$, so that 
$\vartheta_r(\tau,\overline \tau)$ and hence $\mathcal Z_r^\circ(\tau)$ 
are finite sums of products of a holomorphic
and an antiholomorphic function (in fact, it is easy to see that $n=2pq$). 
For example, 
the vacuum vector $\Omega$ is contained in
the tensor product $\mathcal V(pq)\otimes \mathcal V(pq)^*$
where for $s\in \Bbb Z_{>0}$ we defined
$\mathcal V(s):=\oplus_{m\in \Bbb Z}\mathcal F_{m\sqrt{2s}}$. 
The space $\mathcal V(s)$ is a vertex algebra 
called the {\it lattice vertex algebra}\label{lattice vertex algebra}
attached to the even lattice $\sqrt{2s}\Bbb Z$. 
This algebra is generated by the vertex operators 
$X(m\sqrt{2s},z)$ (which, as we know, satisfy the bosonic version of space locality). 

\begin{example} Consider the case $r=1$. In this case we have two sectors, 
the vacuum sector $\mathcal V(2)\otimes \mathcal V(2)^*$ 
and another one, $\mathcal W\otimes \mathcal W^*$, 
where $\mathcal W=\oplus_{n\in 2\Bbb Z+1}\mathcal F_{\frac{n}{\sqrt{2}}}$. 
The particles corresponding to $\mathcal F_{\frac{n}{\sqrt{2}}}$ 
for odd $n$ are anyons with statistics $\frac{1}{4}$, so they 
satisfy the braided commutativity relation 
of the form $X(z)Y(w)=iY(w)X(z)$. 

It is not difficult to show that the Fourier modes of the vertex operators 
$X(\sqrt{2},z)$ and $X(-\sqrt{2},z)$ generate a projective action 
of the Lie algebra $\mathfrak{sl}_2[z,z^{-1}]$ on $\mathcal V(2)=\Bbb L_0$ and on $\mathcal W=\Bbb L_1$, which are exactly the irreducible integrable representations of the affine Kac-Moody algebra $\widehat{\mathfrak{sl}}_2=\mathfrak{sl}_2[t,t^{-1}]\oplus \Bbb CK$ (the universal central extension of $\mathfrak{sl}_2[t,t^{-1}]$ at level $k=1$ (i.e., $K$ acts by $1$), namely 
$X(\sqrt{2},z)$, $X(-\sqrt{2},z)$, $\sqrt{2} a(z)$ give the currents $e(z)$, $f(z)$ and $h(z)$, where for $b\in {\mathfrak{sl}}_2$
$$
b(z):=\sum_n (b\otimes t^n)z^{-n-1}.
$$
This is the so called {\it Frenkel-Kac vertex operator construction}\label{Frenkel-Kac vertex operator construction} of level 1 irreducible integrable modules (defined for any finite dimensional simply-laced simple Lie algebra) 
in the simplest special case $\g=\mathfrak{sl}_2$. 
Thus the circle-valued theory of a free boson for 
$r=1$ is the so-called {\it Wess-Zumino-Witten model}\index{Wess-Zumino-Witten model} in the simplest example 
of the Lie algebra $\mathfrak{sl}_2$ and level $1$. 
\end{example} 

\begin{example} Let $r=\sqrt{2}$. In this case we have four sectors: 
$\mathcal V_j\otimes \mathcal V_{-j}^*$, $j=0,1,2,3$, 
where $\mathcal V_j=\oplus_{n\in 4\Bbb Z+j}\mathcal F_{\frac{n}{2}}$. 
In particular, $\mathcal V_0=\mathcal V(4)$ and particles in 
$\mathcal V_2=\mathcal F_1$ are fermions arising in the boson-fermion correspondence. 
\end{example} 

\subsection{Free massless fermions} 

In a similar way to free massless bosons, 
one can describe the theory of a free massless fermion $\xi(z)$. 
As explained in Subsection \ref{ferlag}, in two dimensions it makes sense to consider 
chiral spinors taking values in the tautological representation of 
${\rm Spin}(2)=U(1)$ with kinetic term $(\xi,\bold D \xi)$. 
So we have a single quantum field 
$$
\xi(z)=\sum_{n\in \Bbb Z+\frac{1}{2}}\xi_n z^{-n-\frac{1}{2}}
$$ 
and the conjugate quantum field $\xi_*(\overline z)$. 
The modes of $\xi(z)$ satisfy the relation 
$$
[\xi(z),\xi(w)]_+=\delta(z-w), 
$$
where $[,]_+$ is the supercommutator. 
This yields the Clifford algebra relations 
$$
\xi_n\xi_m+\xi_m\xi_n=\delta_{m,-n}
$$
for $m,n\in \Bbb Z$. This algebra has a unique irreducible positive energy representation 
$\Lambda=\wedge (\xi_{-1/2},\xi_{-3/2},...)$ on which $\xi_j$ acts by multiplications for $j<0$ 
and by differentiations for $j>0$. There is an invariant positive Hermitian inner product on $\Lambda$ 
in which the Clifford monomials in $\xi_j$, $j>0$ form an orthonormal basis (invariance means that 
$\xi_j^\dagger=\xi_{-j}$). Thus the Hilbert space of the theory is the completion of $\mathcal D:=\Lambda\otimes \Lambda^*$, where $\Lambda^*$ is the dual of $\Lambda$ corresponding to antiholomorphic fields. 

The hamiltonian $H$ is supposed to satisfy commutation relations
$$
[H,\xi_n]=-\xi_n,\ [H,\xi_n^*]=\xi_n^*,
$$ 
So we have 
$$
H=H_L+H_R,
$$
where
$$
H_L=\sum_{n>0}n\xi_{-n}\xi_n
$$
and similarly for $H_R$. 
The Virasoro algebra 
is defined by 
$$
L_m=\frac{1}{2}\sum_{n\in \Bbb Z+\frac{1}{2}}n:\xi_n\xi_{-n+m}:,
$$
i.e., $H_L=L_0$. 

\begin{exercise} Show that these operators $L_n$ satisfy the Virasoro commutation relations 
with central charge $c=\frac{1}{2}$. 
\end{exercise}

\printindex

\end{document}